\documentclass[11pt]{article}

\usepackage{graphicx}
\usepackage{natbib}
\usepackage{amsmath,amssymb}
\usepackage{color}
\usepackage{textcomp}
\usepackage{subfigure}
\usepackage{algorithmic,algorithm}
\usepackage{placeins}
\usepackage{slashbox}
\usepackage{hyperref}

\def\C{{\mathcal{C}}}

\def\H{{\mathcal{H}}}

\def\R{{\mathcal{R}}}
\def\E{{\mathcal{E}}}

\def\K{{\mathcal{K}}}

\def\RR{{\mathbb{R}}}
\def\CC{{\mathbb{C}}}

\def\ZZ{{\mathbb{Z}}}

\def\k{{\mathbf k}}
\def\x{{\bm x}}

\def\u{{\bm u}}

\def\x{{\mathbf x}}

\def\u{{\mathbf u}}

\def\0{{\mathbf 0}}

\def\bomega{\boldsymbol{\omega}}
\def\bnabla{\boldsymbol{\nabla}}

\def\Dpartial#1#2{ {\partial #1 \over \partial #2} }

\def\Dpartialn#1#2#3{ {\partial^{#3} #1 \over \partial #2^{#3}} }

\def\Bmp#1{ \begin{minipage}{#1} }
\def\Emp{ \end{minipage} }
\def\Bmpc#1{ \begin{minipage}[c]{#1} }
\def\Bmpt#1{ \begin{minipage}[t]{#1} }
\def\Bmpb#1{ \begin{minipage}[b]{#1} }

\def\tTE{\widetilde{T}_{\E_0}}

\newcommand{\uvec}{\mathbf{u}}

\newcommand{\laplacian}{\Delta}

\newcommand{\tomega}{\widetilde{\omega}}
\newcommand{\tuE}{\widetilde{\mathbf{u}}_{\E_0}}

\newcommand{\tuET}{\widetilde{\uvec}_{0;\E_0,T}}
\newcommand{\tuEtT}{\widetilde{\uvec}_{0;\E_0,\tTE}}

\newcommand{\tuu}{\widetilde{u}_0}
\newcommand{\tuuE}{\widetilde{u}_{0;\E_0}}
\newcommand{\tuuET}{\widetilde{u}_{0;\E_0,T}}

\newcommand{\argmax}{\operatorname{argmax}}

\newcommand{\supp}{\operatorname{supp}}

\newtheorem{theorem}{Theorem}
\newtheorem{proof}{Proof}

\begin{document}
\title{Alignments of Triad Phases \\ in 1D Burgers and 3D Navier-Stokes Flows}

\author{Di Kang$^1$, Bartosz Protas$^1$\thanks{Email address for correspondence: bprotas@mcmaster.ca} \ and Miguel D.~Bustamante$^2$ 
\\ \\
$^1$ Department of Mathematics and Statistics, McMaster University \\
 Hamilton, ON, Canada \\ \\ 
$^2$ School of Mathematics and Statistics, University College Dublin \\
Belfield, Dublin 4, Ireland 
}

\date{\today}

\maketitle

\maketitle

\begin{abstract}
  The goal of this study is to analyze the fine structure of nonlinear
  modal interactions in different 1D Burgers and 3D Navier-Stokes
  flows. This analysis is focused on preferential alignments
  characterizing the phases of Fourier modes participating in triadic
  interactions, which are key to determining the nature of energy
  fluxes between different scales.  We develop novel diagnostic tools
  designed to probe the level of coherence among triadic interactions
  realizing different flow scenarios. We consider extreme (in the
  sense of maximizing the growth of enstrophy in finite time) 1D
  viscous Burgers flows and 3D Navier-Stokes flows which are
  complemented by singularity-forming inviscid Burgers flows as well
  as viscous Burgers flows and Navier-Stokes flows corresponding to
  generic turbulent and simple unimodal initial data, such as the
  Taylor-Green vortex. The main finding is that while the extreme
  viscous Burgers and Navier-Stokes flows reveal the same relative
  level of enstrophy amplification by nonlinear effects, this
  behaviour is realized via modal interactions with vastly different
  levels of coherence. In the viscous Burgers flows the flux-carrying
  triads have phase values which saturate the nonlinearity thereby
  maximizing the energy flux towards small scales. On the other hand,
  in 3D Navier-Stokes flows with the extreme initial data the energy
  flux to small scales is realized by a very small subset of helical
  triads, with the flux-carrying triads showing again a high level of
  coherence, but with their phases taking a broader range of preferred
  values {which form} complex time-dependent patterns. The second main
  finding concerns the role of initial coherence.  Comparison of the
  flows {resulting} from the extreme and generic initial
  conditions shows striking similarities between these two types of
  flows, for the 1D viscous Burgers equation as well as the 3D
  Navier-Stokes equation.  The third main finding concerns 3D
  Navier-Stokes flows {resulting} from the Taylor-Green initial
  condition. On the one hand, as expected, the fluxes in this case are
  several times smaller than in the extreme or generic cases. On the
  other hand, the flux-carrying triads in this case display much more
  rigid patterns of phase coherence, with persistent narrow bands of
  preferred phase values.  Finally, we find that {the paradigm
    based on the instability assumption due to \cite{Waleffe1992}}
  holds true at all scales for the 3D Navier-Stokes flows, with some
  variations. The flows {resulting} from the extreme or generic
  initial conditions show a quantitative agreement {with
    simulations of statistically stationary turbulence by
    \citet{alexakis2018cascades} concerning the inverse-cascade
    behaviour at the inertial range of homochiral triads (i.e., triads
    with the same helical content, or Class I in Waleffe's notation).}
  In contrast, in the flow {resulting} from the Taylor-Green
  initial condition, the inverse-cascade behaviour is dominated by
  heterochiral triads of Class II in Waleffe's notation.
\end{abstract}

%\vspace*{-1.25cm}
%\tableofcontents

%%%%%%%%%%%%%%%%%%%%%%%%%%%%%%%%%%%%%%%%%%%%%%%%%%%%%%%%%%%%%%%%%%%%%%%%%%%%%%%
%%%%%%%%%%%%%%%%%%%%%%%%%%%%%%%%%%%%%%%%%%%%%%%%%%%%%%%%%%%%%%%%%%%%%%%%%%%%%%%
%%%%%%%%%%%%%%%%%%%%%%%%%%%%%%%%%%%%%%%%%%%%%%%%%%%%%%%%%%%%%%%%%%%%%%%%%%%%%%%
\section{Introduction}
\label{sec:intro}

Understanding and quantifying the most extreme forms of behaviour
allowed for by the Navier-Stokes system describing the motion of
viscous incompressible fluids remains a key outstanding problem in
theoretical fluid dynamics. By ``extreme behaviour'' we mean
situations where certain flow quantities attain very large or very
small values approaching or saturating mathematically rigorous bounds
on these quantities. Providing answers to such questions is important
for our understanding of fundamental performance limitations inherent
in elementary physical processes underlying fluid motion such as
transport, mixing, etc. A special case of this class of problems,
which is particularly important from the theoretical point of view,
concerns the question whether or not solutions to the
three-dimensional (3D) Navier-Stokes system and some other related
models may form singularities in finite time \citep{d09}. Singularity
formation implies that the solution develops features preventing it
from satisfying the equation in the classical sense, i.e., pointwise
in space and in time. Such features are usually hypothesized to have
the form of non-differentiable accumulations of vorticity. Needless to
say, should such singularities indeed form, this would invalidate the
equation as a physically consistent model for fluid flow. In
recognition of its significance, the Navier-Stokes regularity problem
has been named by the Clay Mathematics Institute one of its seven
``Millennium Problems'' \citep{f00}. However, progress on this problem
has been rather slow \citep{Robinson2020}.

As highlighted by \citet{Robinson2020} in one of his concluding
remarks, progress concerning possibly singular behaviour in the
Navier-Stokes and related flow models will require a better
understanding of the fine structure of interactions between modes. The
goal of the present study is thus to provide new insights about this
problem by analyzing both generic and extreme flow evolutions using
diagnostic tools specifically designed to shed light on the nature of
triadic interactions. As described in more detail below, the extreme
flows we consider extremize the enstrophy as a physically relevant
quantity which also plays the role of an indicator of the regularity
of solutions.

All models we consider are defined on spatially periodic domains
$\Omega := [0,1]^d$, where $d = 1, 3$ is the dimension and ``$:=$''
means ``equal to by definition'', such that they are subject to
periodic boundary conditions.  Flows of viscous incompressible fluids
in 3D are governed by the Navier-Stokes system
\begin{subequations}
\label{eq:3DNS_u}
\begin{alignat}{2}
\left(\Dpartial{}{t} - \nu \bnabla^2\right) \mathbf{{u}}(\mathbf{x},t) &= - \bnabla p(\mathbf{x},t) - \mathbf{{u}}(\mathbf{x},t) \cdot\bnabla\mathbf{{u}}(\mathbf{x},t) \qquad\qquad & & \x \in \Omega, \ t > 0 \label{eq:3DNS_uA} \\   
\bnabla \cdot \mathbf{{u}}(\mathbf{x},t) &= 0 & & \x \in \Omega, \ t > 0 \label{eq:3DNS_uB} \\
\u(\x,0) &= \u_0(\x) & & \x \in \Omega, \ t = 0,
\end{alignat}
\end{subequations}
where $\u(\x,t)$ and $p(\x,t)$ are the velocity vector field and the
scalar pressure field {whereas $\nu$ is the kinematic viscosity.}
Equations \eqref{eq:3DNS_uA} and \eqref{eq:3DNS_uB} represent,
respectively, the conservation of momentum and mass, whereas $\u_0$ is
the initial condition. Key quantities characterizing solutions of the
Navier-Stokes system \eqref{eq:3DNS_u} include the kinetic energy
$\K(\u)$, enstrophy $\E(\u)$ and helicity $\H(\u)$ defined as
\begin{subequations}
\label{eq:KEH}
\begin{align}
\K(\u) &:= \frac{1}{2} \int_{\Omega} |\u|^2  \, \mathrm{d}\x,  \label{eq:K} \\ 
\E(\u) &:= \frac{1}{2} \int_{\Omega} |\bomega|^2 \, \mathrm{d}\x {= \frac{1}{2}\int_{\Omega} |\bnabla \u|^2 \mathrm{d}\x,} \label{eq:E} \\  
\H(\u) &:= \int_{\Omega} \u \cdot \bomega \, \mathrm{d}\x,  \label{eq:H}
\end{align}
\end{subequations}
where $\boldsymbol{\omega} := \bnabla \times \mathbf{{u}}$ is the
vorticity. With a slight abuse of notations, we will sometimes write
$\E(t) = \E(\u(t))$. We will also use the Sobolev space $H^1(\Omega)$
of functions with square-integrable gradients with the norm defined as
$\| \u \|_{H^1}^2 := \int_{\Omega} |\u|^2 + |\bnabla \u|^2 \, \mathrm{d}\x$
\citep{af05}.

As a commonly-used simplified model for the Navier-Stokes system
\eqref{eq:3DNS_u}, we will also consider the one-dimensional (1D)
viscous Burgers system
\begin{subequations}
\label{eq:Burgers}
\begin{alignat}{2}
\Dpartial{u(x,t)}{t}+u(x,t)\Dpartial{u(x,t)}{x} - \nu \Dpartialn{u(x,t)}{x}{2}&= 0 
\qquad\qquad & & x \in \Omega, \ t > 0,  \label{eq:Burgersa} \\
u(x,0)&=u_0(x) & & x \in \Omega, \ t = 0 \label{eq:Burgersb,} 
\end{alignat}
\end{subequations}
and its inviscid version
\begin{subequations}
\label{eq:Burgers0}
\begin{alignat}{2}
\Dpartial{u(x,t)}{t}+u(x,t)\Dpartial{u(x,t)}{x}&= 0 
\qquad\qquad & & x \in \Omega, \ t > 0,  \label{eq:Burgers0a} \\
u(x,0)&=u_0(x) & & x \in \Omega, \ t = 0, \label{eq:Burgers0b} 
\end{alignat}
\end{subequations}
where $u_0$ is the initial condition. In the Burgers flows the kinetic
energy and the enstrophy are defined analogously to \eqref{eq:K} and
\eqref{eq:E}.

The question of the existence of solutions to the 1D Burgers problems
is well understood \citep{kl04}.  While the viscous problem
\eqref{eq:Burgers} is globally well-posed, the inviscid Burgers
problem \eqref{eq:Burgers0} is known to produce a finite-time blow-up
for all nonconstant initial data $u_0 \in H^1(\Omega)$ (more
specifically, the solutions blow up by developing an infinitely steep
front). On the other hand, as already indicated above, the question of
existence of classical (smooth) solutions to the 3D Navier-Stokes
system \eqref{eq:3DNS_u} remains open \citep{d09,Robinson2020}. One of
the most useful conditional regularity results asserts that the
solution $\u(t)$ of \eqref{eq:3DNS_u} remains smooth on the time
interval $[0,T]$ provided $\E(t) < \infty$, $\forall t \in [0,T]$
\citep{ft89}, making the enstrophy a convenient indicator of the
regularity of solutions. The evolution of the kinetic energy and the
enstrophy is governed by the system \citep{d09}
\begin{subequations}
\label{eq:dKdtdEdt}
\begin{align}
\frac{d\K(\u(t))}{dt} & =  -2\nu \E(\u(t)), \label{eq:dKdt_system}\\
\frac{d\E(\u(t))}{dt} & =  \R(\u(t)) 
\leq -\nu \frac{\E(\u(t))^2}{\K(\u(t))} + \frac{c}{\nu^3} \E(\u(t))^3, \label{eq:dEdt_system}
\end{align}
\end{subequations}
where $\R(\u(t)) := -\nu\int_\Omega |\laplacian\u|^2\,\mathrm{d}\x +
\int_{\Omega} \u\cdot\nabla\u\cdot\laplacian\u\, d\x$ and $c$ is a
known constant. Using \eqref{eq:dEdt_system}, one can then derive the estimate 
\begin{equation}
\E(\u(t)) \leq \frac{\E_0}{\sqrt{1 - \frac{27}{4\,\pi^4\,\nu^3}\,\E_0^2\, t}},
\label{eq:Et}
\end{equation}
where $\E_0 := \E(\u_0)$, which is one of the best results of this
type available to-date. However, since the upper bound on the
right-hand side (RHS) of this estimate becomes unbounded when $t
\rightarrow 4\,\pi^4\,\nu^3 / (27\,\E_0^2)$, finite-time blow-up
cannot be ruled out based on this estimate.

In order to probe the question whether the enstrophy might become
unbounded in finite time in 3D Navier-Stokes flows, the following
family of PDE-constrained optimization problems was solved by
\citet{KangYumProtas2020}
\begin{equation}
\tuET =  \argmax_{\u_0 \in H^1(\Omega), \ \bnabla\cdot\u_0 = 0} \E(\u(T)) \quad 
\textrm{subject to} \
\left\{\begin{aligned}
& \mbox{{\textrm System}} \ \eqref{eq:3DNS_u} \\
& \E(\u_0) = \E_0 
\end{aligned}\right. ,
\label{eq:maxE}
\end{equation}
where, for specified values of $T > 0$ and $\E_0 > 0$, the {\em
  extreme} initial conditions $\tuET$ with fixed enstrophy $\E_0$ were
found, such that the corresponding enstrophy $\E(\u(T))$ at the final
time $T$ is maximal. Optimization problems of this type are
efficiently solved numerically using adjoint-based gradient methods
\citep{pbh04}.  While no evidence was found for unbounded growth of
enstrophy in such extreme Navier-Stokes flows, solving problem
\eqref{eq:maxE} for a broad range of $\E_0$ and $T$ revealed the
following empirical relation characterizing how the largest attained
enstrophy scales with the initial enstrophy $\E_0$ in the most extreme
scenarios \citep{KangYumProtas2020}
\begin{equation}
\max_{T>0} \E(T)  \ \sim \ \E_0^{3/2}.
\label{eq:maxT_vs_E0}
\end{equation}

While blow-up is ruled out in viscous Burgers flows, the question
whether the corresponding a priori estimates on the growth of
enstrophy are sharp or have room for improvement is quite pertinent as
they are obtained in a similar way to \eqref{eq:Et}. In order to
address this question, a family of optimization problems of the type
\eqref{eq:maxE} was solved for the 1D viscous Burgers system
\eqref{eq:Burgers} by \citet{ap11a}. Interestingly, the extreme
behaviour found in this way also obeys relation \eqref{eq:maxT_vs_E0}.
Our present study is motivated by the observation that extreme
behaviour in 1D Burgers systems and 3D Navier-Stokes systems must be
characterized by nonlinear energy transfers towards modes
corresponding to small spatial scales. We aim to understand and
compare the fine structure of the nonlinear interactions between modes
which give rise to the extreme behaviour described by relation
\eqref{eq:maxT_vs_E0} in 1D Burgers and 3D Navier-Stokes flows. We
will also compare these flows to the flow evolutions corresponding to
generic and some simple initial data, as well as to inviscid Burgers
flows which do form singularities in finite time.

Since the seminal work by \cite{kraichnan_1959}, it has been
established that nonlinear triad interactions play a pivotal role in
the mechanisms of energy transfer across spatial scales in 2D or 3D
Navier-Stokes flows, leading to cascades of energy, enstrophy and
helicity with direct or inverse directions, depending on the
dimension. Specifically for 3D Navier-Stokes flows, the work by
\cite{Waleffe1992} based on analysis of helical triads established a
paradigm termed ``the instability assumption'' whereby unstable modes
in helical triads determine the statistically expected directions of
spectral energy fluxes across scales. A recent result by
\cite{moffatt_2014} shows that a Galerkin-truncated system consisting
of a single triad does not produce a reasonable time evolution when
compared with the evolution of the real system. It thus indicates that
triad interactions must always be considered in the context of a
network of multiple triad interactions, where energy exchanges in all
directions are in principle allowed, as otherwise a non-physical
recurrent behaviour would be obtained.  Waleffe's paradigm has been
confirmed in several works \citep{PhysRevFluids.2.054607,
  alexakis2017helically, alexakis2018cascades} and has motivated new
research concerning the effect of triad interactions on the directions
of the cascades of energy and other invariants. These questions were
recently investigated by
\citet{PhysRevLett.108.164501,biferale2013split,sahoo2015disentangling}
who used alternative models to numerically ``probe'' the effect of
triad interactions by arbitrarily eliminating some modes or some triad
interactions from the equations, depending on various criteria such as
the relative helicities of the modes involved.
%Of course, shell models deserve a special mention [cite a lot].

A vast majority of the works cited above have focused on theoretical
studies and/or on {analyzing the energy and helicity spectra of
  numerically computed flows,} along with their corresponding spectral
fluxes.  As the spectral fluxes are the main quantities of interest
when looking at energy cascades across scales, it is interesting to
note the following research gap: an analysis of these spectral fluxes
in terms of the dynamically evolving spectral phases (i.e., the
arguments of the complex spectral coefficients) has been absent from
practically all previous works, except perhaps in shell models where a
relevant quantity that somehow {bridges phases and fluxes} is the
three-point correlator which provides a good diagnostic of flux
cascades \citep{de2015inverse, rathmann2016role}. This lack of studies
of phase dynamics (or statistics thereof) in the literature
{persists} despite the fact that researchers are well aware of
the fact that any physical field whose spectral phases are
{random and uniformly} distributed over time displays Gaussian
statistics along with strictly zero energy transfer
\citep{alexakis2018cascades}, which is equivalent to the statement
that the presence of energy cascades and intermittency implies
nontrivial phase correlations.  This research gap not only applies to
the study of the 2D or 3D Navier-Stokes equations, but also to many
other systems, including the 1D Burgers equations.  In our present
study we intend to close this gap by introducing and studying relevant
quantities related to the spectral phases. {One of us was a
  collaborator on a number of preliminary studies of the 1D Burgers
  equation and models thereof \citep{buzzicotti2016phase,
    MurrayBustamante2018},} where the spectral fluxes of energy were
obtained explicitly in terms of the magnitudes of the Fourier
coefficients along with the less known Fourier triad phases which are
linear combinations of the three phases of the complex Fourier modes
involved in a triad interaction.  In these works it was evident that
the main mechanism responsible for the strong energy fluxes across
scales, observed in the 1D Burgers system, is the so-called
\emph{alignment} of triad phases, namely, the tendency of triad phases
to cluster around values that maximize the spectral fluxes towards
small scales. This alignment is dynamical and is controlled in part by
the presence of stable fixed points in the evolution equations {for
  the phases}.  Moreover, in the 1D Burgers system this alignment
seems to entail sustained \emph{synchronizations} of the triad phases,
namely, a collective behaviour whereby a large proportion of triad
phases align over a wide range of spatial scales and over long time
intervals. Going beyond that, \citet{MurrayPhDthesisNew} applied this
idea to the study of 3D Navier-Stokes flows with stochastic forcing at
large spatial scales.  These analyses showed that, statistically,
there is a slight preference for helical triad phases to align so as
to maximize the fluxes towards small scales.  Specifically, when
restricting attention to $5\%$ of the most energetic modes at each
reciprocal length scale, the probability density functions of helical
triad phases show wide alignment peaks {with values of only 7\% above
  the uniform value.}  Remarkably, such a slight preference for
alignment {of triad phases} is enough to produce significant energy
flux cascades.  A possible explanation for this is that the coherent
structures associated with aligned triad phases have a low
dimensionality, perhaps even close to one.

Our present study aims to understand, from the point of view of triad
phases, {the fine structure of mode interactions arising} in extreme
transient solutions where strong energy transfers are observed. Our
most important result is that while the extreme 1D Burgers and 3D
Navier-Stokes flows obey the same scaling relation
\eqref{eq:maxT_vs_E0} for the maximum growth of enstrophy, this
behaviour is achieved through vastly different levels of
synchronization in the two cases.  In 1D Burgers flows all triads
involving energy-containing Fourier modes align so as to maximize the
energy flux towards small scales.  On the other hand, in 3D
Navier-Stokes flows energy transfer to small scales is realized by
only a small subset of the triads corresponding to preferred phase
angles forming complex spatio-temporal patterns. The second main
finding is that removing the spatial coherence from the extreme
initial data in both 1D Burgers and 3D Navier-Stokes flows (by
randomizing the Fourier phases while retaining the magnitudes of the
Fourier coefficients) does not profoundly change the nature of triadic
interactions and the resulting fluxes in these flows.

The structure of the paper is as follows: the specific flow problems
we will analyze, specified in terms of their initial data, are defined
in the next section; then, in \S\,\ref{sec:diagnostics}, we
introduce the diagnostics we will employ to  characterize the
structure of the modal interactions; our results are presented in
\S\,\ref{sec:results}, whereas discussion and final conclusions
are deferred to \S\,\ref{sec:final}; some technical material is
collected in three appendices.

%%%%%%%%%%%%%%%%%%%%%%%%%%%%%%%%%%%%%%%%%%%%%%%%%%%%%%%%%%%%%%%%%%%%%%%%%%%%%%%
%%%%%%%%%%%%%%%%%%%%%%%%%%%%%%%%%%%%%%%%%%%%%%%%%%%%%%%%%%%%%%%%%%%%%%%%%%%%%%%
%%%%%%%%%%%%%%%%%%%%%%%%%%%%%%%%%%%%%%%%%%%%%%%%%%%%%%%%%%%%%%%%%%%%%%%%%%%%%%%
\section{Flow problems}
\label{sec:problems}

In this section we introduce the flow problems we will investigate.
Each of these problems is defined by a combination a governing system
(1D viscous or inviscid Burgers
\eqref{eq:Burgers}--\eqref{eq:Burgers0}, or the 3D Navier-Stokes
system \eqref{eq:3DNS_u}) and one of the initial conditions given
below. Information about which initial conditions are used with the
different models is summarized in table \ref{tab:cases}. Some
combinations of governing systems and initial data are not considered
as they do not provide additional interesting insights. {In
  agreement with the studies in which the different initial conditions
  were originally obtained, the viscosity coefficient will be equal to
  $\nu = 0.001$ in the 1D viscous Burgers system \eqref{eq:Burgers}
  and to $\nu = 0.01$ in the 3D Navier-Stokes system
  \eqref{eq:3DNS_u}.}

\subsection{Unimodal initial conditions}
\label{sec:unimodal}

The simplest initial conditions have the form of the product of single
Fourier harmonics with low wavenumbers in each spatial variable. For
the 1D Burgers equation, both inviscid and viscous, it will therefore
take the form
\begin{equation}
u_0(x) = A \sin(2 \pi x), \quad x \in \Omega,
\label{eq:Bu0sin}
\end{equation}
where $A > 0$ is a constant adjusted such that the initial condition
has prescribed enstrophy $\E_0$. We note that for the inviscid Burgers
equation \eqref{eq:Burgers0} the constant $A$ plays no role, since the
solution corresponding to $A = 1$ is obtained by rescaling the time
$t$ as $A^{-1} t$.

For the 3D Navier-Stokes system \eqref{eq:3DNS_u} the corresponding
unimodal initial condition will have the form of the Taylor-Green
vortex \citep{tg37} with the velocity components given by
\begin{subequations}
\label{eq:TG}
\begin{align} 
u_1(x_1,x_2,x_3) & =  A\sin(2\pi x_1)\cos(2\pi x_2)\cos(2\pi x_3), \\
u_2(x_1,x_2,x_3) & =  -A\cos(2\pi x_1)\sin(2\pi x_2)\cos(2\pi x_3), \qquad [x_1,x_2,x_3]^T \in \Omega, \\
u_3(x_1,x_2,x_3) & =  0, 
\end{align}
\end{subequations}
where $A > 0$ is again an adjustable constant. We remark that this
initial condition has a long history in studies of extreme behaviour in
incompressible flows \citep{bmonmu83,b91,bb12,ap16}.

\subsection{Extreme  initial conditions}
\label{sec:ext}

For the Navier-Stokes systems \eqref{eq:3DNS_u}, the ``extreme''
initial condition is defined as the initial data which for a given
initial enstrophy $\E_0$ produces the largest possible growth of
enstrophy in finite time. It is thus given by $\tuE = \max_{T>0}
\tuET$, where $\tuET$ is a solution of the optimization problem
\eqref{eq:maxE} with fixed $T$ and $\E_0$.  For each value of $\E_0$,
$\tTE := \argmax_{T>0} \tuuET$ is then the time when the largest
growth of enstrophy is achieved. An example of such an extreme initial
condition $\tuE$ obtained at $\E_0 = 500$ is shown in figure
\ref{fig:NSext} \citep{KangYumProtas2020}. We see that this field has
the form of three perpendicular pairs of antiparallel vortex tubes
and, as discussed in detail by \cite{KangYumProtas2020}, the resulting
flow evolution (which maximizes the enstrophy at the final time
$\tTE$) is marked by a series of reconnection events.  Our analysis of
the 3D Navier-Stokes flows in \S\,\ref{sec:results3D} will focus on
the case with $\E_0 = 250$.  The extreme initial conditions for the 1D
viscous Burgers system \eqref{eq:Burgers} are obtained analogously as
$\tuuE = \max_{T>0} \tuuET$, where $\tuuET$ are the solutions of a
suitably adapted optimization problem \eqref{eq:maxE} with fixed
$\E_0$ and $T$, cf.~figure \ref{fig:Bu0ext0} \citep{ap11a}.  Our
analysis of the 1D Burgers flows in \S\,\ref{sec:results1D} will focus
on the case with $\E_0 = 100$.

\begin{figure}
\begin{center}
\mbox{\subfigure[$\tomega_1$]{\includegraphics[width=0.3\textwidth]{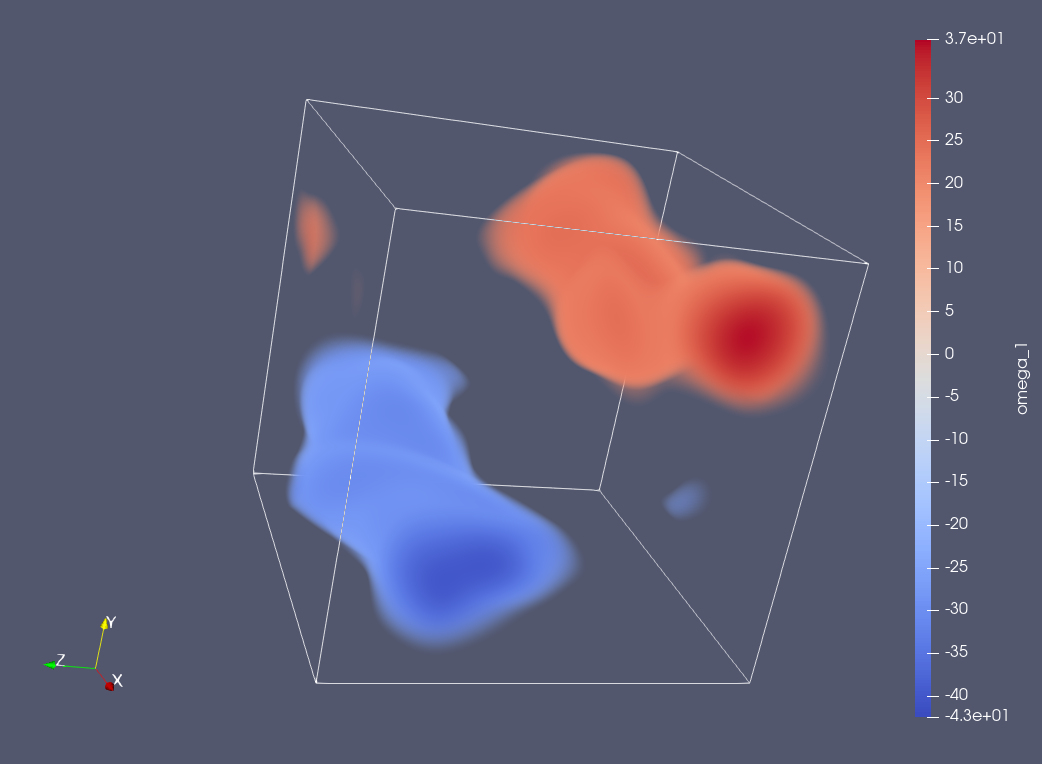}}\qquad
\subfigure[$\tomega_2$]{\includegraphics[width=0.3\textwidth]{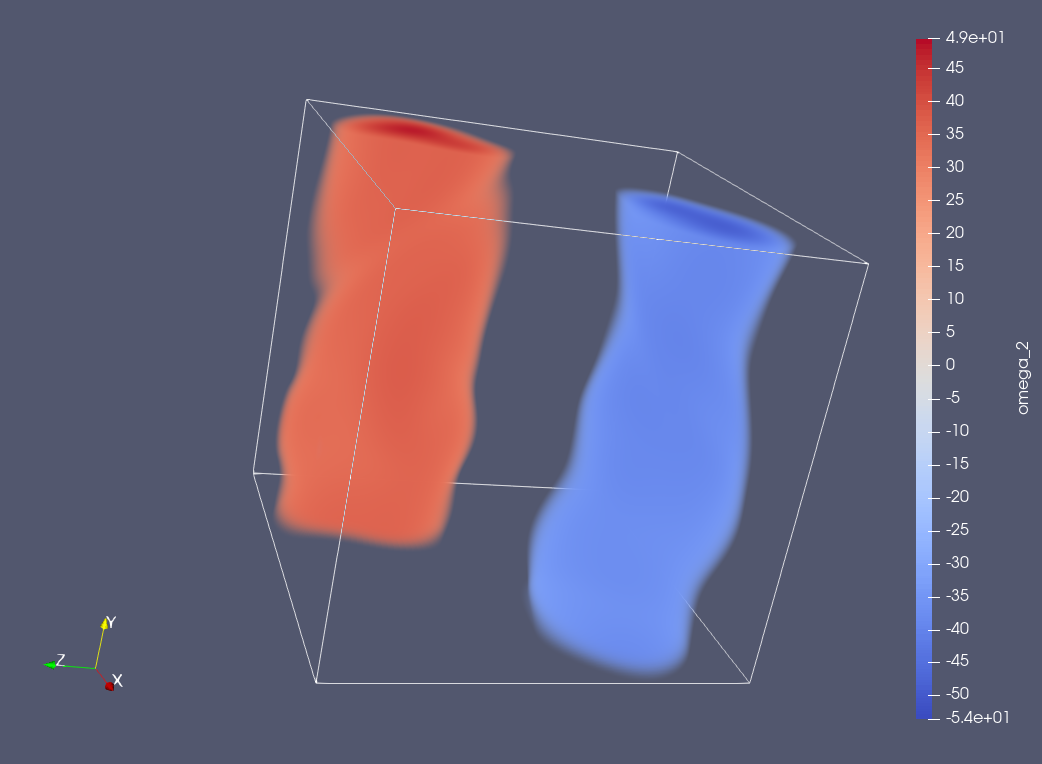}}\qquad
\subfigure[$\tomega_3$]{\includegraphics[width=0.3\textwidth]{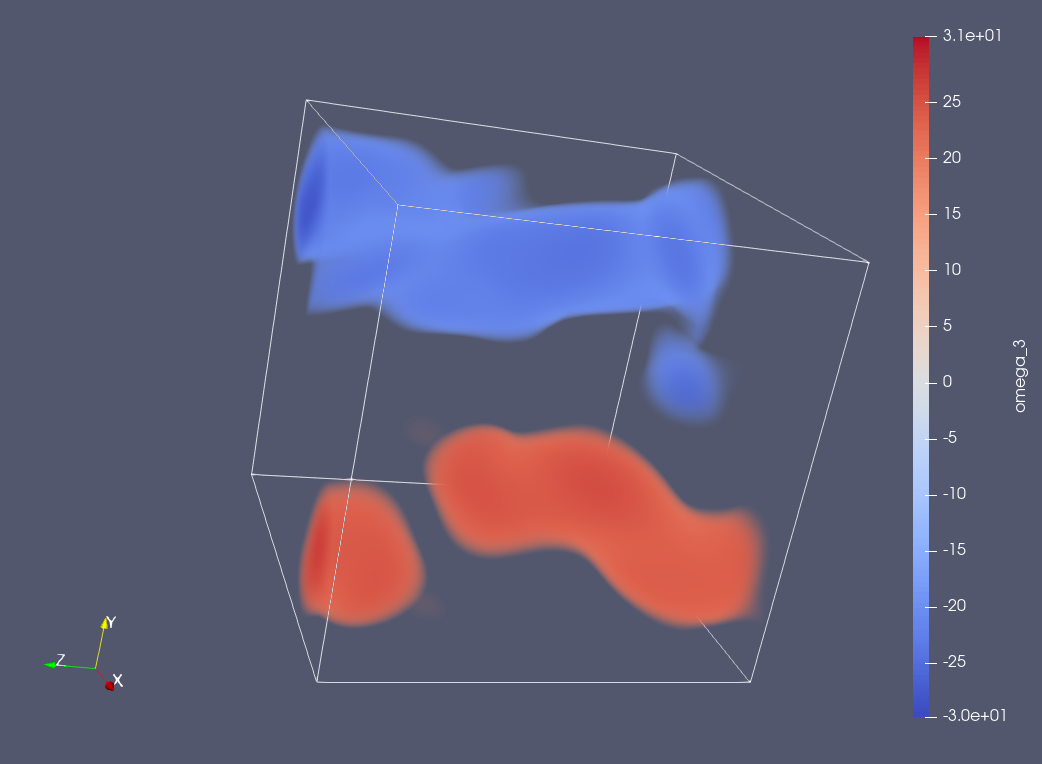}}}
\caption{Vorticity components of the optimal initial condition
  $\tuEtT$ obtained by solving the finite-time optimization problem
  \eqref{eq:maxE} for the initial enstrophy $\E_0 = 500$ and the
  corresponding optimal length $\tTE = 0.17$ of the time interval
  \citep{KangYumProtas2020}.}
\label{fig:NSext}
\bigskip
\includegraphics[width=0.6\textwidth]{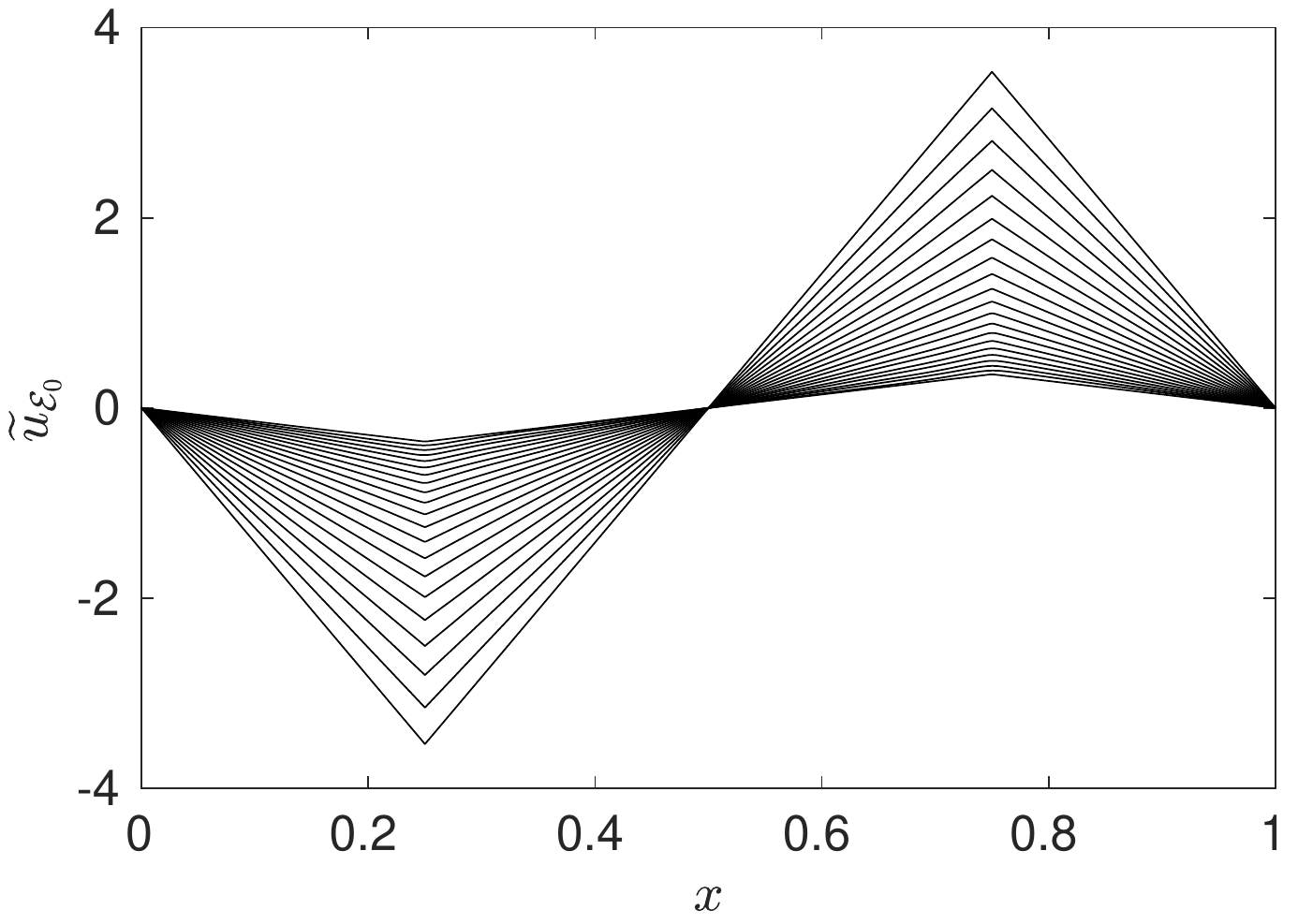}
\caption{Extreme initial conditions $\tuuE$ for the 1D Burgers
  equation \eqref{eq:Burgers} obtained by solving a family of
  finite-time optimization problems analogous to \eqref{eq:maxE} for
  different values of the initial enstrophy $\E_0 \in [1,100]$ and the
  corresponding optimal time windows $\tTE$ (larger values of $\E_0$
  correspond to increased magnitude of $\tuuE$, \citet{ap11a}).}
\label{fig:Bu0ext0}
\end{center}
\end{figure}

Since for all nonconstant initial conditions $u_0$ the inviscid Burgers
system \eqref{eq:Burgers0} develops a singularity in finite time, in
the sense that $\E(t) \rightarrow \infty$ as $t \nearrow t^*$, where
$t^*$ is the blow-up time \citep{kl04}, an optimization problem of the
type \eqref{eq:maxE} is not well-defined for this system. As a
counterpart to the optimization problem \eqref{eq:maxE} for the
inviscid Burgers system \eqref{eq:Burgers0}, one might therefore think
of the problem of finding initial condition $u_0 \in H^1(\Omega)$ with
prescribed enstrophy $\E_0 > 0$ such that the corresponding solution
blows up in the shortest time $\tilde{t}^*$.  However, as shown in
Appendix \ref{sec:Bu0ext0}, such a problem is not in fact well posed
since $\inf_{u_0} t^*(u_0) = 0$ with $\E(u_0) = \E_0 < \infty$. However,
formally extrapolating the extreme initial data $\tuuE$ shown in
figure \ref{fig:Bu0ext0} to the limit $\E_0 \rightarrow \infty$ leads
to the following extreme initial data for the inviscid Burgers system
\begin{equation}
\tuu(x) = A \left( 4\left| x - \frac{1}{2} \right| - 1 \right), \quad x \in \Omega.
\label{eq:Bu0ext0}
\end{equation}
We note that $\tuu$ in \eqref{eq:Bu0ext0} is piecewise-linear, hence
it is a $H^1$ function.

\subsection{Generic initial conditions}
\label{sec:generic}

For the viscous Burgers system \eqref{eq:Burgers}, the ``generic''
initial condition is obtained from the extreme initial condition
$\tuuE$ by decomposing this field in a Fourier series
\eqref{eq:Fourier1D} and randomizing the phases of the Fourier
coefficients. This is done by replacing the original phase in each of
the Fourier coefficients with a random number uniformly distributed
over $[0, 2\pi]$ while retaining the magnitude of the coefficient. For
the Navier-Stokes system \eqref{eq:3DNS_u} the generic initial
condition is obtained from the extreme initial data $\tuE$ in a
similar manner, except that randomization is performed on the phases
of the coefficients $u_{\k}^{\pm}$ of the helical-wave decomposition
\eqref{eq:heli_decomp}, rather than using the Fourier decomposition
\eqref{eq:Fourier3D}, which ensures that the resulting field remains
divergence-free (see \S\,\ref{sec:helical} for additional
details).

{We note that, as a result of the parity symmetry
  \eqref{eq:oddity} characterizing the three initial conditions for
  the 3D Navier-Stokes system, the helicity remains identically zero,
  $\H(t)=0, \; \forall t\ge 0$, in the resulting flows. As shown in
  Appendix \ref{sec:conjugate_triads}, this observation will
  significantly simplify our analysis of the results in
  \S\,\ref{sec:results3D}.}

\begin{table}
  \begin{center}
    % \hspace*{-1.1cm}
    \begin{tabular}{|l|c|c|c|} \hline
      \backslashbox{initial data}{model}   & \Bmp{2.5cm} \centering 1D inviscid \\ Burgers \eqref{eq:Burgers0} \Emp  &  \Bmp{2.5cm} \centering 1D viscous Burgers \eqref{eq:Burgers} \Emp  & 3D Navier-Stokes \eqref{eq:3DNS_u} \\ \hline
      unimodal \eqref{eq:Bu0sin} &  $\checkmark$  & & {$\checkmark$} \\ 
      extreme \eqref{eq:maxE}   &  $\checkmark$   & $\checkmark$  & $\checkmark$  \\ 
      generic                   &                 & $\checkmark$  & $\checkmark$  \\ \hline
    \end{tabular}
  \end{center}
  \caption{Summary information specifying which initial conditions are used with different
    governing systems.}
  \label{tab:cases}
\end{table}

%%%%%%%%%%%%%%%%%%%%%%%%%%%%%%%%%%%%%%%%%%%%%%%%%%%%%%%%%%%%%%%%%%%%%%%%%%%%%%%
%%%%%%%%%%%%%%%%%%%%%%%%%%%%%%%%%%%%%%%%%%%%%%%%%%%%%%%%%%%%%%%%%%%%%%%%%%%%%%%
%%%%%%%%%%%%%%%%%%%%%%%%%%%%%%%%%%%%%%%%%%%%%%%%%%%%%%%%%%%%%%%%%%%%%%%%%%%%%%%
\section{Diagnostics}
\label{sec:diagnostics}

In this section we introduce a number of diagnostic quantities which
will allow us to characterize the dynamics of energy transfers in
terms of the evolution of triad phases in our model problems. These
derivations are based on \citet{MurrayBustamante2018} for 1D Burgers
flows (\S\,\ref{sec:diagn1D}) and on \citet{MurrayPhDthesisNew} for
3D Navier-Stokes flows (\S\,\ref{sec:diagn3D}).

\subsection{Fourier Triad Phases and Fluxes in 1D Burgers Flows}
\label{sec:diagn1D}

The solution $u(x,t)$ of the 1D Burgers equation, either inviscid
\eqref{eq:Burgers0} or viscous \eqref{eq:Burgers}, on a periodic
domain can be conveniently represented in terms of its Fourier series
\begin{equation}
u(x,t) =\sum_{k \in \mathbb{Z}} e^{2\pi i k x} \widehat {u}_{k}(t),
\label{eq:Fourier1D}
\end{equation}
where $k$ are the wavenumbers restricted to the set of integers
$\mathbb{Z}$ due to the assumed periodicity.  The dynamic content of
the $k$th mode is expressed by its complex-valued Fourier coefficient
$\widehat {u}_{k}(t)$ which, given that the solution $u(x,t)$ is
real-valued, is subject to the condition of conjugate symmetry $\widehat
{u}_{-k}(t) = \widehat {u}_{k}^*(t)$, $\forall k \in \mathbb{Z}$, where
${}^*$ denotes complex conjugation. Using this representation, the
governing partial differential equation (PDE) can be decomposed into a
set of ordinary differential equations (ODEs) which describe the time
evolution of each individual Fourier mode. We focus our discussion
here on the viscous problem \eqref{eq:Burgers} with the results for
the inviscid problem \eqref{eq:Burgers0} obtained simply by setting
$\nu = 0$. For each Fourier mode we use the polar representation $\widehat
u_k(t) = a_k(t) e^{i \phi_k(t)}$, where $a_k(t) = |\widehat u_k(t)|$ is
the mode amplitude and $\phi_{k}(t) = \arg [\widehat{u}_{k}(t)]$ the mode
phase. In terms of the Fourier-space representation, equation
\eqref{eq:Burgersa} then takes the form $\frac{\partial \widehat
  {u}_k}{\partial t} = -{\pi i k} \sum_{k_1}\, \widehat {u}_{k_1} \widehat
{u}_{k-k_1} - 4 \nu \pi^2 k^2 \widehat {u}_k$. The core of the dynamics is
the quadratic convolution term, a term that globally conserves energy
by redistributing it amongst Fourier modes via \emph{triad}
interactions, i.e., interactions between groups of 3 modes.  In order
to elucidate the role of these triadic interactions, especially their
phase dynamics, in the mechanisms of energy transfers, we derive
evolution equations for the mode amplitudes $a_k$ and the
corresponding phases $\phi_k$:
\begin{eqnarray}
\hspace{-0.1cm}
\label{eq:ampl_Burg}
\frac{d{a}_k}{dt} &=& \pi {k}  \sum_{k_1,k_2} a_{k_1} a_{k_2} \, \sin(\varphi_{k_1,k_2}^{k}) \, \delta_{k_1+k_2,k} - 4 \nu \pi^2 k^2 a_{k} \, , \quad k \in \ZZ^+, \\
\label{eq:phas_Burg}
\frac{d{\phi}_k}{dt} &=& - \pi {k}  \sum_{k_1,k_2} \frac{a_{k_1} a_{k_2}}{a_k} \,\cos(\varphi_{k_1,k_2}^{k}) \, \delta_{k_1+k_2,k},
\end{eqnarray}
where $\delta_{kj}$ is the Kronecker symbol and in the light of the
conjugate symmetry of the Fourier coefficients the wavenumber $k$ is
restricted to positive integers only (however, the indices $k_1$ and
$k_2$ in the sums run over both positive and negative values).

\subsubsection{Key degrees of freedom and energy budget} 

Upon examination of equations
\eqref{eq:ampl_Burg}--\eqref{eq:phas_Burg}, the key dynamical degrees
of freedom consist of the modes' real amplitudes $a_k(t)$ along with
the \emph{triad phases}:
\begin{equation}
{\varphi}_{k_1,\,k_2}^{k_3}(t) = {\phi}_{k_1}(t) + {\phi}_{k_2}(t) - {\phi}_{k_3}(t),
\end{equation}
where the wavenumbers $k_1$, $k_2$, $k_3$ satisfy a ``closed-triad''
condition: $k_1 + k_2 = k_3$ (in other words, modes with wavenumbers
satisfying this condition form a ``triad'' within which they can
exchange energy). An explicit form of the evolution equations for the
triad phases is obtained simply as appropriate linear combinations of
equation (\ref{eq:phas_Burg}) for the relevant wavenumbers.

Let us define the energy spectrum and the rate-of-change of its budget
up to wavenumber $k$ as:
\begin{equation}
{\K}_{k}(t) :=  \widehat {u}_{k}(t) \widehat {u}_{-k}(t) = a_k^2(t)\,, \qquad \frac{d}{dt}\left(\sum_{k'=1}^k {\K}_{k'}(t)\right) =  - \Pi(k,t) - \mathcal{D}(k,t)\,,
\label{eq:def_spec} 
\end{equation}
where $\mathcal{D}(k,t) := 8 \nu \pi^2 \sum_{k'=1}^k (k')^2 {\K}_{k'}$ is the
dissipation rate and $\Pi(k,t)$ is the flux across wavenumber $k$
towards large wavenumbers.  An explicit expression for the flux
$\Pi(k)$ in terms of the variables $a_k$ and
${\varphi}_{k_1,\,k_2}^{k_3}$ was derived by
\citet{buzzicotti2016phase}
\begin{equation}
\Pi(t,k) = \sum_{k_1=1}^k \sum_{k_3=k+1}^{\infty} 4 \pi k_1 \, {a}_{k_1} {a}_{k_2} {a}_{k_3} \,  \sin(\varphi_{k_1,k_2}^{k_3}). 
\label{eq:flux_sin}
\end{equation}
We remark here that each term on the right-hand side (RHS) in this
equation is due to a single triad interaction between modes with
wavenumbers $k_1, k_2$ and $k_3$, such that $ k_1 + k_2 = k_3$. This
form of $\Pi(k,t)$ encodes the particular wavenumber ordering $0 < k_1
< k_3, \quad 0 < k_2 < k_3$. While this choice of ordering is only a
convention, it is useful because the triad phases
$\varphi_{k_1,k_2}^{k_3}$, \emph{with such ordering}, are found to
take values preferentially near $\pi/2 \pmod {2\pi}$. In our analysis
below, it will be informative to split the flux given in
\eqref{eq:flux_sin} into its positive and negative parts $\Pi(t,k) =
\Pi^+(t,k) + \Pi^-(t,k)$, where $\Pi^+(t,k) \ge 0$ and $\Pi^-(t,k) <
0$, $\forall t,k$.

As discussed in detail by \citet{buzzicotti2016phase}, equation
\eqref{eq:flux_sin} demonstrates that robust energy flux towards small
scales must be accompanied by triad phases preferentially aligning
around the angle $\pi / 2$. More specifically, since $ k_1 \,
a_{k_1}a_{k_2}a_{k_3} >0$, at any given time a positive contribution
to the energy flux is maximized if the triad phase takes the value
$\pi/2 \pmod{ 2\pi}$. Time-dependent distributions of triad phases
occurring in inviscid and viscous Burgers flows corresponding to
different initial conditions, cf.~table \ref{tab:cases}, will be
analyzed in \S\,\ref{sec:results1D}. As the contribution of a
given triad to flux \eqref{eq:flux_sin} is proportional to $4 \pi k_1
\, a_{k_1}a_{k_2}a_{k_3}$, in \S\,\ref{sec:results1D} we will also
consider probability distributions of triad phases weighted by this
factor.

% In \S\,\ref{sec:FPresults} we show that this triad phase alignment to $\pi/2 \pmod{ 2\pi}$ occurs for all triads in the so-called ``inertial-range'' of scales, which is characterised by a power-law spectrum  $a_k \propto k^{-1}$ (see \cite{becreview}). 
% (THIS LAST PARAGRAPH TO BE IMPROVED).
%Finally, we define the total dissipation $\epsilon$ and total energy as
%\begin{equation}
%\epsilon(t) \equiv \mathcal{D}(\infty,t) = 2 \nu \sum_{k=1}^{\infty} k^2 E_k \quad , \quad E(t) = \sum_{k=1}^{\infty}E_{k}.\\ 
%\label{eq:def_tot}
%\end{equation}

%%%%%%%%%%%%%%%%%%%%%%%%%%%%%%%%%%%%%%%%%%%%%%%%%%%%%%%%%%%%%%%%%%%%%%%%%%%%%%%
\subsection{Helical triad phases and fluxes in 3D Navier-Stokes flows}
\label{sec:diagn3D}

The definition of triad phases in 3D incompressible flows needs to be
modified since we no longer deal with a scalar field as in the 1D
case. In order to account for the incompressibility condition, we will
employ a helical decomposition \citep{craya1957contribution,
  lesieur1972decomposition, herring1974approach,
  constantin1988beltrami,Waleffe1992} to represent the velocity field
resulting in two complex degrees of freedom at each wavevector. From
this we can deduce two \textit{helical} Fourier phases and when triad
interactions are considered, there will be 8 distinct types of helical
triads. However, only four of them are independent in the case of
zero-helicity flows (cf.~Appendix \ref{sec:conjugate_triads}) {and we
  also identify two important subtypes of ``boundary'' triads. In
  regard to 3D Navier-Stokes flows,} our interest will be mostly in
the forward energy cascade from large to small scales. We present a
more detailed derivation of a formula obtained by
\citet{MurrayPhDthesisNew} for the energy flux into a subset of modes
and find that its efficiency depends on the phases of these different
triad types. Then, in our analysis of the results in
\S\,\ref{sec:results3D}, the behaviour of each triad type will be
considered leading to distinct differences in the degree {of coherence
  they exhibit in different flows.}

\subsubsection{Helical-wave decomposition}
\label{sec:helical}

To examine the dynamics of system \eqref{eq:3DNS_u} at different
scales it is natural to begin with the Fourier decomposition of the
velocity field
\begin{equation}
\mathbf{{u}}(\mathbf{x},t) = \sum_{\mathbf{k}\in \mathbb{Z}^3\setminus \mathbf{0}} \mathbf{\widehat{u}}_{\mathbf{k}}(t) \exp(2\pi i \mathbf{k}\cdot\mathbf{x})\, ,
\label{eq:Fourier3D}
\end{equation}
where $\widehat{\u}_{\k}(t) \in \CC^3$ are the Fourier coefficients.
Since representation \eqref{eq:Fourier3D} constructed using arbitrary
Fourier coefficients satisfying the condition of conjugate symmetry
$\widehat{\u}_{-\k}(t) = \widehat{\u}_{\k}^*(t)$ will {\em not} in
general satisfy the incompressibility condition \eqref{eq:3DNS_uB}, we
consider the helical basis proposed by
\citet{constantin1988beltrami,Waleffe1992} and recently employed by
\citet{chen2003joint,biferale2013split, alexakis2017helically,
  sahoo2018energy}. Representation of vector fields in terms of the
helical basis satisfies the incompressibility condition by
construction and therefore involves the smallest number of independent
degrees of freedom consistent with this constraint. The helical
decomposition is applied in Fourier space as follows:
\begin{equation}
\label{eq:heli_decomp}
\mathbf{\widehat{u}}_{\mathbf{k}}(t) = \mathbf{h}_{\mathbf{k}}^{+} u_{\mathbf{k}}^{+}(t)+ \mathbf{h}_{\mathbf{k}}^- u_{\mathbf{k}}^-(t) \equiv \sum_{s \in \{ \pm\}} \mathbf{h}_{\mathbf{k}}^s u_{\mathbf{k}}^s(t)\,, \qquad \mathbf{k}\in \mathbb{Z}^3\setminus \mathbf{0}, 
\end{equation}
where $u_\k^+(t), u_\k^-(t) \in \CC$ are coefficients and the helical
basis vectors $\mathbf{h}_{\mathbf{k}}^s$, $s \in \{+,-\}$, have
complex components and are defined as the set of eigenmodes of the
curl operator:
\begin{equation*}
i \mathbf{k} \times \mathbf{h}_{\mathbf{k}}^s = s\, k \,\mathbf{h}_{\mathbf{k}}^s\,, \qquad k := |\mathbf{k}|\,. 
\end{equation*}
One of the key aspects behind this formulation is that due to the
choice of the basis vectors, the incompressibility condition
\eqref{eq:3DNS_uB} is satisfied automatically.  Since the basis
elements $\mathbf{h}_\k^s$ are not defined uniquely, we use the
proposal given by \citet{Waleffe1992}
\begin{equation}
\mathbf{h}_{\mathbf{k}}^s =\frac{1}{k} \mathbf{v} \times \mathbf{k} + i s \mathbf{v}\,, \qquad \mathbf{v} \equiv \frac{\mathbf{z}\times\mathbf{k}}{|\mathbf{z}\times\mathbf{k}|}\,,
\label{eq:hks}
\end{equation}
where $\mathbf{z} \in \RR^3$ is an arbitrary fixed vector. We note
here that there are now only four dynamical degrees of freedom per
wavevector in equation \eqref{eq:heli_decomp}, represented by the two
complex scalar helical modes $u_{\mathbf{k}}^{+}(t)$ and
$u_{\mathbf{k}}^{-}(t)$.

We now note the following identities:
\begin{equation}
\mathbf{h}_{-\mathbf{k}}^s = \mathbf{h}_{\mathbf{k}}^{-s} = [\mathbf{h}_{\mathbf{k}}^s]^*\,, \qquad \mathbf{h}_{\mathbf{k}}^{s_1} \cdot [\mathbf{h}_{\mathbf{k}}^{s_2}]^* = 2 \delta_{s_1 s_2}, 
\quad s,s_1,s_2 \in \{+,-\} \, .
\label{eq:hks_sym}
\end{equation}
From these identities along with the reality of the original 3D
velocity field it follows that the helical modes $u_{\mathbf{k}}^s(t)$
satisfy
\begin{equation*}
u_{-\mathbf{k}}^s(t) = [u_{\mathbf{k}}^s(t)]^*.
\end{equation*}

Another key aspect of the helical formulation is that the energy
  and helicity per wave-vector $\mathbf{k}$ are obtained in a
  ``diagonal'' form in terms of the helical modes as
\begin{subequations}
\begin{align}
{\K}_{\mathbf{k}}(t) & = u_{\mathbf{k}}^+(t) [u_{\mathbf{k}}^+(t)]^* + u_{\mathbf{k}}^-(t) [u_{\mathbf{k}}^-(t)]^* = \sum_{s \in \{ \pm\}} u_{\mathbf{k}}^s(t) [u_{\mathbf{k}}^s(t)]^*, \label{eq:heli_E_k} \\ 
{\H}_{\mathbf{k}}(t) & = k\big(u_{\mathbf{k}}^+(t) [u_{\mathbf{k}}^+(t)]^* - u_{\mathbf{k}}^-(t) [u_{\mathbf{k}}^-(t)]^* \big) = k \sum_{s \in \{ \pm\}} s u_{\mathbf{k}}^s(t) [u_{\mathbf{k}}^s(t)]^*. \label{eq:heli_H_k}
\end{align}
\end{subequations}
Based on \eqref{eq:heli_E_k}, one can define the energy spectrum as 
\begin{equation}
%e(k) := \sum_{|\k|-1/2 \le k \le |\k|+1/2} \K_\k
e(k) :=  \sum_{\substack{\mathbf{k} \in \mathbb{Z}^3\\ k-1/2 < |\mathbf{k}| < k+1/2}} {\mathcal{K}}_{\mathbf{k}}
\label{eq:ek}
\end{equation}
and the helicity spectrum can be defined in an analogous manner using
\eqref{eq:heli_H_k}. Energy and helicity are important because in
unforced inviscid flows the total energy and helicity (i.e., the
respective sums of ${\K}_{\mathbf{k}}(t)$ and ${\H}_{\mathbf{k}}(t)$
over all wave-vectors $\mathbf{k}$) are constants of motion
\citep{Moffatt1969}.  It is known that on average helicity flows from
small to large scales and energy has a flux from large to small
scales. We will focus on the energy flux from here on.

To find how the energy spectrum evolves in time we need to know the
evolution equation for the helical modes. Following \citet{Waleffe1992}
we get:
\begin{equation*}
\left(\Dpartial{}{t} + 4 \nu \pi^2 k^2\right) [u_{\mathbf{k}}^s(t)]^* = - \frac{\pi}{2} \!\!\sum_{ \substack{\mathbf{k}_1, \mathbf{k}_2 \in \mathcal{U}\\ \mathbf{k}_1 + \mathbf{k}_2 + \mathbf{k} = 0}} \!\!\sum_{s_{1},s_{2} \in \{\pm\}}\!\!(s_{1} k_1 - s_{2} k_2) [\mathbf{h}_{\mathbf{k}_1}^{s_{1}} \times \mathbf{h}_{\mathbf{k}_2}^{s_{2}} \cdot \mathbf{h}_{\mathbf{k}}^{s}] \,  u_{\mathbf{k}_1}^{s_{1}}(t) \, u_{\mathbf{k}_2}^{s_{2}}(t)\,,
\end{equation*}
where $\mathcal{U} = \mathbb{Z}^3\setminus \mathbf{0} \, $.  Let
$\mathcal{C} \subset \mathcal{U}$ be a subset of wavevectors such that
$-\mathcal{C} =\mathcal{C}$, where $- \mathcal{C} := \{-\mathbf{k}|
\mathbf{k} \in \mathcal{C}\}$ (notice that
$-\mathcal{U}=\mathcal{U}$). Define the energy in the set
$\mathcal{C}$ as
\begin{equation}
\label{eq:clus_E_k}
{\K}_{\mathcal{C}}(t) \equiv \sum_{\mathbf{k} \in \mathcal{C}} {\K}_{\mathbf{k}}(t) = \sum_{\substack{\mathbf{k} \in \mathcal{C} \\ s \in \{ \pm\}}} \, u_{\mathbf{k}}^{s}(t) [u_{\mathbf{k}}^{s}(t)]^*\,.
\end{equation}

\subsubsection{Fluxes}

We consider the rate of change of energy in the set $\mathcal{C}$.  As
the nonlinear term globally conserves energy, using
Eq.~\eqref{eq:clus_E_k} we obtain the time derivative of the energy
${\K}_{\mathcal{C}}(t)$ as:
\begin{equation*}
\dot{{\K}}_{\mathcal{C}}(t) = \Pi_{\mathcal{C}} - \mathcal{\epsilon}_{\mathcal{C}}\,,
\end{equation*}
where $\Pi_{\mathcal{C}}$ and $\mathcal{\epsilon}_{\mathcal{C}}$ are,
respectively, the nonlinear energy flux into the subset of modes
$\mathcal{C}$ and the energy dissipation in $\mathcal{C}$.  We now
examine $\Pi_{\mathcal{C}}$ more closely:
\begin{equation}
\Pi_{\mathcal{C}} =\frac{\pi}{2} \!\!\!\!\sum_{\substack{\mathbf{k}_3 \in \mathcal{C}\\ \mathbf{k}_1, \mathbf{k}_2 \in \mathcal{U}\\ \mathbf{k}_1 + \mathbf{k}_2 + \mathbf{k}_3 = 0}}  \!\!\sum_{s_1, s_2, s_3 \in \{\pm\}}\!\!(s_{2} k_2 - s_{1} k_1) [\mathbf{h}_{\mathbf{k}_1}^{s_{1}} \times \mathbf{h}_{\mathbf{k}_2}^{s_{2}} \cdot \mathbf{h}_{\mathbf{k}_3}^{s_{3}}]  u_{\mathbf{k}_1}^{s_{1}}(t) u_{\mathbf{k}_2}^{s_{2}}(t)u_{\mathbf{k}_3}^{s_{3}}(t)+ \mathrm{c.c.}
\label{eq:engy_flux_basic}
\end{equation}
To obtain a more detailed expression for this flux we use the
amplitude-phase representation of helical modes:
\begin{equation*}
u_{\mathbf{k}}^{s} = a_{\mathbf{k}}^{s} \exp(i \phi_{\mathbf{k}}^{s})\,,
\end{equation*}
where $a_{\mathbf{k}}^{s} := |u_{\mathbf{k}}^{s}| \geq 0$ is the
helical mode amplitude and $\phi_{\mathbf{k}}^{s} :=
\arg[u_{\mathbf{k}}^{s}] \in [0,2\pi]$ is the helical mode phase.
Notice that the reality of the original velocity field gives rise to
the identities:
\begin{equation*}
a_{-\mathbf{k}}^{s} = a_{\mathbf{k}}^{s}\,, \qquad \phi_{-\mathbf{k}}^{s} = -\phi_{\mathbf{k}}^{s} \,, \qquad  s \in \{ +, - \}\,.
\end{equation*}
We thus obtain
\begin{equation*}
\Pi_{\mathcal{C}} = \frac{\pi}{2} \!\!\!\!\sum_{\substack{\mathbf{k}_3 \in \mathcal{C}\\ \mathbf{k}_1, \mathbf{k}_2 \in \mathcal{U}\\ \mathbf{k}_1 + \mathbf{k}_2 + \mathbf{k}_3 = 0}}  \!\!\sum_{s_1, s_2, s_3 \in \{\pm\}}\!\!(s_{2} k_2 - s_{1} k_1) [\mathbf{h}_{\mathbf{k}_1}^{s_{1}} \times \mathbf{h}_{\mathbf{k}_2}^{s_{2}} \cdot \mathbf{h}_{\mathbf{k}_3}^{s_{3}}]   a_{\mathbf{k}_1}^{s_{1}} a_{\mathbf{k}_2}^{s_{2}} a_{\mathbf{k}_3}^{s_{3}}  \mathrm{e}^{i \varphi_{\mathbf{k}_1\mathbf{k}_2\mathbf{k}_3}^{s_{1} s_{2} s_{3}}}+ \mathrm{c.c.}\, ,
\end{equation*}
where we define the \textit{helical triad} phase as
\begin{equation*}
\varphi_{\mathbf{k}_1\mathbf{k}_2\mathbf{k}_3}^{s_1 s_2 s_3} := \phi_{\mathbf{k}_1}^{s_{1}} + \phi_{\mathbf{k}_2}^{s_{2}} + \phi_{\mathbf{k}_3}^{s_{3}}\,.
\end{equation*}
%We now introduce a useful short-hand notation for a generic variable $b$:
%\begin{equation*}
%b_{\mathbf{k}_1}^{s_{1}} \to b_{s_{\mathbf{k}_1}}\,,
%\end{equation*}
%so for example $\phi_{s_{\mathbf{k}_1}} =
%\phi_{\mathbf{k}_1}^{s_{1}}$,
%$\mathbf{h}_{s_{\mathbf{k}_1}}=\mathbf{h}_{\mathbf{k}_1}^{s_{1}} $ and
%so on. Similarly, we can reduce tensors $\varphi_{s_{\mathbf{k}_1}
%  s_{\mathbf{k}_2} s_{\mathbf{k}_3}}=
%\varphi_{\mathbf{k}_1\mathbf{k}_2\mathbf{k}_3}^{s_{1} s_{2}s_{3}} =
%\phi_{\mathbf{k}_1}^{s_{1}} + \phi_{\mathbf{k}_2}^{s_{2}} +
%\phi_{\mathbf{k}_3}^{s_{3}}.$ The energy flux then simplifies to the
%expression
%\begin{equation*}
%\Pi_{\mathcal{C}} = \frac{\pi}{2} \!\!\!\!\sum_{\substack{\mathbf{k}_3 \in \mathcal{C}\\ \mathbf{k}_1, \mathbf{k}_2 \in \mathcal{U}\\ \mathbf{k}_1 + \mathbf{k}_2 + \mathbf{k}_3 = 0}}  \!\!\sum_{s_1, s_2, s_3 \in \{\pm\}}\!\!(s_{2} k_2 - s_{1} k_1) [\mathbf{h}_{s_{\mathbf{k}_1}} \times \mathbf{h}_{s_{\mathbf{k}_2}} \cdot \mathbf{h}_{s_{\mathbf{k}_3}}]  a_{s_{\mathbf{k}_1}} a_{s_{\mathbf{k}_2}} a_{s_{\mathbf{k}_3}}  \mathrm{e}^{i \varphi_{s_{\mathbf{k}_1} s_{\mathbf{k}_2} s_{\mathbf{k}_3}}}+ \mathrm{c.c.}\, .
%\end{equation*}
From the fact that the prefactors $(s_{2} k_2 - s_{1} k_1)
[\mathbf{h}_{\mathbf{k}_1}^{s_{1}} \times
\mathbf{h}_{\mathbf{k}_2}^{s_{2}} \cdot
\mathbf{h}_{\mathbf{k}_3}^{s_{3}}]$ cancel out under total
symmetrization between the indices $1,2,3$ while the terms involving
amplitudes and phases are completely symmetric, we deduce that the
terms in the sum such that $\mathbf{k}_1,\mathbf{k}_2,\mathbf{k}_3 \in
\mathcal{C}$ do not contribute to the flux into $\mathcal{C}$.
Therefore, we can decompose the resultant flux into a sum of three terms:
\begin{align*}
\Pi_{\mathcal{C}} &= \frac{\pi}{2} \!\!\!\!\sum_{\substack{\mathbf{k}_3 \in \mathcal{C}\\ \mathbf{k}_1, \mathbf{k}_2 \in \mathcal{U}\setminus\mathcal{C}\\ \mathbf{k}_1 + \mathbf{k}_2 + \mathbf{k}_3 = 0}}  \!\!\sum_{s_1, s_2, s_3 \in \{\pm\}}\!\!(s_{2} k_2 - s_{1} k_1) [\mathbf{h}_{\mathbf{k}_1}^{s_{1}} \times \mathbf{h}_{\mathbf{k}_2}^{s_{2}} \cdot \mathbf{h}_{\mathbf{k}_3}^{s_{3}}]  a_{\mathbf{k}_1}^{s_{1}} a_{\mathbf{k}_2}^{s_{2}} a_{\mathbf{k}_3}^{s_{3}}  \mathrm{e}^{i \varphi_{\mathbf{k}_1\mathbf{k}_2\mathbf{k}_3}^{s_{1} s_{2} s_{3}}}+ \mathrm{c.c.}\\
&+ \frac{\pi}{2} \!\!\!\! \sum_{\substack{\mathbf{k}_2, \mathbf{k}_3 \in \mathcal{C}\\ \mathbf{k}_1 \in \mathcal{U}\setminus\mathcal{C}\\ \mathbf{k}_1 + \mathbf{k}_2 + \mathbf{k}_3 = 0}}  \!\!\sum_{s_1, s_2, s_3 \in \{\pm\}}\!\!(s_{2} k_2 - s_{1} k_1) [\mathbf{h}_{\mathbf{k}_1}^{s_{1}} \times \mathbf{h}_{\mathbf{k}_2}^{s_{2}} \cdot \mathbf{h}_{\mathbf{k}_3}^{s_{3}}]  a_{\mathbf{k}_1}^{s_{1}} a_{\mathbf{k}_2}^{s_{2}} a_{\mathbf{k}_3}^{s_{3}}  \mathrm{e}^{i \varphi_{\mathbf{k}_1\mathbf{k}_2\mathbf{k}_3}^{s_{1} s_{2} s_{3}}}+ \mathrm{c.c.}\\
&+ \frac{\pi}{2} \!\!\!\! \sum_{\substack{\mathbf{k}_1, \mathbf{k}_3 \in \mathcal{C}\\ \mathbf{k}_2 \in \mathcal{U}\setminus\mathcal{C}\\ \mathbf{k}_1 + \mathbf{k}_2 + \mathbf{k}_3 = 0}}  \!\!\sum_{s_1, s_2, s_3 \in \{\pm\}}\!\!(s_{2} k_2 - s_{1} k_1) [\mathbf{h}_{\mathbf{k}_1}^{s_{1}} \times \mathbf{h}_{\mathbf{k}_2}^{s_{2}} \cdot \mathbf{h}_{\mathbf{k}_3}^{s_{3}}]  a_{\mathbf{k}_1}^{s_{1}} a_{\mathbf{k}_2}^{s_{2}} a_{\mathbf{k}_3}^{s_{3}}  \mathrm{e}^{i \varphi_{\mathbf{k}_1\mathbf{k}_2\mathbf{k}_3}^{s_{1} s_{2} s_{3}}}+ \mathrm{c.c.}
\end{align*}
and, because the summands above are symmetric under the permutation
between the indices $1$ and $2$, the last two sums combine into one,
leading to:
\begin{align*}
\Pi_{\mathcal{C}} &= \frac{\pi}{2} \!\!\!\!\sum_{\substack{\mathbf{k}_3 \in \mathcal{C}\\ \mathbf{k}_1, \mathbf{k}_2 \in \mathcal{U}\setminus\mathcal{C}\\ \mathbf{k}_1 + \mathbf{k}_2 + \mathbf{k}_3 = 0}}  \!\!\sum_{s_1, s_2, s_3 \in \{\pm\}}\!\!(s_{2} k_2 - s_{1} k_1) [\mathbf{h}_{\mathbf{k}_1}^{s_{1}} \times \mathbf{h}_{\mathbf{k}_2}^{s_{2}} \cdot \mathbf{h}_{\mathbf{k}_3}^{s_{3}}]  a_{\mathbf{k}_1}^{s_{1}} a_{\mathbf{k}_2}^{s_{2}} a_{\mathbf{k}_3}^{s_{3}}  \mathrm{e}^{i \varphi_{\mathbf{k}_1\mathbf{k}_2\mathbf{k}_3}^{s_{1} s_{2} s_{3}}}+ \mathrm{c.c.}\\
&+ {\pi} \!\!\!\! \sum_{\substack{\mathbf{k}_2, \mathbf{k}_3 \in \mathcal{C}\\ \mathbf{k}_1 \in \mathcal{U}\setminus\mathcal{C}\\ \mathbf{k}_1 + \mathbf{k}_2 + \mathbf{k}_3 = 0}}  \!\!\sum_{s_1, s_2, s_3 \in \{\pm\}}\!\!(s_{2} k_2 - s_{1} k_1) [\mathbf{h}_{\mathbf{k}_1}^{s_{1}} \times \mathbf{h}_{\mathbf{k}_2}^{s_{2}} \cdot \mathbf{h}_{\mathbf{k}_3}^{s_{3}}]  a_{\mathbf{k}_1}^{s_{1}} a_{\mathbf{k}_2}^{s_{2}} a_{\mathbf{k}_3}^{s_{3}}  \mathrm{e}^{i \varphi_{\mathbf{k}_1\mathbf{k}_2\mathbf{k}_3}^{s_{1} s_{2} s_{3}}}+ \mathrm{c.c.} \ .
\end{align*}
Finally, noting that the summation domain in the second sum above is
symmetric under the permutation between the indices $2$ and $3$, we
symmetrize this second sum to get:
\begin{equation}
\begin{aligned}
\Pi_{\mathcal{C}} &= \frac{\pi}{2} \!\!\!\!\sum_{\substack{\mathbf{k}_3 \in \mathcal{C}\\ \mathbf{k}_1, \mathbf{k}_2 \in \mathcal{U}\setminus\mathcal{C}\\ \mathbf{k}_1 + \mathbf{k}_2 + \mathbf{k}_3 = 0}}  \!\!\sum_{s_1, s_2, s_3 \in \{\pm\}}\!\!(s_{2} k_2 - s_{1} k_1) [\mathbf{h}_{\mathbf{k}_1}^{s_{1}} \times \mathbf{h}_{\mathbf{k}_2}^{s_{2}} \cdot \mathbf{h}_{\mathbf{k}_3}^{s_{3}}]  a_{\mathbf{k}_1}^{s_{1}} a_{\mathbf{k}_2}^{s_{2}} a_{\mathbf{k}_3}^{s_{3}}  \mathrm{e}^{i \varphi_{\mathbf{k}_1\mathbf{k}_2\mathbf{k}_3}^{s_{1} s_{2} s_{3}}}+ \mathrm{c.c.}\\
&+\frac{\pi}{2} \!\!\!\! \sum_{\substack{\mathbf{k}_2, \mathbf{k}_3 \in \mathcal{C}\\ \mathbf{k}_1 \in \mathcal{U}\setminus\mathcal{C}\\ \mathbf{k}_1 + \mathbf{k}_2 + \mathbf{k}_3 = 0}}  \!\!\sum_{s_1, s_2, s_3 \in \{\pm\}}\!\!(s_{2} k_2 - s_{3} k_3) [\mathbf{h}_{\mathbf{k}_1}^{s_{1}} \times \mathbf{h}_{\mathbf{k}_2}^{s_{2}} \cdot \mathbf{h}_{\mathbf{k}_3}^{s_{3}}]  a_{\mathbf{k}_1}^{s_{1}} a_{\mathbf{k}_2}^{s_{2}} a_{\mathbf{k}_3}^{s_{3}}  \mathrm{e}^{i \varphi_{\mathbf{k}_1\mathbf{k}_2\mathbf{k}_3}^{s_{1} s_{2} s_{3}}}+ \mathrm{c.c.} \ .
\end{aligned}
\label{eq:PiC}
\end{equation}
We note here a similarity between \eqref{eq:PiC} and expression
\eqref{eq:flux_sin} describing the flux for 1D Burgers flows. However,
in the present case there are 8 distinct helical triad phases
resulting from the permutations of $s_{1}$, $s_{2}$ and $s_{3}$.

%%%%%%%%%%%%%%%%%%%%%%%%%%%%%%%%%%%%%%%%%%%%%%%%%%%%%
%% combining the phase with coefficient correction
%%%%%%%%%%%%%%%%%%%%%%%%%%%%%%%%%%%%%%%%%%%%%%%%%%%%%
Looking at the static coefficients (of interaction) in the terms in
the flux equation \eqref{eq:PiC}, we note that, due to the fact that
basis vectors are complex valued, these coefficients will be complex
too. Thus, to move all dependence of the sign of the contribution of
each term onto the helical triad phase we must define a correction to
the phase accounting for the phase of the complex coefficient.
Therefore, for a triad $\mathbf{k}_1 + \mathbf{k}_2 + \mathbf{k}_3 =
\mathbf{0}$, we define the {\em generalized helical triad phases} as
follows:
\begin{equation}
\Phi_{\mathbf{k}_1\mathbf{k}_2\mathbf{k}_3}^{s_{1} s_{2} s_{3}} = \varphi_{\mathbf{k}_1\mathbf{k}_2\mathbf{k}_3}^{s_{1} s_{2} s_{3}} + \Delta_{\mathbf{k}_1\mathbf{k}_2\mathbf{k}_3}^{s_{1} s_{2} s_{3}}, 
\label{eq:GHTP}
\end{equation}
where
\begin{equation}
\label{eq:delta_ENS}
\Delta_{\mathbf{k}_1\mathbf{k}_2\mathbf{k}_3}^{s_{1} s_{2} s_{3}} := \begin{cases}
\arg\Big((s_{2} k_2 - s_{1} k_1) [\mathbf{h}_{\mathbf{k}_1}^{s_{1}} \times \mathbf{h}_{\mathbf{k}_2}^{s_{2}} \cdot \mathbf{h}_{\mathbf{k}_3}^{s_{3}}] \Big), & \text{if} \quad \mathbf{k}_3 \in \mathcal{C}; \mathbf{k}_1, \mathbf{k}_2 \in \mathcal{U}\setminus \mathcal{C}  \\
 \arg\Big( (s_{2} k_2 - s_{3} k_3)[ \mathbf{h}_{\mathbf{k}_1}^{s_{1}} \times \mathbf{h}_{\mathbf{k}_2}^{s_{2}} \cdot \mathbf{h}_{\mathbf{k}_3}^{s_{3}}] \Big), & \text{if} \quad \mathbf{k}_3, \mathbf{k}_2 \in \mathcal{C}; \mathbf{k}_1 \in \mathcal{U}\setminus \mathcal{C} \,\end{cases}
\end{equation}
is the required correction.  This means that the sign of the
contribution of each term in the flux expression will depend only on
$\cos(\Phi_{\mathbf{k}_1\mathbf{k}_2\mathbf{k}_3}^{s_{1} s_{2}
  s_{3}})$, and a value of
$\Phi_{\mathbf{k}_1\mathbf{k}_2\mathbf{k}_3}^{s_{1} s_{2} s_{3}} \in
(-\frac{\pi}{2},\frac{\pi}{2})$ will result in a positive contribution
(flux into $\mathcal{C}$), while
$\Phi_{\mathbf{k}_1\mathbf{k}_2\mathbf{k}_3}^{s_{1} s_{2} s_{3}} \in
(\frac{\pi}{2},\pi) \cup (-\pi,-\frac{\pi}{2})$ will produce a
negative contribution (flux out of $\mathcal{C}$).

All these observations lead to the following formula for the flux
whose direction depends only on the generalized helical triad phases:
\begin{equation}
\begin{aligned}
\Pi_{\mathcal{C}} &= {\pi}\!\!\!\!\sum_{\substack{\mathbf{k}_3 \in \mathcal{C} \\ \mathbf{k}_1, \mathbf{k}_2 \in \mathcal{U}\setminus\mathcal{C} \\ \mathbf{k}_1 + \mathbf{k}_2 + \mathbf{k}_3 = 0 \\ s_{1},s_{2},s_{3} \in \{ \pm \}}} \big|(s_{2} k_2 - s_{1} k_1) [\mathbf{h}_{\mathbf{k}_1}^{s_{1}} \times \mathbf{h}_{\mathbf{k}_2}^{s_{2}} \cdot \mathbf{h}_{\mathbf{k}_3}^{s_{3}}]\big|  a_{\mathbf{k}_1}^{s_{1}} a_{\mathbf{k}_2}^{s_{2}} a_{\mathbf{k}_3}^{s_{3}} \cos(\Phi_{\mathbf{k}_1\mathbf{k}_2\mathbf{k}_3}^{s_{1} s_{2} s_{3}}) \\
&+ {\pi}\!\!\!\!\sum_{\substack{\mathbf{k}_2, \mathbf{k}_3 \in \mathcal{C} \\ \mathbf{k}_1 \in \mathcal{U}\setminus\mathcal{C} \\ \mathbf{k}_1 + \mathbf{k}_2 + \mathbf{k}_3 = 0 \\ s_{1},s_{2},s_{3}\in \{ \pm \}}} \big|(s_{2} k_2 - s_{3} k_3) [\mathbf{h}_{\mathbf{k}_1}^{s_{1}} \times \mathbf{h}_{\mathbf{k}_2}^{s_{2}} \cdot \mathbf{h}_{\mathbf{k}_3}^{s_{3}}]\big|   a_{\mathbf{k}_1}^{s_{1}} a_{\mathbf{k}_2}^{s_{2}} a_{\mathbf{k}_3}^{s_{3}} \cos(\Phi_{\mathbf{k}_1\mathbf{k}_2\mathbf{k}_3}^{s_{1} s_{2} s_{3}})\,. \\
\end{aligned}
\label{eq:engy_flux_final}
\end{equation}
While in practice flux $\Pi_{\mathcal{C}}$ can be evaluated more
efficiently using the pseudospectral representation of the nonlinear
advection term in \eqref{eq:3DNS_uA}, {as is usually done in the
  numerical solution of system \eqref{eq:3DNS_u},} expression
\eqref{eq:engy_flux_final} is interesting as it explicates the
dependence of the flux on the (generalized) triad phases for phases of
different types. We observe that for a given triad $\k_1 + \k_2 + \k_3
= \0$, its contribution to the flux in \eqref{eq:engy_flux_final} is a
sum of 8 distinct parts corresponding to all possible combinations of
$s_1,s_2,s_3 \in \{+,-\}$. However, as demonstrated in Appendix
\ref{sec:conjugate_triads}, under the assumption of oddity under
parity transformation \eqref{eq:oddity}, which the Navier-Stokes flows
discussed in \S\,\ref{sec:problems} satisfy, there are only 4
independent classes of mode interactions associated with each triad,
namely, $\{+ + +,+ - -,+ - +,+ + -\}$, because the ``conjugate''triads
$\{- - -,- + +,- + -,- - +\}$ obtained by changing $s_j \rightarrow
-s_j$, $j=1,2,3$, produce the same fluxes and are thus
indistinguishable form the original ones.

How does this reduction to 4 helical triad types relate to Waleffe's
analysis \citep{Waleffe1992} in terms of helical triad types?  It turns
out that we recover Waleffe's classes I to IV, plus two new ``boundary''
classes. To begin with, Waleffe's analysis of helical triad types not
only considers the helicity of each mode, but also the relative size
of the wavevectors in each mode. By performing a direct resummation of
equation \eqref{eq:engy_flux_final} and using the results from
Appendix \ref{sec:conjugate_triads}, we can write the flux as a sum of
6 terms:
$$\Pi_\mathcal{C} = \Pi_\mathcal{C}^{+++} + \Pi_\mathcal{C}^{+--}  + \Pi_\mathcal{C}^{+-+} + \Pi_\mathcal{C}^{++-} + \Pi_\mathcal{C}^{(+-)+} + \Pi_\mathcal{C}^{+(+-)},$$
where the first four helical fluxes in the above sum are defined by
\begin{equation}
\begin{aligned}
\Pi_\mathcal{C}^{s_1 s_2 s_3} := &2{\pi}\!\!\!\!\sum_{\substack{\mathbf{k}_3 \in \mathcal{C} \\ \mathbf{k}_1, \mathbf{k}_2 \in \mathcal{U}\setminus\mathcal{C} \\ \mathbf{k}_1 + \mathbf{k}_2 + \mathbf{k}_3 = 0 \\
|\mathbf{k}_1| < |\mathbf{k}_2| < |\mathbf{k}_3|}} \big|(s_{2} k_2 - s_{1} k_1) [\mathbf{h}_{\mathbf{k}_1}^{s_{1}} \times \mathbf{h}_{\mathbf{k}_2}^{s_{2}} \cdot \mathbf{h}_{\mathbf{k}_3}^{s_{3}}]\big|  a_{\mathbf{k}_1}^{s_{1}} a_{\mathbf{k}_2}^{s_{2}} a_{\mathbf{k}_3}^{s_{3}} \cos(\Phi_{\mathbf{k}_1\mathbf{k}_2\mathbf{k}_3}^{s_{1} s_{2} s_{3}}) \\
+ &2{\pi}\!\!\!\!\sum_{\substack{\mathbf{k}_2, \mathbf{k}_3 \in \mathcal{C} \\ \mathbf{k}_1 \in \mathcal{U}\setminus\mathcal{C} \\ \mathbf{k}_1 + \mathbf{k}_2 + \mathbf{k}_3 = 0 \\
|\mathbf{k}_1| < |\mathbf{k}_2| < |\mathbf{k}_3|}} \big|(s_{2} k_2 - s_{3} k_3) [\mathbf{h}_{\mathbf{k}_1}^{s_{1}} \times \mathbf{h}_{\mathbf{k}_2}^{s_{2}} \cdot \mathbf{h}_{\mathbf{k}_3}^{s_{3}}]\big|   a_{\mathbf{k}_1}^{s_{1}} a_{\mathbf{k}_2}^{s_{2}} a_{\mathbf{k}_3}^{s_{3}} \cos(\Phi_{\mathbf{k}_1\mathbf{k}_2\mathbf{k}_3}^{s_{1} s_{2} s_{3}})\,,
\end{aligned}
\label{eq:engy_flux_triad}
\end{equation}
while the last two helical fluxes correspond to the boundary classes,
with the symbol $(+-)$ representing the corresponding ambiguity
between two parameters:
\begin{equation}
\begin{aligned}
\Pi_\mathcal{C}^{(+-)+} := &2{\pi}\!\!\!\!\sum_{\substack{\mathbf{k}_3 \in \mathcal{C} \\ \mathbf{k}_1, \mathbf{k}_2 \in \mathcal{U}\setminus\mathcal{C} \\ \mathbf{k}_1 + \mathbf{k}_2 + \mathbf{k}_3 = 0 \\
|\mathbf{k}_1| = |\mathbf{k}_2| < |\mathbf{k}_3| \\
s_{1} \in \{ \pm \}}} 2 k_1 \big|[\mathbf{h}_{\mathbf{k}_1}^{s_{1}} \times \mathbf{h}_{\mathbf{k}_2}^{-s_{1}} \cdot \mathbf{h}_{\mathbf{k}_3}^{+}]\big|  a_{\mathbf{k}_1}^{s_{1}} a_{\mathbf{k}_2}^{-s_{1}} a_{\mathbf{k}_3}^{+} \cos(\Phi_{\mathbf{k}_1\mathbf{k}_2\mathbf{k}_3}^{s_{1}, -s_{1}, +}) \\
+ & 2{\pi}\!\!\!\!\sum_{\substack{\mathbf{k}_2, \mathbf{k}_3 \in \mathcal{C} \\ \mathbf{k}_1 \in \mathcal{U}\setminus\mathcal{C} \\ \mathbf{k}_1 + \mathbf{k}_2 + \mathbf{k}_3 = 0 \\
|\mathbf{k}_1| = |\mathbf{k}_2| < |\mathbf{k}_3| \\
s_{1} \in \{ \pm \}}} \big|(s_{1} k_1 + k_3) [\mathbf{h}_{\mathbf{k}_1}^{s_{1}} \times \mathbf{h}_{\mathbf{k}_2}^{-s_{1}} \cdot \mathbf{h}_{\mathbf{k}_3}^{+}]\big|   a_{\mathbf{k}_1}^{s_{1}} a_{\mathbf{k}_2}^{-s_{1}} a_{\mathbf{k}_3}^{+} \cos(\Phi_{\mathbf{k}_1\mathbf{k}_2\mathbf{k}_3}^{s_{1}, -s_{1}, +})\,,
\end{aligned}
\label{eq:engy_flux_triad_(PM)P}
\end{equation}
\begin{equation}
\begin{aligned}
\Pi_\mathcal{C}^{+ (+-)} := &2{\pi}\!\!\!\!\sum_{\substack{\mathbf{k}_3 \in \mathcal{C} \\ \mathbf{k}_1, \mathbf{k}_2 \in \mathcal{U}\setminus\mathcal{C} \\ \mathbf{k}_1 + \mathbf{k}_2 + \mathbf{k}_3 = 0 \\
|\mathbf{k}_1| < |\mathbf{k}_2| = |\mathbf{k}_3| \\
s_{2} \in \{ \pm \}}} \big|(s_{2} k_2 - k_1) [\mathbf{h}_{\mathbf{k}_1}^{+} \times \mathbf{h}_{\mathbf{k}_2}^{s_{2}} \cdot \mathbf{h}_{\mathbf{k}_3}^{-s_{2}}]\big|  a_{\mathbf{k}_1}^{+} a_{\mathbf{k}_2}^{s_{2}} a_{\mathbf{k}_3}^{-s_{2}} \cos(\Phi_{\mathbf{k}_1\mathbf{k}_2\mathbf{k}_3}^{+, s_{2}, -s_{2}}) \\
+ & 2{\pi}\!\!\!\!\sum_{\substack{\mathbf{k}_2, \mathbf{k}_3 \in \mathcal{C} \\ \mathbf{k}_1 \in \mathcal{U}\setminus\mathcal{C} \\ \mathbf{k}_1 + \mathbf{k}_2 + \mathbf{k}_3 = 0 \\
|\mathbf{k}_1| < |\mathbf{k}_2| = |\mathbf{k}_3| \\
s_{2} \in \{ \pm \}}} 2 k_2 \big|[\mathbf{h}_{\mathbf{k}_1}^{+} \times \mathbf{h}_{\mathbf{k}_2}^{s_{2}} \cdot \mathbf{h}_{\mathbf{k}_3}^{-s_{2}}]\big|   a_{\mathbf{k}_1}^{+} a_{\mathbf{k}_2}^{s_{2}} a_{\mathbf{k}_3}^{-s_{2}} \cos(\Phi_{\mathbf{k}_1\mathbf{k}_2\mathbf{k}_3}^{+, s_{2}, -s_{2}})\,,
\end{aligned}
\label{eq:engy_flux_triad_P(PM)}
\end{equation}
We remark that flux contributions corresponding to these last two
boundary classes normally have fewer terms compared with the other 4
classes, although details will depend on the size of the set
$\mathcal{C}$. To further simplify the notation, hereafter we will use
``PPP'' instead of $+++$, etc.  In terms of Waleffe's notation
\citep{Waleffe1992}, our $+++$ (PPP) class thus corresponds to Class I
triads, while $+--$ (PMM) corresponds to Class II, $+-+$ (PMP) to
Class III and $++-$ (PPM) to Class IV.  See figure \ref{fig:Waleffe}
for a depiction of this classification, including the expected
directions of energy transfer according to Waleffe's ``instability
assumption''. The new boundary class $+(+-)$ (or P(PM)) is at the
boundary between PPM and PMP, namely, between Class III and Class IV.
The new boundary class $(+-)+$ (or (PM)P) is at the boundary between
PMP and PMM, namely, between Class II and Class III. For each of these
6 types of helical fluxes it is useful to
define the following time-dependent diagnostic quantities: \\
\noindent (i) The probability density function (PDF for short)  ${\mathcal{P}}_{\mathcal{C}}^{s_1 s_2 s_3}(\Phi)(t)$ of the triad phases, representing a normalized histogram of the phases that contribute to the flux in the above sums.\\
\noindent (ii) The ``flux density'' $F_{\mathcal{C}}^{s_1 s_2 s_3}(\Phi)(t)$ representing the contribution from generalized helical triads with  phase angle  $\Phi$ to the flux $\Pi_\mathcal{C}^{s_1 s_2 s_3}(t)$ as given by formulas \eqref{eq:engy_flux_triad}--\eqref{eq:engy_flux_triad_P(PM)}, namely,
\begin{equation}
\Pi_\mathcal{C}^{s_1 s_2 s_3}(t) = \int_{-\pi}^\pi F_{\mathcal{C}}^{s_1 s_2 s_3}(\Phi)(t) \cos(\Phi) \mathrm{d}\Phi\,.
\label{eq:FC}
\end{equation}
\noindent (iii) The weighted PDF (wPDF for short) ${\mathcal{W}}_{\mathcal{C}}^{s_1 s_2 s_3}(\Phi)(t)$, representing a normalized version of the flux densities
\begin{equation}
{\mathcal{W}}_{\mathcal{C}}^{s_1 s_2 s_3}(\Phi)(t) = \frac{{F}_{\mathcal{C}}^{s_1 s_2 s_3}(\Phi)(t)}{\int_{-\pi}^{\pi} {F}_{\mathcal{C}}^{s_1 s_2 s_3}(\Phi')(t)\mathrm{d}\Phi'}\,.
\label{eq:WC}
\end{equation}
By analyzing these densities, in \S\,\ref{sec:results3D} we will
assess the relative significance of the 6 distinct classes of helical
triad interactions for the transfer of energy in 3D Navier-Stokes
flows with different initial conditions, cf.~table \ref{tab:cases}.

\begin{figure}
\begin{center}
\includegraphics[width=\textwidth]{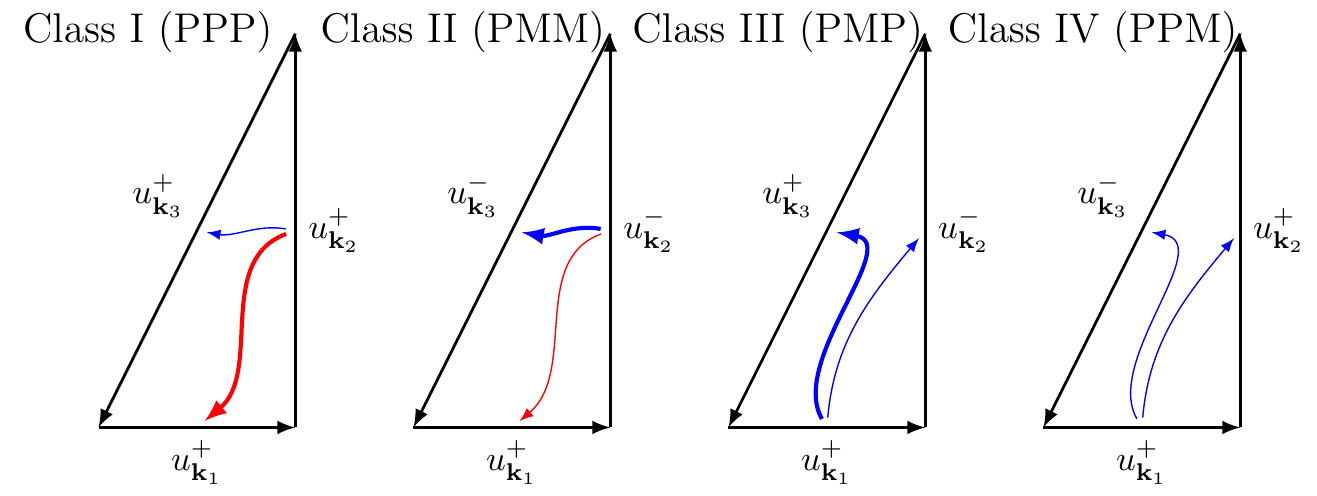}
\caption{Depiction of Waleffe's classification of helical triads along
  with the new letter notation used in this paper for each class. In
  our convention, the wavevectors are ordered by size, i.e.,
  $|\mathbf{k}_1| < |\mathbf{k}_2| < |\mathbf{k}_3|$ (just like in
  Waleffe's {convention}), and the helicity of the smallest
  wavevector is always ``$+$''. This is possible as we are dealing
  with fields that are odd under parity and therefore one can identify
  triads that are related by a global helicity reversal. The letter
  notation {``P'' and ``M'' indicates, respectively, ``$+$'' and
    ``$-$'' helicity.}  The blue and red arrows indicate the preferred
  directions for energy fluxes according to {the ``instability
    assumption'' of \citet{Waleffe1992}:} blue represents the forward
  cascade (from large to small spatial scales) while red represents
  the inverse cascade (from small to large scales); the {line
    thickness reflects the} relative strength of the energy flux. The
  new boundary-type classes P(PM) and (PM)P, not shown {in the
    figure}, correspond, respectively, to the intersection between
  Class III and Class IV ($|\mathbf{k}_2| = |\mathbf{k}_3|$) and the
  intersection between Class II and Class III ($|\mathbf{k}_1| =
  |\mathbf{k}_2|$).}
\label{fig:Waleffe}
\end{center}
\end{figure}

\section{Results}
\label{sec:results}

%%%%%%%%%%%%%%%%%%%%%%%%%%%%%%%%%%%%%%%%%%%%%%%%%%%%%%%%%%%%%%%%%%%%%%%%%%%%%%%

{In this section we employ the diagnostics developed in
  \S\,\ref{sec:diagnostics} to analyze the 1D Burgers and 3D
  Navier-Stokes flows with the different initial conditions specified
  in table \ref{tab:cases}. Solutions to the inviscid Burgers system
  \eqref{eq:Burgers0} are obtained analytically with the
  Fourier-Lagrange formula described in Appendix \ref{sec:FL}, whereas
  the viscous Burgers and Navier-Stokes systems \eqref{eq:Burgers} and
  \eqref{eq:3DNS_u} are solved numerically with standard
  pseudo-spectral approaches for which details are provided by
  \citet{ap11a,KangYumProtas2020}.  }

%%%%%%%%%%%%%%%%%%%%%%%%%%%%%%%%%%%%%%%%%%%%%%%%%%%%%%%%%%%%%%%%%%%%%%%%%%%%%%%
\subsection{Results for 1D Burgers flows}
\label{sec:results1D}

We begin the presentation of the results for the 1D Burgers flows by
briefly describing the evolution of the solutions corresponding to
different initial conditions, cf.~table \ref{tab:cases}, in the
physical and spectral space. We also present the time history of the
enstrophy $\E(t)$ for the different cases.

The time evolution of the inviscid Burgers system \eqref{eq:Burgers0}
in the physical space with the unimodal \eqref{eq:Bu0sin} and extreme
\eqref{eq:Bu0ext0} initial conditions is presented in figures
\ref{fig:B0ut}(a) and \ref{fig:B0ut}(b), respectively. While the
formation of a steep front where the derivative of the solution
$u(x,t)$ with respect to $x$ becomes unbounded can be observed in both
cases, we note that in the case with the extreme initial condition the
solution also becomes discontinuous as the blow-up time is approached.
This latter effect can be attributed to the simultaneous collapse of
characteristics corresponding to initial points in a set of measure
greater than zero (in contrast, the singularity evident in figure
\ref{fig:B0ut}(a) is the result of the crossing of characteristics
corresponding to points with an infinitesimal separation). From the
evolution of the enstrophy $\E(t)$ shown in figure \ref{fig:B0sEt}(b)
we conclude that in the latter case the singularity formation, marked
by an unbounded growth of enstrophy, occurs much faster.  As expected
from the piecewise-linear form of the extreme initial condition in
figure \ref{fig:B0ut}(b), cf.~equation \eqref{eq:Bu0ext0}, its Fourier
coefficients behave as $\left|
  \left[\widehat{\widetilde{u}}_0\right]_k \right| \sim k^{-2}$,
whereas at the time of blow-up $t = t^*$ the spectrum of the solution
has the form $|\widehat{u}_k(t^*)| \sim k^{-1}$, see figure
\ref{fig:B0sEt}(a).

\begin{figure}
\begin{center}
\mbox{\subfigure[]{\includegraphics[width=0.45\textwidth]{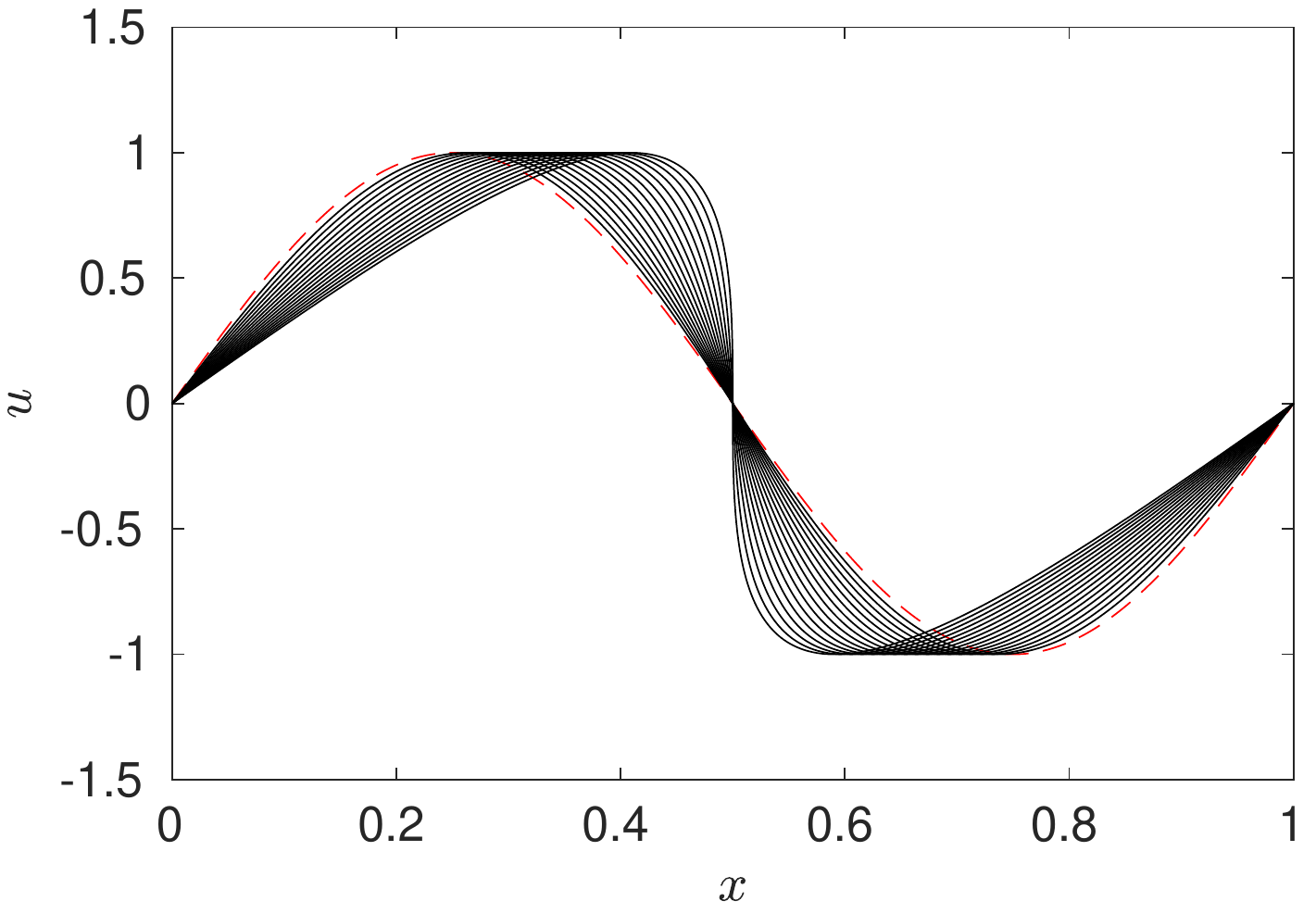}}\qquad
\subfigure[]{\includegraphics[width=0.45\textwidth]{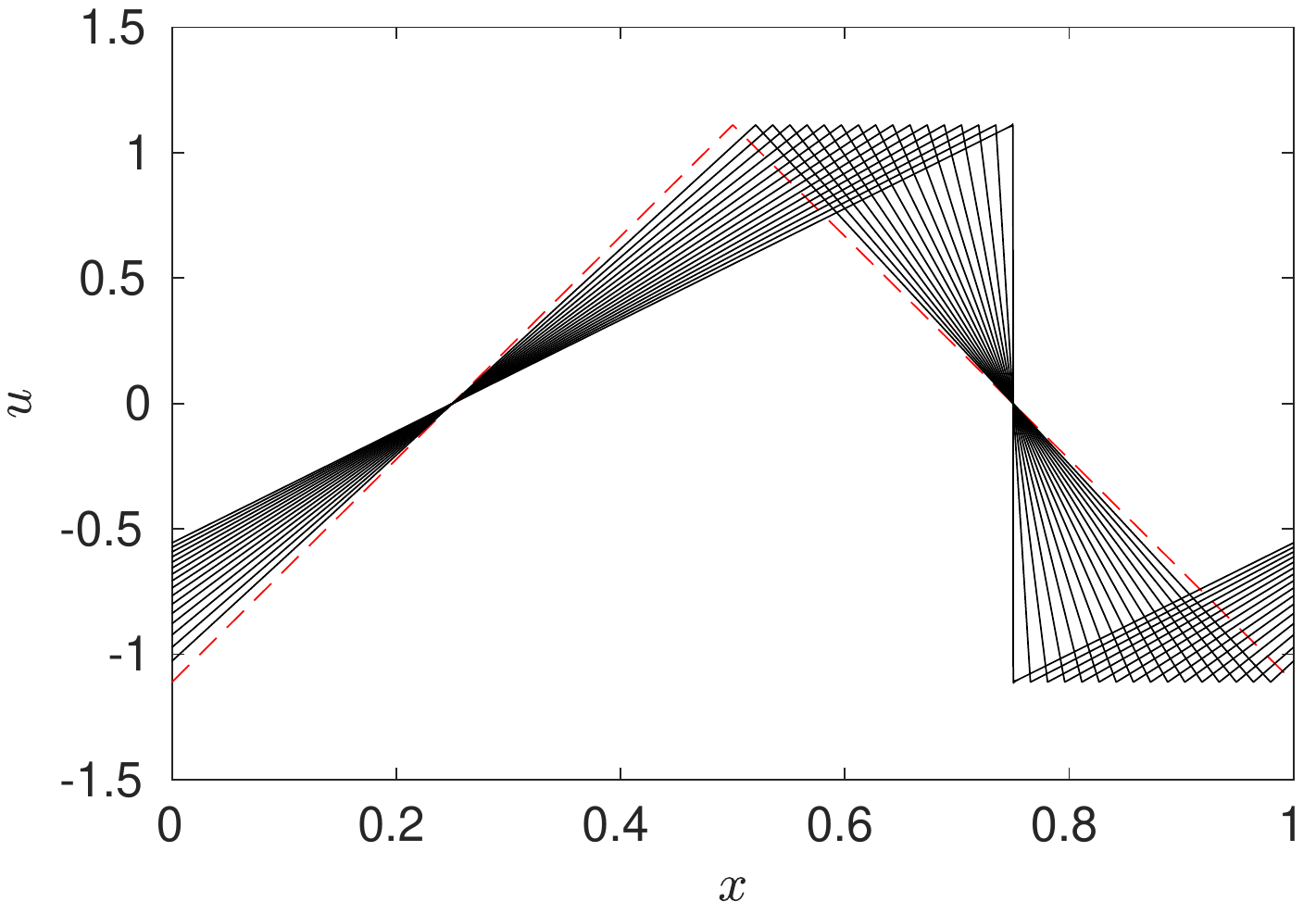}}}
\caption{Solution of the inviscid Burgers equation \eqref{eq:Burgers0}
  with (a) the unimodal initial condition \eqref{eq:Bu0sin} and (b)
  the extreme initial condition \eqref{eq:Bu0ext0} for increasing
  times $t \in [0,t^*]$. The initial conditions are marked with red
  dashed lines and the solutions develop progressively steeper fronts
  as the time increases.}
\label{fig:B0ut}
\end{center}
\end{figure}

\begin{figure}
\begin{center}
\mbox{\subfigure[]{\includegraphics[width=0.45\textwidth]{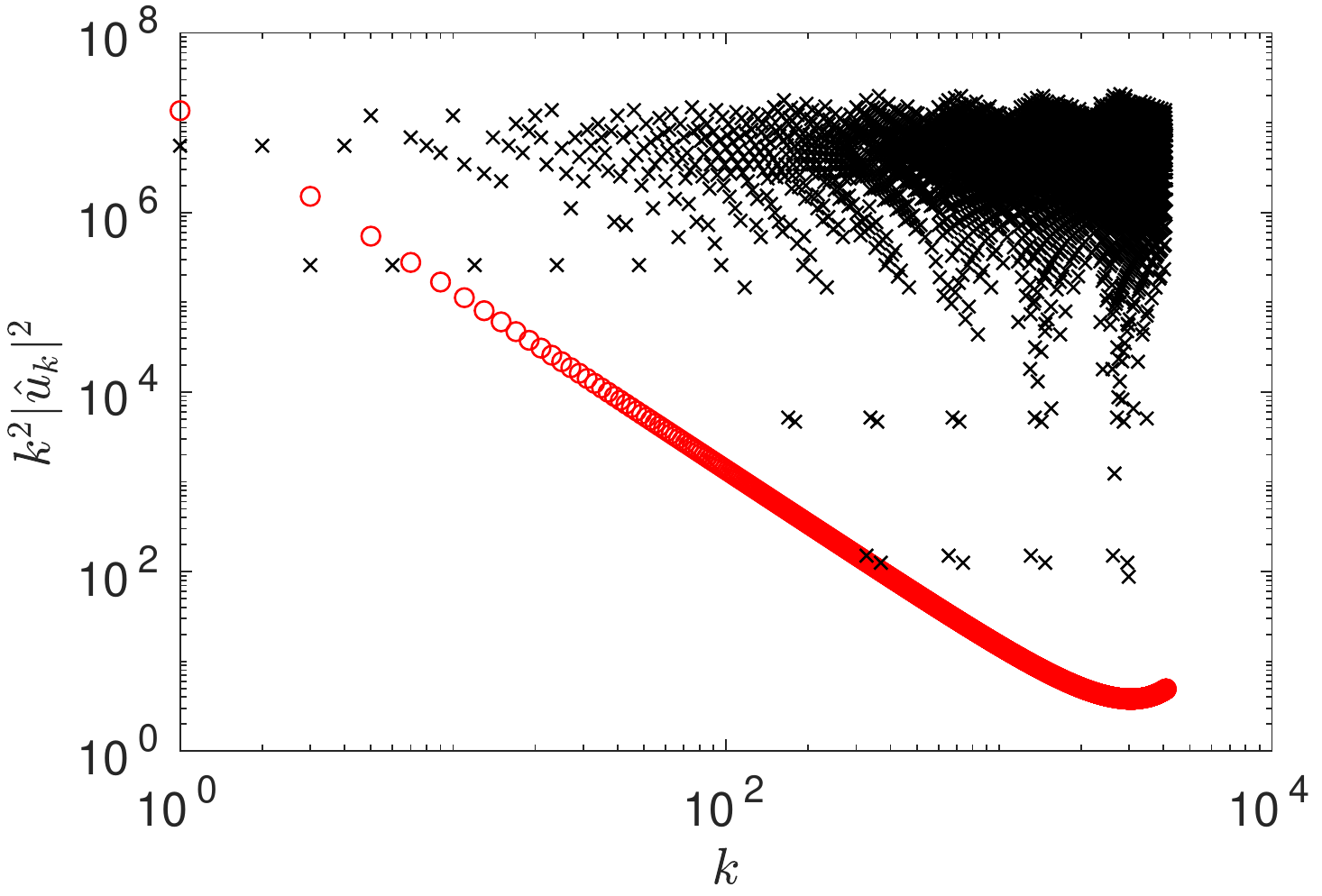}}\qquad
\subfigure[]{\includegraphics[width=0.5\textwidth]{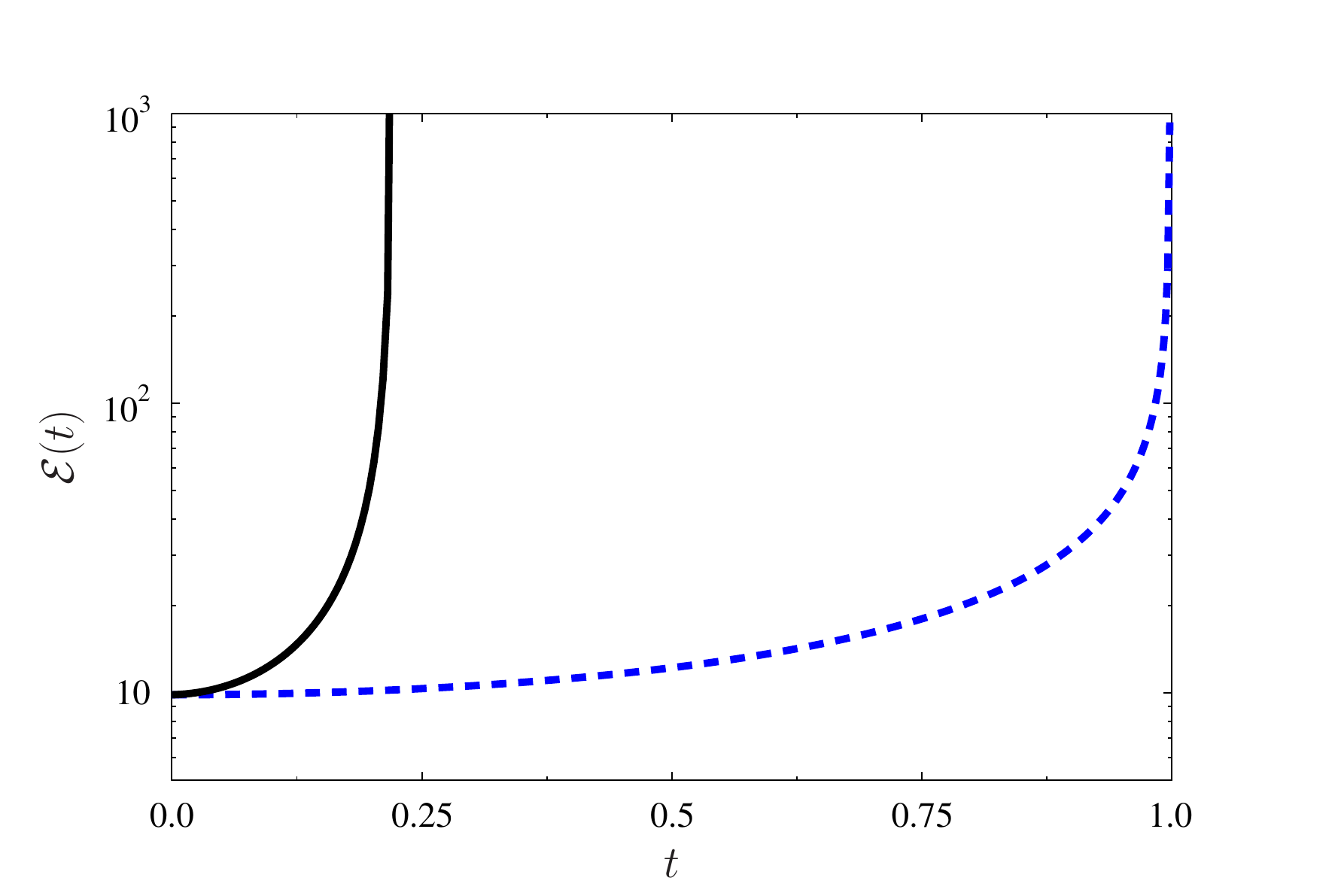}}}
\caption{(a) Compensated Fourier spectra in the solution of the
  inviscid Burgers equation \eqref{eq:Burgers0} with the extreme
  initial condition \eqref{eq:Bu0ext0} at (red circles) the initial
  time $t=0$ and (black crosses) the blow-up time $t^*$. (b) The time
  history of the enstrophy $\E(t)$ in the solution of the inviscid
  Burgers equation \eqref{eq:Burgers0} with (blue dashed line) the
  unimodal initial condition \eqref{eq:Bu0sin} and (black solid line)
  the extreme initial condition \eqref{eq:Bu0ext0}; in both cases the
  constant $A$ is chosen such that $\E(0) = \pi^2$. }
\label{fig:B0sEt}
\end{center}
\end{figure}

The evolution of the viscous Burgers flow with the extreme initial
condition $\tuuE$ shown in the physical space in figure
\ref{fig:But}(b) is qualitatively similar to the inviscid case,
cf.~figure \ref{fig:B0ut}(b), except that a singularity does not form.
This behaviour is reflected in the spectrum of the solution which at
all times $t>0$ decays exponentially for large wavenumbers $k$ (see
figure \ref{fig:BsEt}(a)) and in the fact that the corresponding
enstrophy $\E(t)$ starts to decrease after attaining a maximum at $t =
\tTE$ (see figure \ref{fig:BsEt}(b)). In figure \ref{fig:BsEt}(a) we
also observe that, as expected \citep{bk07}, for intermediate
wavenumbers $k$ the Fourier spectrum scales as $|\widehat{u}_k(t)|
\sim k^{-1}$, $t>0$.  As regards the viscous Burgers flow
corresponding to the generic initial condition shown in the physical
space in figure \ref{fig:But}(a), we observe that while a steep front
is still formed, the maximum attained enstrophy is smaller than in the
flow with the extreme initial condition $\tuuE$ (see figure
\ref{fig:BsEt}(b)).  The ``irregular'' behaviour of the solution
apparent in figure \ref{fig:But}(a) is due to the random form of the
generic initial condition (cf.~\S\,\ref{sec:generic}).

\begin{figure}
\begin{center}
\mbox{\subfigure[]{\includegraphics[width=0.45\textwidth]{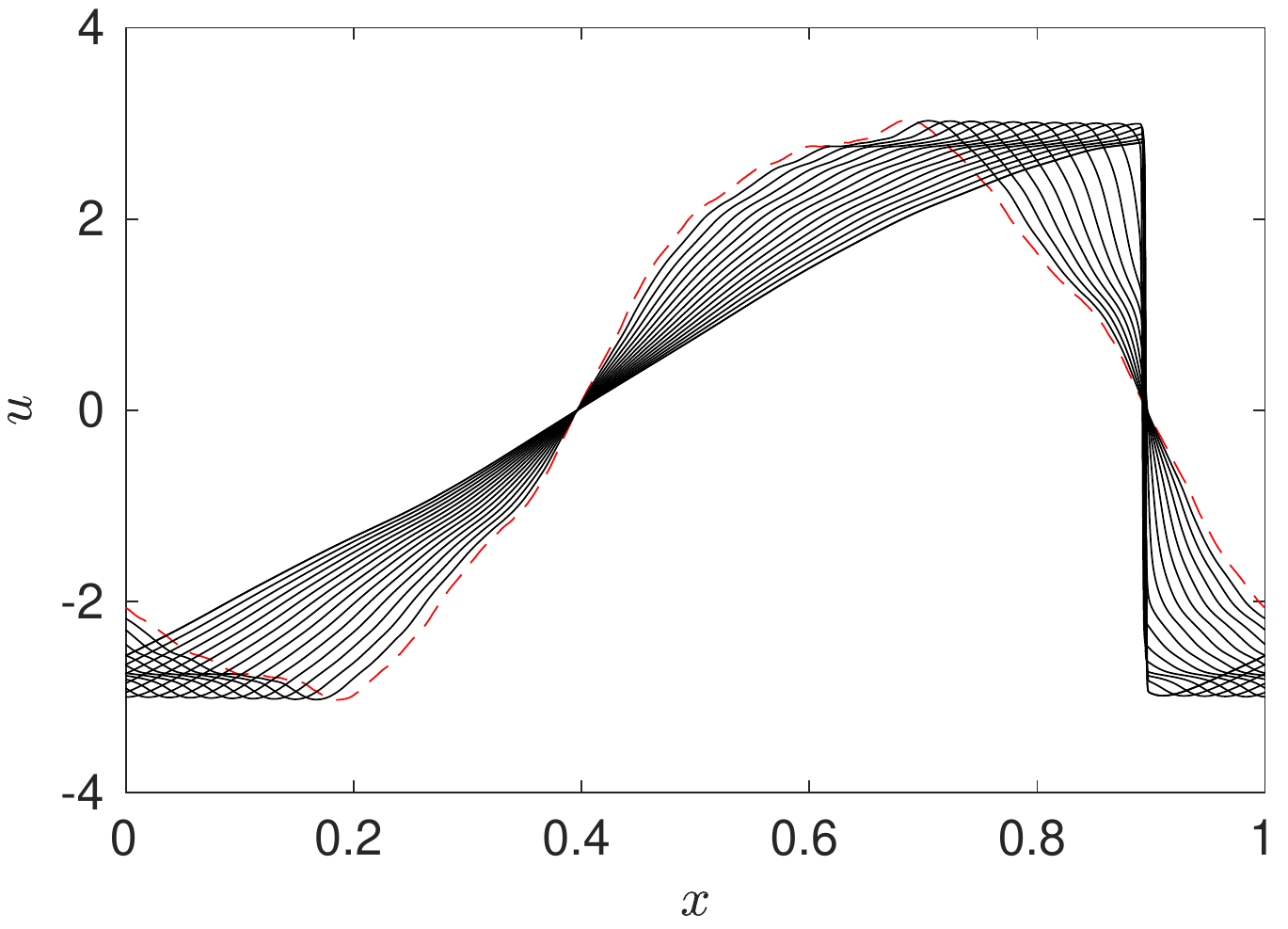}}\qquad
\subfigure[]{\includegraphics[width=0.45\textwidth]{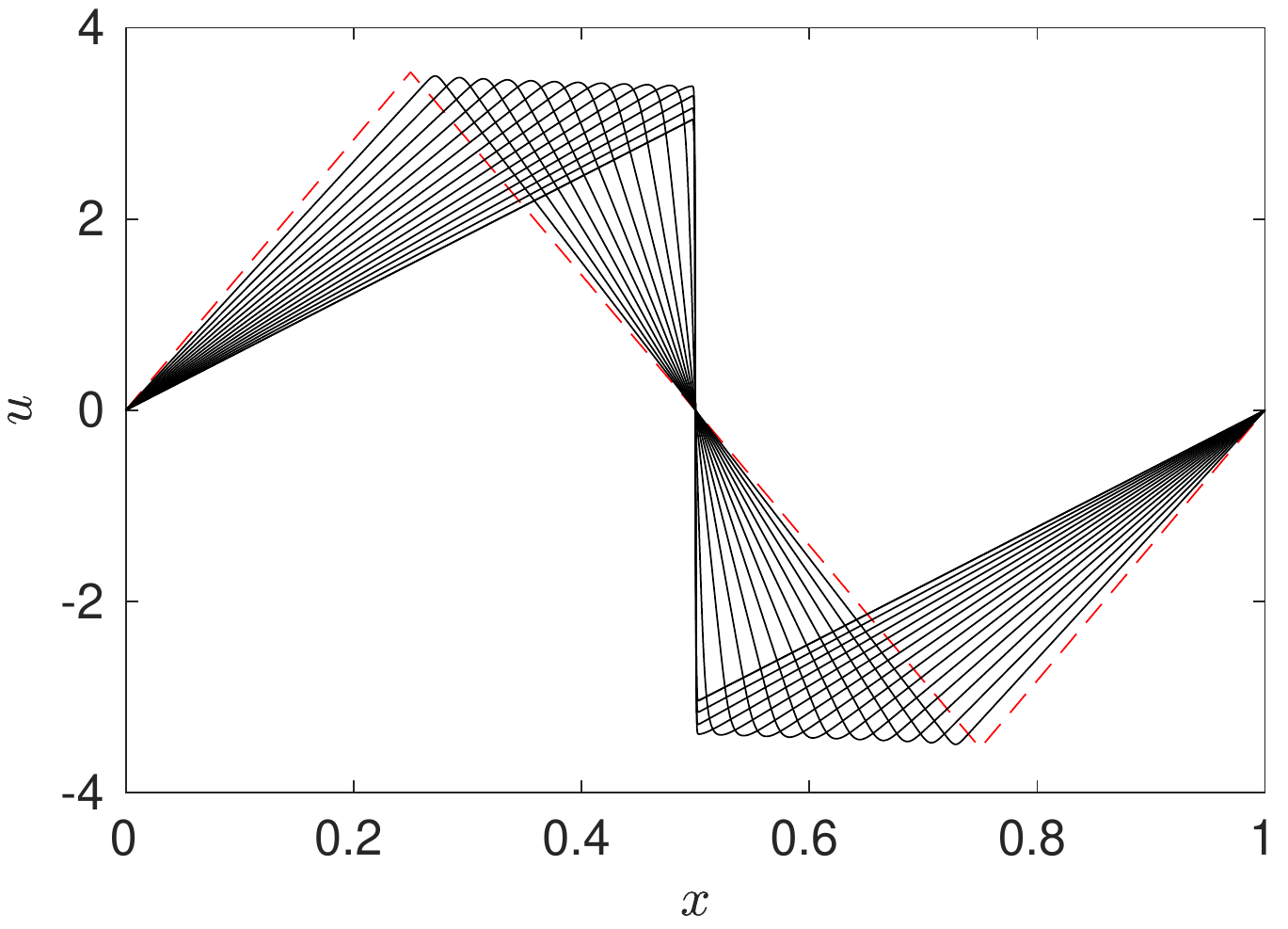}}}
\caption{Solution of the viscous Burgers equation \eqref{eq:Burgers}
  with the (a) generic and (b) extreme initial condition $\tuuE$ for
  increasing times $t \in [0,1.25\, \tTE ]$. The initial conditions
  are marked with red dashed lines and the solutions develop
  progressively steeper fronts as the time increases.}
\label{fig:But}
\end{center}
\end{figure}

\begin{figure}
\begin{center}
\mbox{\subfigure[]{\includegraphics[width=0.45\textwidth]{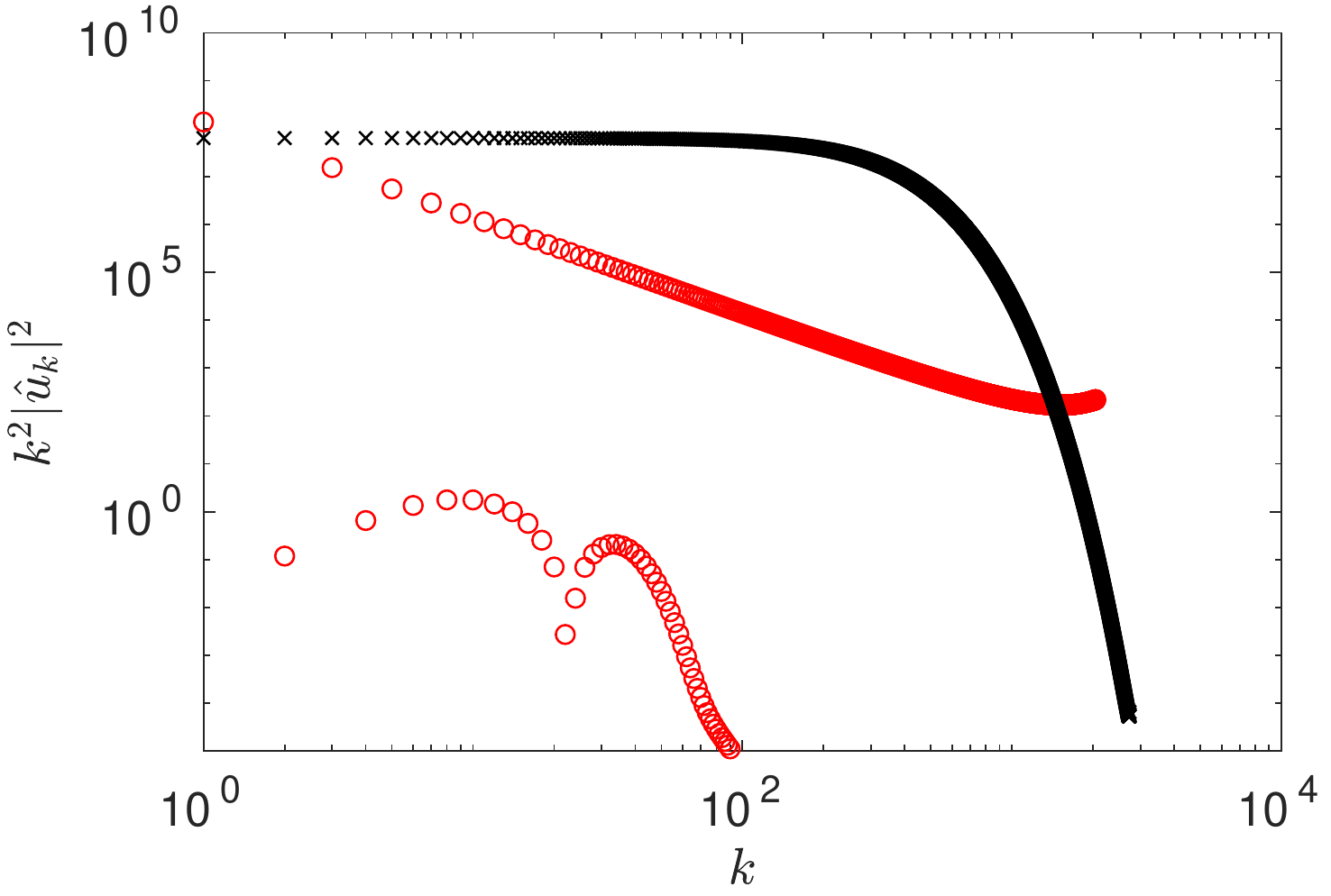}}\qquad
\subfigure[]{\includegraphics[width=0.5\textwidth]{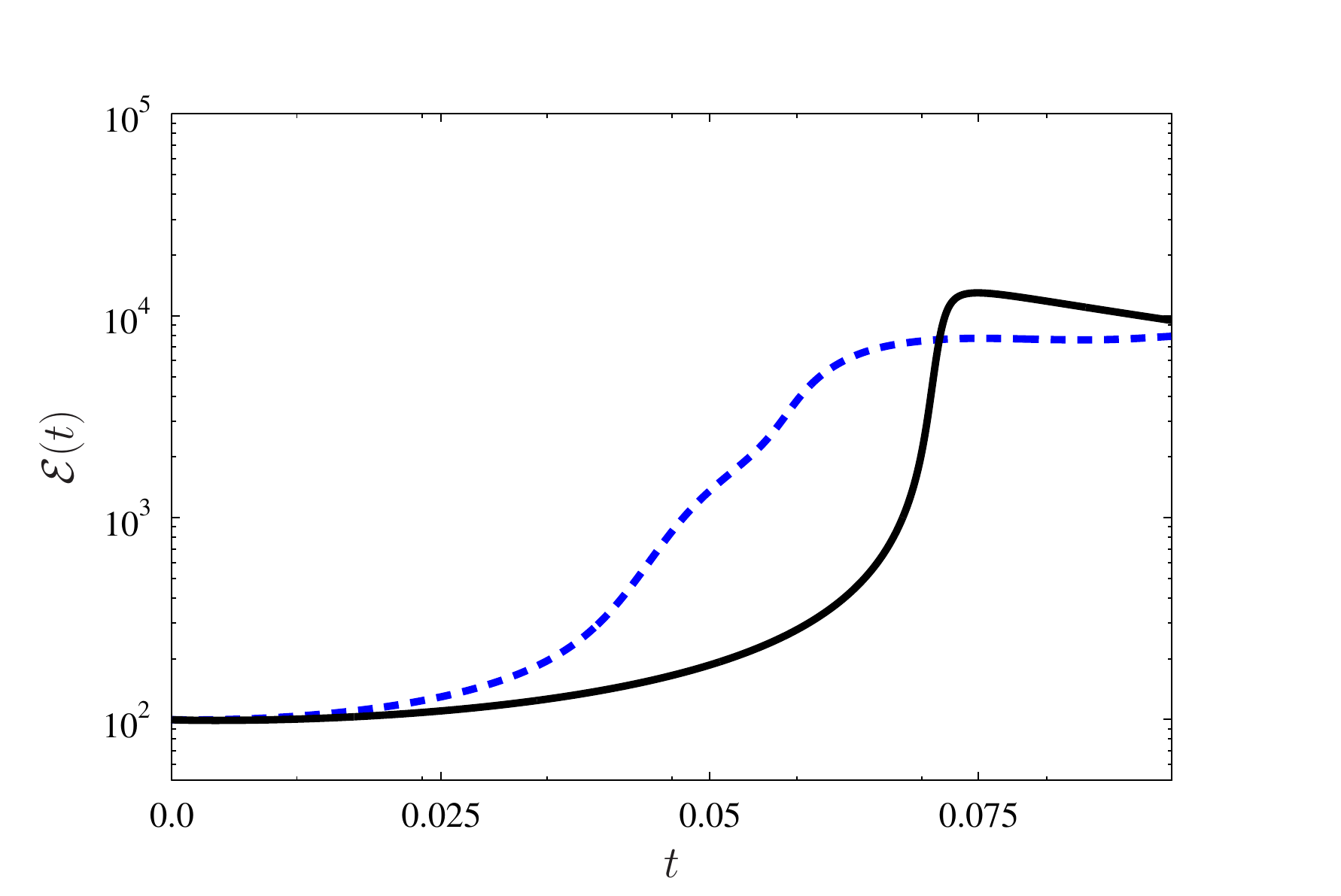}}}
\caption{(a) Compensated Fourier spectra in the solution of the
  viscous Burgers equation \eqref{eq:Burgers} with the extreme initial
  condition $\tuuE$ at (red circles) the initial time $t=0$ and (black
  crosses) the time $\tTE$ when the largest enstrophy is achieved for
  $\E_0 = 100$. (b) The time history of the enstrophy $\E(t)$ in the
  the solution of the viscous Burgers equation \eqref{eq:Burgers} with
  (blue dashed line) the generic and (black solid line) extreme
  initial condition $\tuuE$.}
\label{fig:BsEt}
\end{center}
\end{figure}

We now go on to analyze these flows in terms of the diagnostics
focusing on fluxes and triad interactions introduced in
\S\,\ref{sec:diagn1D}. We begin by noting in figure \ref{fig:B0Pip}
that the solutions of the inviscid Burgers system \eqref{eq:Burgers0}
with the unimodal and extreme initial conditions \eqref{eq:Bu0sin} and
\eqref{eq:Bu0ext0} exhibit only forward energy cascade with
non-negative flux $\Pi(t,k) \ge 0$ at all times $t$ and for all
wavenumbers $k$. In the flow corresponding to the unimodal initial
data \eqref{eq:Bu0sin} the flux $\Pi(t,k)$ is essentially independent
of the wavenumber $k$ and rapidly increases from a small initial value
to a nearly constant level. The small initial values of the flux and
its independence of the wavenumber can be attributed to the fact that
the unimodal initial condition \eqref{eq:Bu0sin} contains a single
Fourier component only. On the other hand, for the extreme initial
data \eqref{eq:Bu0ext0} the flux is an increasing function of the
wavenumber $k$ at early times $t$. Moreover, in this case the flux
$\Pi(t,k)$ increases more rapidly with time and reaches larger values
as the blow-up time $t^*$ is approached.

The PDFs of triad phases $\varphi_{k_1,k_2}^{k_3}$ in inviscid Burgers
flows are shown in figure \ref{fig:B0pdf}.  We notice that in the
flows corresponding to both the unimodal and extreme initial
conditions \eqref{eq:Bu0sin} and \eqref{eq:Bu0ext0} a significant
fraction of these phases is aligned at the angle $-\pi/2$, which in
the light of formula \eqref{eq:flux_sin} indicates that these triads
may provide negative contributions to the flux. In fact, in the
solution corresponding to the extreme initial condition
\eqref{eq:Bu0ext0}, alignment at $\pm \pi /2$ is close to being
equally probable for times $t \ll t^*$, but all triad phases collapse
at $\pi/2$ as $t \rightarrow t^*$, cf.~figure \ref{fig:B0pdf}(a).
However, the weighted PDF presented in figure \ref{fig:B0Wpdf}(a)
indicates that the triads with phases aligned at $-\pi / 2$ do not in
fact participate in the energy transfer. The reason is that these
triads involve modes with vanishing Fourier coefficients, such that
their contributions to flux \eqref{eq:flux_sin} is effectively zero.
On the other hand, the weighted PDF for the flow with the extreme
initial condition \eqref{eq:Bu0ext0} provides evidence for inverse
energy transfer occurring at the beginning of the flow evolution.

\begin{figure}
\begin{center}
\mbox{\subfigure[]{\includegraphics[width=0.45\textwidth]{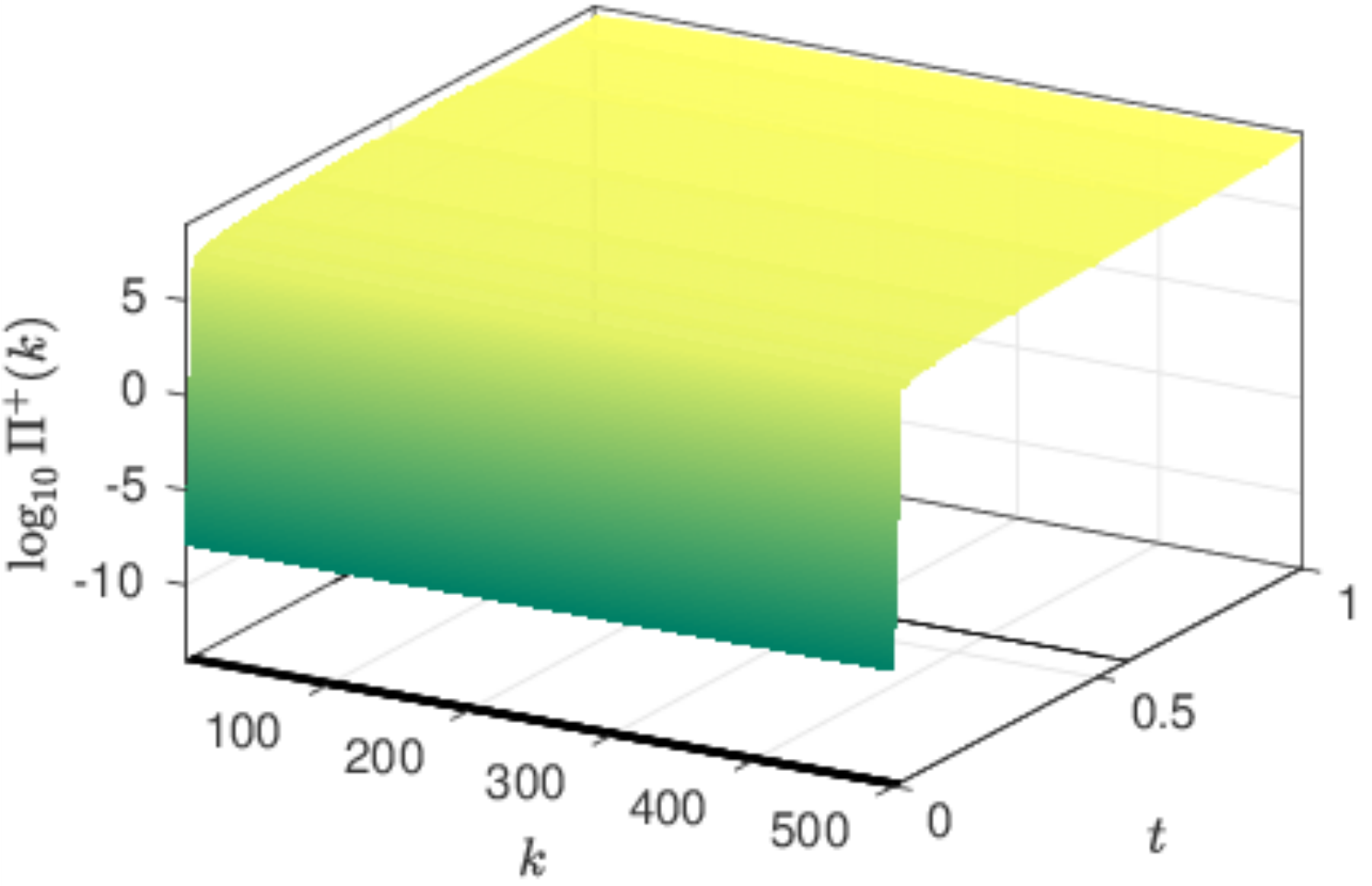}}\qquad
\subfigure[]{\includegraphics[width=0.45\textwidth]{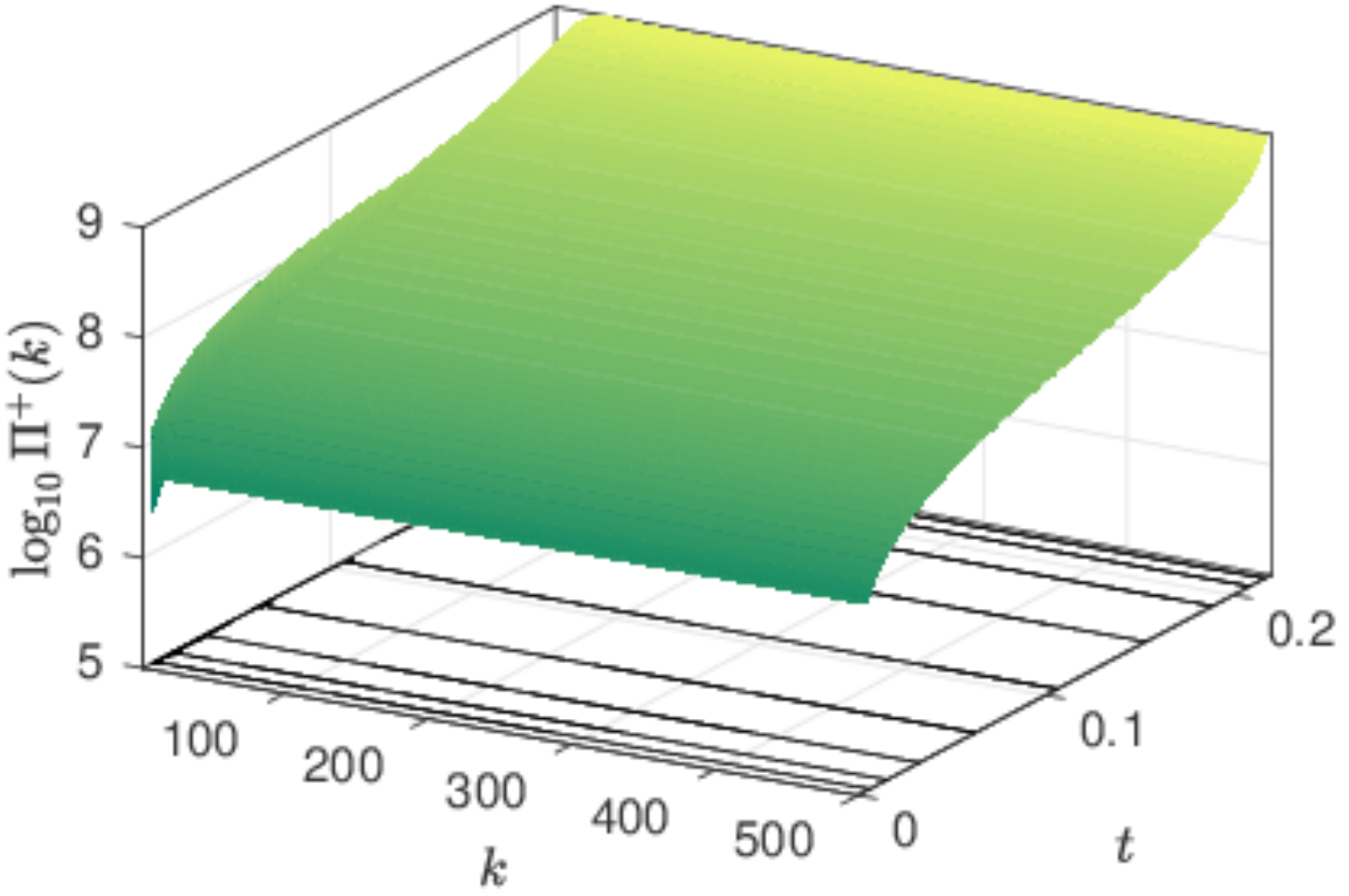}}}
\caption{Flux $\Pi(t,k)$ as function of the time $t \in [0,t^*]$ and the
  wavenumber $k$ in the solution of the inviscid Burgers equation
  \eqref{eq:Burgers0} with (a) the unimodal initial condition
  \eqref{eq:Bu0sin} and (b) the extreme initial condition
  \eqref{eq:Bu0ext0}. }
\label{fig:B0Pip}

\mbox{\subfigure[]{\includegraphics[width=0.45\textwidth]{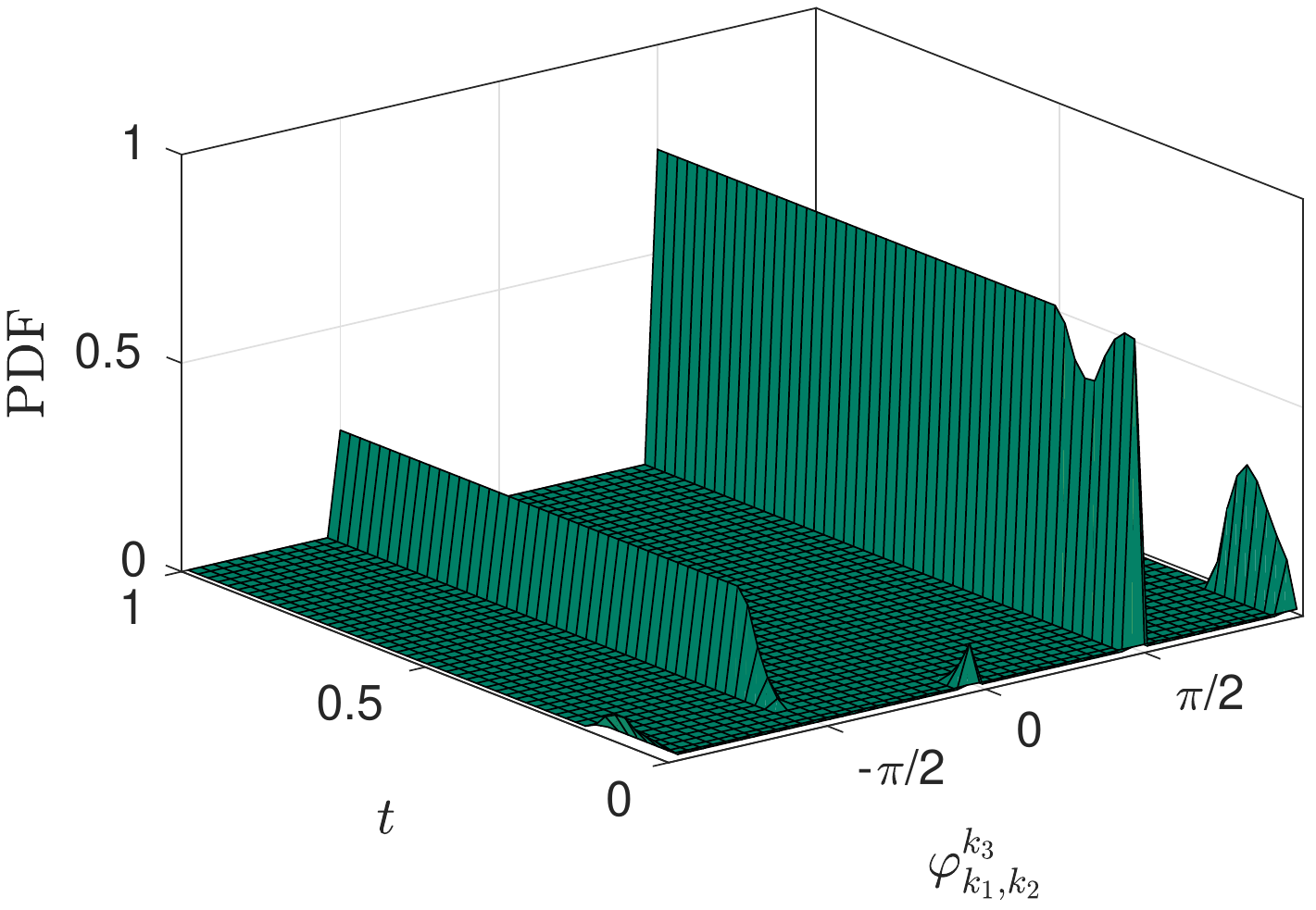}}\qquad
\subfigure[]{\includegraphics[width=0.45\textwidth]{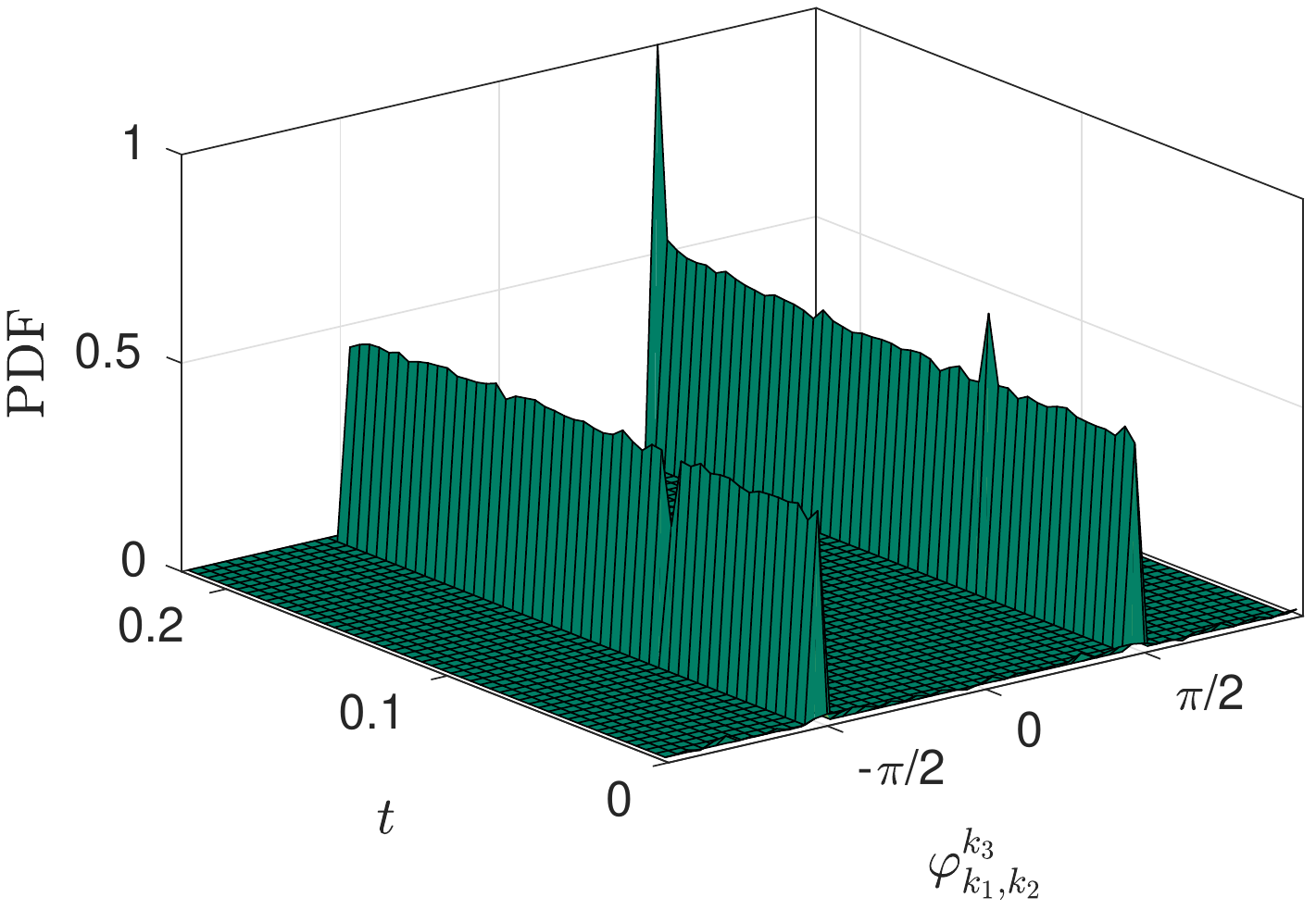}}}
\caption{Time evolution of the PDF of the triad phase
  $\varphi_{k_1,k_2}^{k_3}$ in the solution of the inviscid Burgers
  equation \eqref{eq:Burgers0} with (a) the unimodal initial condition
  \eqref{eq:Bu0sin} and (b) the extreme initial condition
  \eqref{eq:Bu0ext0} for  $t \in [0,t^*]$.}
\label{fig:B0pdf}
\mbox{\subfigure[]{\includegraphics[width=0.45\textwidth]{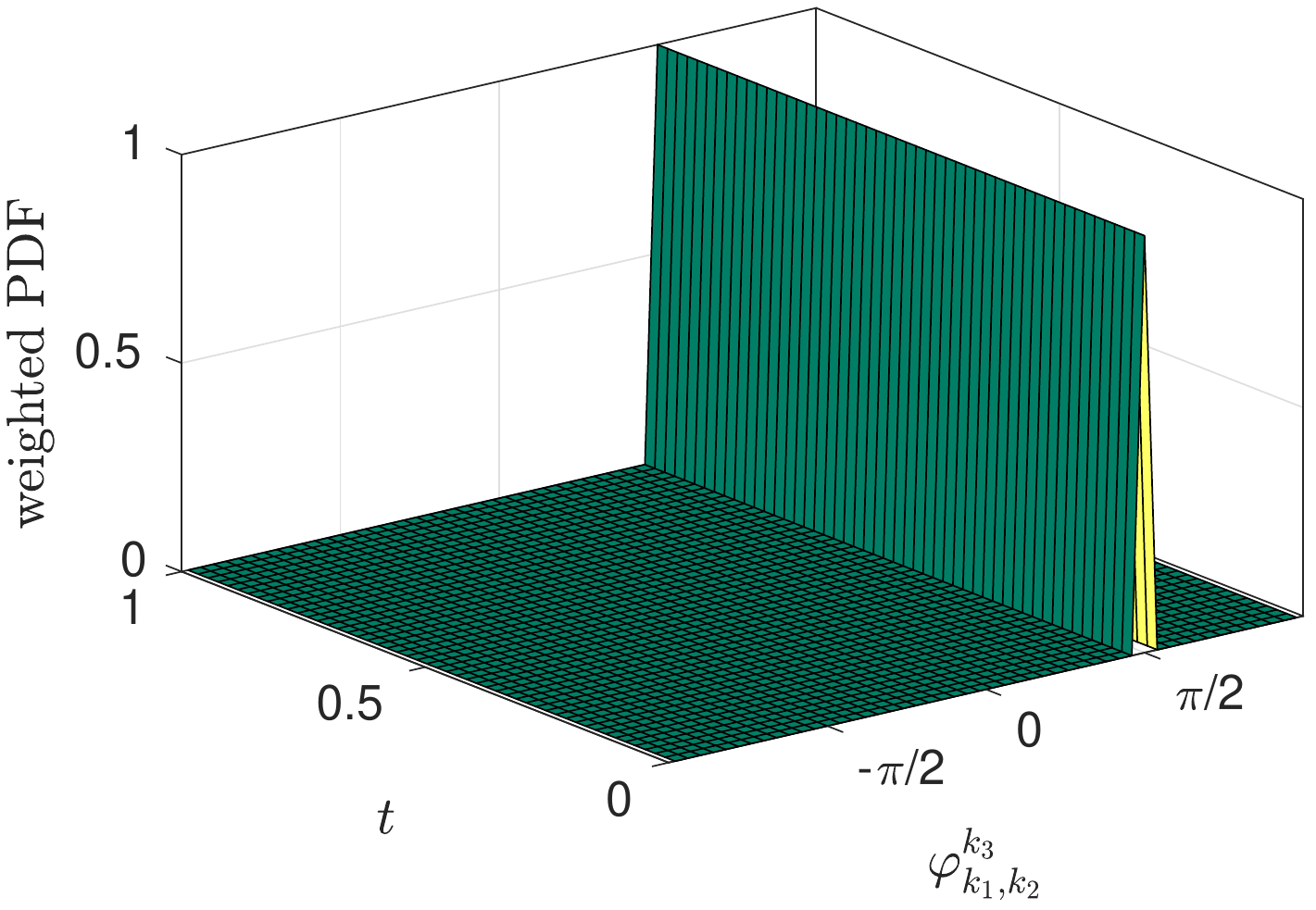}}\qquad
\subfigure[]{\includegraphics[width=0.45\textwidth]{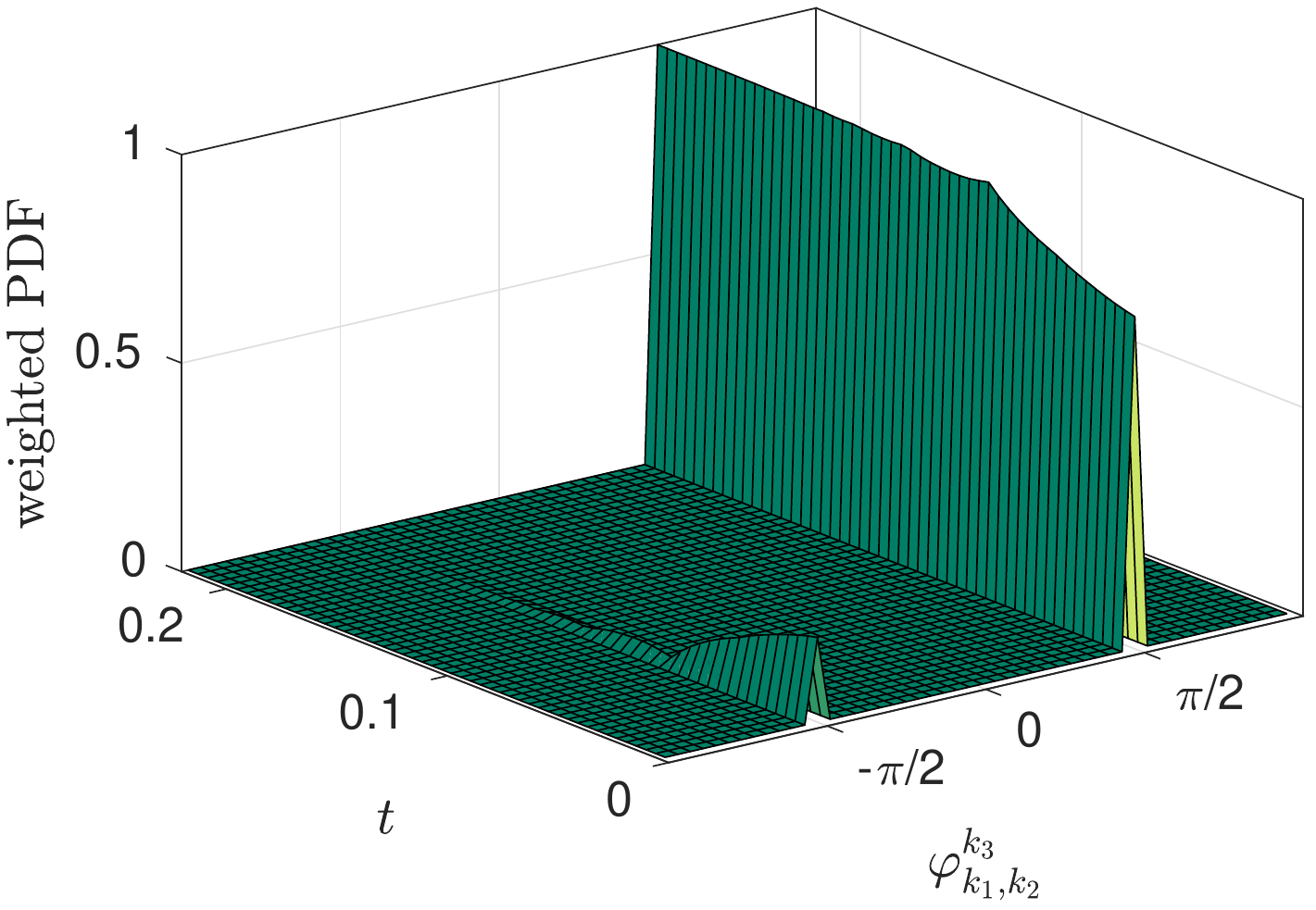}}}
\caption{Time evolution of the weighted PDF of the triad phase
  $\varphi_{k_1,k_2}^{k_3}$ in the solution of the inviscid Burgers
  equation \eqref{eq:Burgers0} with (a) the unimodal initial condition
  \eqref{eq:Bu0sin} and (b) the extreme initial condition
  \eqref{eq:Bu0ext0} for  $t \in [0,t^*]$.}
\label{fig:B0Wpdf}
\end{center}
\end{figure}

% To reduce the file size ...
% convert case10_Pip.pdf -compress Zip case10_Pip2.pdf
 
As regards solutions of the viscous Burgers system \eqref{eq:Burgers}
with both the generic and extreme initial conditions, in figure
\ref{fig:BPip} we note that the positive flux $\Pi^+(t,k)$ exhibits
nontrivial behaviour. It becomes less intense at intermediate times
and increases around the time $\tTE$ when the maximum enstrophy is
reached in both cases, at which point it also becomes almost
independent of the wavenumber $k$. {In the flow corresponding to
  the extreme initial condition the increase of the flux is more rapid
  which is also reflected in the steeper growth of enstrophy around
  that time, cf.~figure \ref{fig:BsEt}(b).}  Interestingly, Figure
\ref{fig:BPin} shows evidence of negative flux $\Pi^-(t,k)$ for small
wavenumbers $k$ and at early times $t$ which is stronger in the flow
with the extreme initial condition $\tuuE$. The time evolutions of the
PDFs and the weighted PDFs of the triad phase angle
$\varphi_{k_1,k_2}^{k_3}$ in the solutions of the viscous Burgers
system with the two initial conditions are overall similar to the
inviscid cases, cf.~figures \ref{fig:B0pdf} and \ref{fig:B0Wpdf}
versus figures \ref{fig:Bpdf} and \ref{fig:BWpdf}.  One noteworthy
difference is the virtual absence of triads with phases aligned at
$-\pi / 2$ in the viscous Burgers flow corresponding to the generic
initial condition, cf.~figure \ref{fig:Bpdf}(a). Instead, triads
initially exhibit a broad range of phase angles before suddenly
aligning at $\pi / 2$ which is accompanied by the emergence of a sharp
front evident in figure \ref{fig:But}(a). In addition, it is also
interesting to note that in the solution of the viscous problem with
the extreme initial condition $\tuuE$ all triad phases remain aligned
at $\pi/2$ for $t > \tTE$, i.e., after the enstrophy has reached its
maximum, cf.~figure \ref{fig:Bpdf}(b).  This mirrors the behaviour in
the inviscid Burgers equation with the extreme initial condition,
cf.~figure \ref{fig:B0pdf}(b).

\begin{figure}
\begin{center}
\mbox{\subfigure[]{\includegraphics[width=0.45\textwidth]{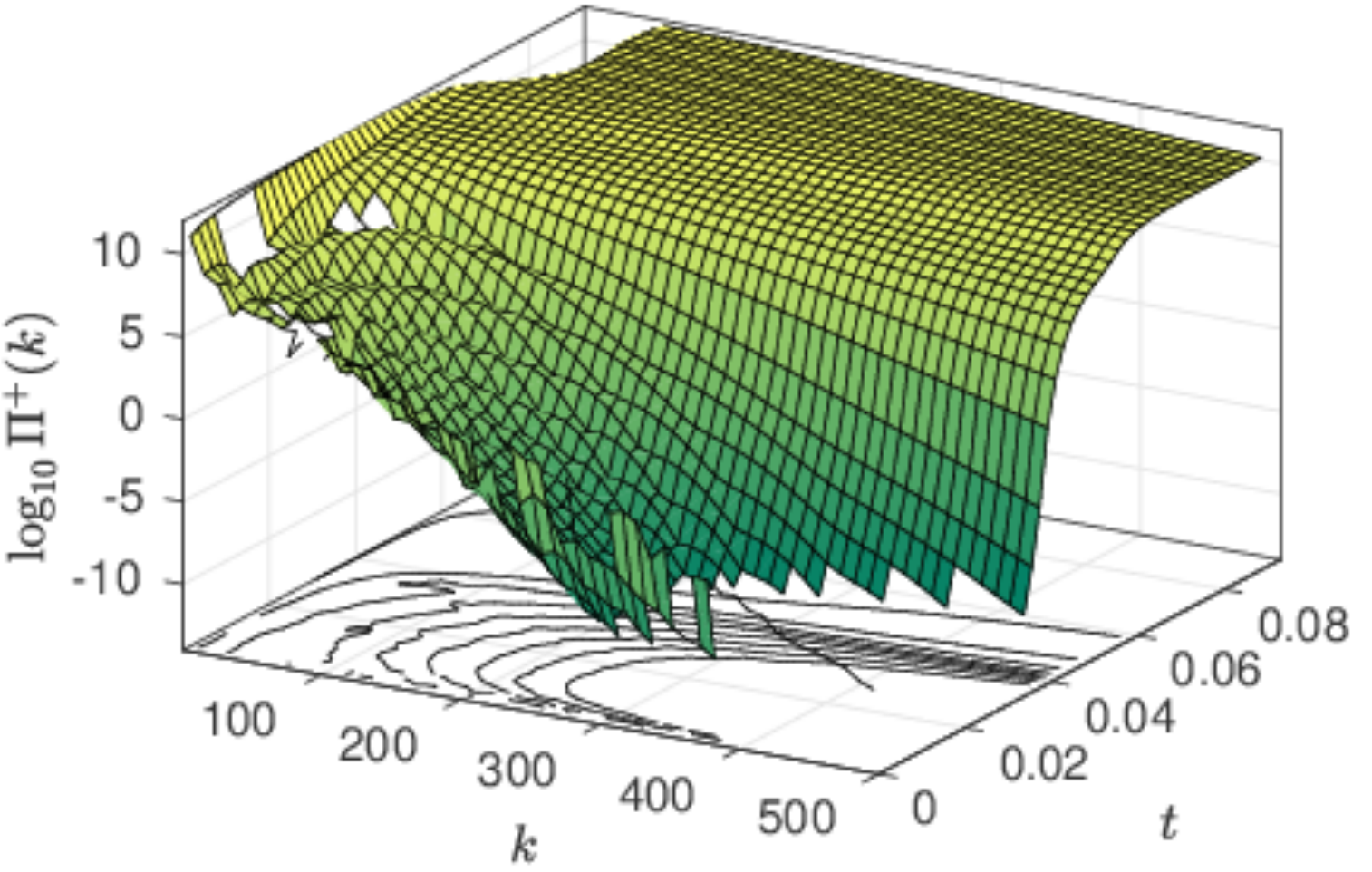}}\qquad
\subfigure[]{\includegraphics[width=0.45\textwidth]{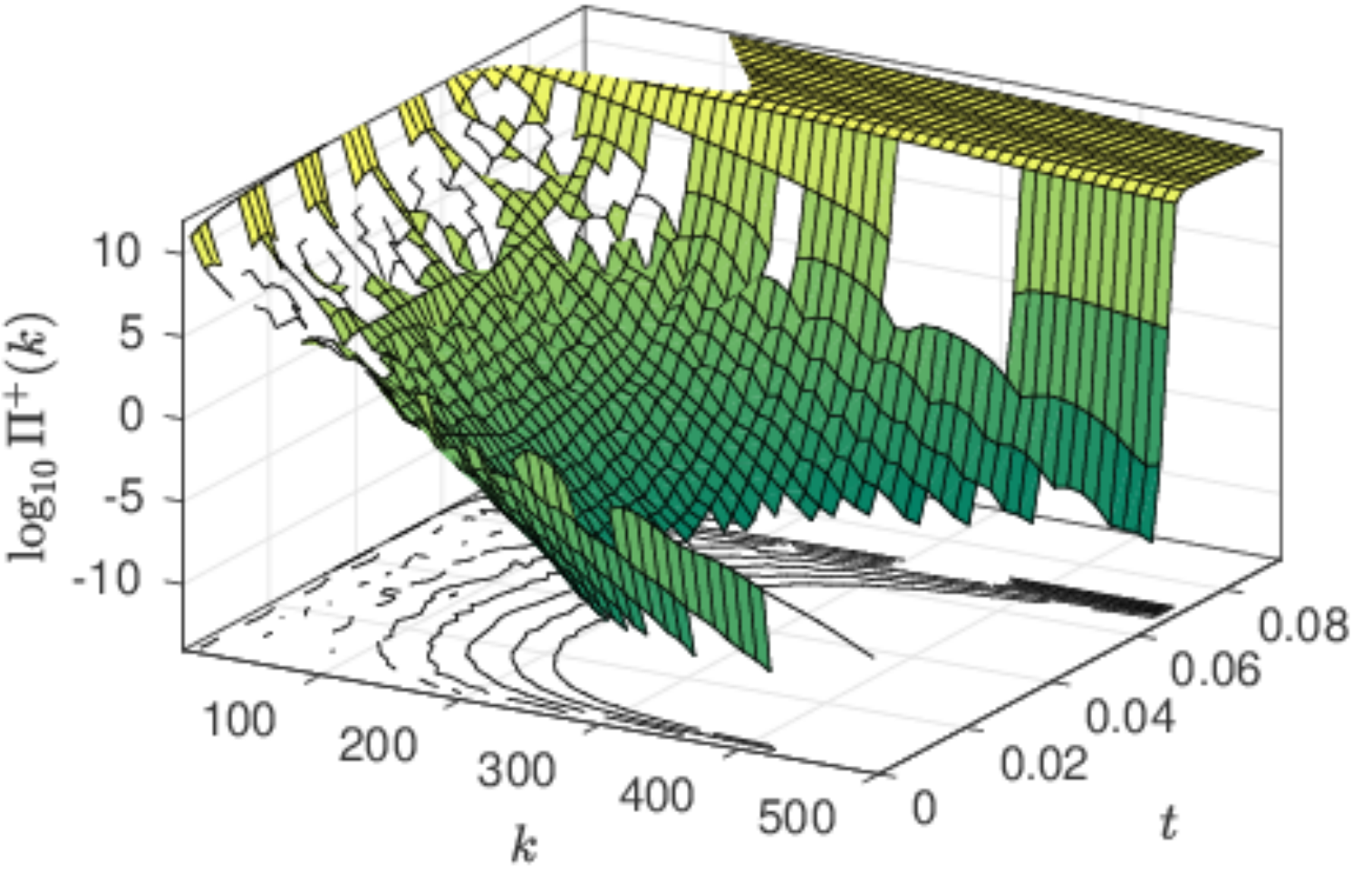}}}
\caption{Positive flux $\Pi^+(t,k)$ as function of the time $t \in [0,
  1.25 \, \tTE]$ and the wavenumber $k$ in the solution of the viscous
  Burgers equation \eqref{eq:Burgers} with the (a) generic and (b)
  extreme initial condition $\tuuE$. }
\label{fig:BPip}
\mbox{\subfigure[]{\includegraphics[width=0.45\textwidth]{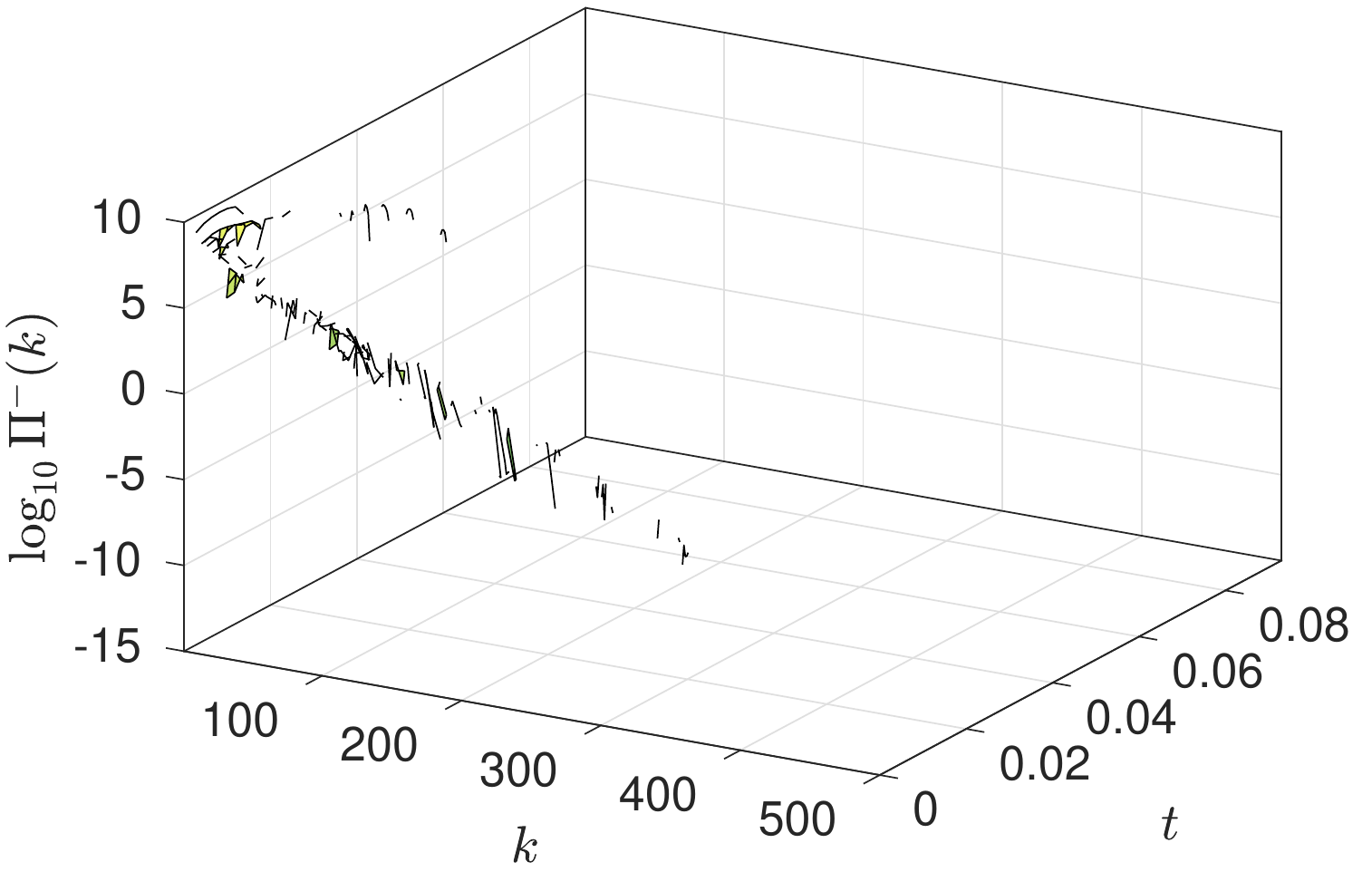}}\qquad
\subfigure[]{\includegraphics[width=0.45\textwidth]{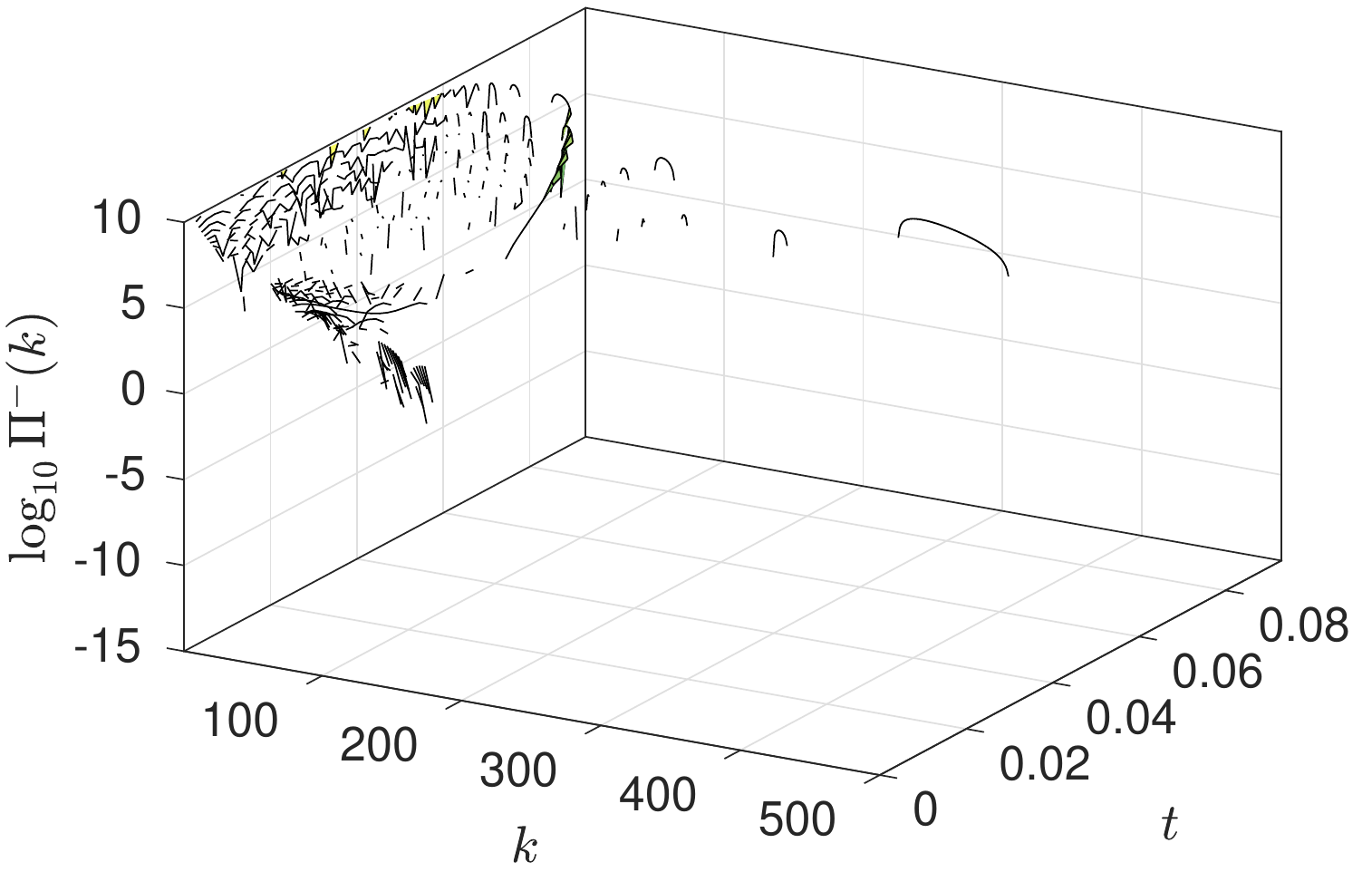}}}
\caption{Negative flux $\Pi^-(t,k)$ as function of the time $t \in [0,
  1.25 \, \tTE]$ and the wavenumber $k$ in the solution of the viscous
  Burgers equation \eqref{eq:Burgers} with the (a) generic and (b)
  extreme initial condition $\tuuE$. }
\label{fig:BPin}
\end{center}
\end{figure}

\begin{figure}
\begin{center}
\mbox{\subfigure[]{\includegraphics[width=0.45\textwidth]{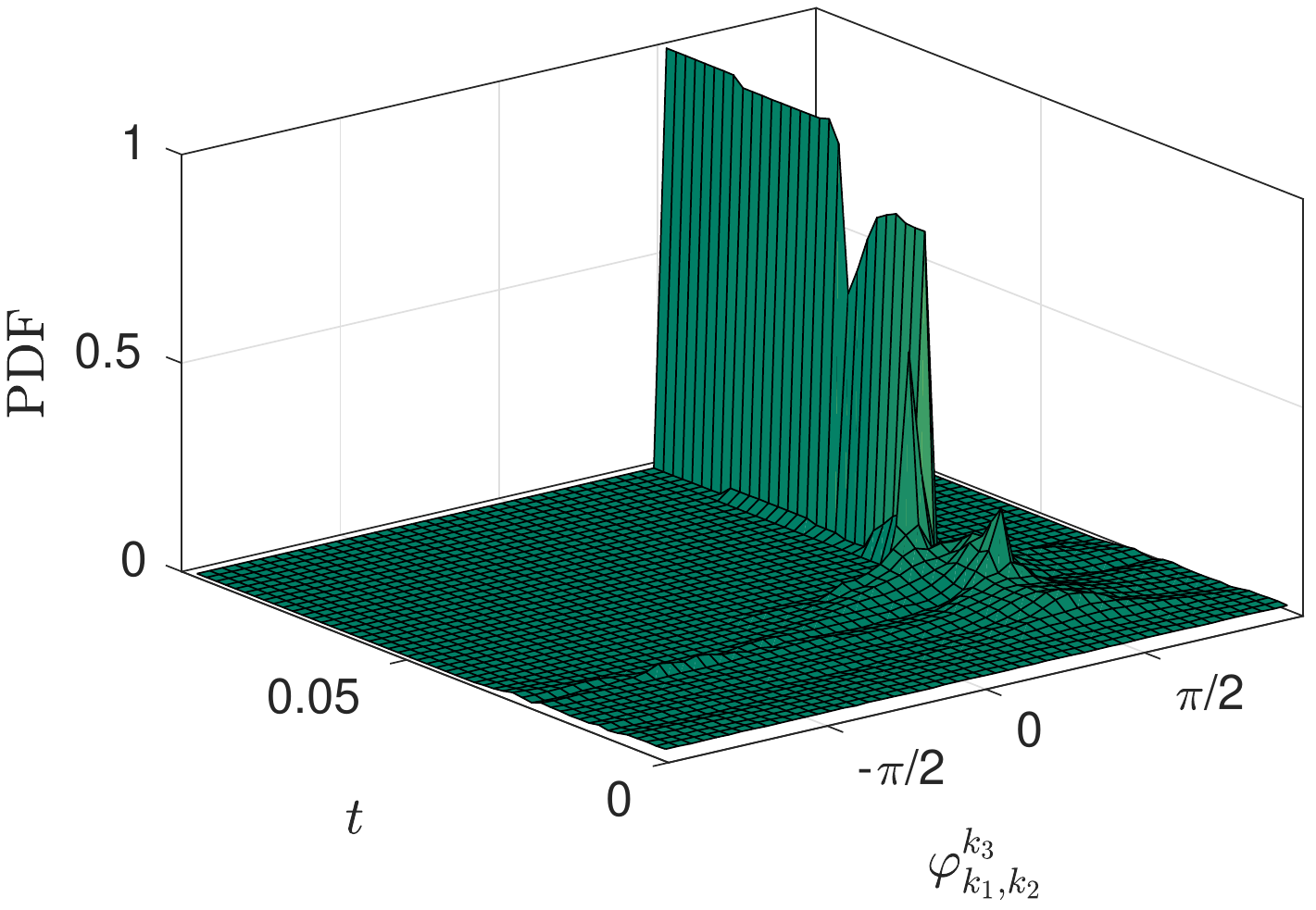}}\qquad
\subfigure[]{\includegraphics[width=0.45\textwidth]{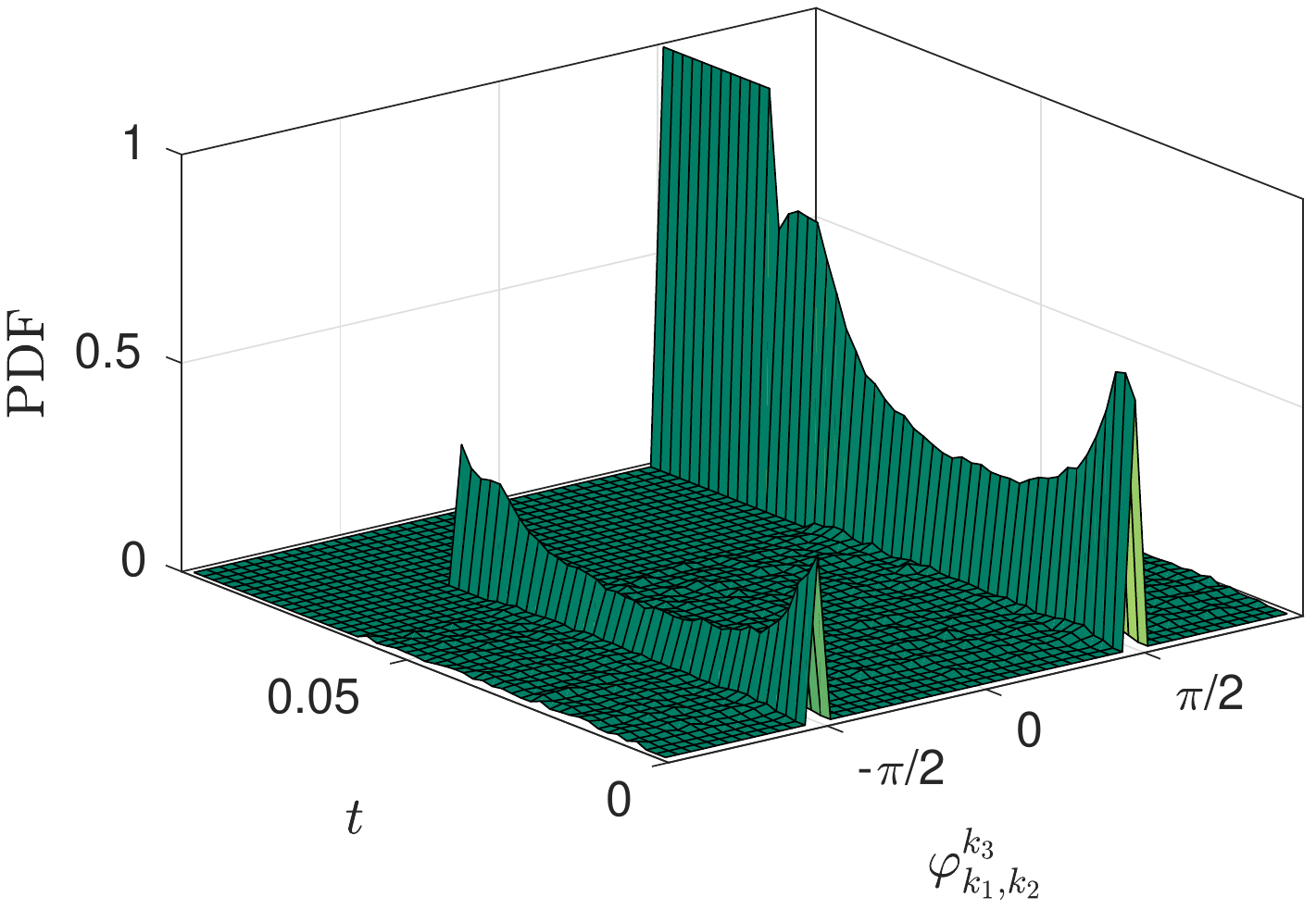}}}
\caption{Time evolution of the PDF of the triad phase
  $\varphi_{k_1,k_2}^{k_3}$ in the solution of the viscous Burgers
  equation \eqref{eq:Burgers} with the (a) generic and (b)
  extreme initial condition $\tuuE$.}
\label{fig:Bpdf}
\mbox{\subfigure[]{\includegraphics[width=0.45\textwidth]{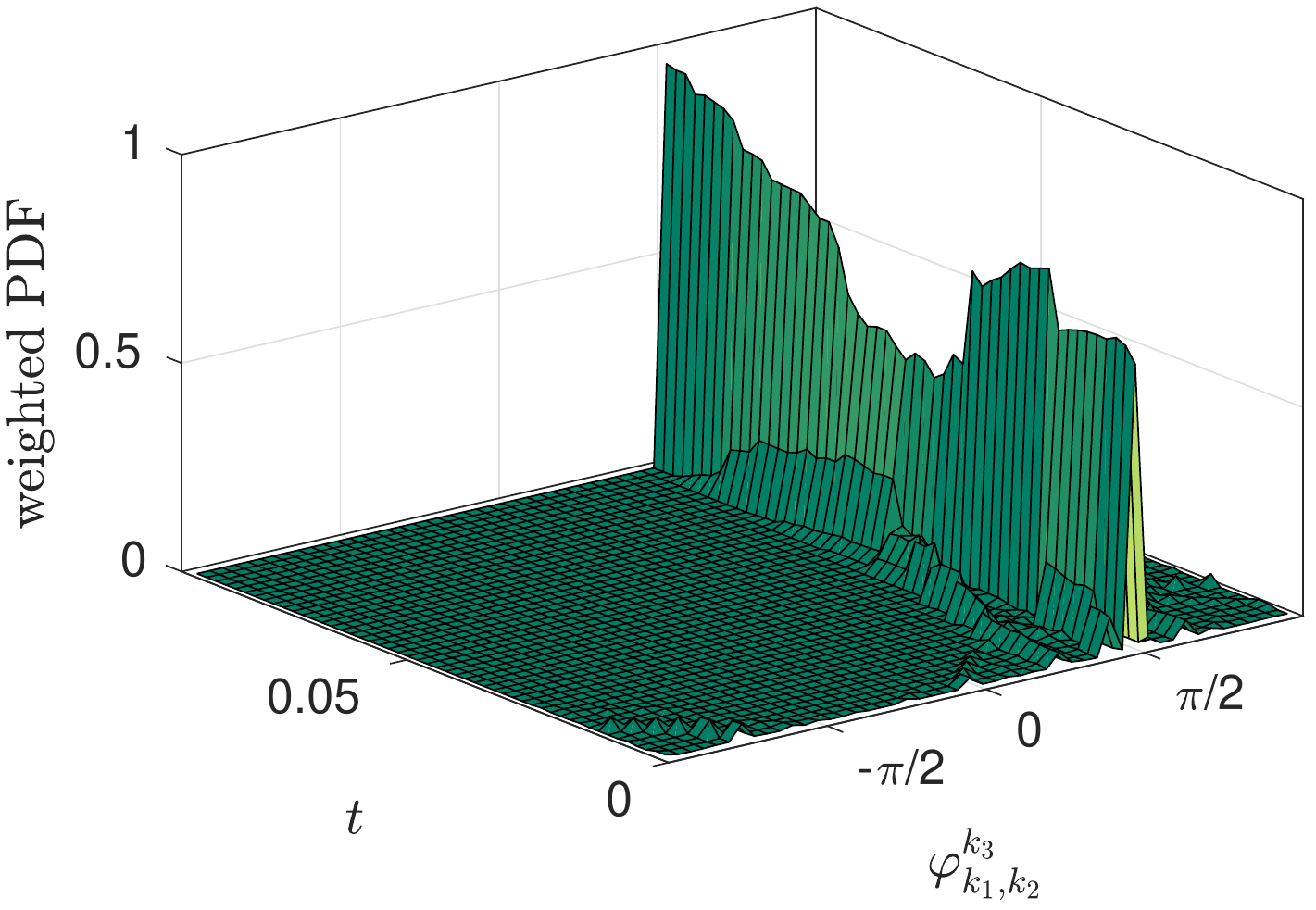}}\qquad
\subfigure[]{\includegraphics[width=0.45\textwidth]{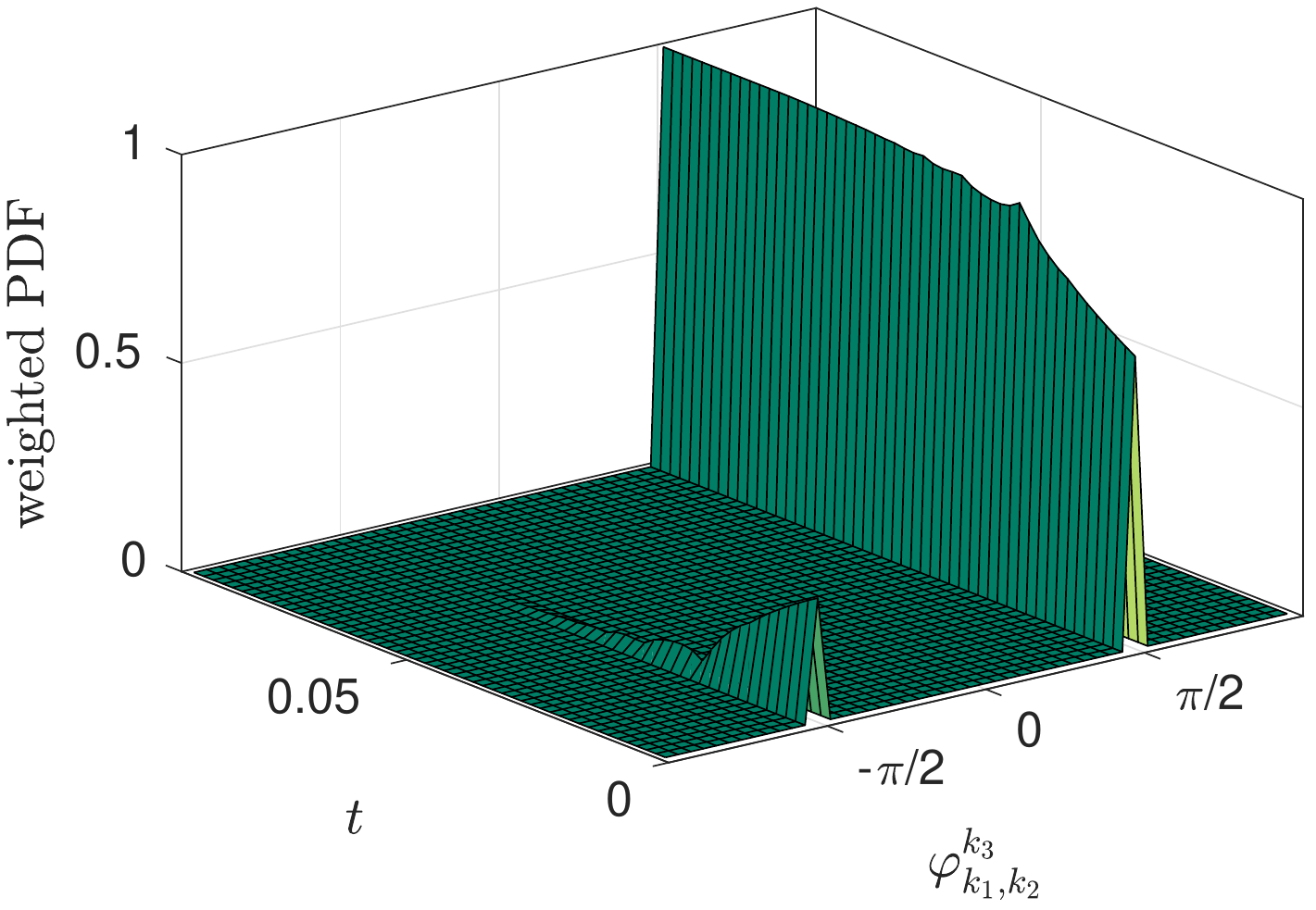}}}
\caption{Time evolution of the weighted PDF of the triad phase
  $\varphi_{k_1,k_2}^{k_3}$ in the solution of the viscous Burgers
  equation \eqref{eq:Burgers} with the (a) generic and (b) extreme
  initial condition $\tuuE$.}
\label{fig:BWpdf}
\end{center}
\end{figure}

\begin{figure}
\begin{center}
\mbox{\subfigure[$t_1 = 0.044$]{\includegraphics[width=0.5\textwidth]{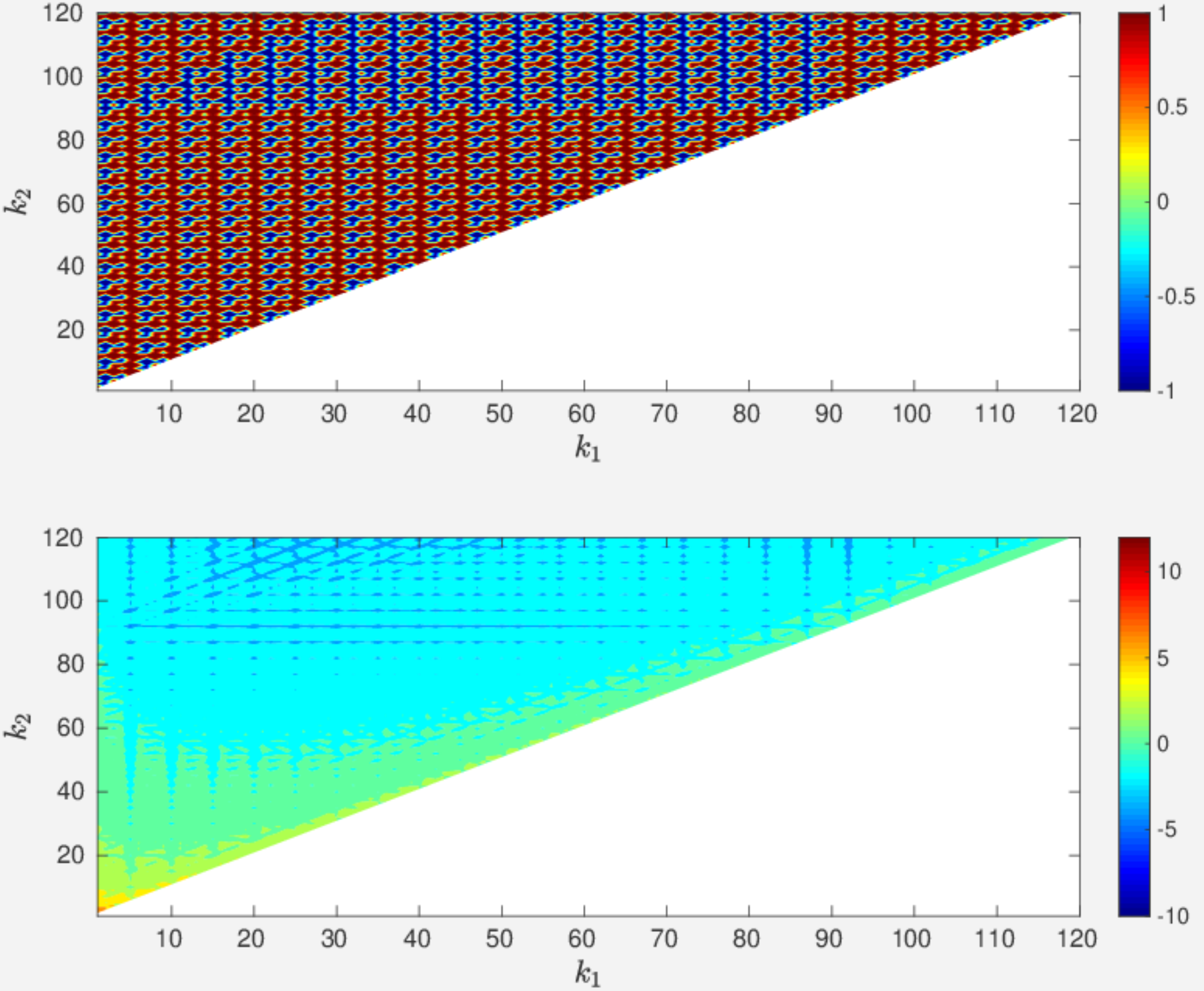}}\quad
\subfigure[$t_2 = 0.176$]{\includegraphics[width=0.5\textwidth]{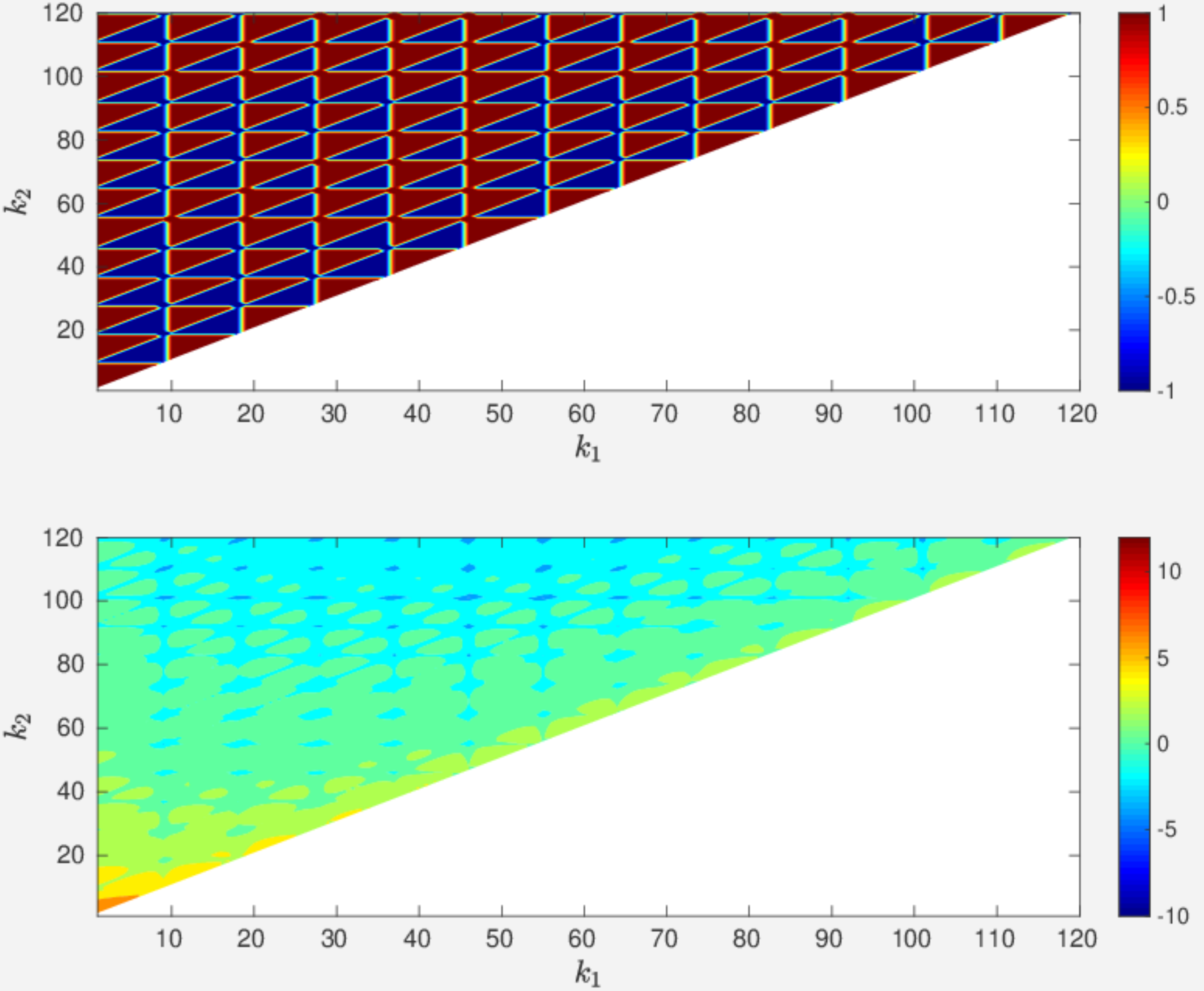}}} \\
\bigskip\bigskip
\mbox{\subfigure[$t_3 = 0.211$]{\includegraphics[width=0.5\textwidth]{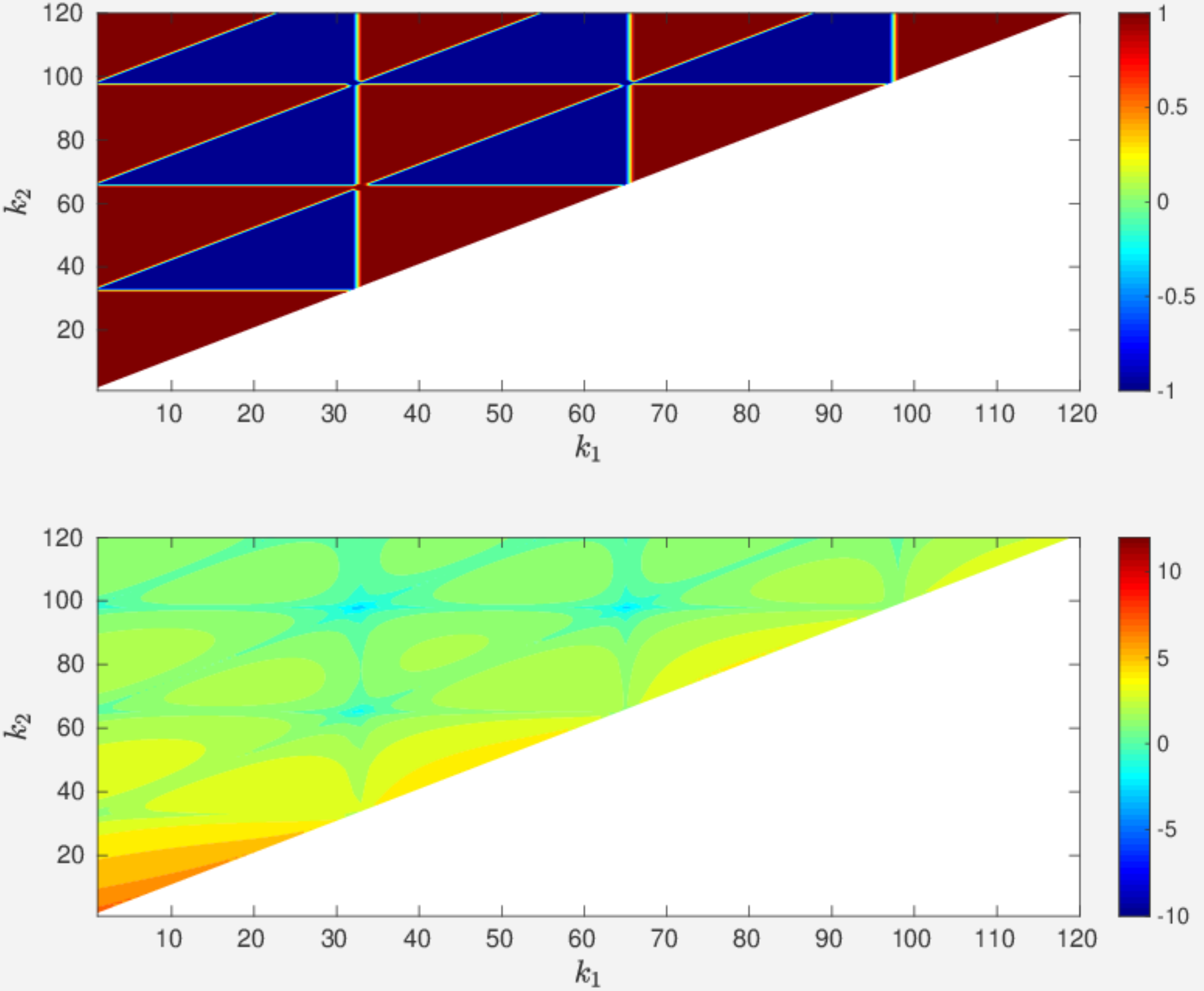}}\quad
\subfigure[$t_4 = 0.224$]{\includegraphics[width=0.5\textwidth]{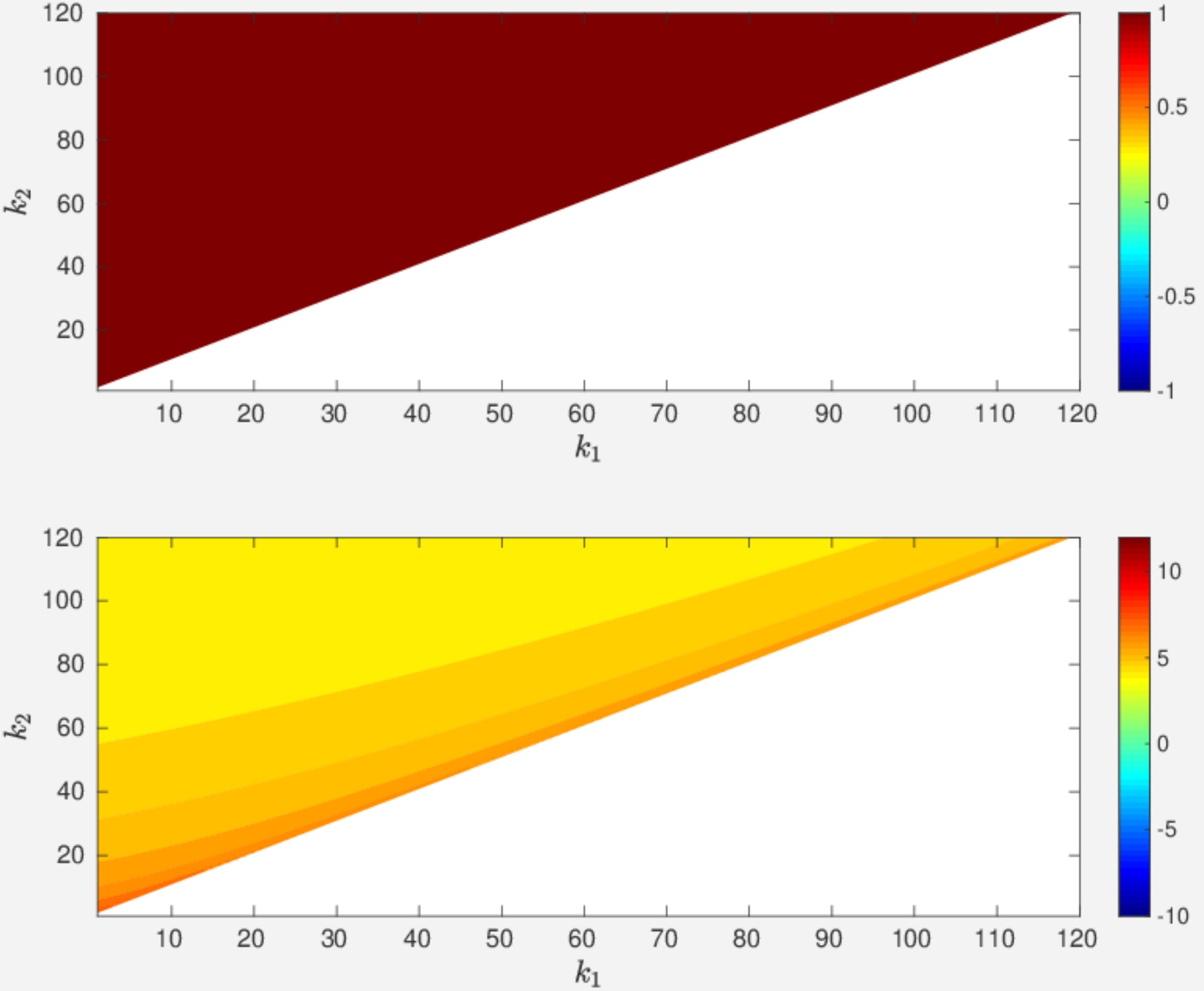}}}
\caption{Dependence of (top panels)
  $\sin\left(\varphi_{k_1,k_2}^{k_3}\right)$ and (bottom panels)
  $\log_{10} \left(4 \pi k_1 \, {a}_{k_1} {a}_{k_2} {a}_{k_3} \right)$
  on the wavenumbers $(k_1,k_2)$ at the indicated times $t_i \in
  (0,t^*)$, $i=1,2,3,4$, during the evolution of the inviscid Burgers
  flow with the extreme initial condition \eqref{eq:Bu0ext0}. An
  animated version of this figure is available
  \href{https://www.youtube.com/watch?v=wcSKBZhpHMc}{on-line}.}
\label{fig:tri_fl}
\end{center}
\end{figure}

In order to provide additional {interesting} insights about the
triadic interactions in the inviscid Burgers flow with the extreme
initial condition \eqref{eq:Bu0ext0}, in figure \ref{fig:tri_fl} we
show the quantity $\sin\left(\varphi_{k_1,k_2}^{k_3}\right)$ and the
logarithm of the corresponding amplitude-dependent factor $\log_{10}
\left(4 \pi k_1 \, {a}_{k_1} {a}_{k_2} {a}_{k_3} \right)$,
cf.~\eqref{eq:flux_sin}, as functions of the wavenumbers $k_1$ and
$k_2$ (where $k_3 = - (k_1 + k_2)$) at different times $t \in (0,t^*)$
(an animated version of this figure is available
\href{https://www.youtube.com/watch?v=wcSKBZhpHMc}{on-line}).  We note
that the triads with phases aligned such that $\varphi_{k_1,k_2}^{k_3}
= \pm \frac{\pi}{2}$ exhibit an intriguing ``triangular'' pattern in
the half plane $\{ (k_1,k_2), \ k_1 \in \ZZ^+, \ k_2 \ge k_1 \}$ as
the blow-up time $t^*$ is approached (for better resolution, we have
only shown the range $k_1,k_2 \in [1,120]$, however, the same pattern
is also evident for larger values of the wavenumbers). More
specifically, this pattern can be empirically described as
\begin{equation}
\forall k_1 \in \ZZ^+, \quad
\varphi_{k_1,k_2}^{k_3}(t) = \left\{ 
\begin{alignedat}{2}
& \frac{\pi}{2}, & \quad & k_2 \in \bigcup_{p=0}^{\infty} 
\left[k_1 + p\kappa(t),  \left\lfloor \frac{k_1}{\kappa(t)} \right\rfloor + (p+1)\kappa(t) \right] \\
- & \frac{\pi}{2}, & & \text{otherwise} 
\end{alignedat} \right. ,
\label{eq:triangle}
\end{equation}
where the "offset" $\kappa(t) \in \ZZ^+$ increases without bound as $t
\rightarrow t^*$ (it can be interpreted as the length of the side of
an elementary triangle in the pattern). This latter property signals
an increasing synchronization of the triads as the blow-up time is
approached. This pattern is also evident, albeit less pronounced due
to a wider spread of values, in the corresponding flux contributions
shown in the bottom panels in figure \ref{fig:tri_fl}. We add that
triad interactions in the viscous Burgers flow corresponding to the
extreme initial data $\tuuE$ (not shown here) exhibit similar trends
to those visible in figure \ref{fig:tri_fl}. On the other hand, the
inviscid Burgers flow with the unimodal initial condition
\eqref{eq:Bu0sin} (also not shown here) reveals a nearly instantaneous
synchronization with $\kappa(t)$ becoming unbounded very quickly, such
that the distributions of $\sin\left(\varphi_{k_1,k_2}^{k_3}\right)$
and $\log_{10} \left(4 \pi k_1 \, {a}_{k_1} {a}_{k_2} {a}_{k_3}
\right)$ at all times $t \in (0,t^*)$ resemble those in figure
\ref{fig:tri_fl}(d). We remark that similar patterns have also been
observed in viscous Burgers flows in the presence of stochastic
forcing \citep{MurrayPhDthesisNew,MurrayBustamante2018}.

\FloatBarrier

\subsection{Results for 3D Navier-Stokes flows}
\label{sec:results3D}

%\textcolor{blue}{[IDEAS: (1) Can we show time evolution of total energy? (or just for us and we can simply mention it). (2) How about the enstrophy spectrum? Even though enstrophy is not conserved, we can still look at its spectrum; it may make more sense as we are looking at maximal enstrophy growth. (3) To note that studying the spectral flux of enstrophy does not make sense as there is no perfect within-the-triad cancellation of the terms when the three wavenumbers are in $\mathcal{C}$.]
%}

{We present the results for 3D Navier-Stokes flows by first briefly
  discussing classical diagnostics, namely, the time evolution of
  Fourier spectra, enstrophy and energy fluxes, applied to flows with
  the three initial conditions (extreme, generic and Taylor-Green,
  cf.~\S\,\ref{sec:problems}) before moving on to analyze these flows
  in terms of the diagnostics based on the triad interactions
  introduced in \S\,\ref{sec:diagn3D}.}  Figure \ref{fig:NSEt}(a)
shows the initial ($t=0$) and final ($t=\tTE$) {energy spectra
  \eqref{eq:ek}} for the three initial conditions. As for the initial
spectra, the extreme case has a distribution across a range of
wavenumbers and the spectrum of the initial condition in the generic
case is by definition identical. On the other hand, the initial
spectrum in the Taylor-Green case has energy concentrated in a small
subset only of low wavenumbers. As for the final spectra, in all cases
the energy has been transferred to higher wavenumbers, with the
Taylor-Green case showing a less dramatic transfer to high
wavenumbers; the extreme and generic cases look quite similar, but the
extreme case shows slightly larger transfers to high wavenumbers.
However, none of these energy spectra exhibits a wide inertial range
which is due to a relatively low Reynolds number. It is evident from a
comparison with the $k^{-5/3}$ line in figure \ref{fig:NSEt}(a) that
an inertial range spans the wavenumbers $k=1$ to $k=3$, while a
dissipative range involves wavenumbers $k\geq10$ (this {will be
  further justified by analysis of the fluxes shown} in figure
\ref{fig:FluxTf}).  The helicity spectra for the three types of flows
considered are identically zero due to the fact that the flows are odd
under the parity transformation. Figure \ref{fig:NSEt}(b) shows the
time evolution of the total enstrophy for each of the three cases. It
is evident that the Taylor-Green case shows a {modest only growth of
  enstrophy}, while the generic case is quite close to the extreme
case.

\begin{figure}
\begin{center}
\mbox{\subfigure[]{\includegraphics[width=0.45\textwidth]{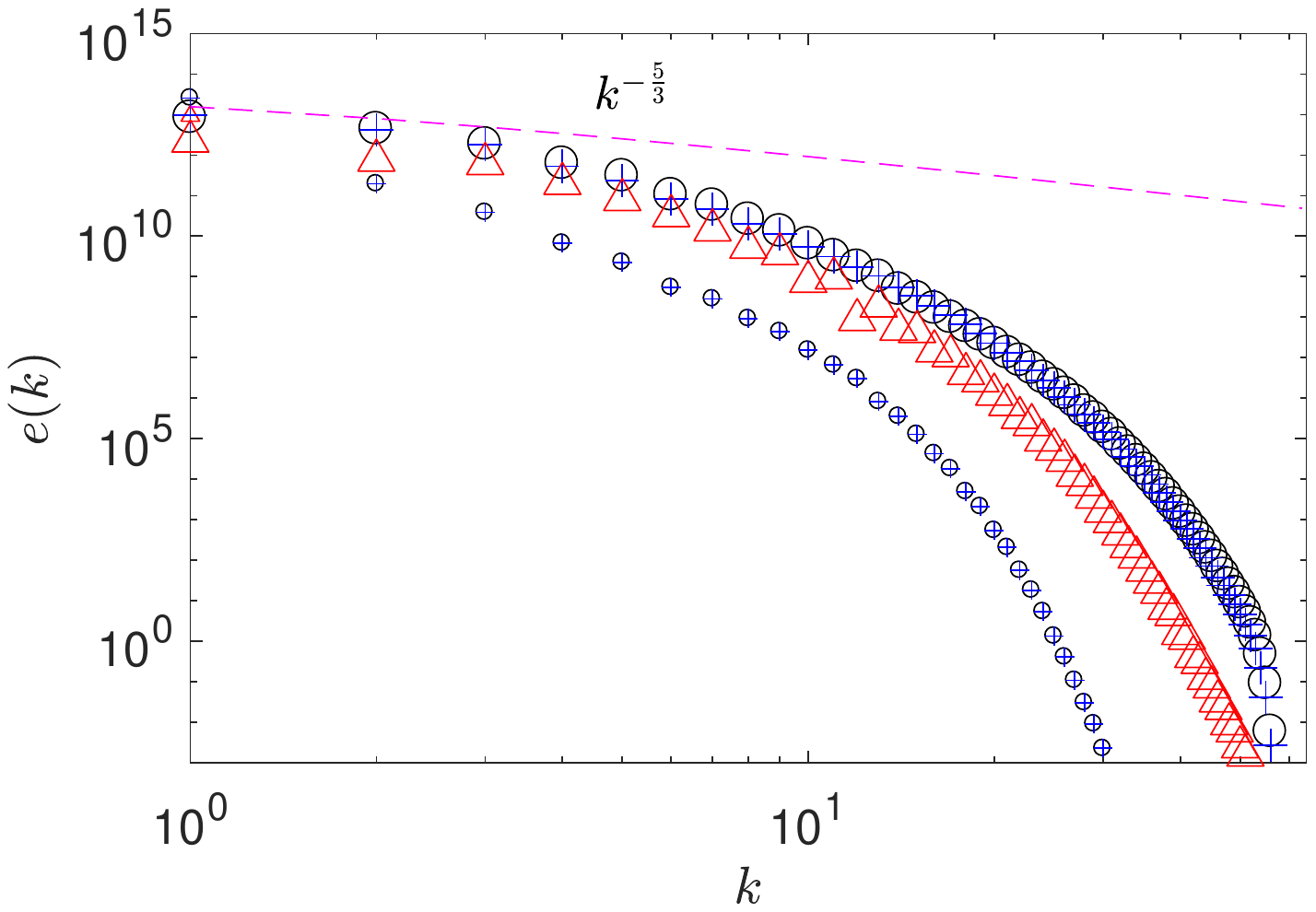}}\qquad
\subfigure[]{\includegraphics[width=0.5\textwidth]{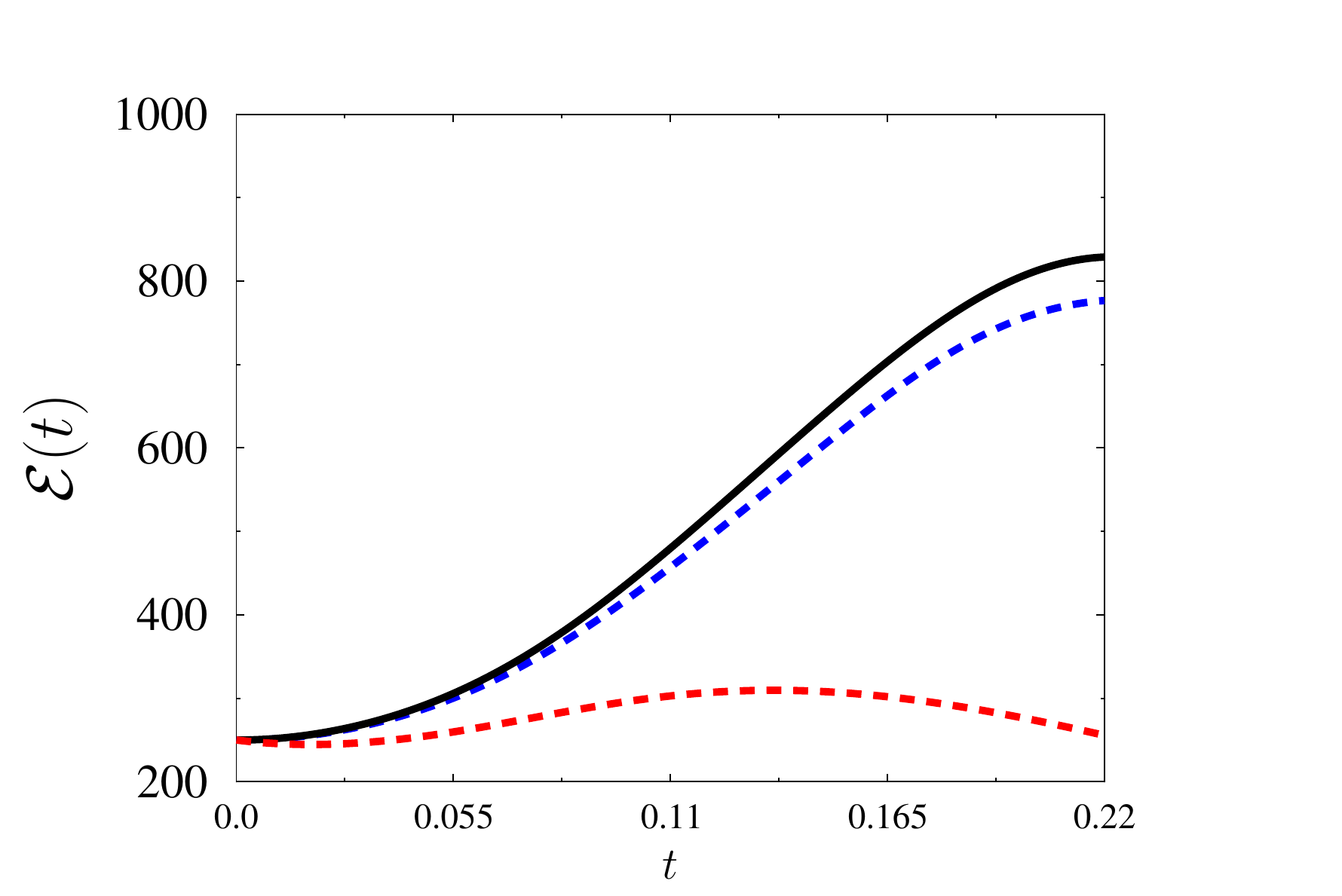}}}
\caption{(a) {Energy spectra \eqref{eq:ek}} of the solution of
  the Navier-Stokes equation \eqref{eq:3DNS_u} with the extreme
  initial condition (black circles), the generic initial condition
  (blue crosses) and the Taylor-Green initial condition (red
  triangles) at the initial time $t=0$ (small symbols) and the time
  $t=\tTE$ (large symbols) when the largest enstrophy is achieved for
  $\E_0 = 250$.  {The straight line represents the expression
    $k^{-5/3}$.} (b) The time history of the enstrophy $\E(t)$ in the
  solution of the Navier-Stokes equation \eqref{eq:3DNS_u} with the
  extreme initial condition (black solid line), the generic initial
  condition (blue dashed line) and the Taylor-Green initial condition
  (red dashed line) for $\E_0 = 250$.}
\label{fig:NSEt}
\end{center}
\end{figure}

{We now discuss the time evolution of the energy flux in the
  flows with the three initial conditions considered.  In each case
  the spectral flux of energy is evaluated as}
\begin{equation}
\Pi(t,k) := \Pi_{{\mathcal C}_k}(t), \qquad {\mathcal C}_k = \{\mathbf{k} \in \mathbb{Z}^3 \setminus \{0\} \quad | \quad 
|\mathbf{k}| > k \}\,,
\label{eq:Ck}
\end{equation}
where $\Pi_{{\mathcal C}}$ is defined in its basic form in equation
\eqref{eq:engy_flux_basic}. To separate positive and negative
contributions {representing, respectively, the direct and inverse
  energy cascades,} we decompose the flux as $\Pi(t,k) = \Pi^+(t,k) -
\Pi^-(t,k)$, where $\Pi^+(t,k), \Pi^-(t,k) \geq 0$ and $\Pi^+(t,k)
\Pi^-(t,k) = 0$ for all wavenumbers $k$ and times $t$. These
quantities are plotted in figures \ref{fig:NSflux},
\ref{fig:NSflux_generic} and \ref{fig:NSflux_TG} for the extreme,
generic and Taylor-Green case.  {We remark that for practical reasons
  the data in these plots is obtained directly by evaluating the
  nonlinear term in the Fourier representation of the Navier-Stokes
  equation \eqref{eq:3DNS_uA} rather than by computing the sums in
  \eqref{eq:engy_flux_basic}.}

\begin{figure}
\begin{center}
\mbox{\subfigure[]{\includegraphics[width=0.45\textwidth]{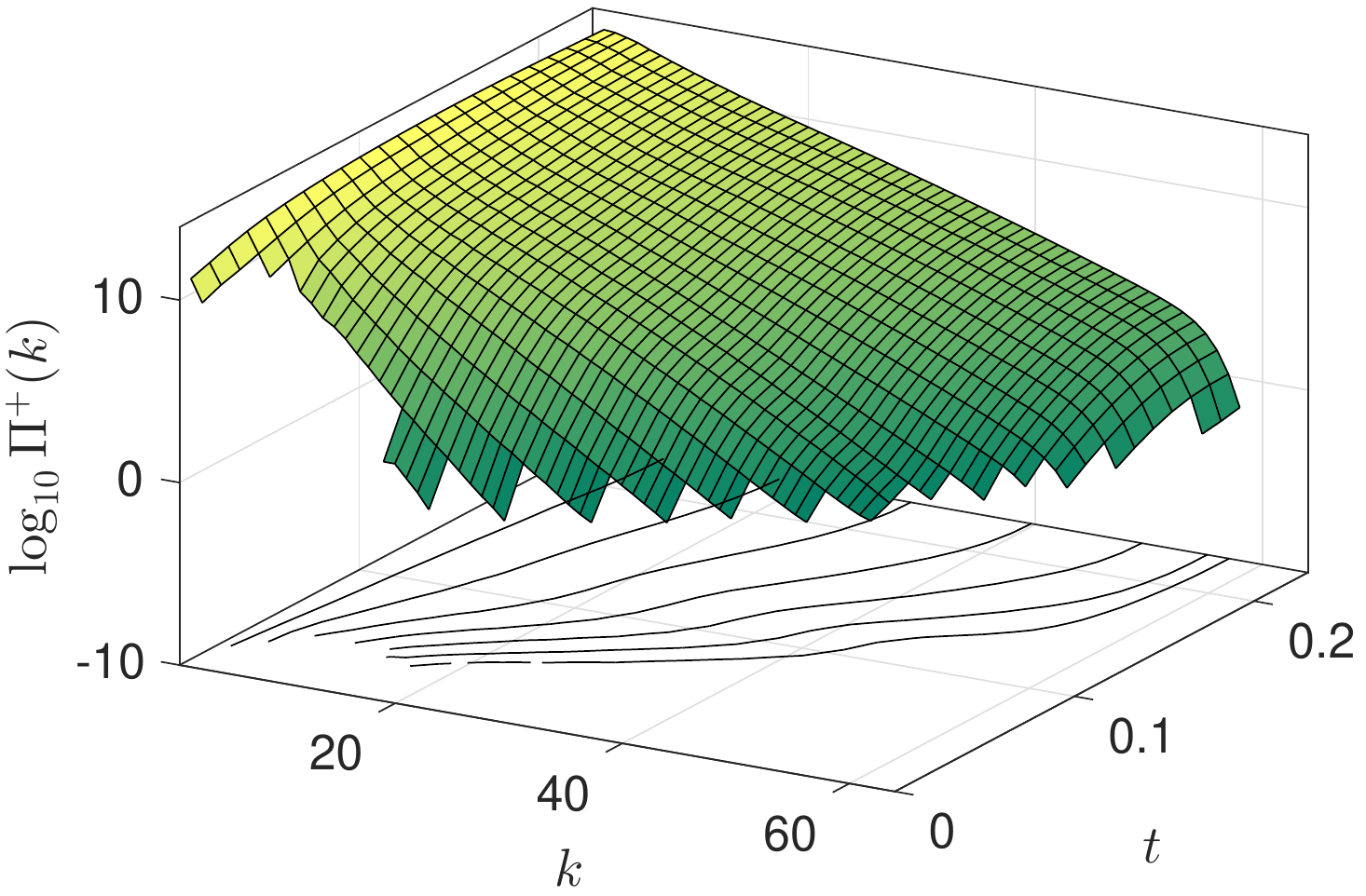}}\qquad
\subfigure[]{\includegraphics[width=0.45\textwidth]{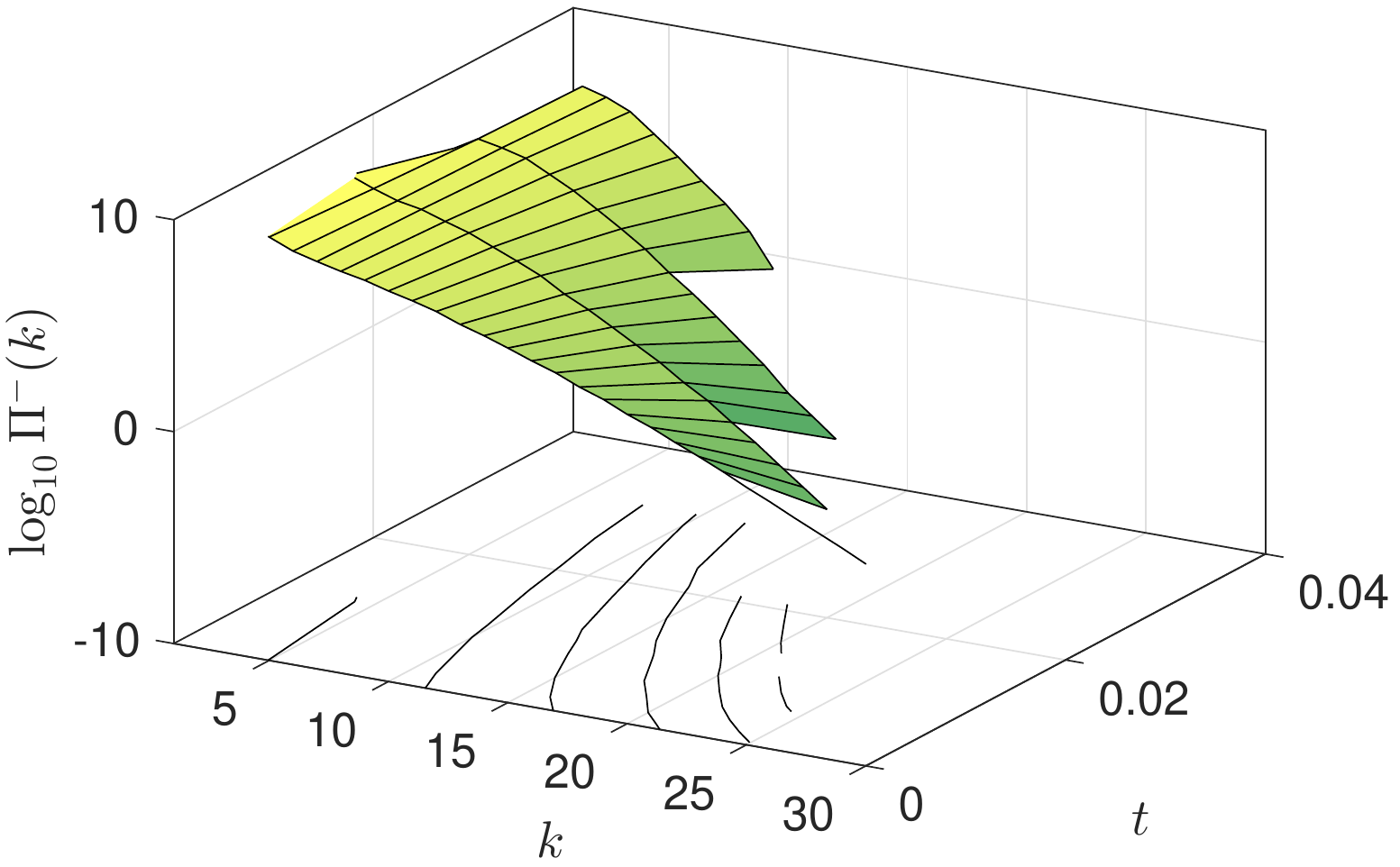}}}
\caption{{[Extreme case:]} (a) Positive and (b) negative part of
  flux $\Pi(t,k)$ as function of the time $t \in [0,\tTE]$ and the
  wavenumber $k$ in the solution of the Navier-Stokes equation with
  the indicated initial condition. {Contours in the $(t,k)$ plane
    are the level sets of $\Pi^+(t,k)$ and $\Pi^-(t,k)$ corresponding
    to the values $10^0,10^2,\dots,10^{12}$.}}
\label{fig:NSflux}
% \end{center}
% \end{figure}
% \begin{figure}
% \begin{center}
\medskip
\mbox{\subfigure[]{\includegraphics[width=0.45\textwidth]{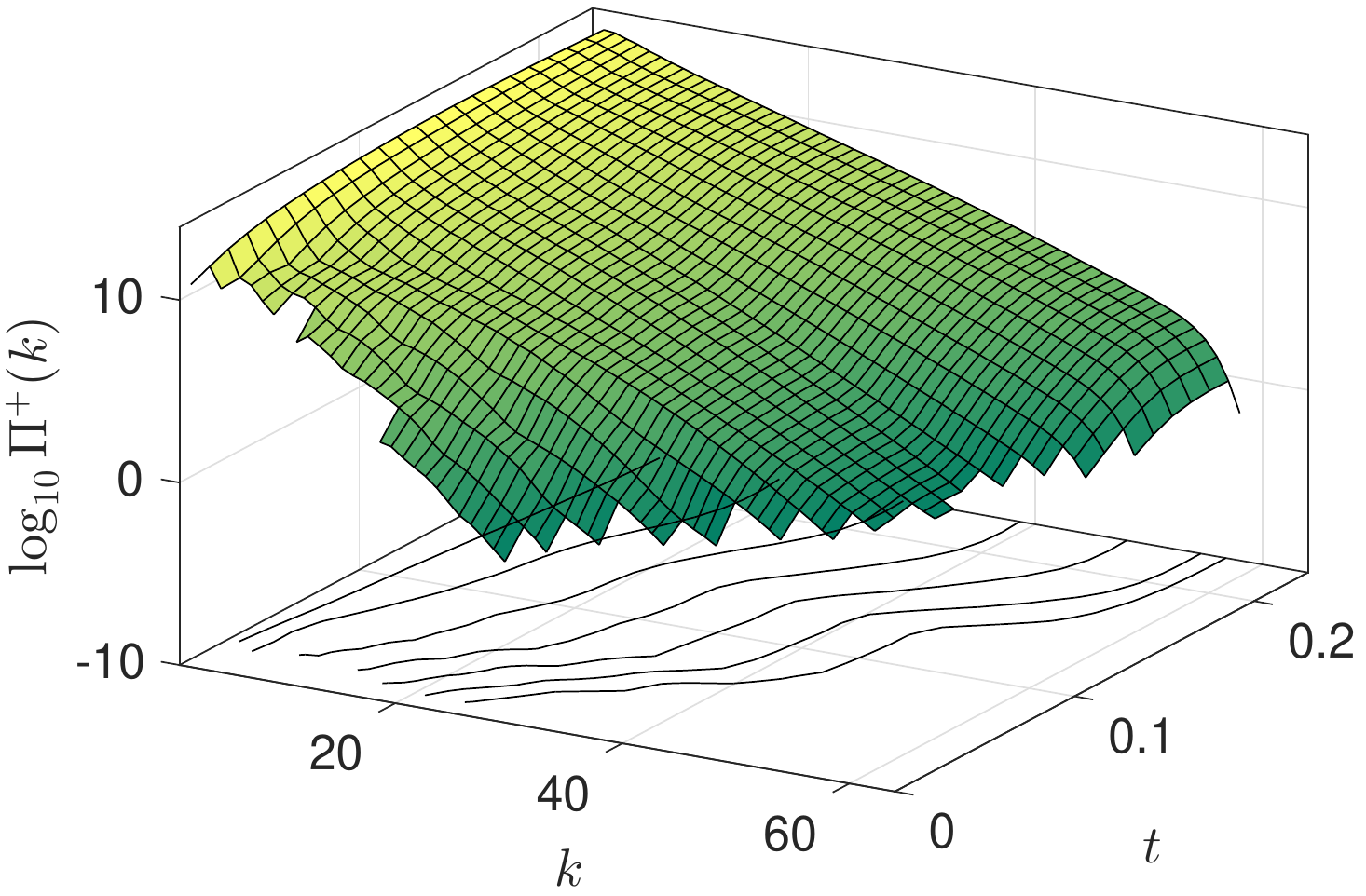}}\qquad
\subfigure[]{\includegraphics[width=0.45\textwidth]{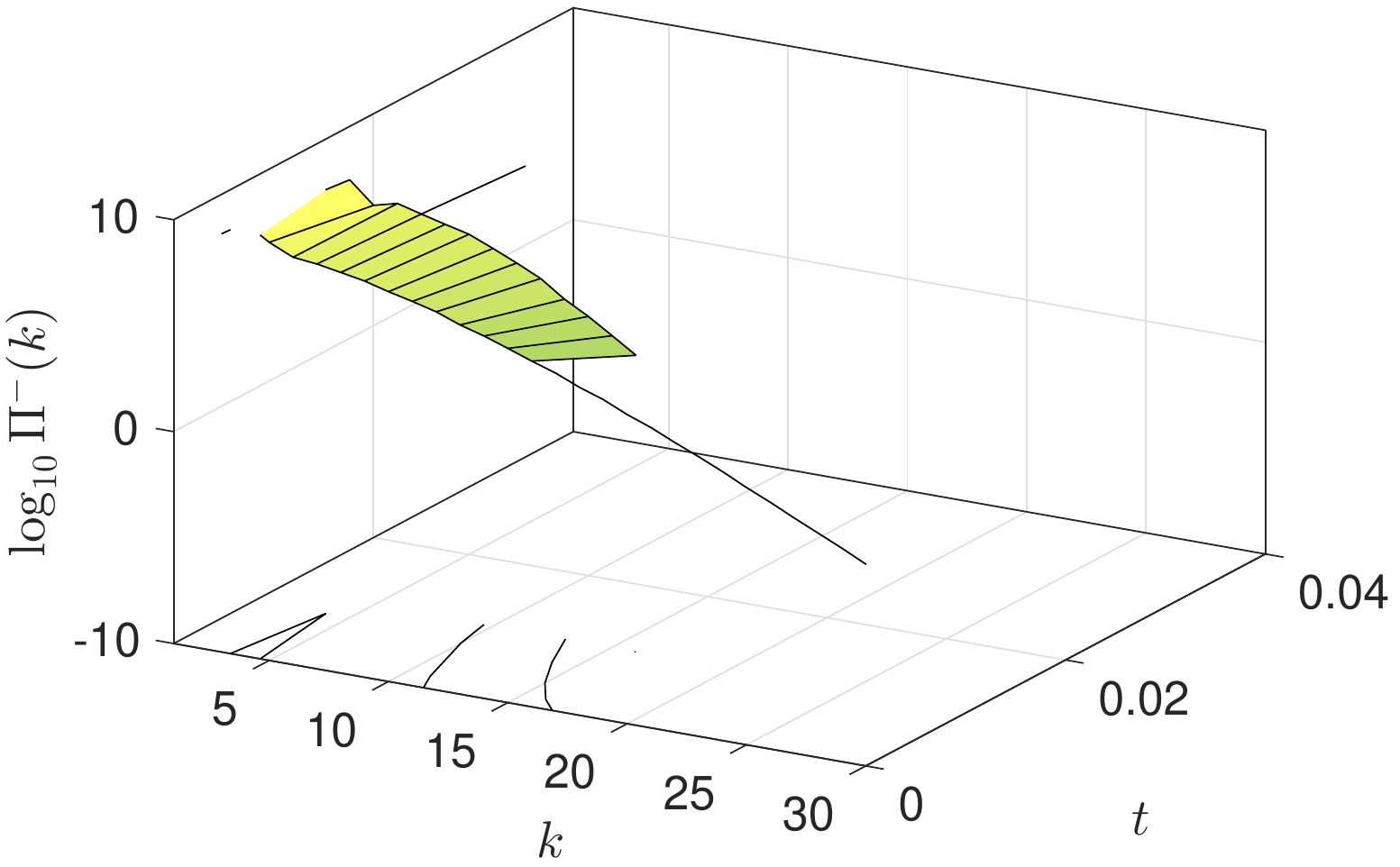}}}
\caption{{[Generic case:]} See the caption of figure \ref{fig:NSflux} for details.}
\label{fig:NSflux_generic}
% \end{center}
% \end{figure}
% \begin{figure}
% \begin{center}
\medskip
\includegraphics[width=0.45\textwidth]{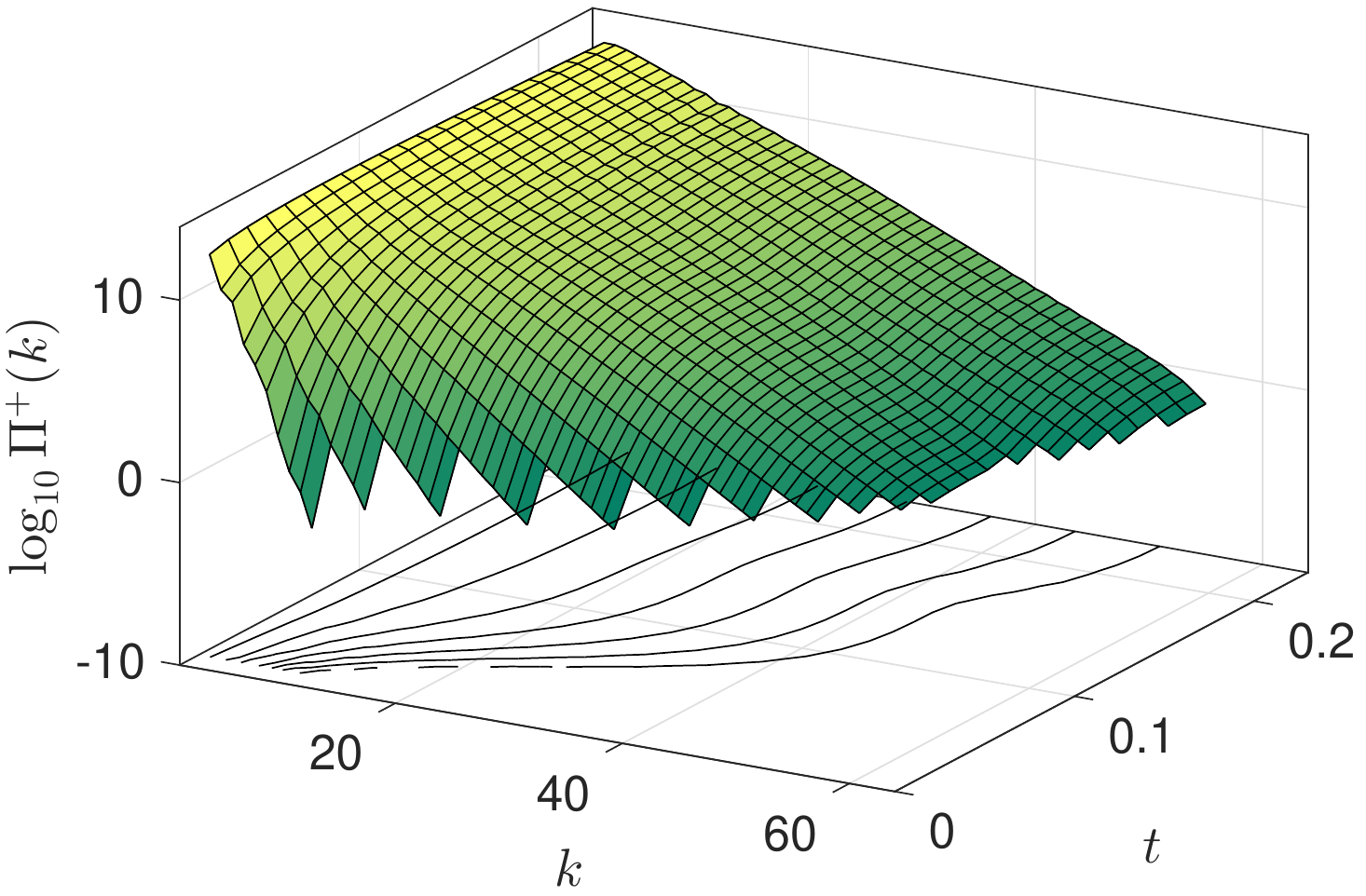}
%\subfigure[]{\includegraphics[width=0.45\textwidth]{Figs2/FluxN_TG.pdf}}
\caption{{[Taylor-Green case:]} Positive part of flux $\Pi(t,k)$
  as function of the time $t \in [0, {\tTE}]$ and the wavenumber $k$ in
  the solution of the Navier-Stokes equation with the Taylor-Green
  initial condition. The negative part of flux $\Pi(t,k)$ vanishes for
  the Taylor-Green case. {Contours in the $(t,k)$ plane are the
    level sets of $\Pi^+(t,k)$ corresponding to the values
    $10^0,10^2,\dots,10^{12}$.}}
\label{fig:NSflux_TG}
\end{center}
\end{figure}

Figures \ref{fig:NSflux}(a) and \ref{fig:NSflux}(b) show,
respectively, positive and negative fluxes for the case of the extreme
initial condition. It is evident that at early times there is a
coherent negative flux at small wavenumbers.
% To better understand the positive
% fluxes, a projection to a horizontal plane is used which shows contour
% lines of constant flux, at values $(0,2,4,6,8,10,12)$ in logarithmic
% scale. 
For a fixed time $t$, closely packed level sets represent a strong
dependence on the wavenumber $k$.
% whereas more spread out level sets indicate a more
% uniform dependence on wavenumber $k$. 
It is evident that, as the time approaches the end of the time window,
{the level sets become more spread out which corresponds to fluxes
shifting towards larger wavenumbers $k$ and to a tendency towards a
uniform flux.}

\begin{figure}
\begin{center}
\includegraphics[width=0.45\textwidth]{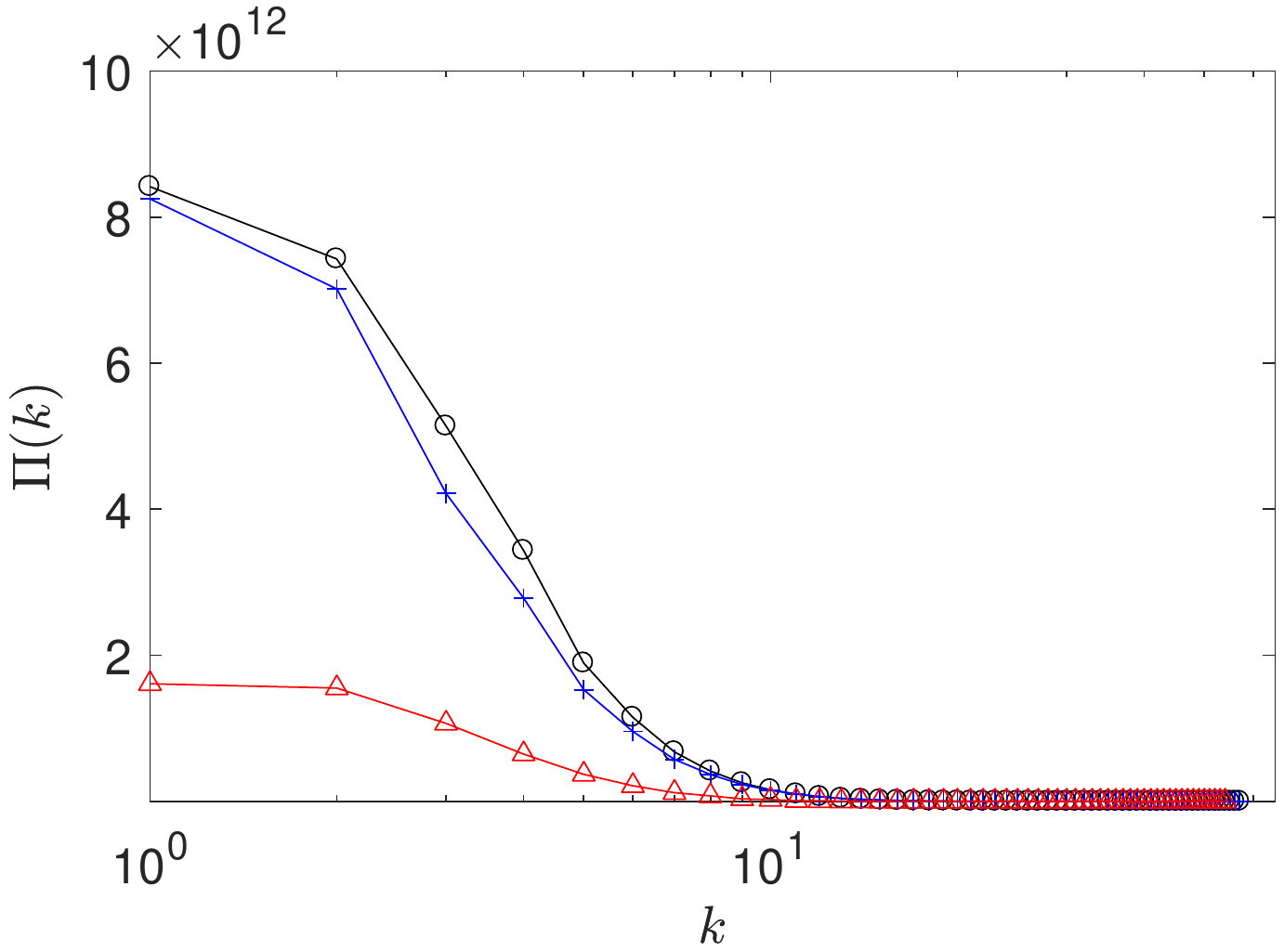}
\caption{{Flux $\Pi(k)$ as a function of the wavenumber $k$ at
    the final time $t=\tTE$ in the solution of the Navier-Stokes
    equation \eqref{eq:3DNS_u} with the extreme initial condition
    (black circles), the generic initial condition (blue crosses) and
    the Taylor-Green initial condition (red triangles).}}
\label{fig:FluxTf}
\end{center}
\end{figure}

In a similar way, figures \ref{fig:NSflux_generic}(a) and
\ref{fig:NSflux_generic}(b) show, respectively, positive and negative
fluxes for the case of the generic initial condition. While the
behaviour of the fluxes in this case appears similar to the extreme
case, there are several differences of quantitative and qualitative
nature. Comparing the negative flux contributions in figures
\ref{fig:NSflux}(b) and \ref{fig:NSflux_generic}(b), it is evident
that while {both cases initially show} negative flux at small
wavenumbers, this flux is short-lived in the generic case.  Comparing
now the positive flux contributions in figures \ref{fig:NSflux}(a) and
\ref{fig:NSflux_generic}(a), we see that initially the contours are
more spread out in the generic case.  This is solely due to the
randomization of the phases of the generic initial condition,
cf.~\S\,\ref{sec:generic}, which takes away the coherence required to
{obtain localized fluxes}.  Second, {the dependence of the
  flux level sets on time} is more erratic in the generic case, while
in the extreme case there is a clear coherent {shift of the level
  sets} towards large wavenumbers.

Figure \ref{fig:NSflux_TG} shows the positive fluxes for the case of
the Taylor-Green initial conditions (there are no negative fluxes in
this case). One observes a rapid initial spreading of the level sets,
mainly due to the fact that the Taylor-Green initial condition has its
energy concentrated at small wavenumbers. After this initial phase,
an intermediate regime occurs characterized by a quasi-steady state
where the level sets change very little in time. Then, a final growth
phase occurs which is however much milder than in the previous two
cases. This can be inferred from the facts that: (i) the level set
representing the smallest flux values reaches a lower wavenumber, and
(ii) the level sets are less spread out at the end of the time window.

{From} a comparison of the time evolution of the positive fluxes
for all cases, cf.~figures \ref{fig:NSflux}(a),
\ref{fig:NSflux_generic}(a) and \ref{fig:NSflux_TG}, we notice that
both the generic and Taylor-Green cases display clear {periods of
  stagnation} in the time evolution of the level sets corresponding to
large wavenumbers.  They occur at intermediate times $t=0.75$--$0.125$
in the generic case and $t=0.1$--$0.15$ in the Taylor-Green case. This
is another piece of evidence for the irregularity of the flux
evolution across wavenumbers in these cases: when flux is steady
(instead of growing in time) {near a certain} wavenumber $k_0$, there
is more time for the energy to dissipate before it can flow to small
scales. In contrast, the extreme case shows a smooth flux evolution
across wavenumbers, without {stagnation}, which seems to facilitate
the late-time maximization of the enstrophy.

The fluxes at the final time $t=\tTE$ are compared between the three
cases in figure \ref{fig:FluxTf} which confirms that the flux in the
flow with the extreme initial condition is the most intense for all
values of $k$ with the flux in the flow with the generic initial
condition being only slightly smaller. The flux in the flow with the
Taylor-Green initial condition is much weaker, approximately by a
factor of four. The plots of flux in figure \ref{fig:FluxTf} confirm
our earlier assessment (cf.~figure \ref{fig:NSEt}) that an inertial
range is found between wavenumbers $k=1$ and $k=3$ while a dissipative
range is found for wavenumbers $k\geq 10$, as there is only about
$1\%$ of the original flux left at that {latter} scale.  More
quantitatively, the Kolmogorov length scale can be estimated via $\eta
:= (\nu^3/2\nu\mathcal{E}_0)^{1/4}$, with $\nu=0.01$ and
$\mathcal{E}_0=250$, giving $\eta \approx 0.02115$, and in terms of
wavenumbers, $k_{\eta} := 2 \pi / \eta \approx 297 = 3.48 k_{\max}$.
Thus, at $k=10$ we get a ratio $k_{\eta}/k \lessapprox 30$ which can
be considered in the dissipative range.
  
{We now proceed to discuss the flows considered in terms of the
  diagnostics based on the triad interactions introduced in
  \S\,\ref{sec:diagn3D}.  With reference to the fluxes towards small
  spatial scales shown in figures \ref{fig:NSflux},
  \ref{fig:NSflux_generic} and \ref{fig:NSflux_TG}, we choose two
  representative {cases: (i) fluxes towards the wavenumber region with
    $k>2$ corresponding to the inertial range of wavenumbers, i.e.,}
  those wavenumbers where energy is flowing due to nonlinear
  interactions and energy dissipation is negligible; (ii) fluxes
  towards {the wavenumber region with $k>10$ corresponding to the
    dissipative range of wavenumbers, i.e.,} those wavenumbers where
  energy is still flowing due to nonlinear interactions but energy
  dissipation dominates, with a significant proportion of the energy
  having already been dissipated along the way from $k=2$ to $k=10$.
  We begin with the extreme case and in figure \ref{fig:NSpdfkb2}}
show the time evolution of the PDFs $P_{\mathcal{C}_2}^{s_1 s_2
  s_3}(\Phi)$ of the generalized helical triad phases {participating}
in the fluxes out of the sphere {defined in the spectral space by
  $k=2$,} separately for each triad type $(s_1 s_2 s_3)$.  {These
  results should be viewed with reference to equation
  \eqref{eq:engy_flux_triad} providing an expression for
  $\Pi_{\mathcal{C}}^{s_1 s_2 s_3}$ which is the contribution of
  triads of type $(s_1 s_2 s_3)$ to the spectral energy flux into a
  set $\mathcal{C}$, here defined as $\mathcal{C} = \mathcal{C}_2 =
  \{\mathbf{k} \in \mathbb{Z}^3 \setminus \{0\} \quad | \quad
  |\mathbf{k}| > 2\}$.} The figure shows the time evolution of the
PDFs ${\mathcal{P}}_{\mathcal{C}_2}^{s_1 s_2 s_3}(\Phi)$ of four
selected generalized helical triad types, namely, the ones which
{exhibit variability of at least $\pm 5\%$ with respect to the uniform
  distribution equal} to $1/2\pi (\approx 0.16)$: $(s_1,s_2,s_3)=$
PPP, PMM, P(PM), (PM)P, where the last two triad types are of the
boundary types, corresponding to the restrictions
$|\mathbf{k}_1|=|\mathbf{k}_2|$ and $|\mathbf{k}_2|=|\mathbf{k}_3|$,
respectively.  For these triad types the PDFs vary slightly {over the
  range} from about $0.14$ to $0.18$, cf.~table
\ref{tab:PDFbounds_extreme}, third column. The only exception (and a
quite remarkable one) is the case of the boundary triad type (PM)P
whose PDF varies from about $0$ to $0.62$ (table
\ref{tab:PDFbounds_extreme}, third column and bottom row) and is
clearly concentrated near $\Phi=0$ (figure \ref{fig:NSpdfkb2}{(d)}).
{However, as figure \ref{fig:NSpdfkb2} shows clearly,} at later times
the PDFs change.  The most active one is again the boundary triad of
type (PM)P, {whose PDF} exhibits a coherent behaviour over time, with
some local maxima oscillating about $\Phi=0$ and other local maxima
oscillating about $\Phi \approx \pm 3\pi/4$. Also active are the triad
types PPP and PMM whose PDFs reveal a persistent late-time preference
for values near $\pi$ or $-\pi$ (figure \ref{fig:NSpdfkb2}(a)--(b)).
As explained in the text preceding equation
\eqref{eq:engy_flux_final}, {triad phases with values} near $0$ give
rise to positive contributions to the flux, whereas values near
$\pm\pi$ produce negative contributions to the flux. In the light of
this alone, it might appear from these figures that contributions to
the flux {from individual triads} have mixed signs, and thus the total
contribution to the flux may change sign as a function of time.
However, one has to {bear in mind that} contributions to the flux also
depend on {mode amplitudes which are constantly} changing as well.

\begin{table}
  \begin{center}
   \hspace*{-1.1cm}
    \begin{tabular}{|l|c|c|} \hline
       \backslashbox{triad}{cases}  & Fraction for $k=2$ &  Fraction for $k=10$ \\ \hline
     PPP  & $24.944\%$ & $24.986\%$  \\ 
      PPM & $24.944\%$ &  $24.986\%$ \\    
      PMP  &  $24.944\%$ &  $24.986\%$ \\
      PMM  & $24.944\%$  &  $24.986\%$ \\ 
	 P(PM)  & $0.224\%$ & $0.054\% $ \\
	 (PM)P  & $9.48\times10^{-5}\%$ & $1.03\times10^{-3}\%$ \\\hline
    \end{tabular}
  \end{center}
  \caption{{Proportion of triads of different types participating in energy flux to regions $\C_k$ in the wavenumber space, cf.~\eqref{eq:Ck}, with the indicated radii $k$. These fractions are functions of $k$ only and do not depend on the initial condition.}}
  \label{tab:percent}

  \begin{center}
   \hspace*{-1.1cm}
    \begin{tabular}{|l|c|c|c|c|} \hline
       \backslashbox{triad}{cases}  & $k=2$, PDF  & $k=2$, wPDF & $k=10$, PDF &  $k=10$, wPDF \\ \hline
     PPP  & [0.145,0.176]   &  [0.016,0.617] & \textcolor[rgb]{0.7,0.7,0.7}{[0.158,0.161]}  &  [0.068,0.561] \\ 
      PPM &  \textcolor[rgb]{0.7,0.7,0.7}{[0.158,0.161]}  &  [0.025,0.485] & \textcolor[rgb]{0.7,0.7,0.7}{[0.159,0.160]}  &  [0.086,0.355] \\    
      PMP  & \textcolor[rgb]{0.7,0.7,0.7}{[0.156,0.162]} &  [0.019,0.732] & \textcolor[rgb]{0.7,0.7,0.7}{[0.159,0.160]} &   $[0.100,0.283]$ \\
      PMM  & [0.148,0.171] & [0.008,0.780] & \textcolor[rgb]{0.7,0.7,0.7}{[0.157,0.162]} &  $[0.066,0.555]$  \\ 
	 P(PM)  & [0.150,0.167] &  [0.001,0.632] & \textcolor[rgb]{0.7,0.7,0.7}{[0.158,0.161]} &   $[0.027,0.577]$ \\
	 (PM)P  & [0.000,0.620] &  [0.000,1.537] & [0.150,0.170] &   $[0.040,0.550]$ \\\hline
    \end{tabular}
  \end{center}
  \caption{{[Extreme case:]} Upper and lower bounds on the PDFs and wPDFs of triad phases of the different types in the Navier-Stokes flow with the extreme initial data, cf.~figures \ref{fig:NSpdfkb2}--\ref{fig:NSWpdfkb10}. {Shaded} intervals (light gray) represent PDFs that are very close (within $\pm5\%$) to the uniform distribution $1/2\pi \approx 0.16$.}
  \label{tab:PDFbounds_extreme}
\end{table}

%1/2/pi*1.05 = 0.1671126902
%1/2/pi*0.95 = 0.1511971959

\begin{figure}
\begin{center}
\mbox{\subfigure[${\mathcal{P}}_{\mathcal{C}_2}^{PPP}(\Phi)(t)$]{\includegraphics[width=0.45\textwidth]{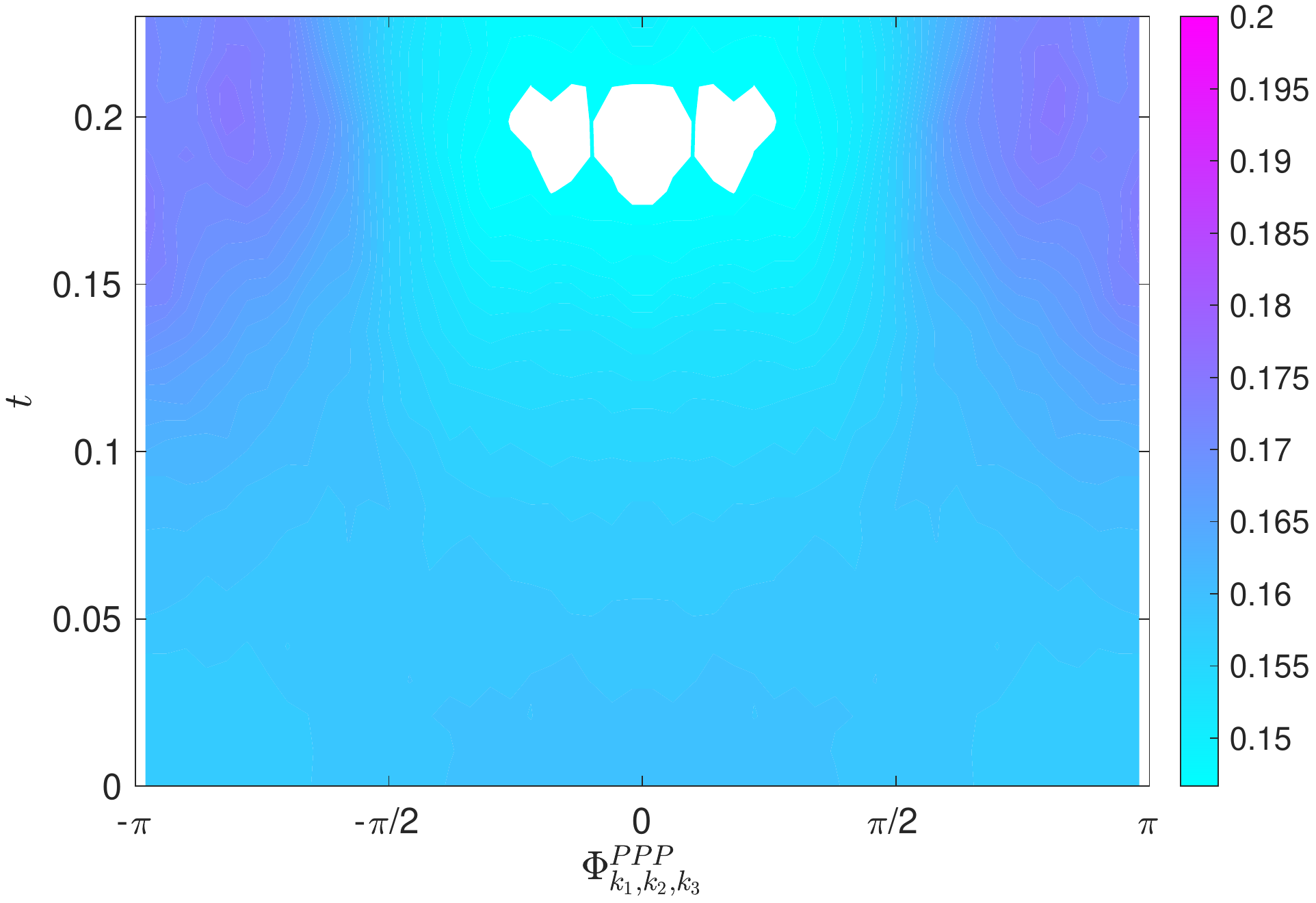}}\qquad
\subfigure[${\mathcal{P}}_{\mathcal{C}_2}^{PMM}(\Phi)(t)$]{\includegraphics[width=0.45\textwidth]{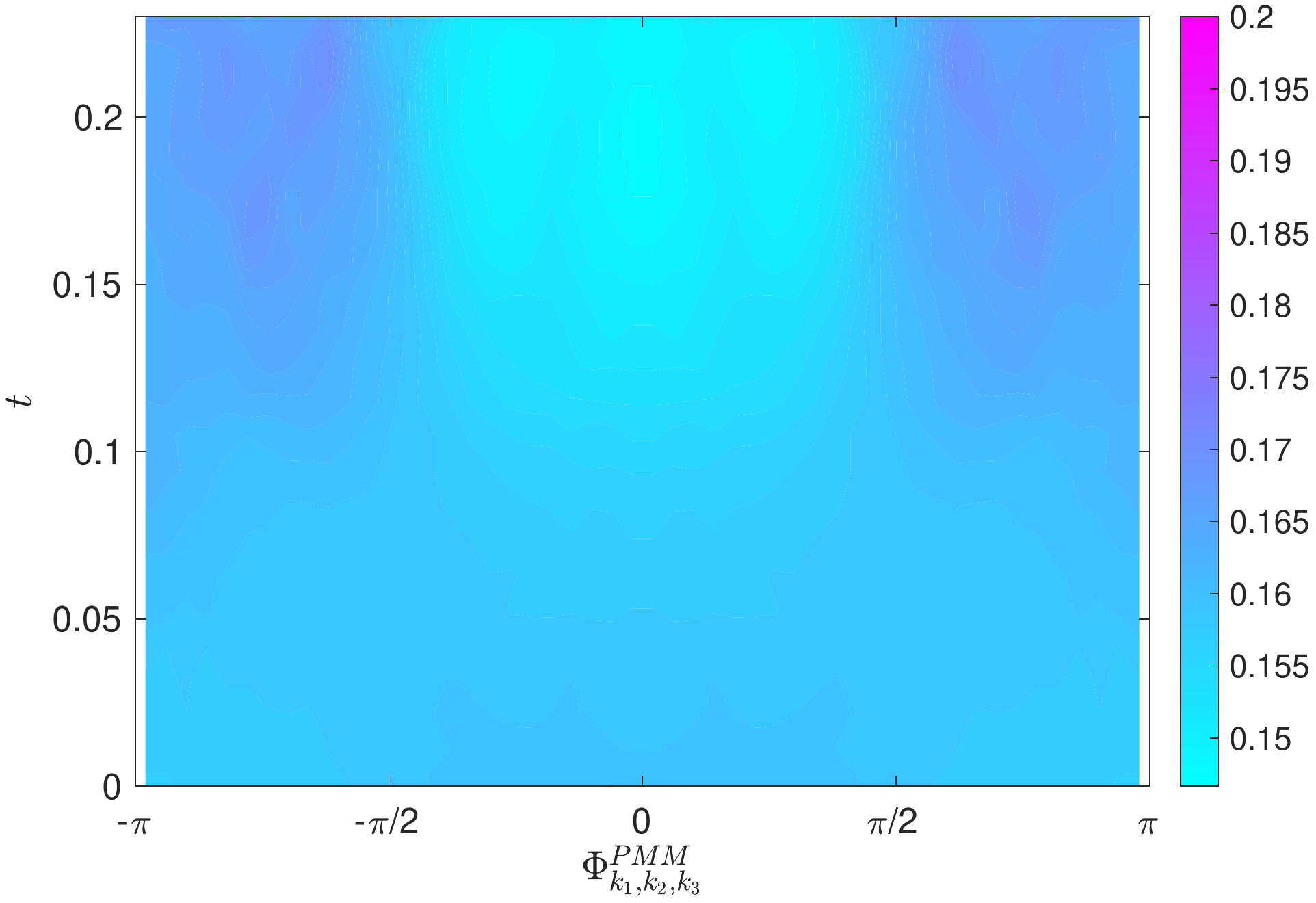}}}
\mbox{\subfigure[${\mathcal{P}}_{\mathcal{C}_2}^{P(PM)}(\Phi)(t)$]{\includegraphics[width=0.45\textwidth]{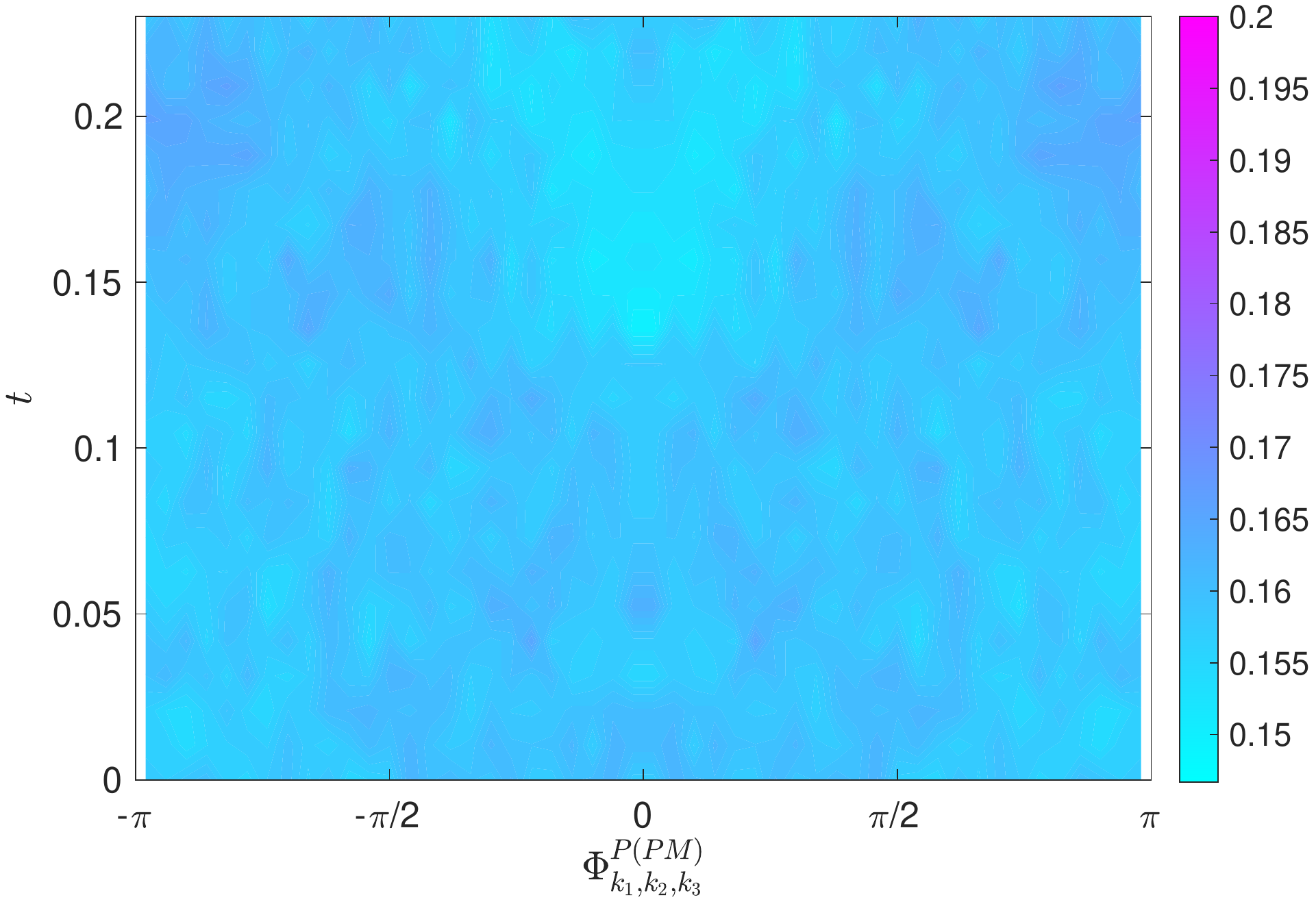}}\qquad
\subfigure[${\mathcal{P}}_{\mathcal{C}_2}^{(PM)P}(\Phi)(t)$]{\includegraphics[width=0.45\textwidth]{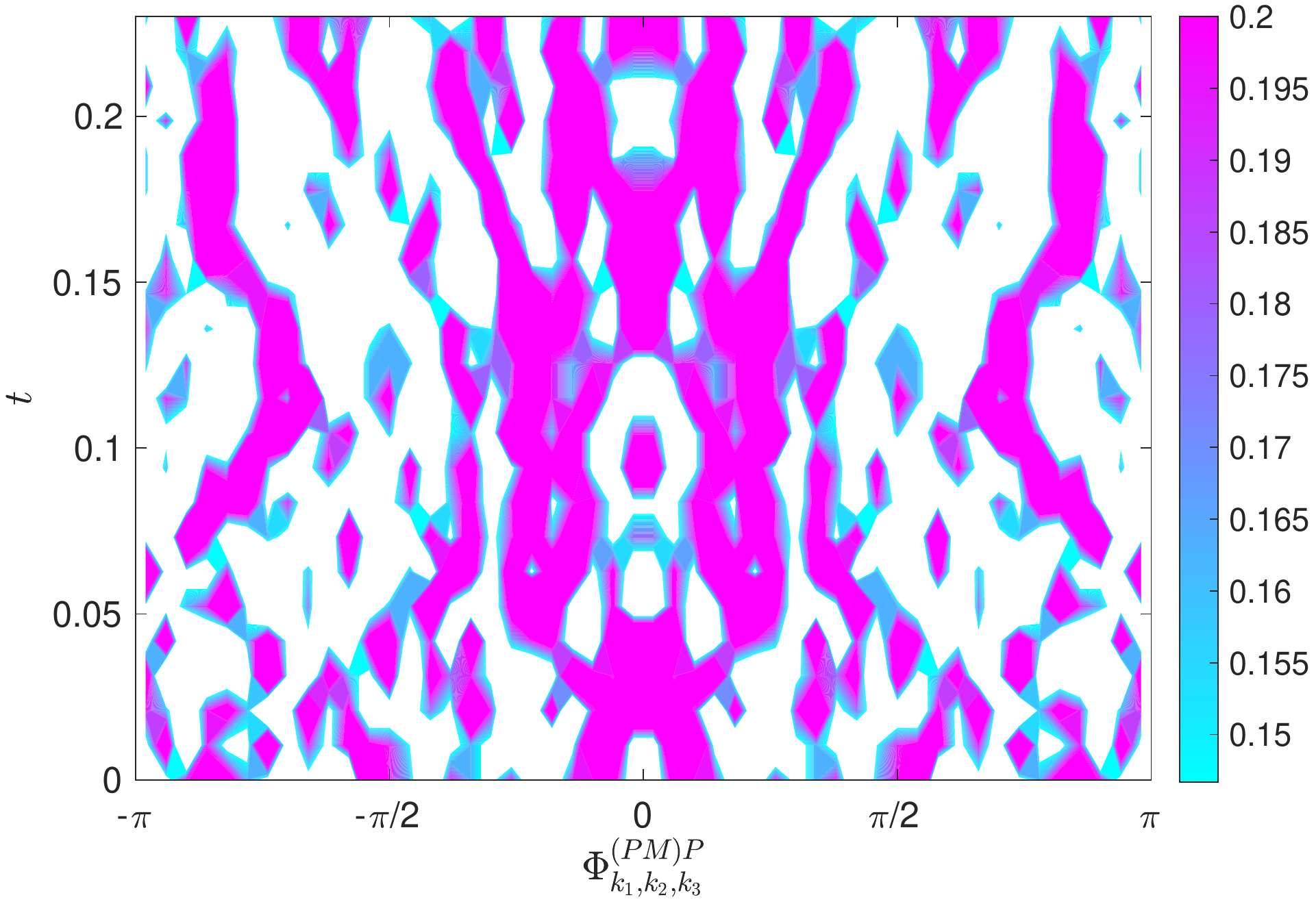}}}
\caption{{[Extreme case, $k>2$:]} Time evolution of the PDFs of the
  triad phase angle ${\mathcal{P}}_{\mathcal{C}_2}^{s_1 s_2
    s_3}(\Phi)$ {for triads of different types in the
    Navier-Stokes flow with the extreme initial data.  Only triad
    types with PDFs revealing variability of at least $\pm5\%$ with
    respect to the uniform distribution $1/2\pi \approx 0.16$
    (corresponding to light blue color in the plots), are shown:} (a)
  ``PPP'', (b) ``PMM'', (c) ``P(PM)'' and (d) ``(PM)P''.}
\label{fig:NSpdfkb2}
\end{center}
\end{figure}

\begin{figure}
\begin{center}
\mbox{\subfigure[Extreme case: Flux towards $k>2$]{\includegraphics[width=0.45\textwidth]{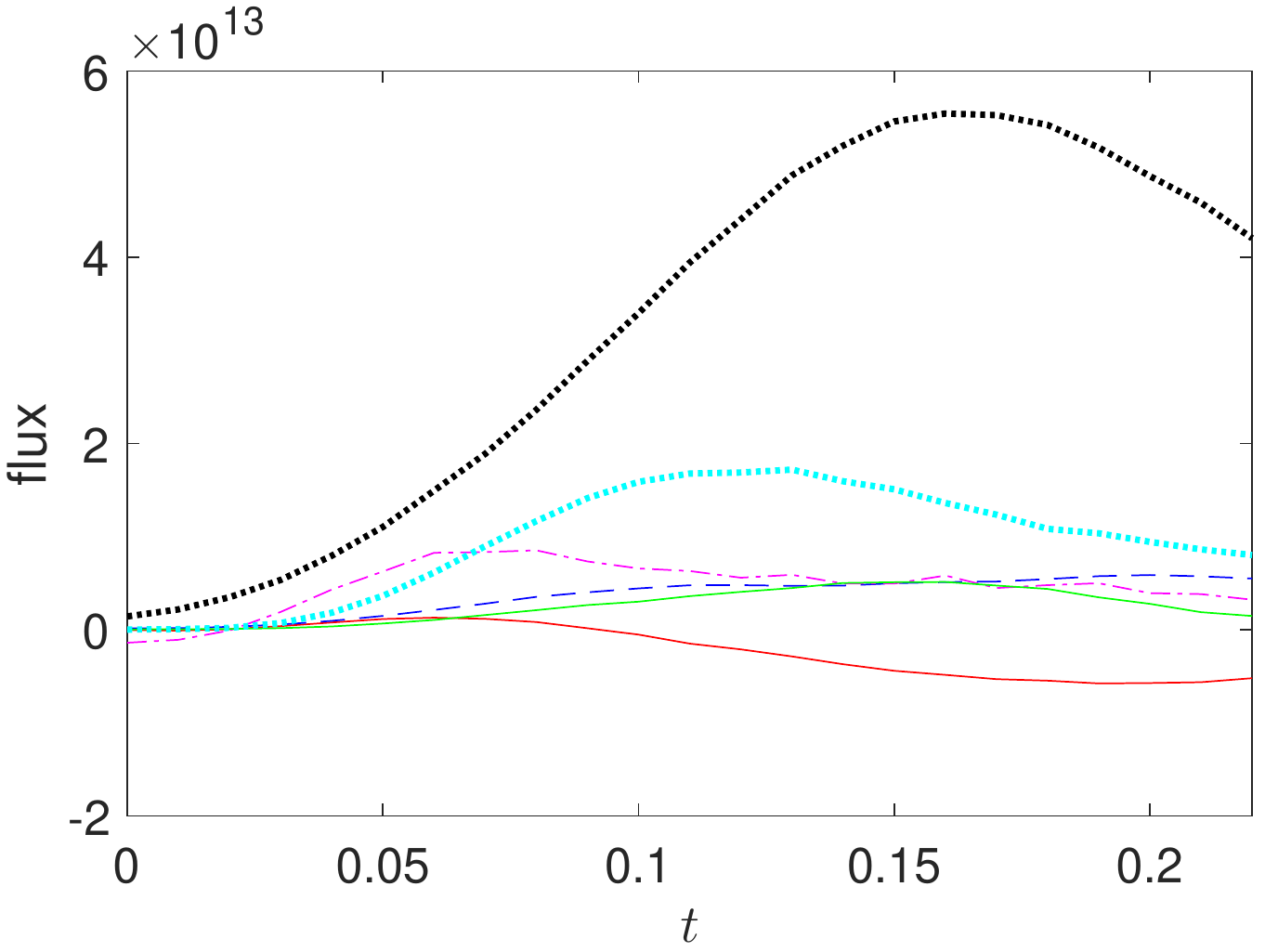}}\qquad
\subfigure[Extreme case: Flux towards $k>10$]{\includegraphics[width=0.45\textwidth]{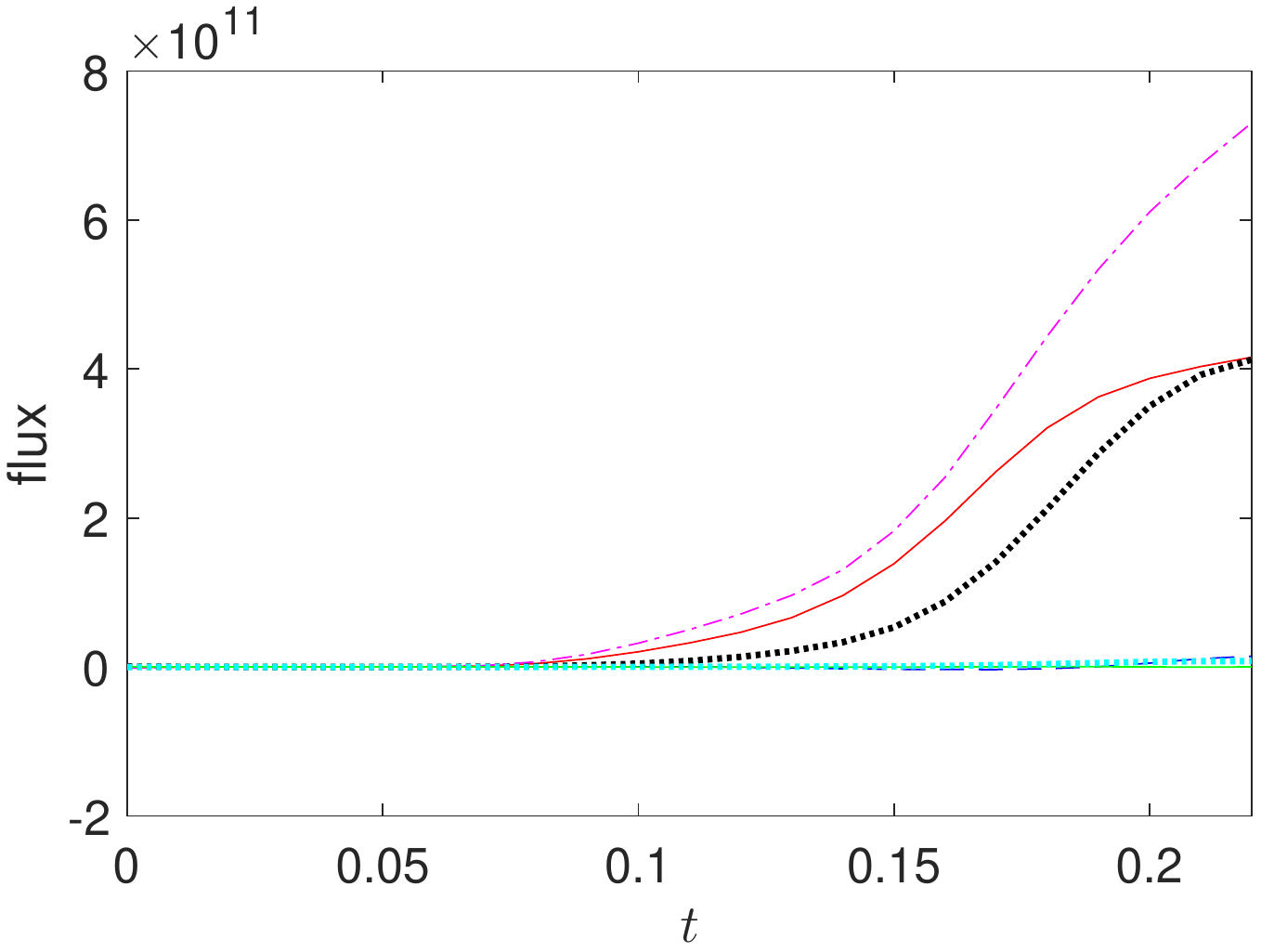}}}\\
\mbox{\subfigure[Generic case: Flux towards $k>2$]{\includegraphics[width=0.45\textwidth]{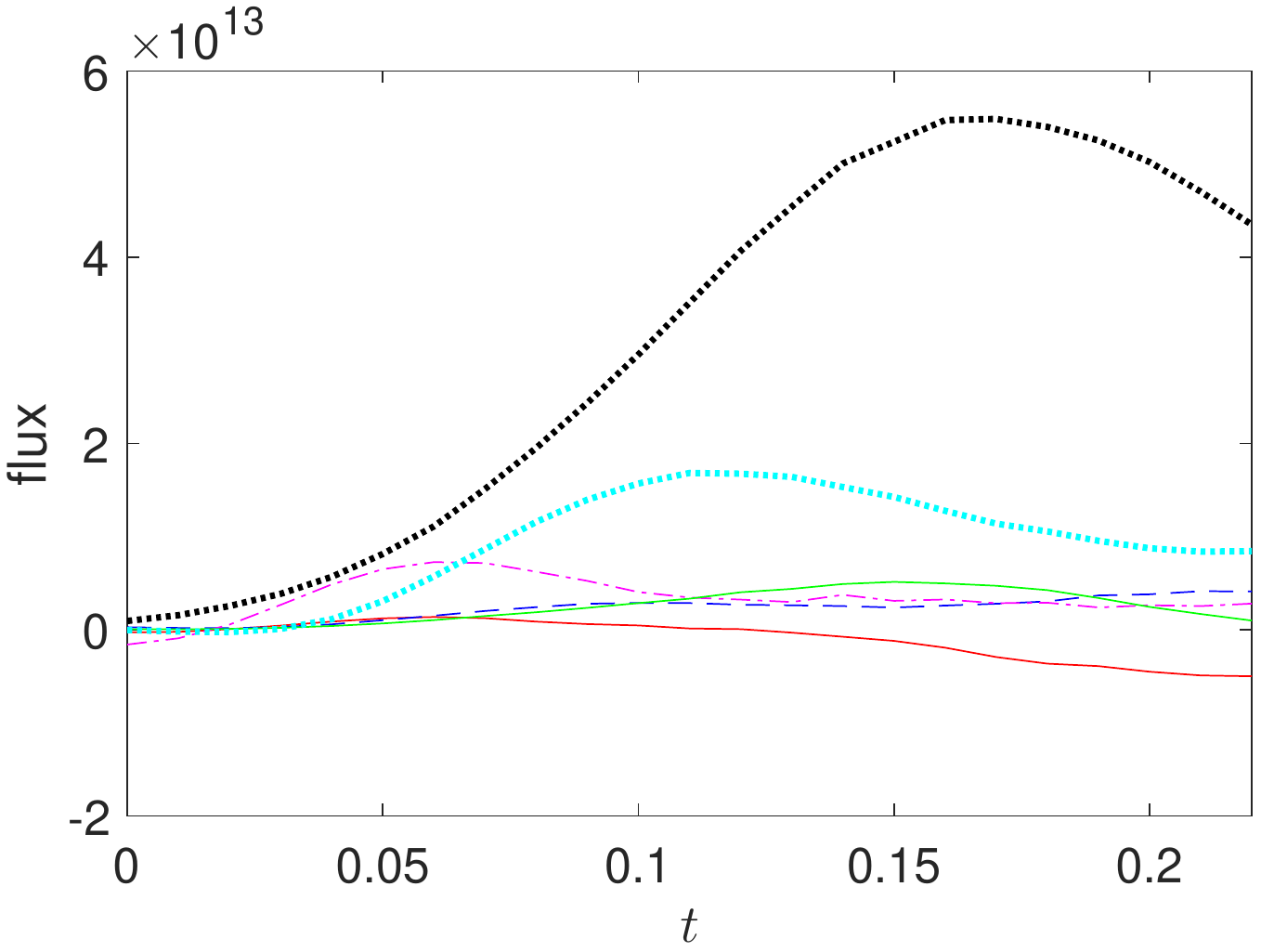}}\qquad
\subfigure[Generic case: Flux towards $k>10$]{\includegraphics[width=0.45\textwidth]{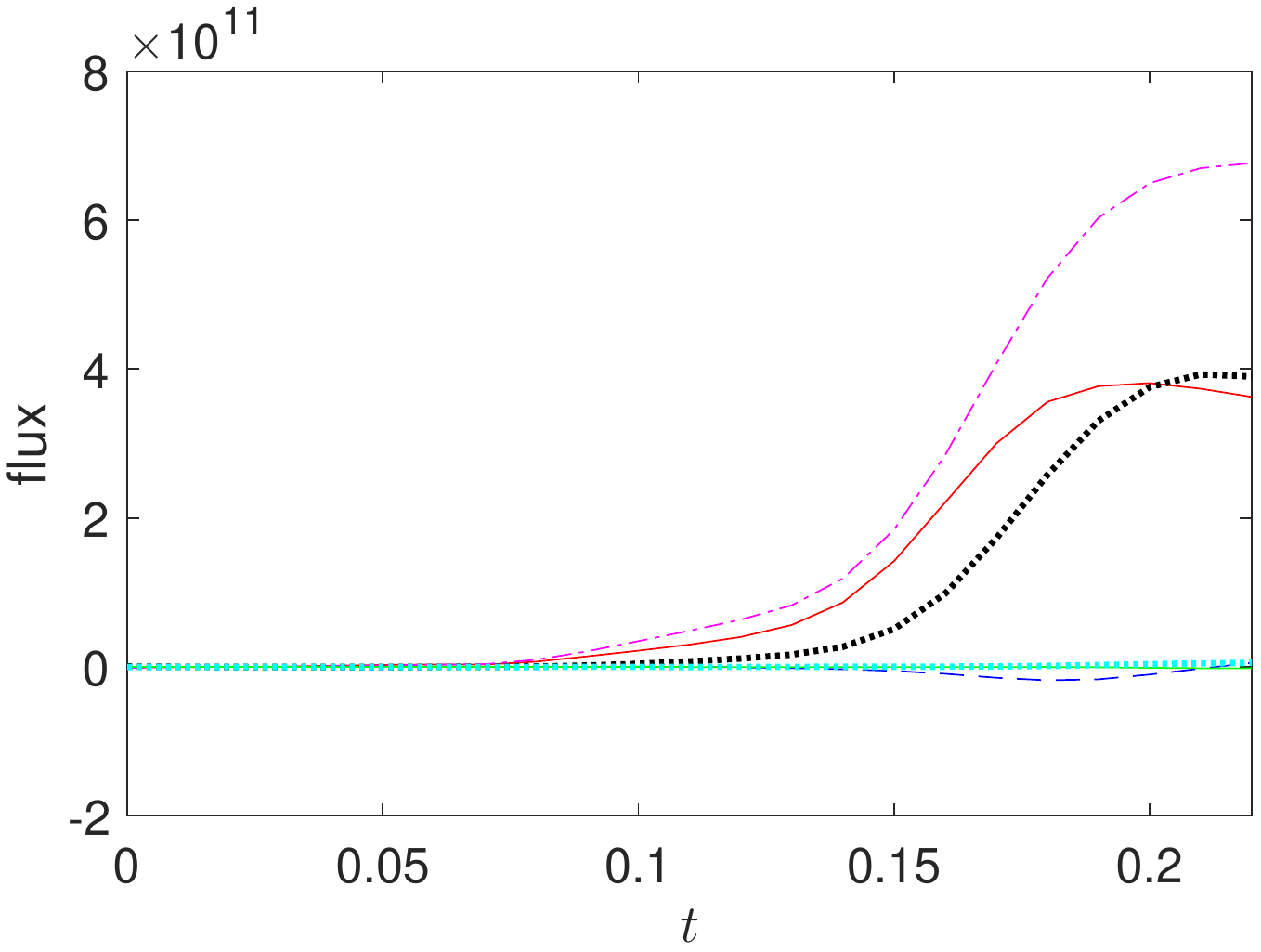}}}
\mbox{\subfigure[Taylor-Green case: Flux towards $k>2$]{\includegraphics[width=0.45\textwidth]{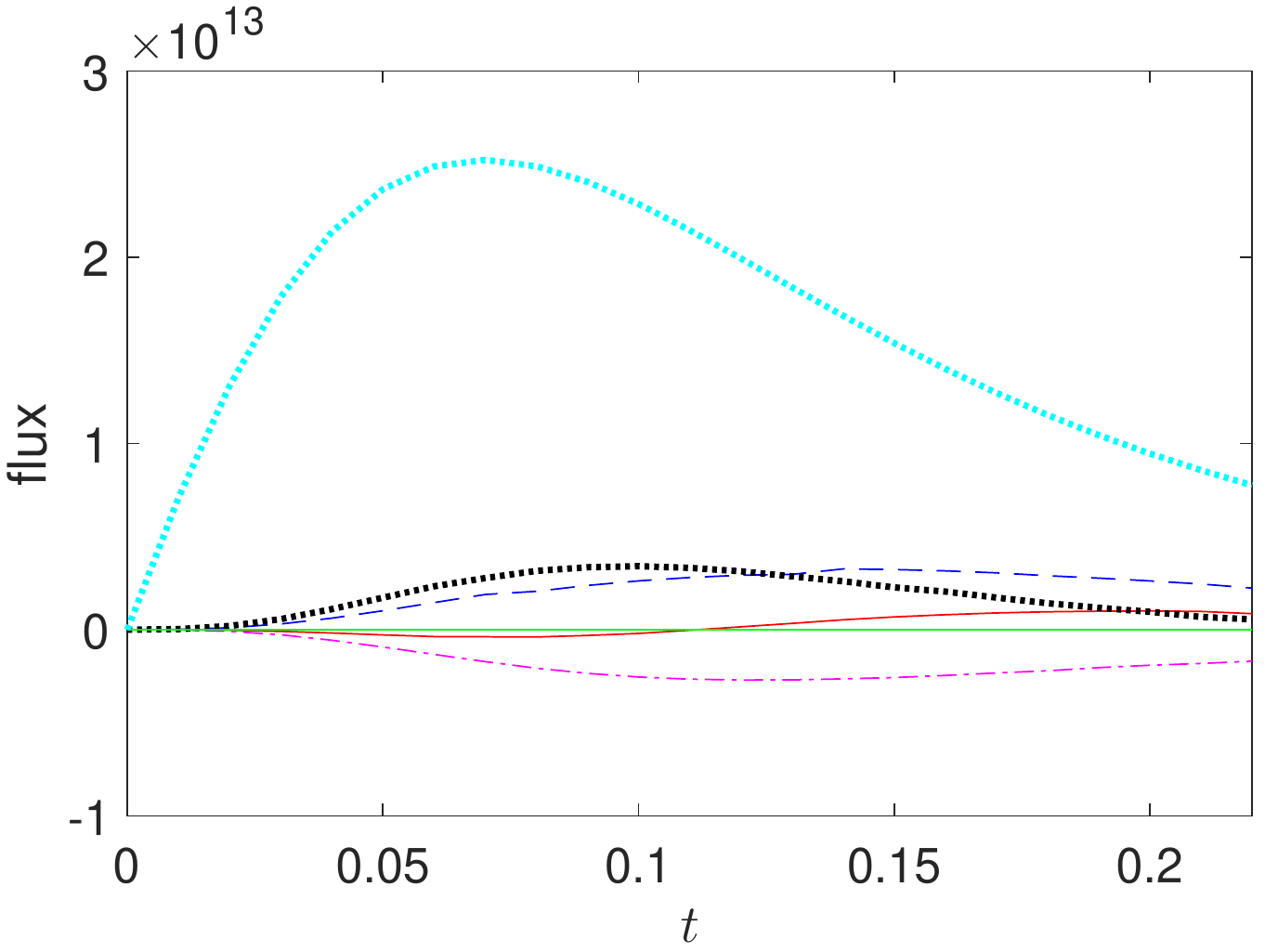}}\qquad
\subfigure[Taylor-Green case: Flux towards $k>10$]{\includegraphics[width=0.45\textwidth]{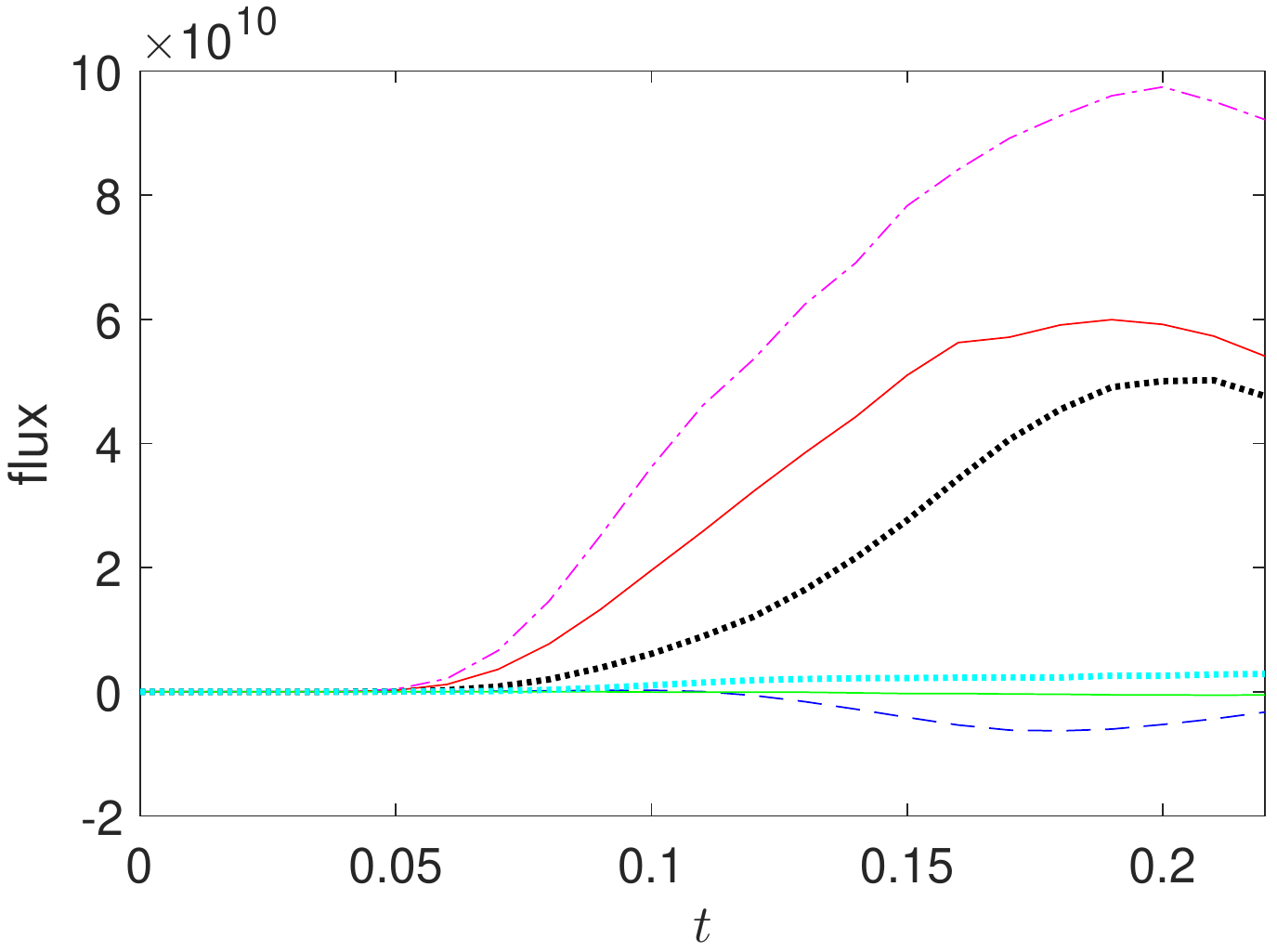}}}
\caption{Time evolution of the fluxes towards $\mathcal{C}_2$ ($k>2$)
  (left panels (a), (c) and (e)) and $\mathcal{C}_{10}$ ($k>10$)
  (right panels (b), (d) and (f)) in the solution of the Navier-Stokes
  system with the extreme initial condition (top panels (a) and (b)),
  generic initial condition (middle panels (c) and (d)) and
  Taylor-Green initial condition (bottom panels (e) and (f))
  contributed by the different triad types: ``PPP'' (red solid line),
  ``PPM'' (blue dashed line), ``PMP'' (black dotted line), ``PMM''
  (magenta dash-dotted line), ``P(PM)'' (green solid line), and
  ``(PM)P '' (cyan dotted line).}
%\caption{Time evolution of the flux of the solution of the Navier-Stokes equation with the different triad types ``PPP'' (red solid line), ``PPM'' (blue dashed line), ``PMP'' (black dotted line) , ``PMM'' (magenta dash-dotted line), ``P(PM)'' (green solid line), and ``(PM)P '' (cyan dotted line), for the initial condition cases
% (a) extreme with $k=2$, (b) extreme with $k=10$, (c) generic with $k=2$, (d) generic with $k=10$, (e) Taylor-Green with $k=2$ and (f) Taylor-Green with $k=10$, for $\E_0 = 250$.}
\label{fig:NSFlux_cases}
\end{center}
\end{figure}

For the above reason we need to look at the positive-definite flux
densities $F_{\mathcal{C}_2}^{s_1 s_2 s_3}(\Phi)$ constructed in terms
of the contribution to the flux $\Pi_{\mathcal{C}_2}^{s_1 s_2 s_3}$
from generalized helical triad phases $\Phi_{\mathbf{k}_1 \mathbf{k}_2
  \mathbf{k}_3}^{s_1 s_2 s_3}$ {with values $\Phi$, cf.~relation
  \eqref{eq:FC}.}  The plot of these fluxes $\Pi_{\mathcal{C}_2}^{s_1
  s_2 s_3}$ as functions of time for each helical triad type is shown
in figure \ref{fig:NSFlux_cases}(a). It is evident from this plot that
there is an early energy transfer out of the sphere defined by $k=2$
due to the PMP triad type, which is the main contributor to the flux
having a maximum flux at a relatively late time ($t\approx 0.17$). The
second most important contributor to the positive flux is the boundary
triad type (PM)P, whose contribution is comparable to the one from the
PMP triads (somewhat surprisingly as the number of (PM)P triads is
relatively small), but has an earlier maximum flux, at about $t\approx
0.11$. The third contributor in importance is the PMM triad type, with
a less significant contribution to the positive flux except for early
times, with a maximum at about $t\approx 0.06$. A notable feature of
this PMM triad type is the evidence of an initially negative
contribution, which nevertheless does not show up in the negative flux
plot in figure \ref{fig:NSflux}(b) as the total flux is positive once
the contribution from the PMP triad type is added.

Having established that the different triad types exhibit different
time dynamics, we now analyze the detailed phase dependence of the
normalized flux densities {${\mathcal{W}}_{\mathcal{C}_2}^{s_1
    s_2 s_3}(\Phi)$, cf.~expression \eqref{eq:WC},} which are shown in
figure \ref{fig:NSWpdfkb2} for the six helical triad types.  The fact
that these plots provide, at any given time, normalized distributions,
is useful for the purpose of detecting phase coherence and complements
the flux plots in figure \ref{fig:NSFlux_cases}(a): in fact, the
normalizing factors above can be reconstructed from these
complementary figures. A striking coherent pattern emerges, showing
for the first time the evolution of the structure of the flux-carrying
helical triads during an extreme event {occurring in the} forward
energy cascade.  Qualitatively, a remarkable feature characterizing
all the six helical triad types is that salient peaks of the weighted
PDFs (darker, redder colours) seem to propagate in time in a
continuous manner, revealing distinctive braid patterns which suggest
that the sets of flux-carrying helical triads slowly change over time.
Although we will not attempt a thorough classification of these sets,
the case of the (PM)P boundary helical triad gives a clue for a
quantitative analysis. There, the weighted PDF (figure
\ref{fig:NSWpdfkb2}(f)) shows a clear and persistent concentration in
the positive-flux range $-\pi/2 < \Phi < \pi/2$, near the time of its
maximum flux ($t\approx 0.11$), and it is evident that these peaks are
well correlated with the peaks in the PDF (figure
\ref{fig:NSpdfkb2}{(d)}), leading to the conclusion that the
majority of the triads in this helical triad type contribute actively
to the flux.  Although these (PM)P triads represent less than the
$0.0001\%$ of all triads {participating in} the flux towards
modes with $k>2$, they contribute with over $20\%$ of the total flux
at $t \approx 0.11$ (these numbers can be calculated from an
assessment of the number of triads in table
\ref{tab:PDFbounds_extreme} (second column) and from the flux plot in
figure \ref{fig:NSFlux_cases}(a)).  The preference for positive flux
contributions is also evident for other triad types.  Figure
\ref{fig:NSWpdfkb2}(d) shows, for the PMM triad type, the strongest
concentration of the weighted PDF at positive-flux phases $\Phi
\approx \pm 0.25 \pi$, near the time of its maximum flux ($t\approx
0.06$).  Finally, the weighted PDF for the main contributor to the
positive flux (triad type PMP, figure \ref{fig:NSWpdfkb2}(c)) shows
strong local maxima near $\Phi=0$ at times $t\approx 0.06$ and
$t\approx0.11$, while a set of weaker local maxima {at
  positive-flux phases $\Phi \approx \pm 0.1 \pi$ occurs} at the time
of the maximum flux ($t\approx 0.17$). This latter case demonstrates
how important it is to complement the flux plots in figure
\ref{fig:NSFlux_cases}(a) with the weighted PDF plots in figure
\ref{fig:NSpdfkb2} {as the former plots do not contain
  information about the flux densities for} different triad types.
Even negative-flux contributions can be identified by these analyses.
Figure \ref{fig:NSWpdfkb2}(a) shows, for the PPP triad type, a
late-time preference for negative-flux contributions (weighted PDF
concentrated near $\Phi=\pm\pi$), which is evidenced in the plot of
the total PPP flux in figure \ref{fig:NSFlux_cases}(a), solid red
line.

\begin{figure}
\begin{center}
\mbox{\subfigure[${\mathcal{W}}_{\mathcal{C}_2}^{PPP}(\Phi)(t)$]{\includegraphics[width=0.45\textwidth]{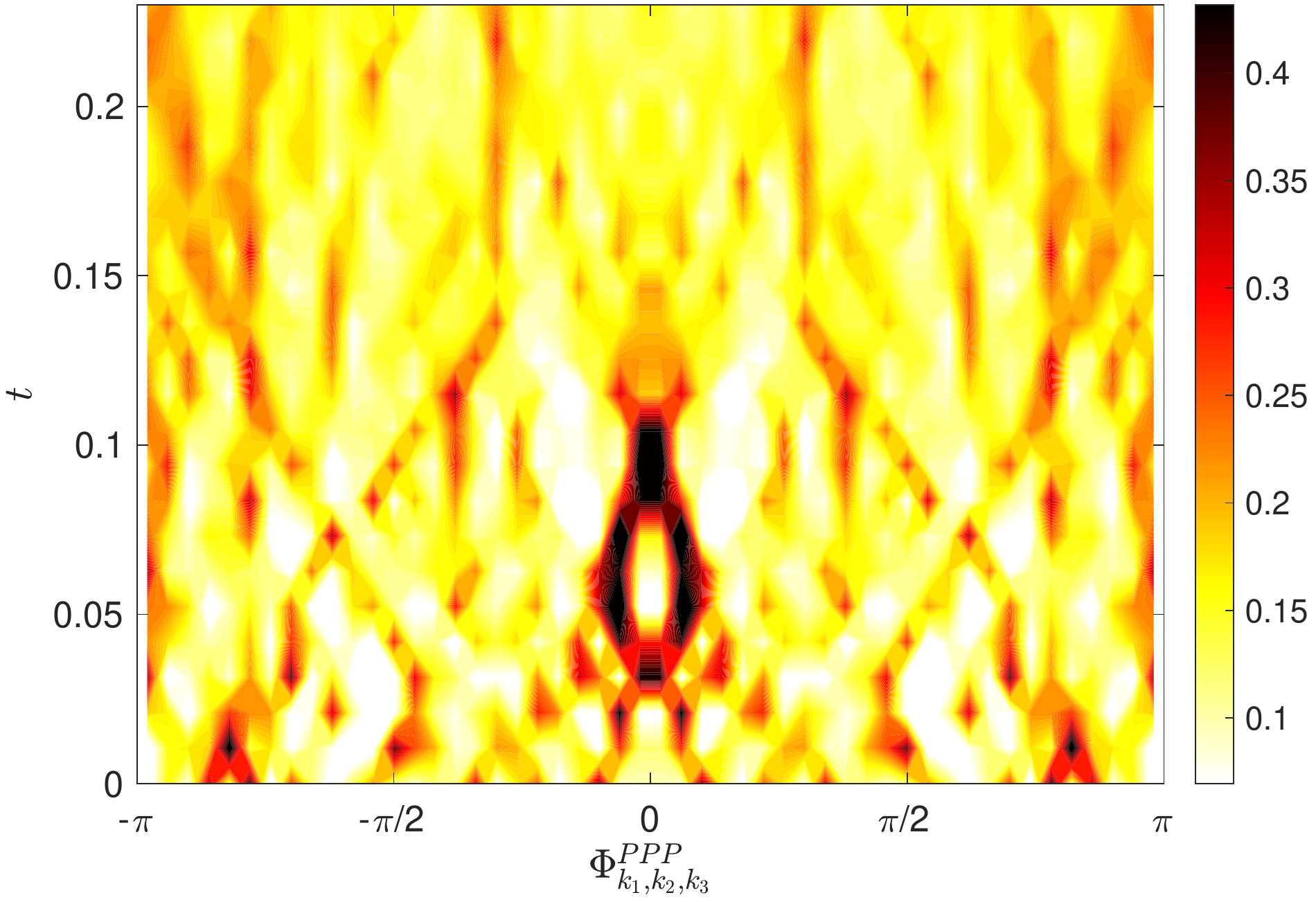}}\qquad
\subfigure[${\mathcal{W}}_{\mathcal{C}_2}^{PPM}(\Phi)(t)$]{\includegraphics[width=0.45\textwidth]{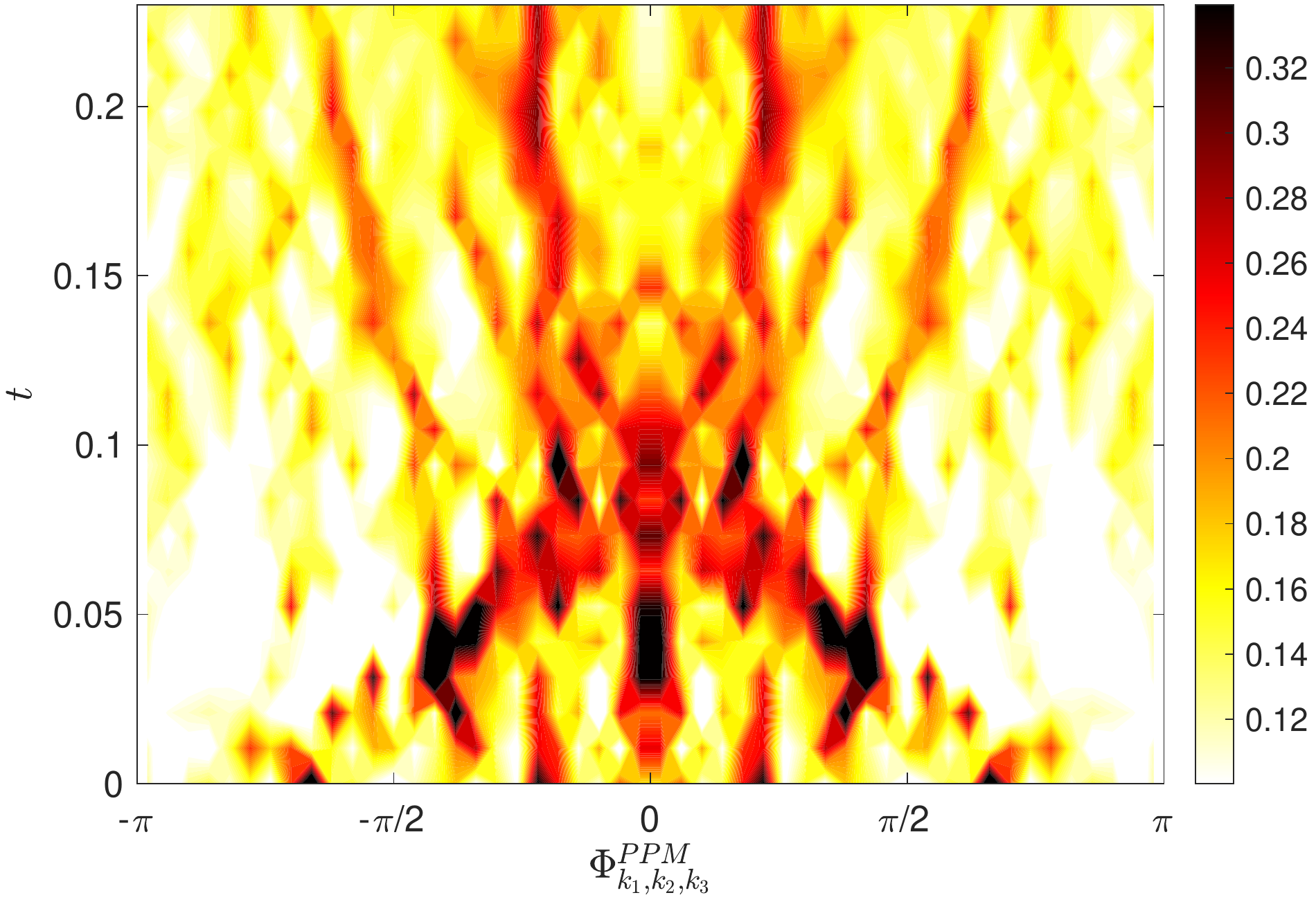}}}\\
\mbox{\subfigure[${\mathcal{W}}_{\mathcal{C}_2}^{PMP}(\Phi)(t)$]{\includegraphics[width=0.45\textwidth]{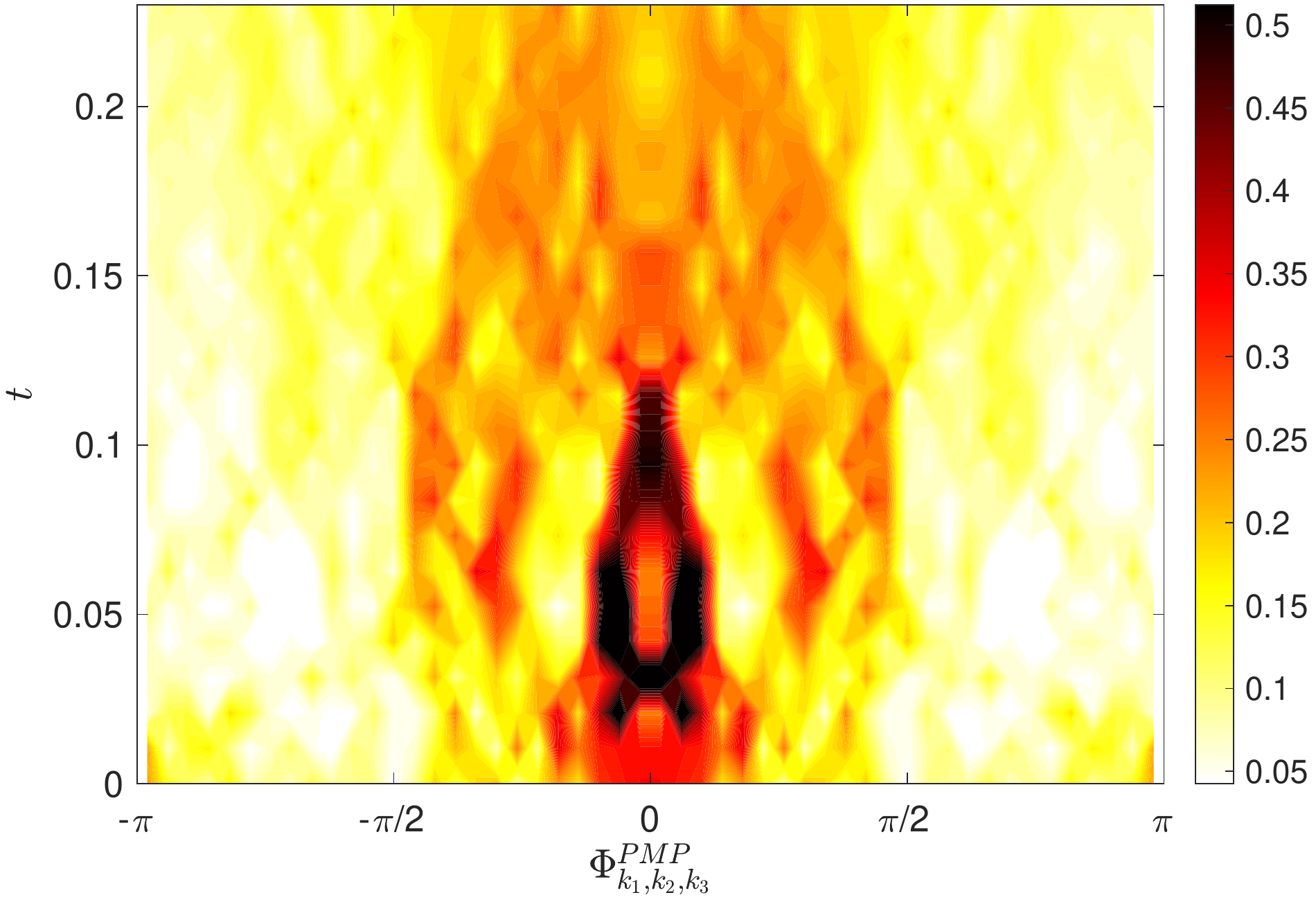}}\qquad
\subfigure[${\mathcal{W}}_{\mathcal{C}_2}^{PMM}(\Phi)(t)$]{\includegraphics[width=0.45\textwidth]{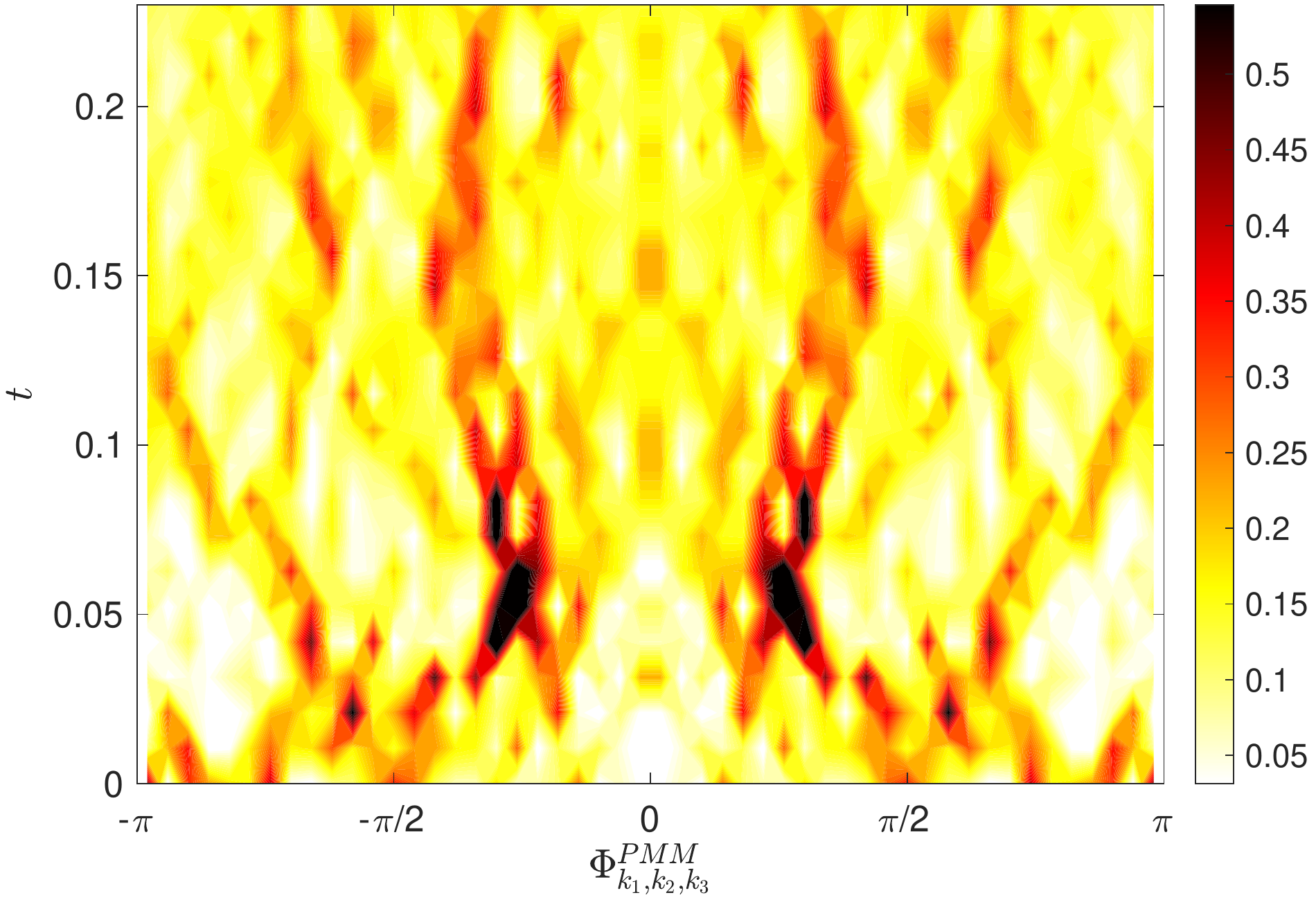}}}
\mbox{\subfigure[${\mathcal{W}}_{\mathcal{C}_2}^{P(PM)}(\Phi)(t)$]{\includegraphics[width=0.45\textwidth]{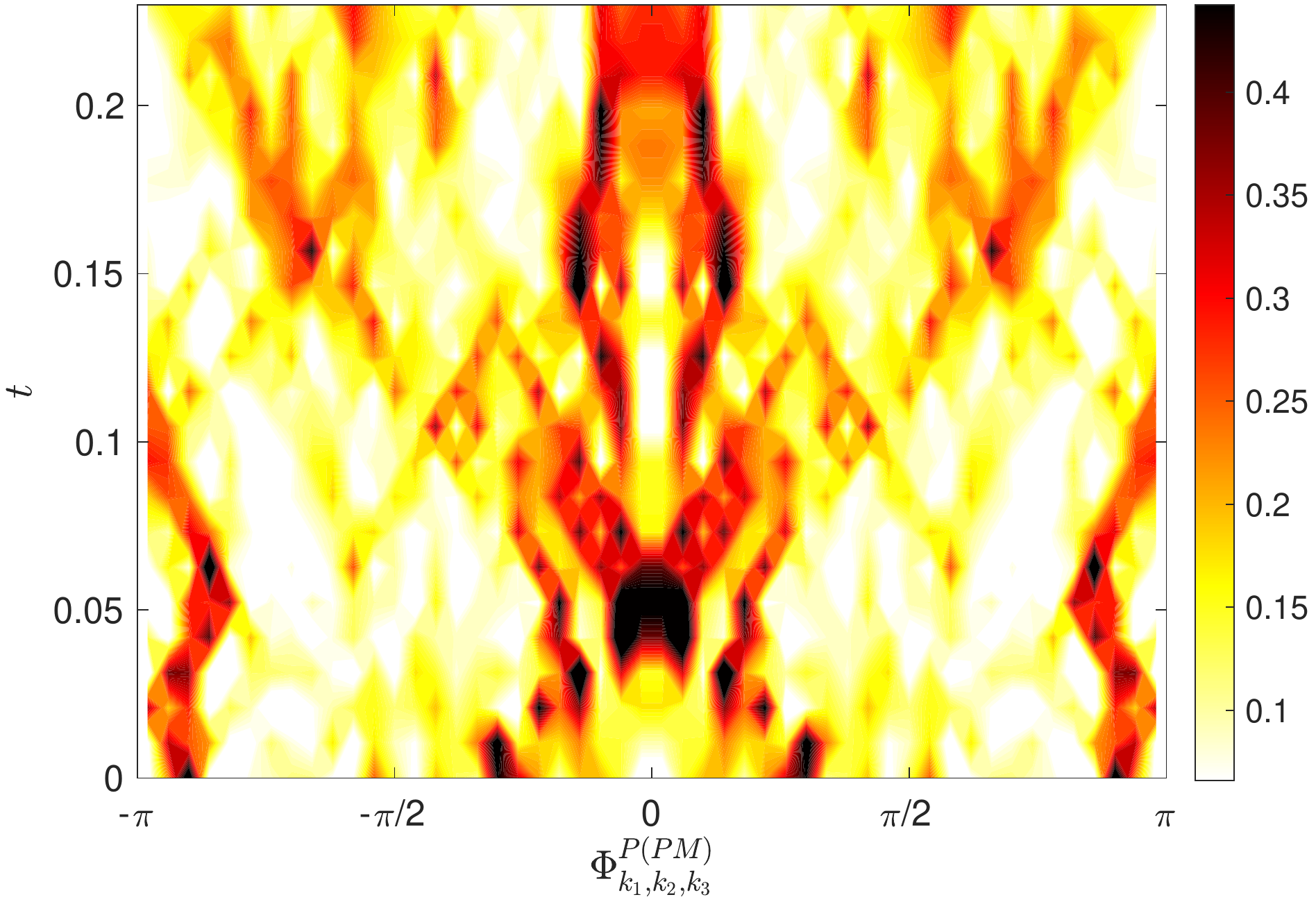}}\qquad
\subfigure[${\mathcal{W}}_{\mathcal{C}_2}^{(PM)P}(\Phi)(t)$]{\includegraphics[width=0.45\textwidth]{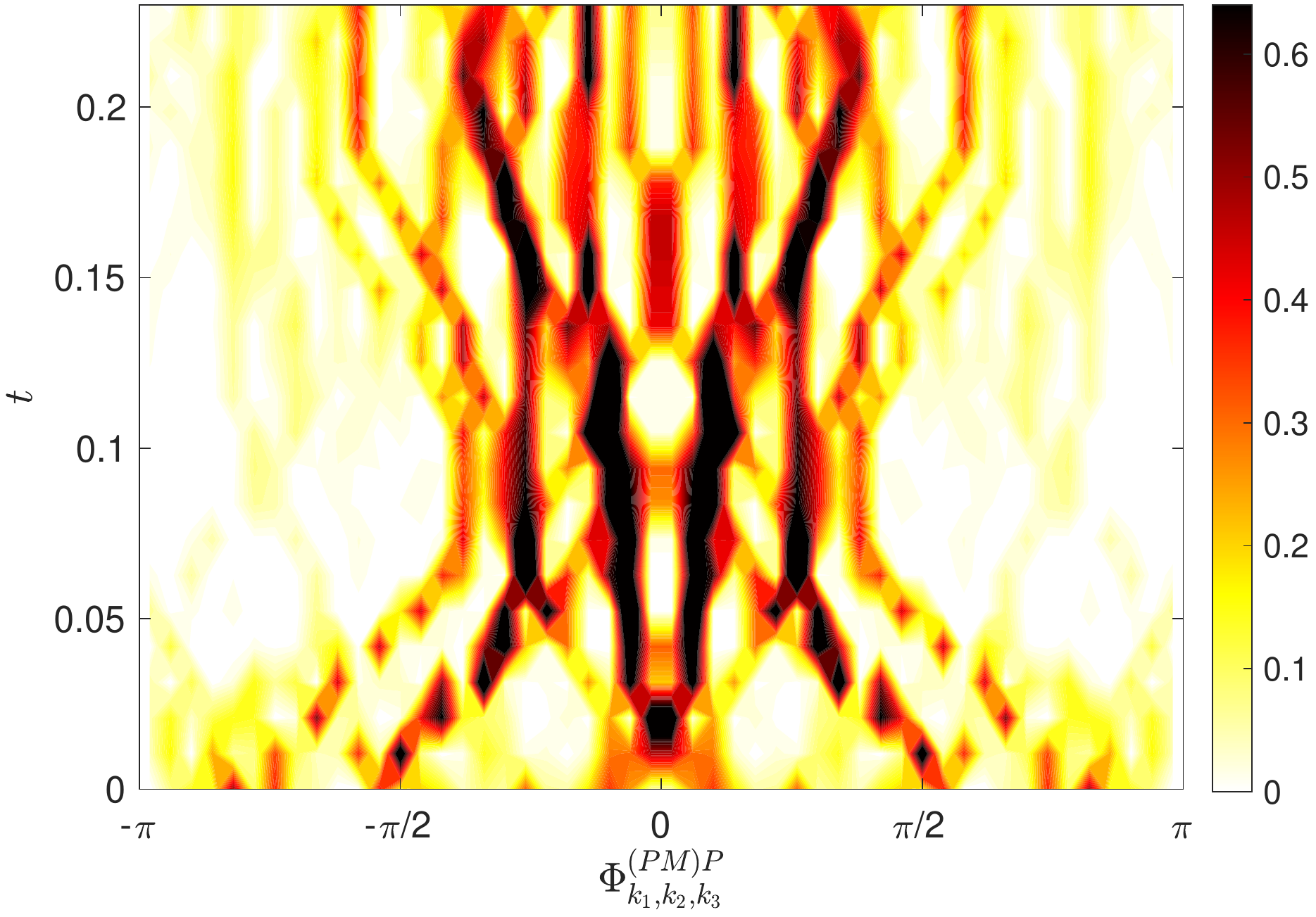}}}

\caption{{[Extreme case, $k>2$:]} Time evolution of the weighted
  PDF of the triad phase angle ${\mathcal{W}}_{\mathcal{C}_2}^{s_1 s_2
    s_3}(\Phi)$, cf.~expression \eqref{eq:WC}, in the Navier-Stokes
  flow with the extreme initial data for the triad types (a) ``PPP'',
  (b) ``PPM'', (c) ``PMP'' , (d) ``PMM'', (e) ``P(PM)'' and (f)
  ``(PM)P''. In these plots the uniform distribution corresponding to
  $1/2\pi \approx 0.16$ is depicted by yellow ($25\%$ of the colour
  bars).}
\label{fig:NSWpdfkb2}
\end{center}
\end{figure}

In order to obtain a better idea of how the energy leaving the sphere
$k=2$ is transported after reaching the dissipative range, we now
perform an analogous study to the one just described, but this time
regarding the energy flux out of the sphere $k=10$, namely towards the
set $\mathcal{C}_{10} = \{\mathbf{k} \in \mathbb{Z}^3 \setminus \{0\}
\quad | \quad |\mathbf{k}| > 10\}$.  Table \ref{tab:PDFbounds_extreme}
(sixth column) shows that there is only one helical triad type --- the
boundary triad (PM)P --- with a non-uniform PDF over the time window
considered, cf.~figure \ref{fig:NSpdfkb10}(a).  In contrast, when we
consider the flux densities and the corresponding weighted PDFs
$W_{\mathcal{C}_{10}}^{s_1 s_2 s_3}(\Phi)$, cf.~relations
\eqref{eq:FC} and \eqref{eq:WC}, shown in figure \ref{fig:NSWpdfkb10},
clear coherent patterns emerge, although they are not as sustained
over time as in the case with $k>2$; rather, they typically consist of
an initial burst of concentrated coherent flux followed by a more
disperse tail. Again, for a complete picture, these are to be
complemented with plots of the fluxes $\Pi_{\mathcal{C}_{10}}^{s_1 s_2
  s_3}$ shown in figure \ref{fig:NSFlux_cases}(b): now, the main
contributor to the flux towards $\mathcal{C}_{10}$ is triad type PMM,
followed by PPP and PMP.  The plot shows that a stage of sustained
growth begins at around $t=0.075$. Could this growth stage be related
to the coherent patterns observed in figure \ref{fig:NSWpdfkb10}?
Possibly yes, as coherent patterns develop near that time among
positive-flux phases, leading to a build-up of energy over time. As
for the main contributor to the energy flux, namely the set of PMM
helical triads, its weighted PDF is shown in figure
\ref{fig:NSWpdfkb10}(d), revealing a strong and coherent event near
$\Phi \approx \pm 0.075 \pi$ which occurs during $t=0.04$--$0.1$ and
retains coherence after that.  As for the second main contributor, the
set of PPP helical triads, its weighted PDF in figure
\ref{fig:NSWpdfkb10}(a) shows a similar, even more concentrated,
pattern near $\Phi=0$ during the same time range. Notably, the
weighted PDFs of both these triad types show early coherence
($t=0$--$0.03$) at triad phases corresponding to negative flux
contributions. At these early stages the actual fluxes are quite small
so one could conjecture that a kind of ``slingshot effect'' is at work
for this extreme initial condition, similarly to the extreme initial
condition in the Burgers case, where an initial stage of negative-flux
phase alignments (triad phases near $-\pi/2$ in the Burgers case, and
near $\pi$ or $-\pi$ in the Navier-Stokes case) is followed by a
rearrangement of the phases to {favor forward energy cascade} (triad
phases near $\pi/2$ in the Burgers case, and near $0$ in the
Navier-Stokes case).  As for the third contributor to the flux, the
set of PMP helical triads, its weighted PDF in figure
\ref{fig:NSWpdfkb10}(c) shows a clear initial positive-flux
concentration near $\Phi=0$, followed by a dip to a close-to-uniform
distribution at $t=0.04$, leading to another event when the positive
flux attains local maxima in the range $t=0.04 - 0.1$, to then become
gradually less concentrated for the rest of the time window.

\begin{figure}
\begin{center}
\mbox{
\subfigure[${\mathcal{P}}_{\mathcal{C}_{10}}^{(PM)P}(\Phi)(t)$ (extreme)]{\includegraphics[width=0.45\textwidth]{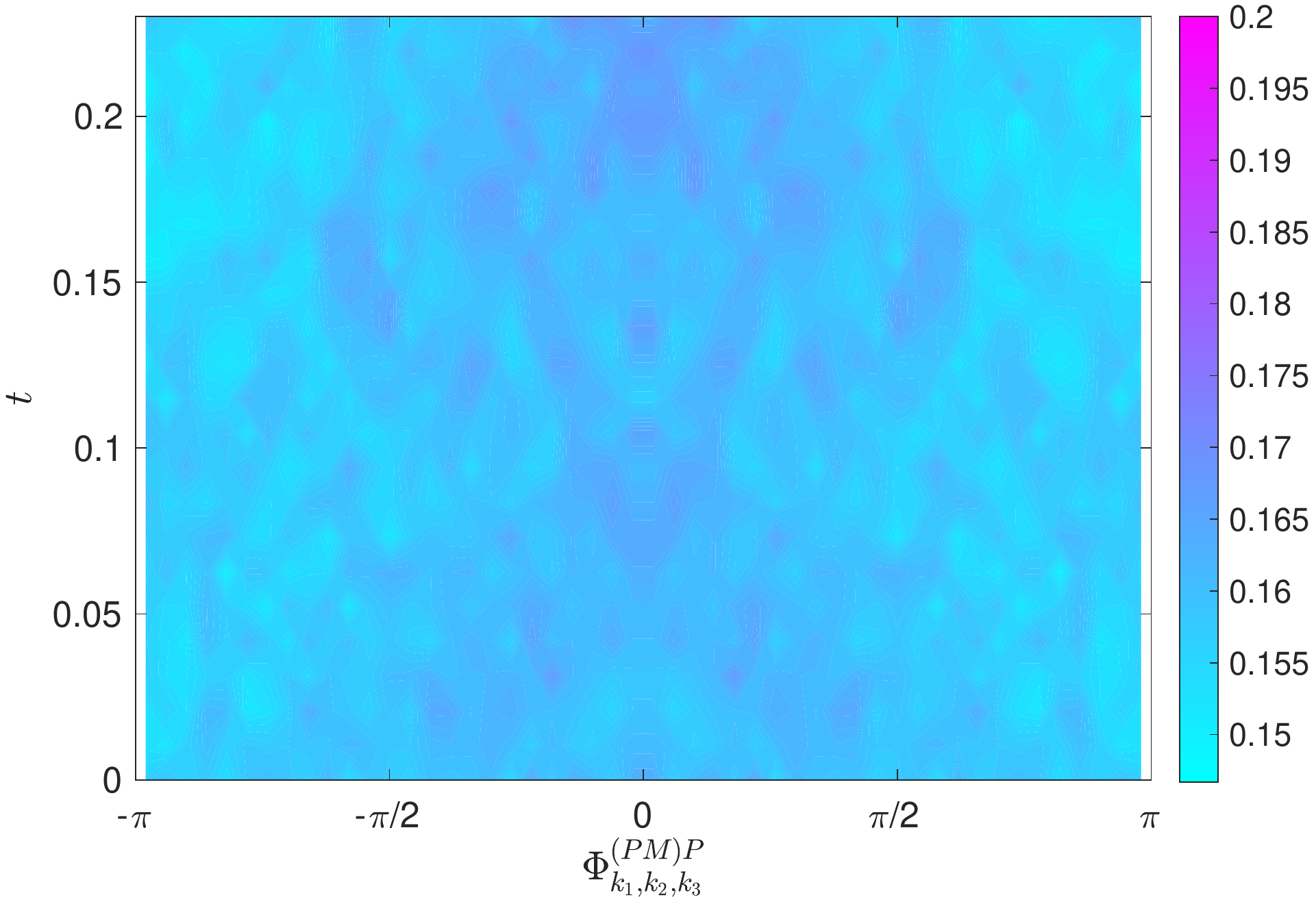}}}
\mbox{\subfigure[${\mathcal{P}}_{\mathcal{C}_{10}}^{(PM)P}(\Phi)(t)$ (generic)]{\includegraphics[width=0.45\textwidth]{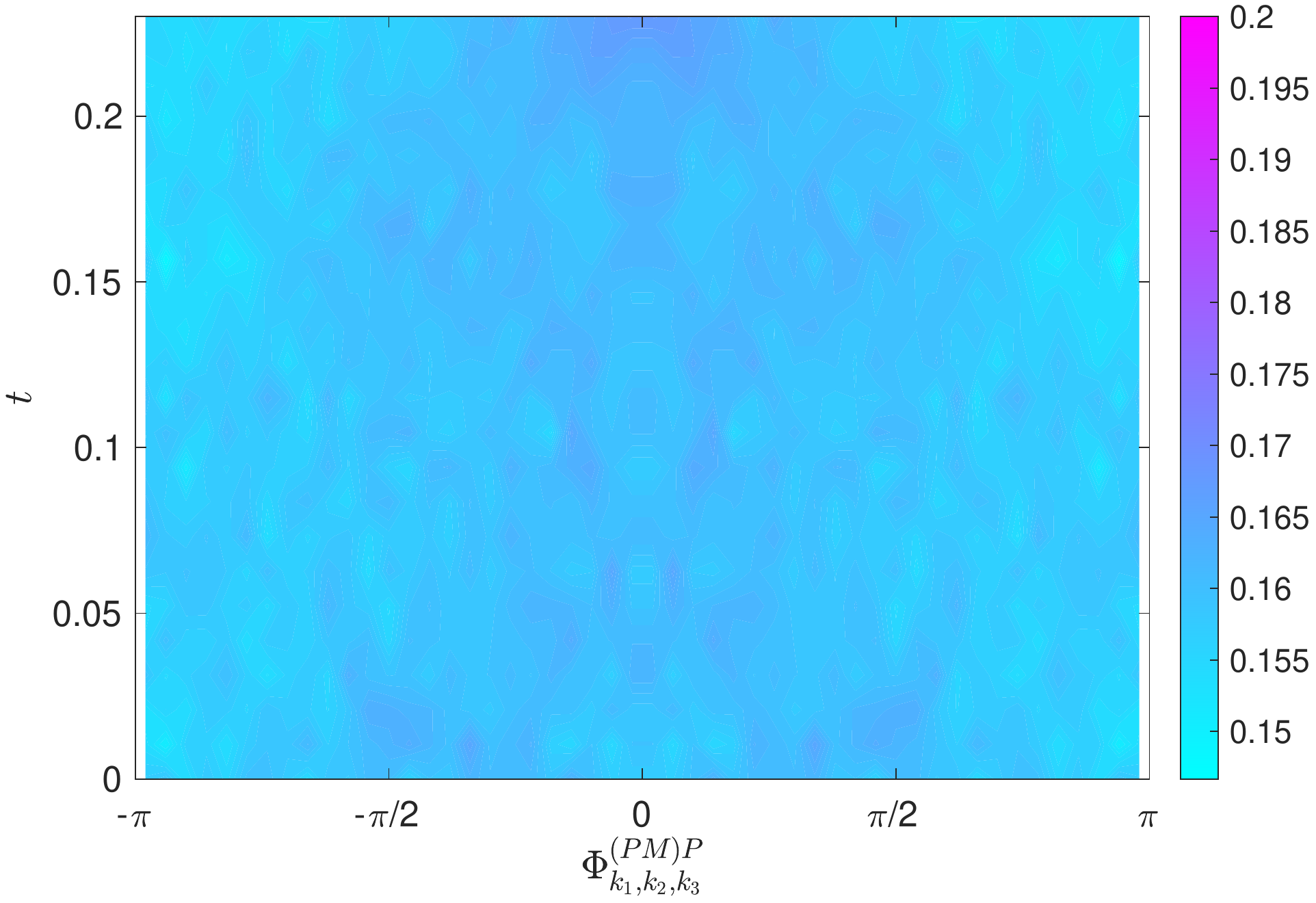}}\qquad
\subfigure[${\mathcal{P}}_{\mathcal{C}_{10}}^{(PM)P}(\Phi)(t)$ (Taylor-Green)]{\includegraphics[width=0.45\textwidth]{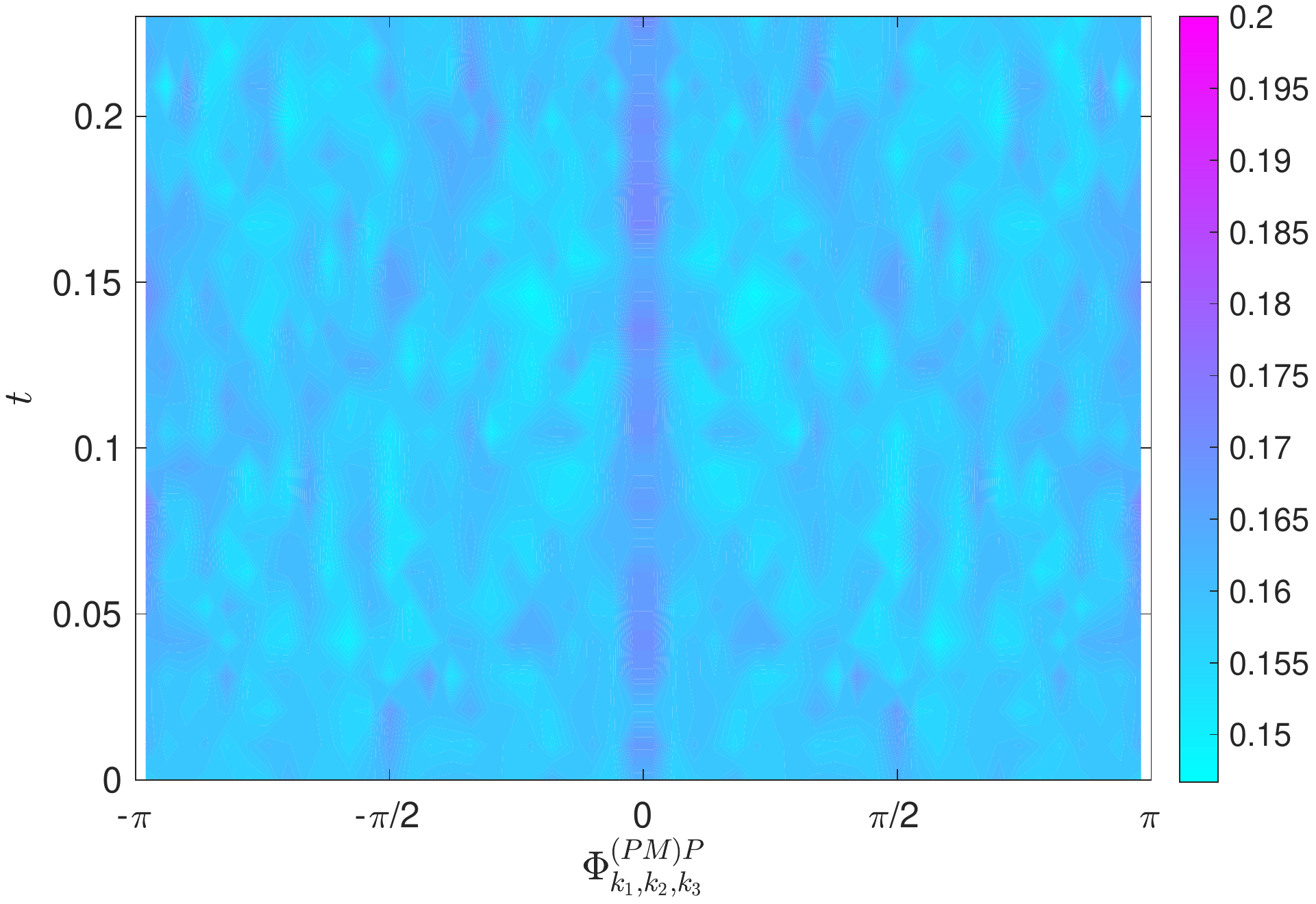}}}
\caption{{[$k>10$:]} Time evolution of the PDFs of the phase triad
  angle ${\mathcal{P}}_{\mathcal{C}_{10}}^{(PM)P}(\Phi)$ in the
  solution of the Navier-Stokes system with the (a) extreme, (b)
  generic, and (c) Taylor-Green initial conditions. {The
    ``(PM)P'' triad is the only triad type exhibiting a PDF with
    variability exceeding $\pm5\%$ with respect to the uniform
    distribution} equal to $1/2\pi$ ($\approx 0.16$, which
  corresponds to light blue colour in the plots). The PDFs for all
  the other five triad types (``PPP'', ``PPM'', ``PMP'', ``PMM'' and
  ``P(PM)'') are not shown, as they are essentially uniform (see
  tables \ref{tab:PDFbounds_extreme}, \ref{tab:PDFbounds_generic} and
  \ref{tab:PDFbounds_TG}).}
\label{fig:NSpdfkb10}
\end{center}
\end{figure}

\begin{figure}
\begin{center}
\mbox{\subfigure[${\mathcal{W}}_{\mathcal{C}_{10}}^{PPP}(\Phi)(t)$]{\includegraphics[width=0.45\textwidth]{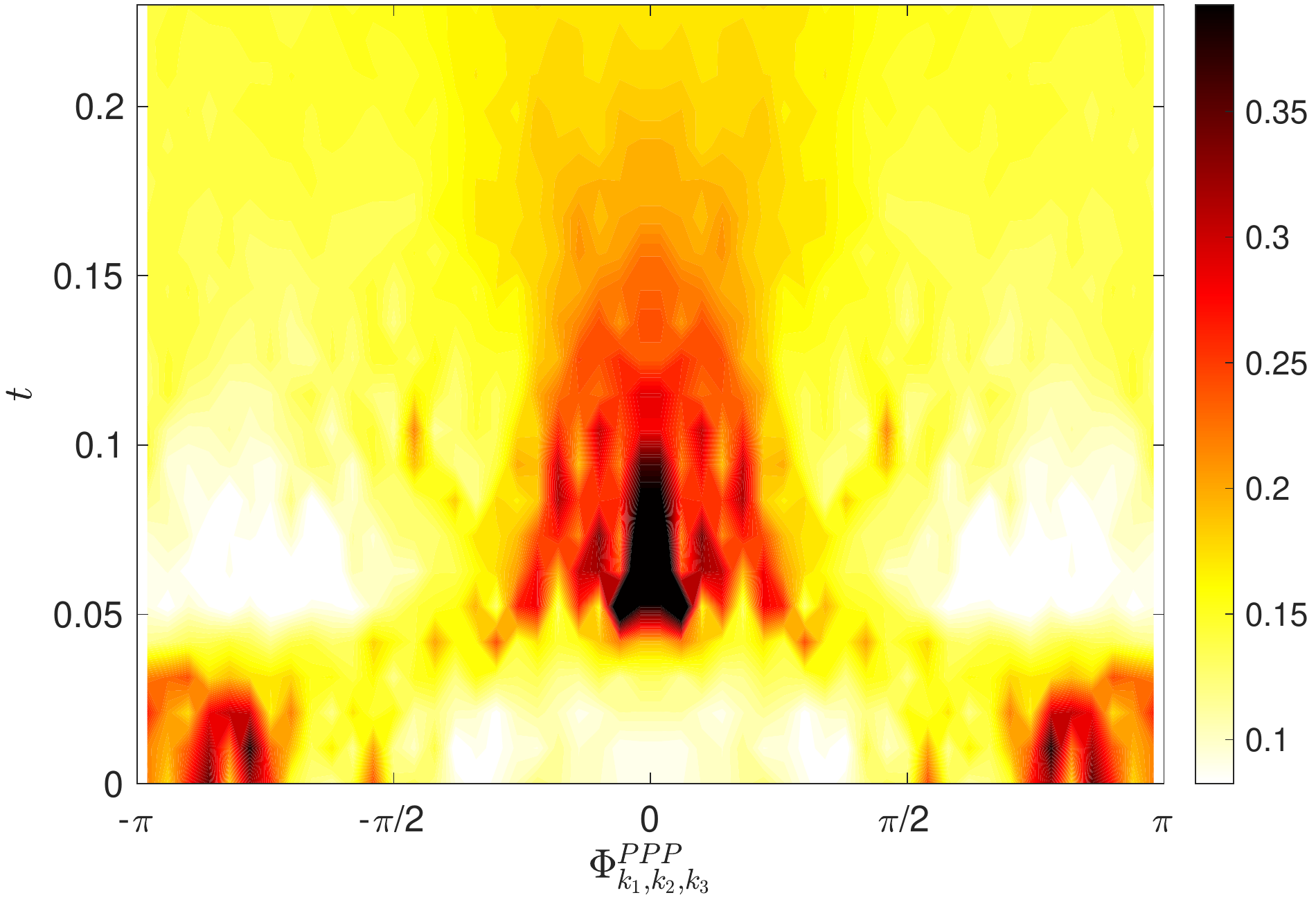}}\qquad
\subfigure[${\mathcal{W}}_{\mathcal{C}_{10}}^{PPM}(\Phi)(t)$]{\includegraphics[width=0.45\textwidth]{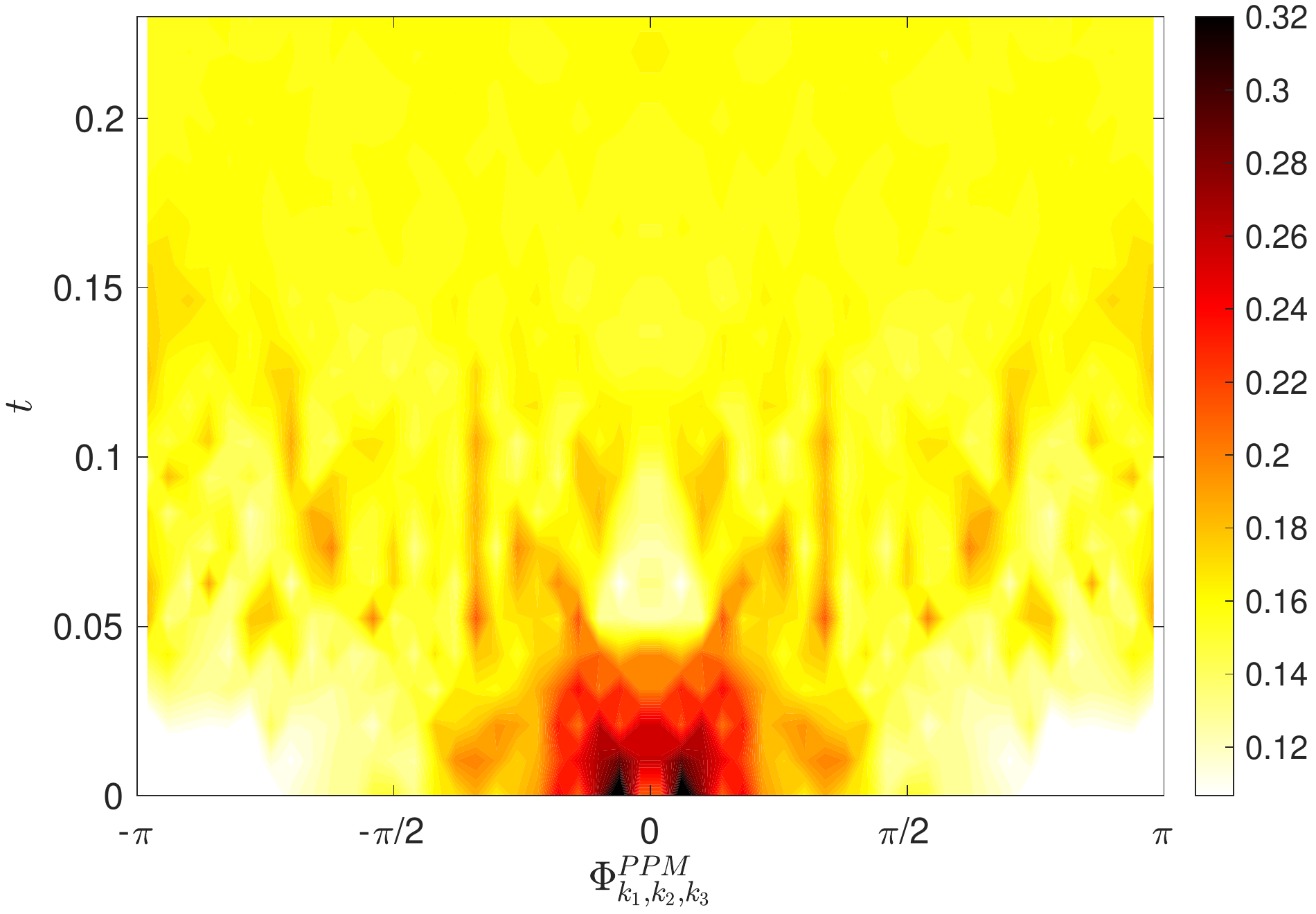}}}\\
\mbox{\subfigure[${\mathcal{W}}_{\mathcal{C}_{10}}^{PMP}(\Phi)(t)$]{\includegraphics[width=0.45\textwidth]{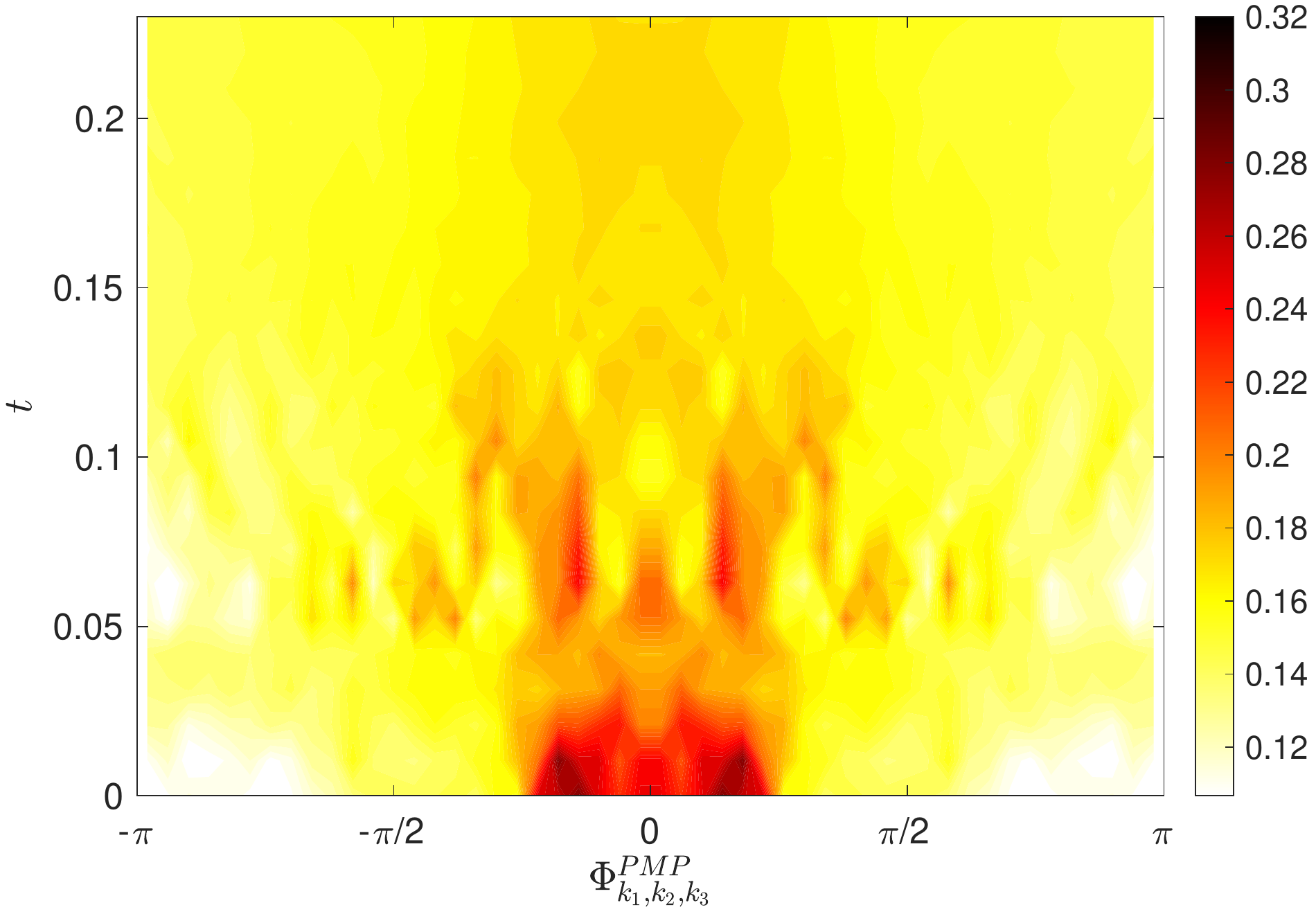}}\qquad
\subfigure[${\mathcal{W}}_{\mathcal{C}_{10}}^{PMM}(\Phi)(t)$]{\includegraphics[width=0.45\textwidth]{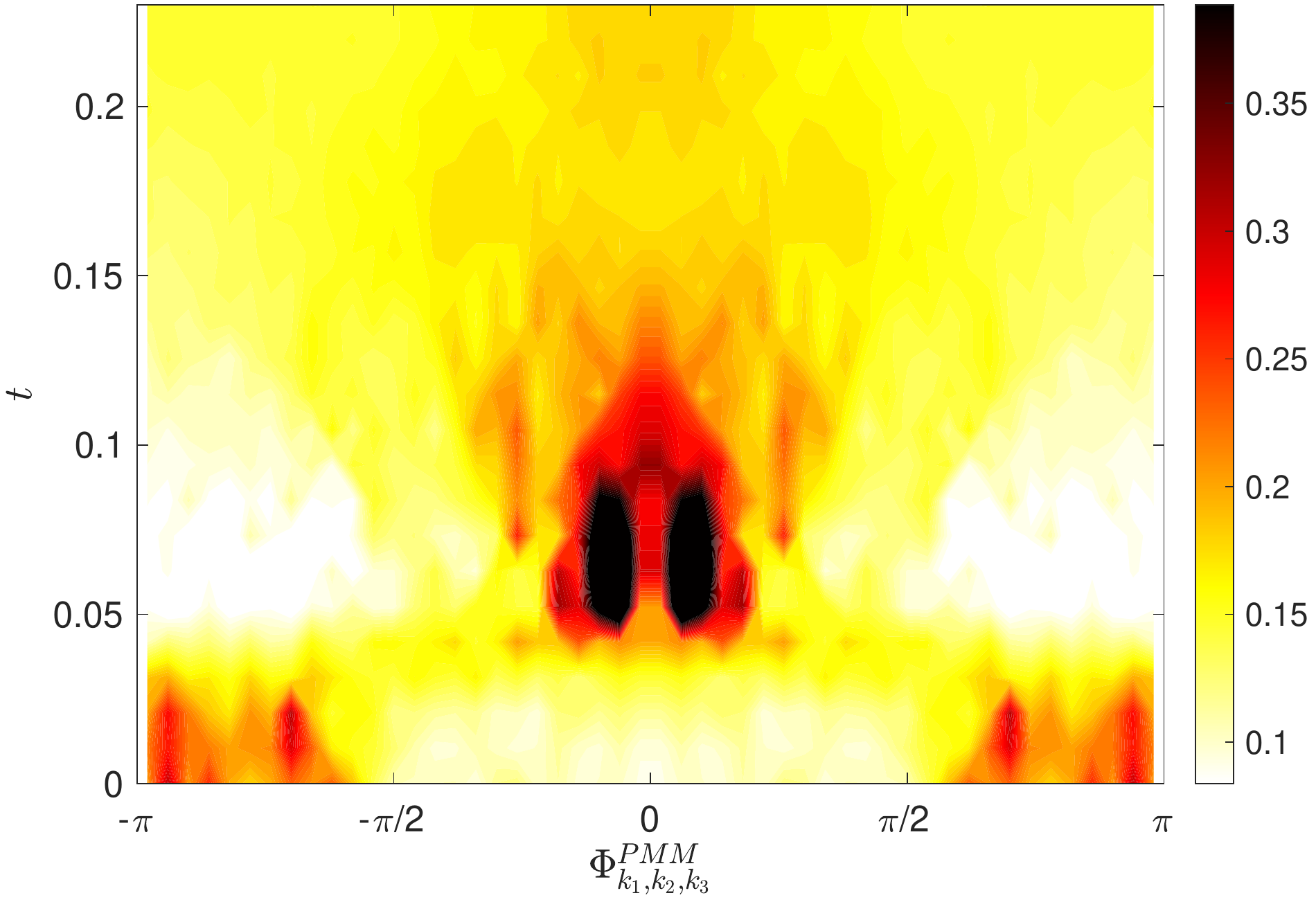}}}
\mbox{\subfigure[${\mathcal{W}}_{\mathcal{C}_{10}}^{P(PM)}(\Phi)(t)$]{\includegraphics[width=0.45\textwidth]{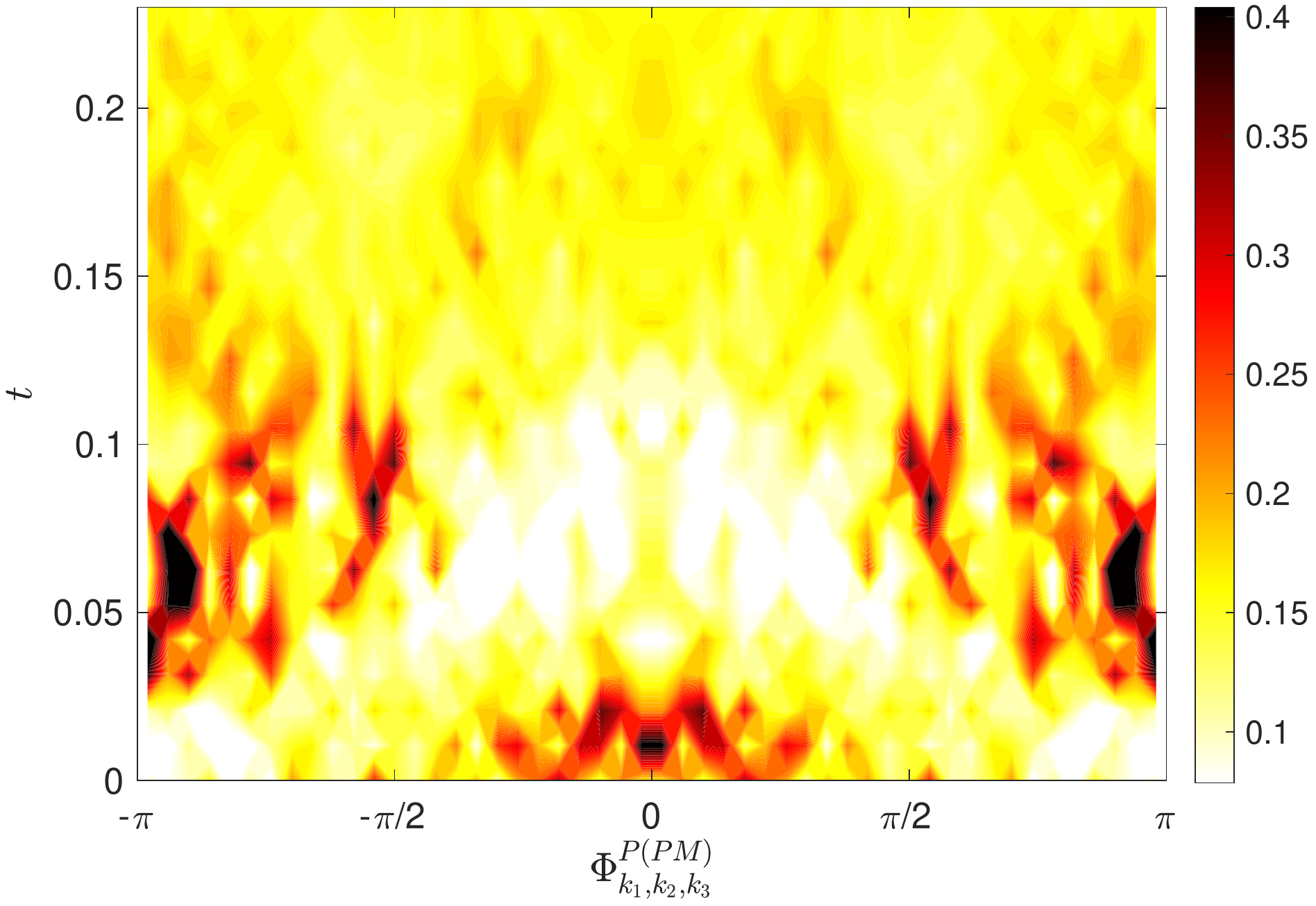}}\qquad
\subfigure[${\mathcal{W}}_{\mathcal{C}_{10}}^{(PM)P}(\Phi)(t)$]{\includegraphics[width=0.45\textwidth]{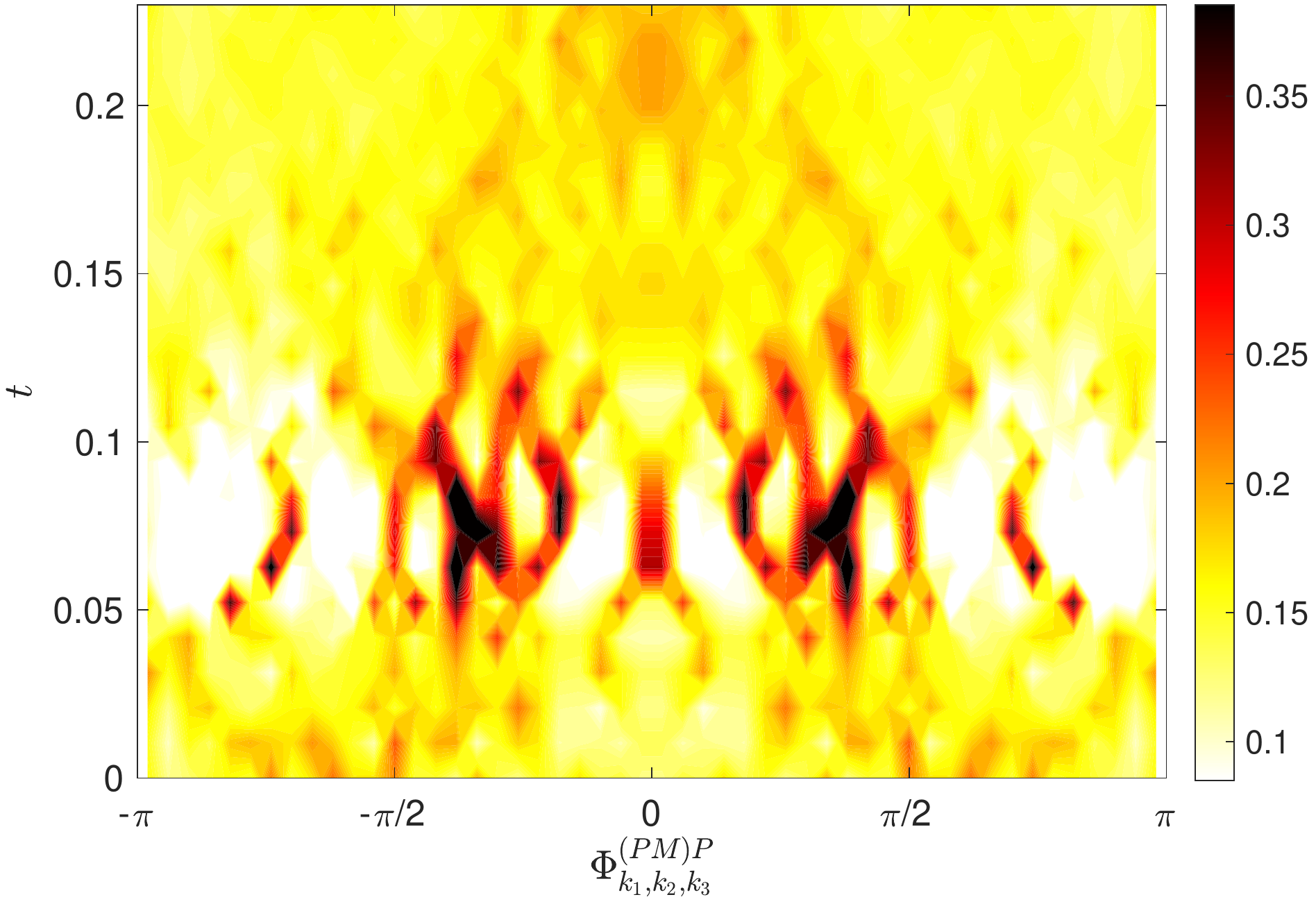}}}

\caption{{[Extreme case, $k>10$:]} Time evolution of the weighted
  PDF of the triad phase angle ${\mathcal{W}}_{\mathcal{C}_{10}}^{s_1
    s_2 s_3}(\Phi)$, cf.~expression \eqref{eq:WC}, in the
  Navier-Stokes flow with the extreme initial data for the triad types
  (a) ``PPP'', (b) ``PPM'', (c) ``PMP'' , (d) ``PMM'', (e) ``P(PM)''
  and (f) ``(PM)P''. In these plots the uniform distribution
  corresponding to $1/2\pi \approx 0.16$ is depicted by yellow ($25\%$
  of the colour bars).}
\label{fig:NSWpdfkb10}
\end{center}
\end{figure}

We now turn to the case of the generic initial condition,
cf.~\S\,\ref{sec:generic}, and compare the fluxes and triads analyses
against the extreme case. Let us recall that in the generic case the
initial amplitudes of the Fourier coefficients are the same as in the
extreme case, while {their phases are initially set as random
  variables with uniform distributions} over $[0,2\pi)$. Therefore,
the generic case offers a method to test whether or not the observed
phase coherence in the extreme case is due to a special initial
condition on the triad phases.

Interestingly, the results for the generic case are hardly
distinguishable from those for the extreme case. To begin with, as
regards the flux towards the wavenumber region $k>2$, table
\ref{tab:PDFbounds_generic} (second column) shows again that the same
four triad types as in the extreme case have PDFs with variability of
$\pm 5\%$ or more with respect to the uniform case. Figure
\ref{fig:NSpdfkb2_generic} shows these PDFs and it is evident that,
apart from details, the generic case displays {essentially} the
same global features as the extreme case, including the coherent
patterns in panel (d) corresponding to (PM)P helical triads. Next, the
time evolution of the fluxes for each of the six triad types towards
the wavenumber region $k>2$ is plotted in figure
\ref{fig:NSFlux_cases}(c). Comparing this with the extreme case in
figure \ref{fig:NSFlux_cases}(a), it is evident that each triad type
shows a similar trend, with small differences related to the timing
and value of flux maxima. In particular, the (PM)P boundary triad type
is again seen to contribute {well} over $20\%$ of the total flux
near $t=0.11$, despite the fact that (PM)P triads constitute only less
than $0.0001\%$ of the total number of triads {participating in
  energy transfer}.  Again, the triad types that contribute the most
to the flux are, in order of importance, PMP, (PM)P and PMM. The next
piece of analysis concerns the weighted PDFs for the generic case
shown in figure \ref{fig:NSWpdfkb2_generic}. A quick qualitative
comparison with the extreme case, cf.~figure \ref{fig:NSWpdfkb2},
demonstrates the same overall patterns of coherence. Focusing now our
attention on the main flux contributors: PMP in panel (c), (PM)P in
panel (f) and PMM in panel (d), we clearly observe an early
{burst} of positive flux at $t=0.05$ in the weighted PDF for the
PMM triads, followed by a {pronounced accumulation of the
  weighted PDF around} $\Phi=0$ at $t=0.05$---$0.1$ and then another
{such event at $t=0.07$---$0.11$ for the (PM)P triads}. This is
qualitatively very similar to the extreme case. However, from the
quantitative point of view, we can state that {in the present
  case} the three main flux contributors show a less extreme behaviour
in their weighted PDFs: comparing the ranges of values of the weighted
PDFs in the the extreme case (table \ref{tab:PDFbounds_extreme},
fourth column) with the same ranges for the generic case (table
\ref{tab:PDFbounds_generic}, third column) for the triad types PMP,
(PM)P and PMM, it is evident that the range in the generic case is
about 30\% to 35\% smaller than in the extreme case.  This could be
attributed to an initial lack of coherence in the phases, which
translates into a less coherent time evolution.

\begin{table}
  \begin{center}
    % \hspace*{-1.1cm}
    \begin{tabular}{|l|c|c|c|c|} \hline
      \backslashbox{triad}{cases}   & $k=2$, PDF  & $k=2$, wPDF  & $k=10$, PDF &  $k=10$, wPDF \\ \hline
     PPP   & [0.145,0.171]   &  [0.022,0.770]  & \textcolor[rgb]{0.7,0.7,0.7}{[0.158,0.161]}  &  [0.039,0.898] \\ 
      PPM  &  \textcolor[rgb]{0.7,0.7,0.7}{[0.157,0.161]}  &  [0.031,0.497]  & \textcolor[rgb]{0.7,0.7,0.7}{[0.159,0.160]}  &  [0.089,0.361] \\    
      PMP    & \textcolor[rgb]{0.7,0.7,0.7}{[0.156,0.162]} &  [0.033,0.537] &  \textcolor[rgb]{0.7,0.7,0.7}{[0.159,0.160]} &   $[0.104,0.281]$ \\
      PMM   & [0.148,0.169] & [0.016,0.505]  & \textcolor[rgb]{0.7,0.7,0.7}{[0.157,0.162]} &  $[0.054,0.806]$  \\ 
	 P(PM)    & [0.148,0.168] &  [0.002,0.747] &  \textcolor[rgb]{0.7,0.7,0.7}{[0.158,0.160]} &   $[0.027,0.660]$ \\
	 (PM)P    & [0.000,0.424] &  [0.000,1.103] &  [0.149,0.169] &   $[0.051,0.520]$ \\\hline
    \end{tabular}
  \end{center}
  \caption{{[Generic case:]} Upper and lower bounds on the PDFs and wPDFs of triad phases of the different types in the Navier-Stokes flow with the generic initial data, cf.~figures  \ref{fig:NSpdfkb2_generic}--\ref{fig:NSWpdfkb10_generic}. {Shaded} intervals (light gray) represent PDFs that are very close (within $\pm5\%$) to the uniform distribution $1/2\pi \approx 0.16$.}
  \label{tab:PDFbounds_generic}
\end{table}

%1/2/pi*1.05 = 0.1671126902
%1/2/pi*0.95 = 0.1511971959

\begin{figure}
\begin{center}
\mbox{\subfigure[${\mathcal{P}}_{\mathcal{C}_2}^{PPP}(\Phi)(t)$]{\includegraphics[width=0.45\textwidth]{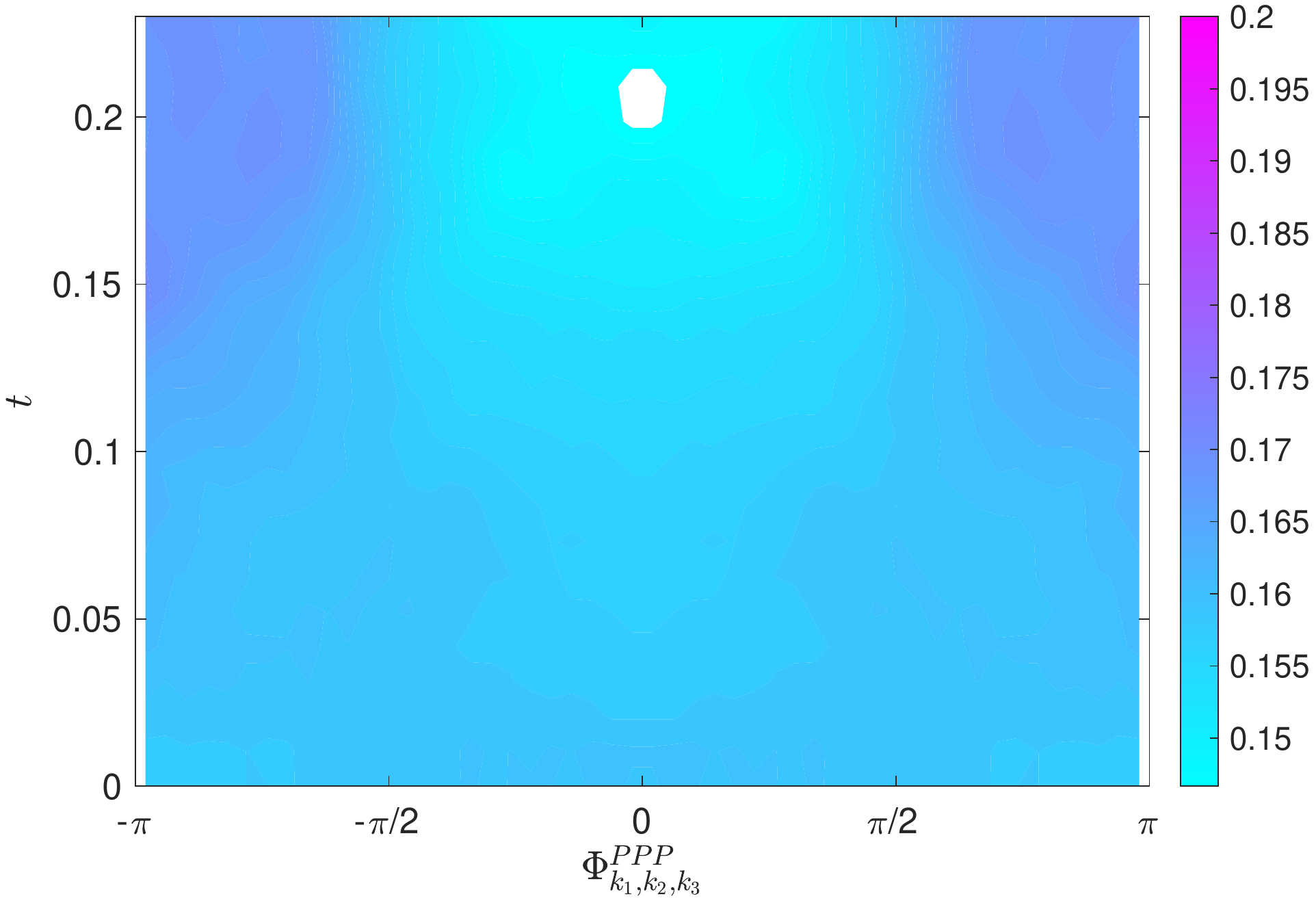}}\qquad
\subfigure[${\mathcal{P}}_{\mathcal{C}_2}^{PMM}(\Phi)(t)$]{\includegraphics[width=0.45\textwidth]{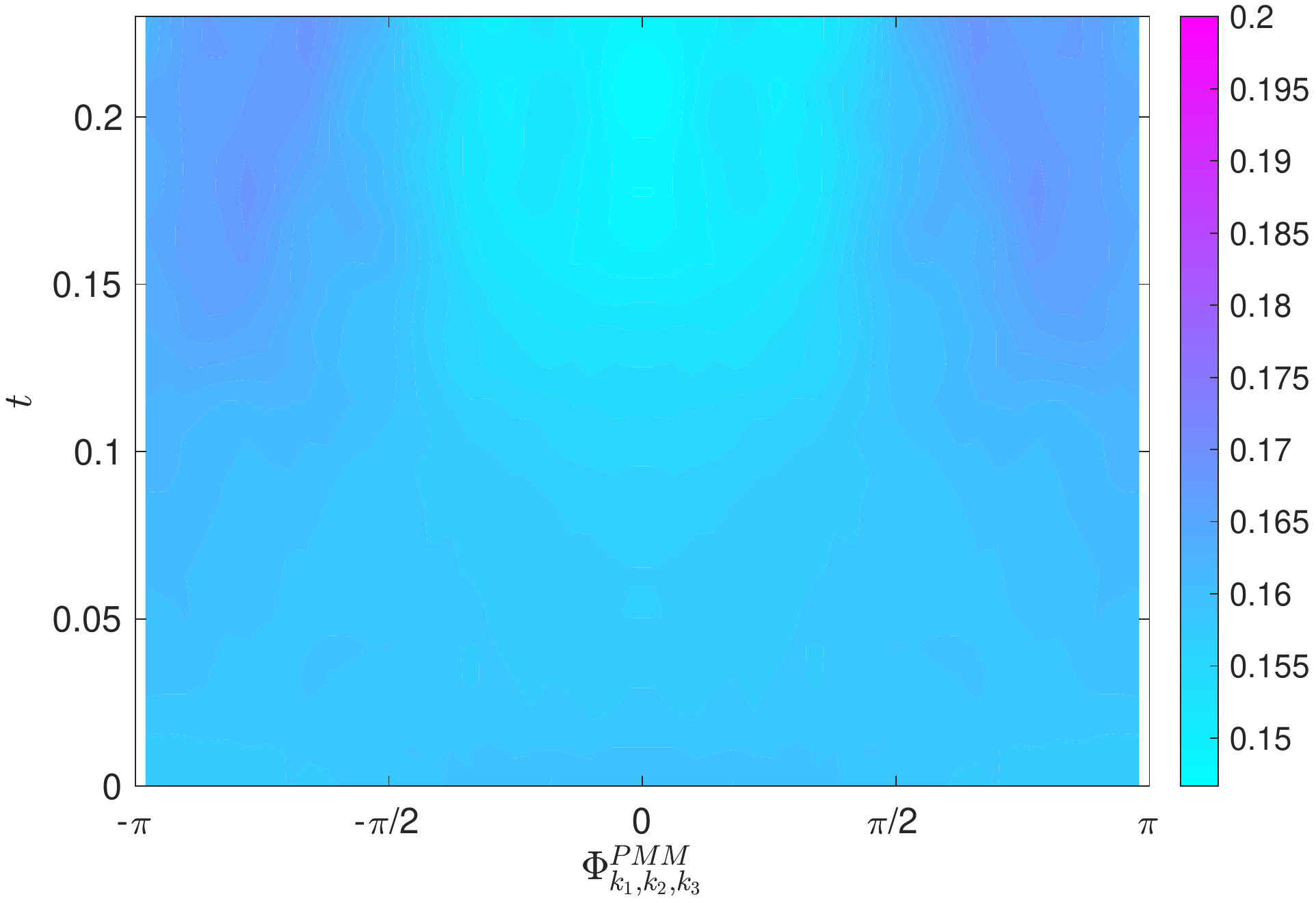}}}
\mbox{\subfigure[${\mathcal{P}}_{\mathcal{C}_2}^{P(PM)}(\Phi)(t)$]{\includegraphics[width=0.45\textwidth]{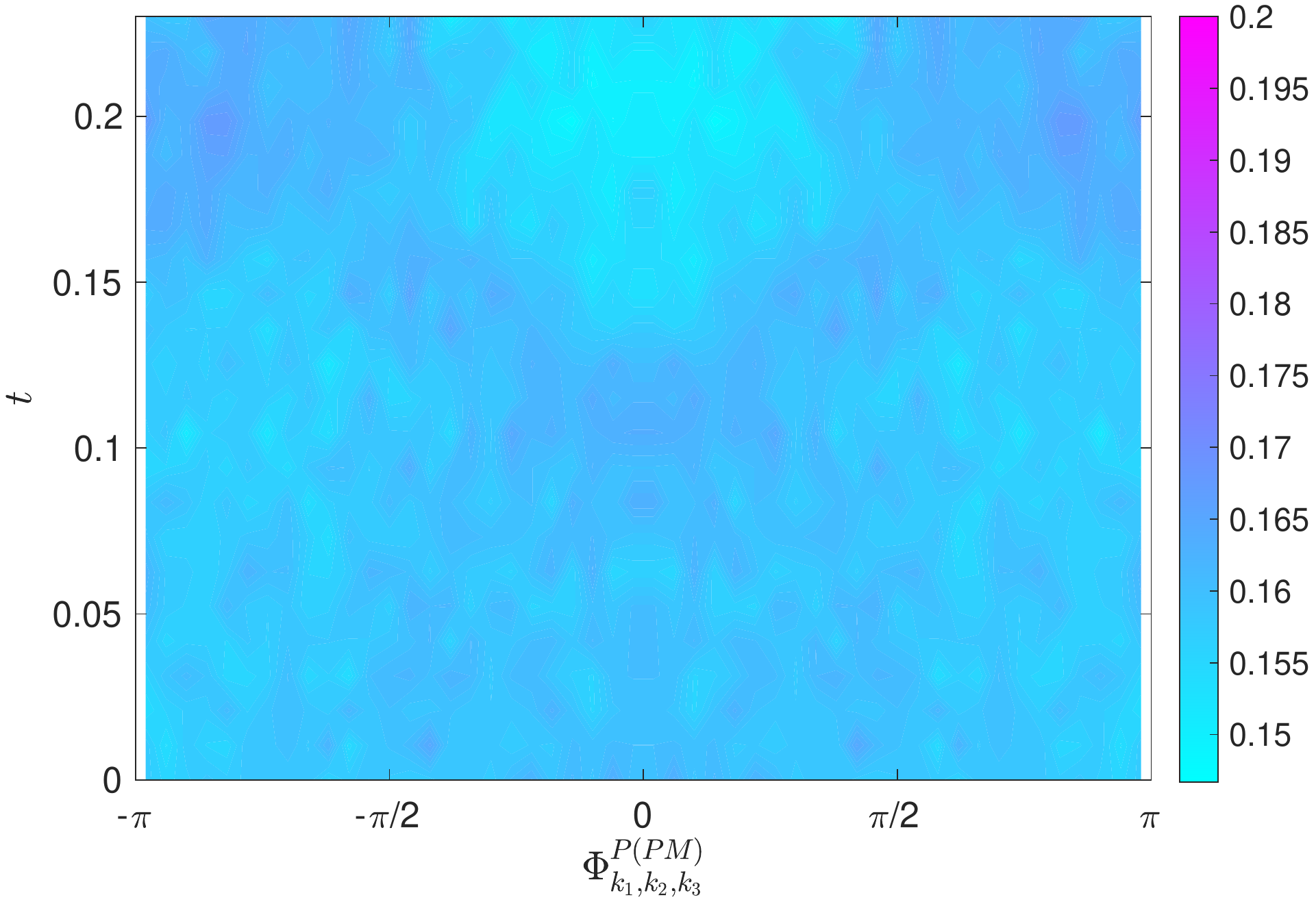}}\qquad
\subfigure[${\mathcal{P}}_{\mathcal{C}_2}^{(PM)P}(\Phi)(t)$]{\includegraphics[width=0.45\textwidth]{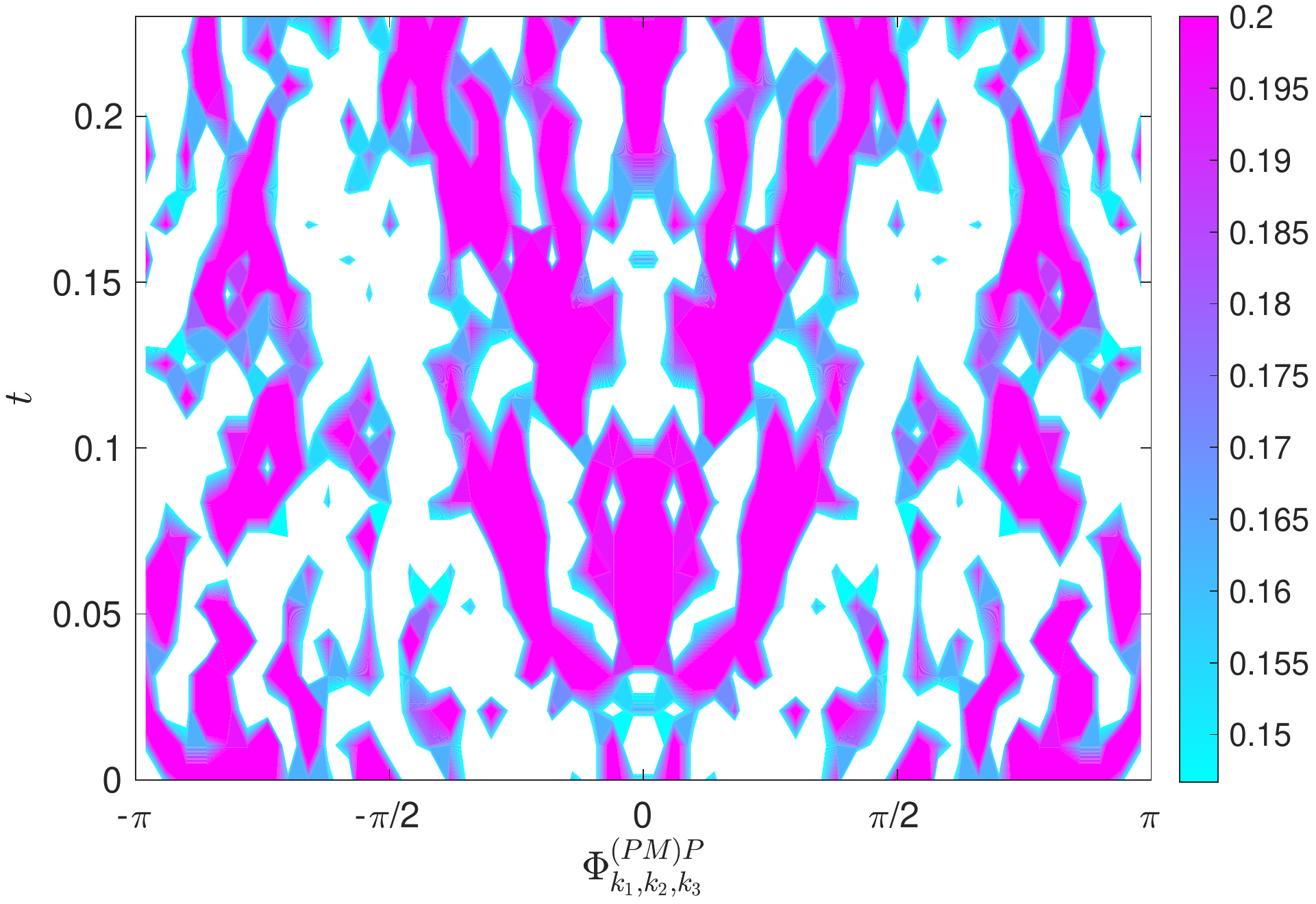}}}
\caption{{[Generic case, $k>2$:]} Time evolution of the PDFs of
  the triad phase angle ${\mathcal{P}}_{\mathcal{C}_2}^{s_1 s_2
    s_3}(\Phi)$ {for triads of different types in the
    Navier-Stokes flow with the generic initial data.  Only triad
    types with PDFs revealing variability of at least $\pm5\%$ with
    respect to the uniform distribution $1/2\pi \approx 0.16$
    (corresponding to light blue colour in the plots), are shown:} (a)
  ``PPP'', (b) ``PMM'', (c) ``P(PM)'' and (d) ``(PM)P''.}
\label{fig:NSpdfkb2_generic}
\end{center}
\end{figure}

%\begin{figure}
%\begin{center}
%\mbox{\subfigure[]{\includegraphics[width=0.45\textwidth]{Figs2/Wpdfpppkb2rand.pdf}}\qquad
%\subfigure[]{\includegraphics[width=0.45\textwidth]{Figs2/Wpdfppmkb2rand.pdf}}}\\
%\mbox{\subfigure[]{\includegraphics[width=0.45\textwidth]{Figs2/Wpdfpmpkb2rand.pdf}}\qquad
%\subfigure[]{\includegraphics[width=0.45\textwidth]{Figs2/Wpdfpmmkb2rand.pdf}}}
%\mbox{\subfigure[]{\includegraphics[width=0.45\textwidth]{Figs2/Wpdfp_pmkb2rand.pdf}}\qquad
%\subfigure[]{\includegraphics[width=0.45\textwidth]{Figs2/Wpdfpm_pkb2rand.pdf}}}
%\caption{Time evolution of the weighted PDF of the triad phase
%  $\varphi_{k_1,k_2}^{k_3}$ in the generic solution of the Navier-Stokes equation with triad type 
% (a) ``PPP'', (b)``PPM'', (c) ``PMP''  , (d) ``PMM'',(e) ``P(PM)'' and (f) ``(PM)P ''  for $\E_0 = 250$ and $k=2$.}
%\label{fig:NSWpdfkb2_generic}
%\end{center}
%\end{figure}

\begin{figure}
\begin{center}
\mbox{\subfigure[${\mathcal{W}}_{\mathcal{C}_2}^{PPP}(\Phi)(t)$]{\includegraphics[width=0.45\textwidth]{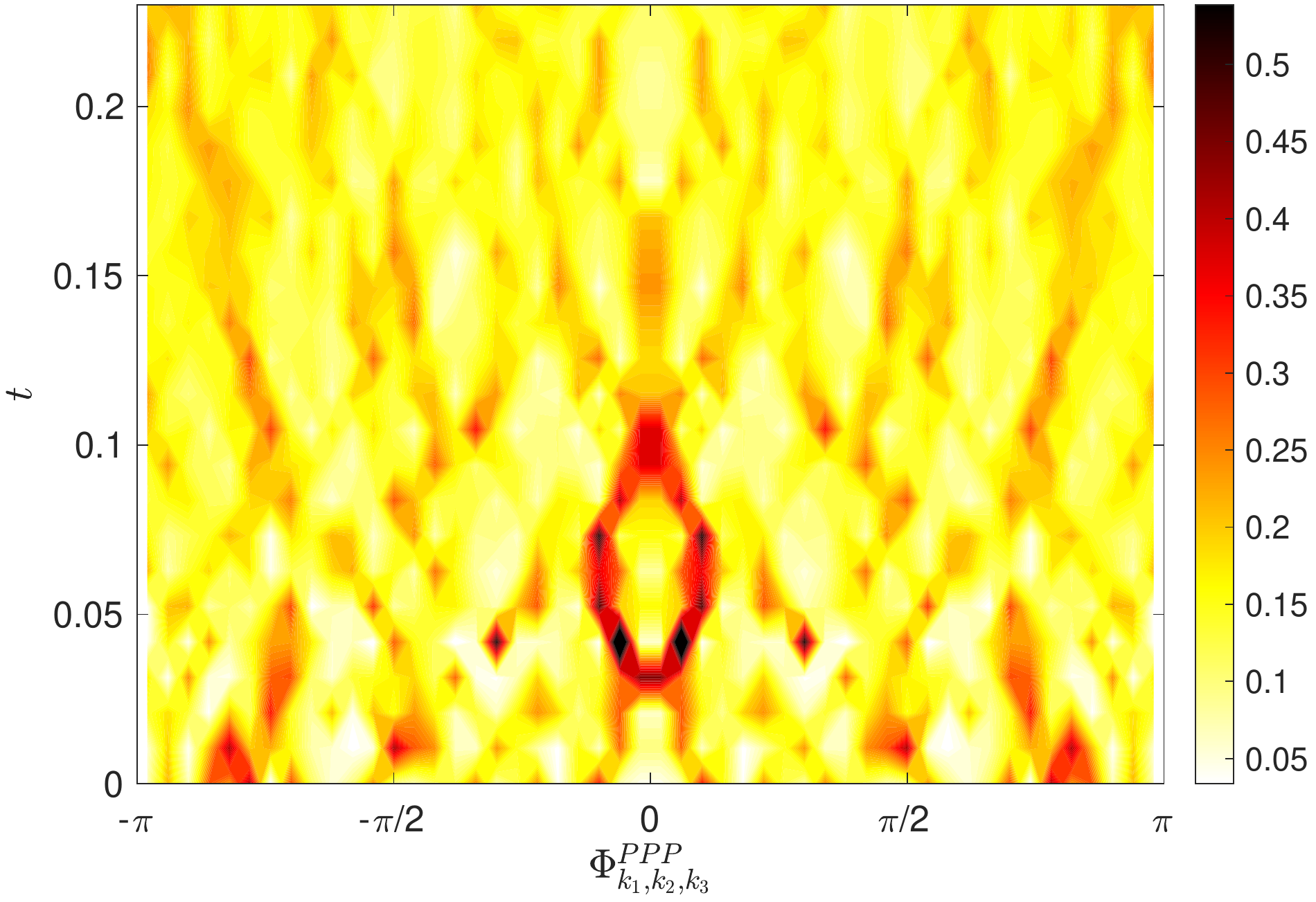}}\qquad
\subfigure[${\mathcal{W}}_{\mathcal{C}_2}^{PPM}(\Phi)(t)$]{\includegraphics[width=0.45\textwidth]{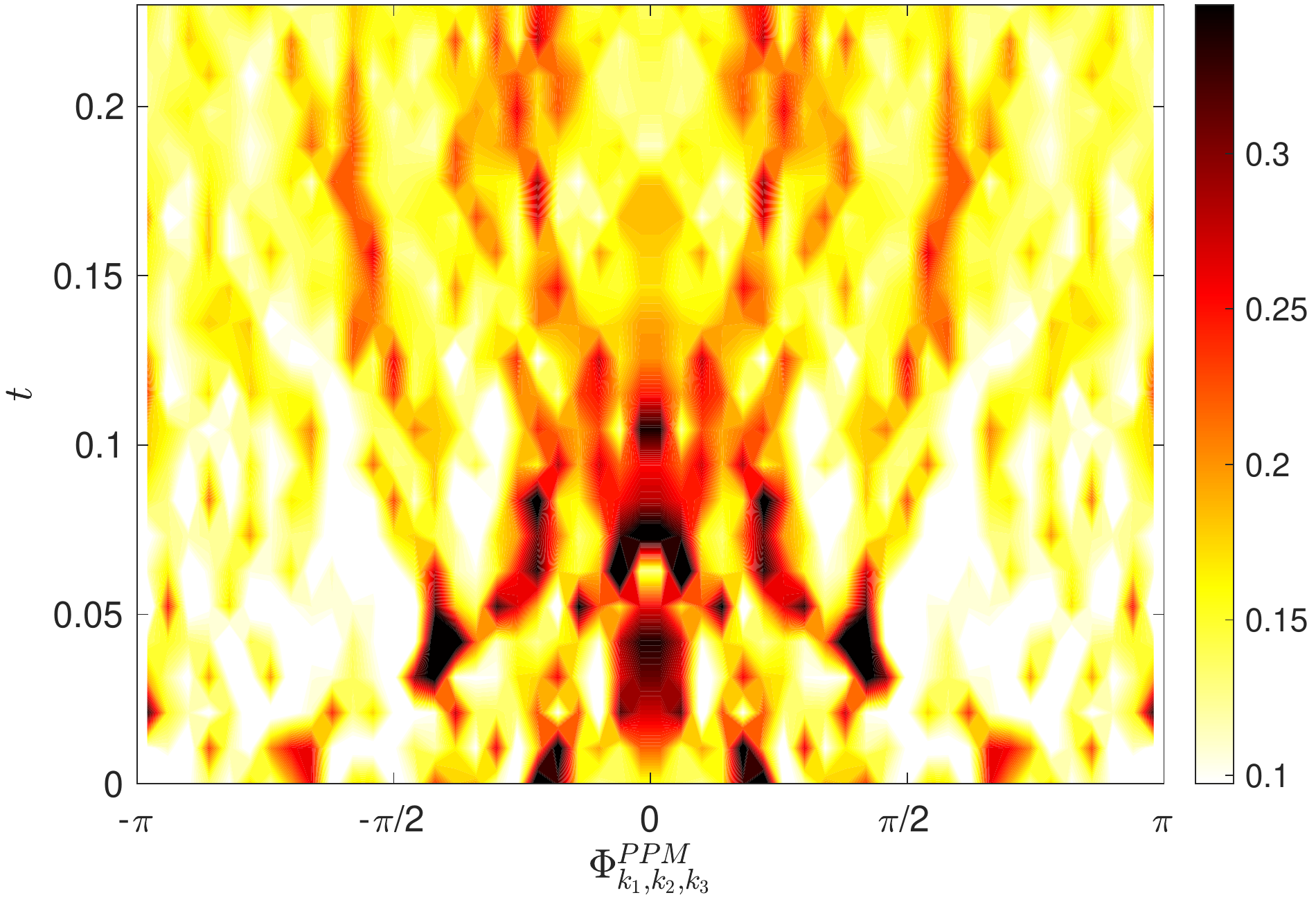}}}\\
\mbox{\subfigure[${\mathcal{W}}_{\mathcal{C}_2}^{PMP}(\Phi)(t)$]{\includegraphics[width=0.45\textwidth]{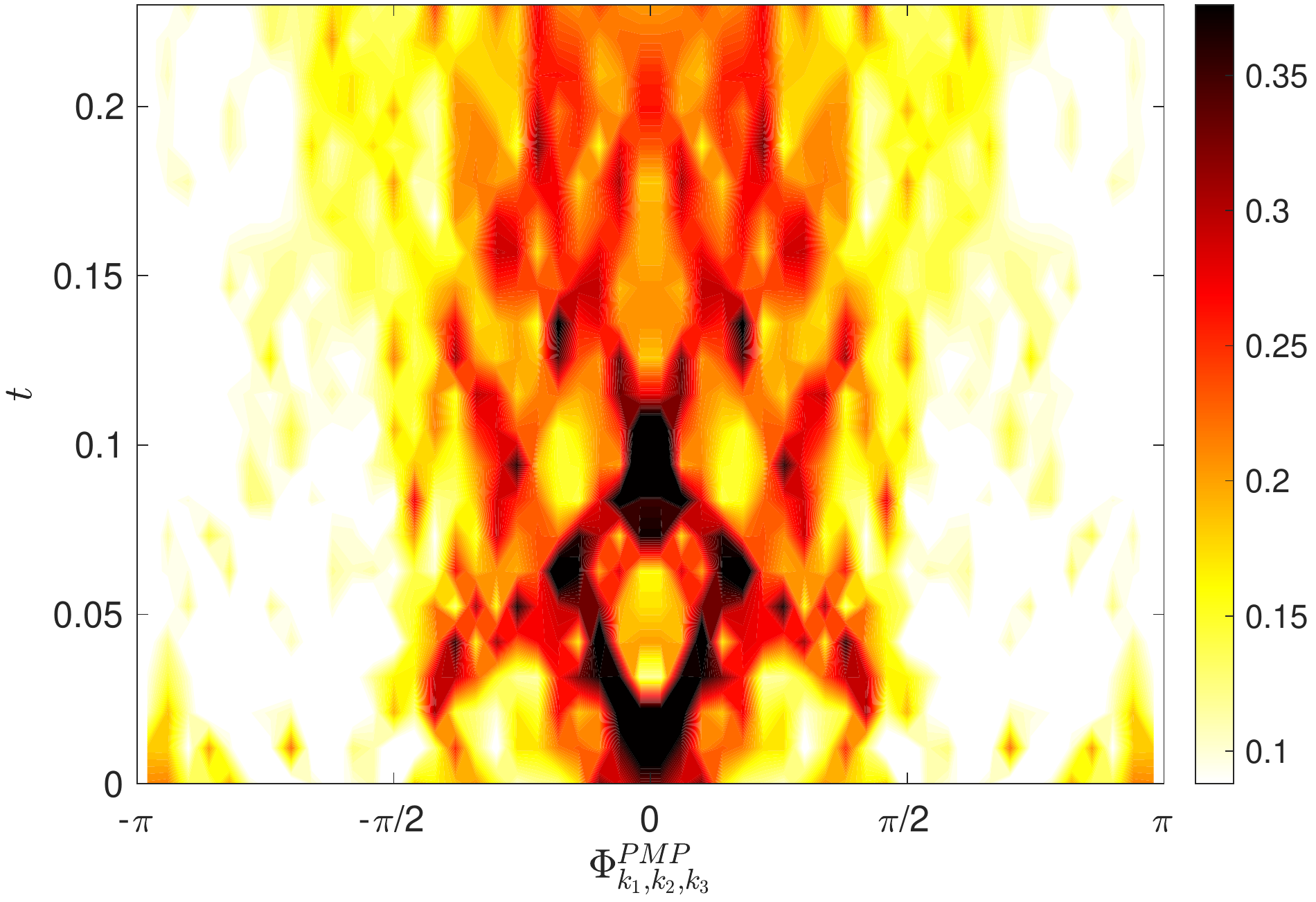}}\qquad
\subfigure[${\mathcal{W}}_{\mathcal{C}_2}^{PMM}(\Phi)(t)$]{\includegraphics[width=0.45\textwidth]{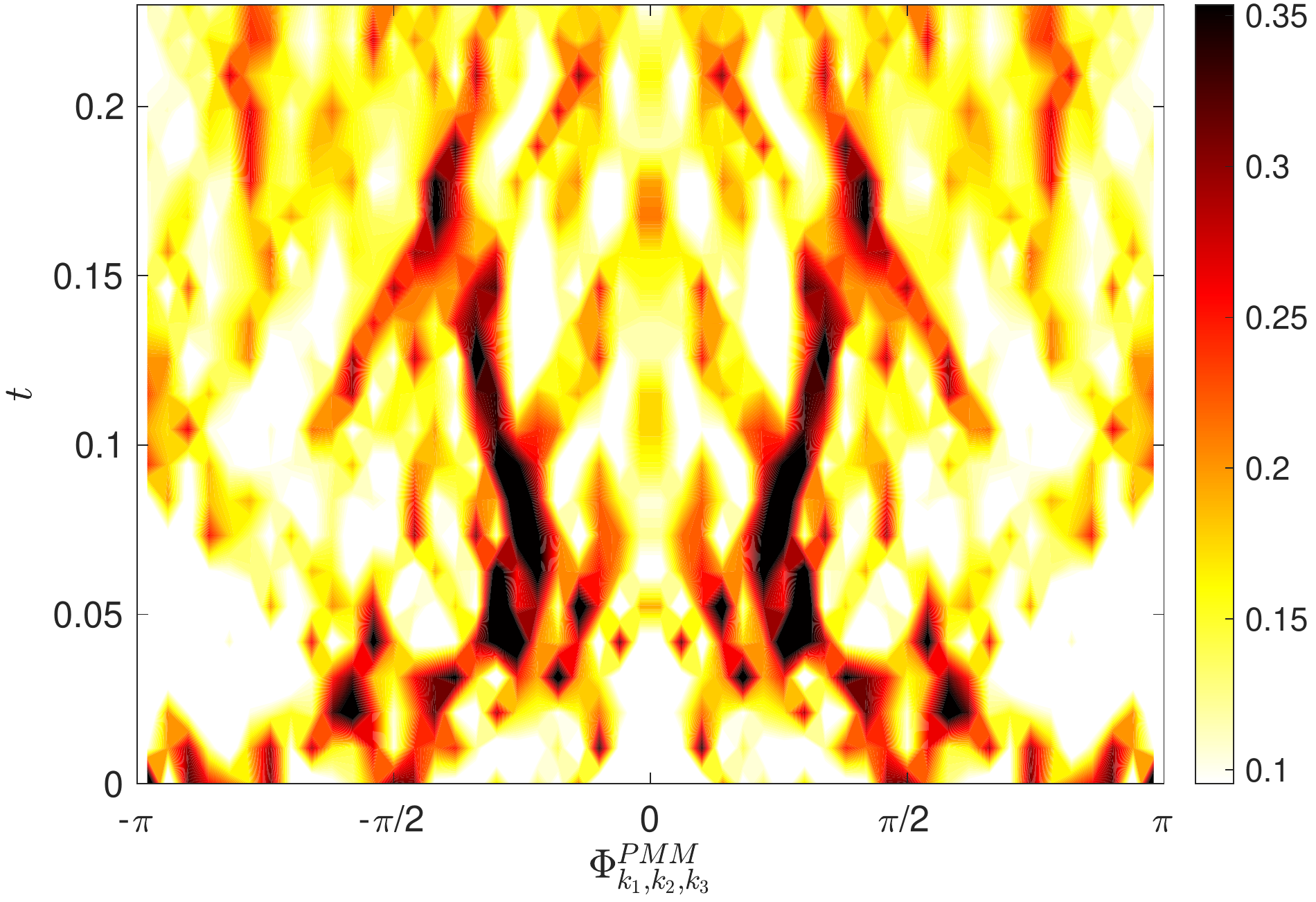}}}
\mbox{\subfigure[${\mathcal{W}}_{\mathcal{C}_2}^{P(PM)}(\Phi)(t)$]{\includegraphics[width=0.45\textwidth]{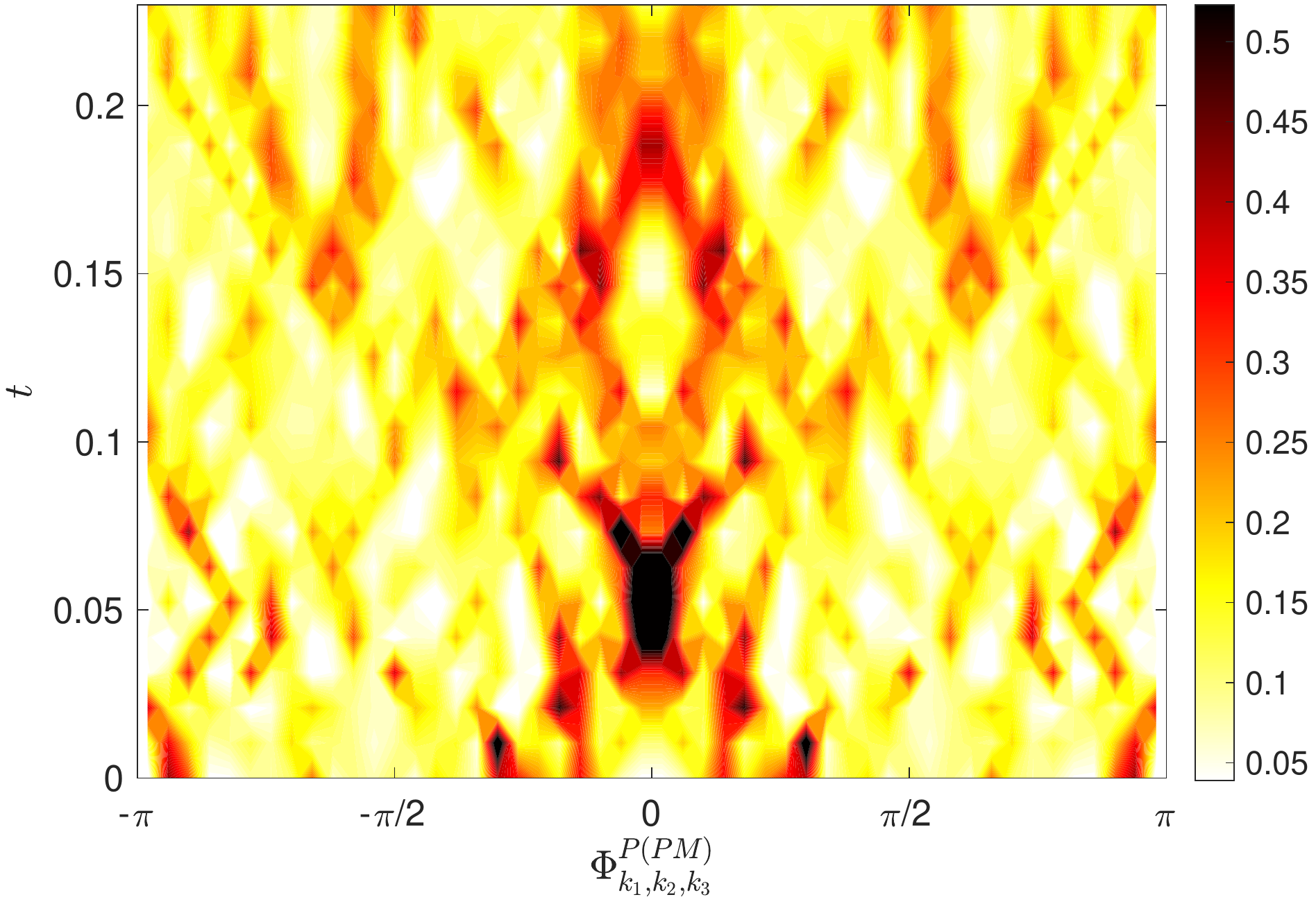}}\qquad
\subfigure[${\mathcal{W}}_{\mathcal{C}_2}^{(PM)P}(\Phi)(t)$]{\includegraphics[width=0.45\textwidth]{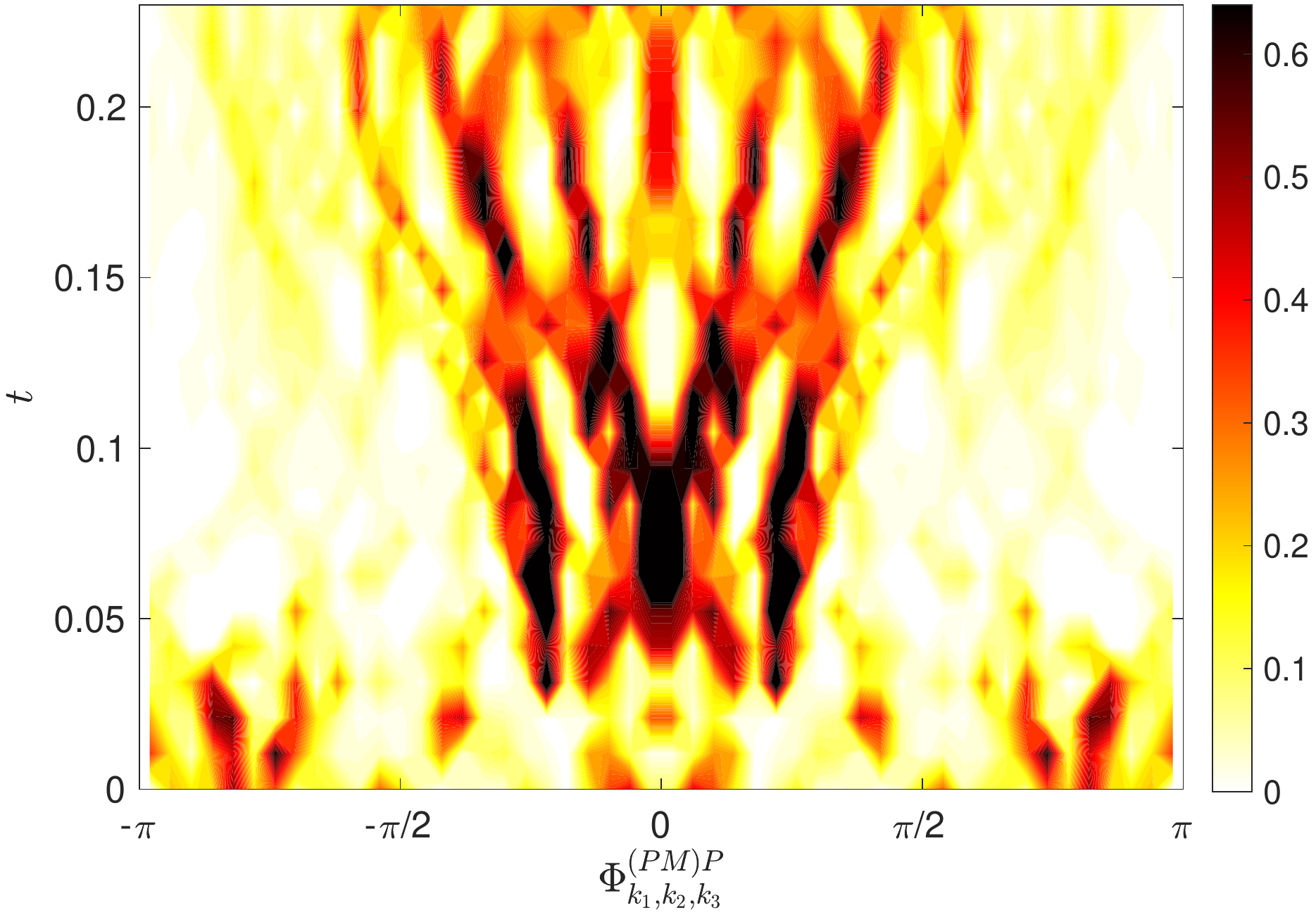}}}
\caption{{[Generic case, $k>2$:]} Time evolution of the weighted
  PDF of the triad phase angle ${\mathcal{W}}_{\mathcal{C}_2}^{s_1 s_2
    s_3}(\Phi)$, cf.~expression \eqref{eq:WC}, in the Navier-Stokes
  flow with the generic initial data for the triad types (a) ``PPP'',
  (b) ``PPM'', (c) ``PMP'' , (d) ``PMM'', (e) ``P(PM)'' and (f)
  ``(PM)P''. In these plots the uniform distribution corresponding to
  $1/2\pi \approx 0.16$ is depicted by yellow ($25\%$ of the colour
  bars).}
\label{fig:NSWpdfkb2_generic}
\end{center}
\end{figure}

\begin{figure}
\begin{center}
\mbox{\subfigure[${\mathcal{W}}_{\mathcal{C}_{10}}^{PPP}(\Phi)(t)$]{\includegraphics[width=0.45\textwidth]{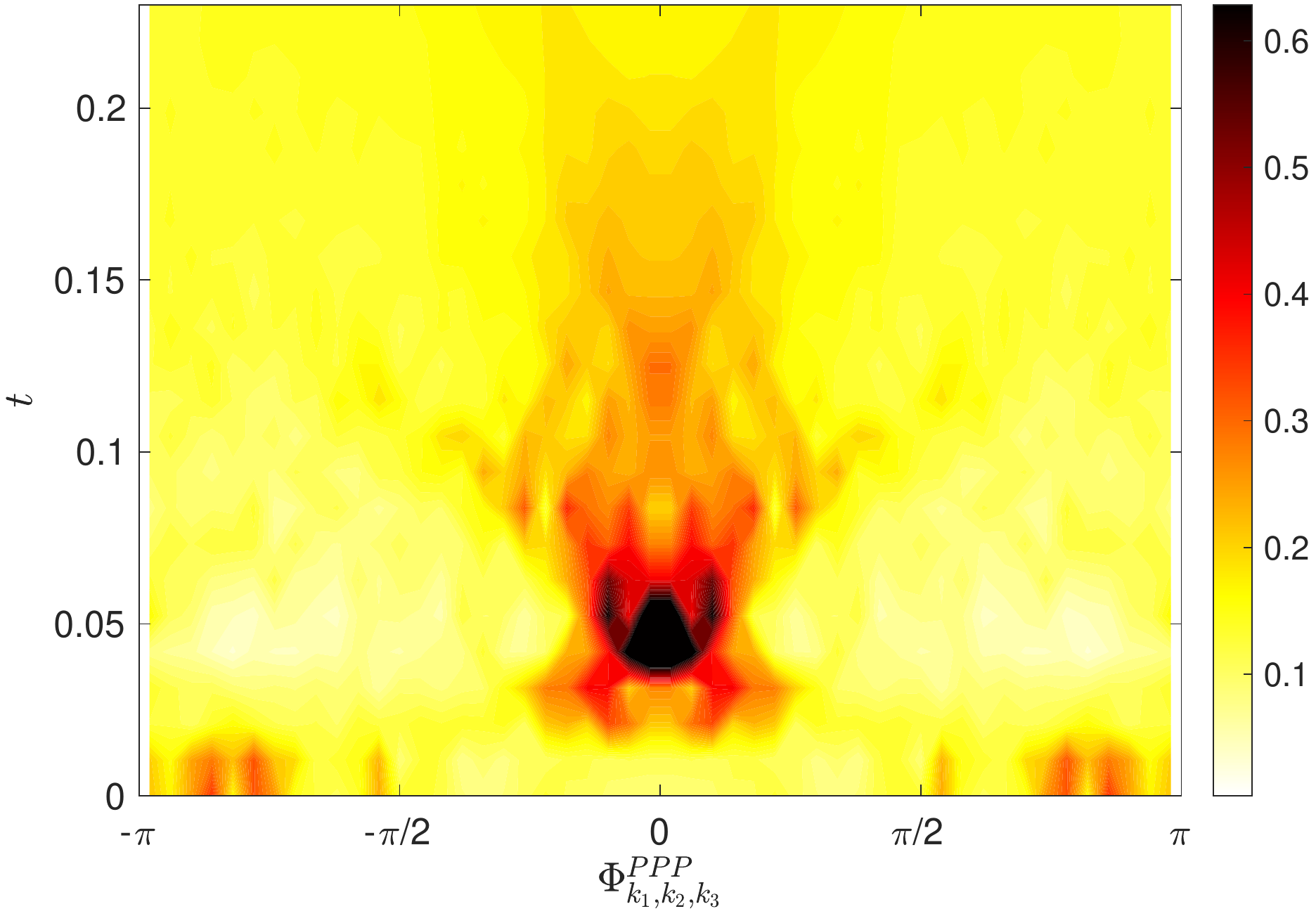}}\qquad
\subfigure[${\mathcal{W}}_{\mathcal{C}_{10}}^{PPM}(\Phi)(t)$]{\includegraphics[width=0.45\textwidth]{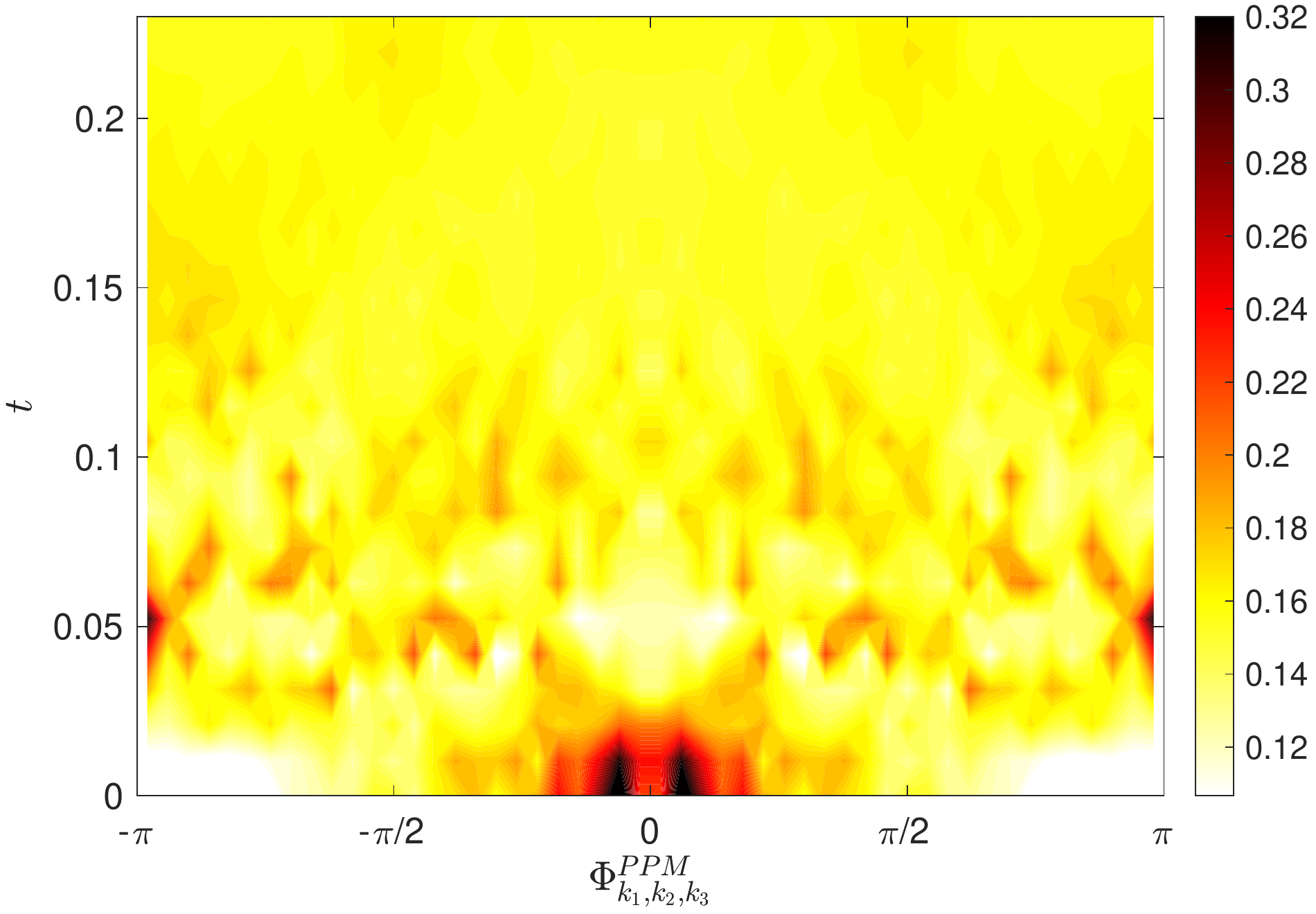}}}\\
\mbox{\subfigure[${\mathcal{W}}_{\mathcal{C}_{10}}^{PMP}(\Phi)(t)$]{\includegraphics[width=0.45\textwidth]{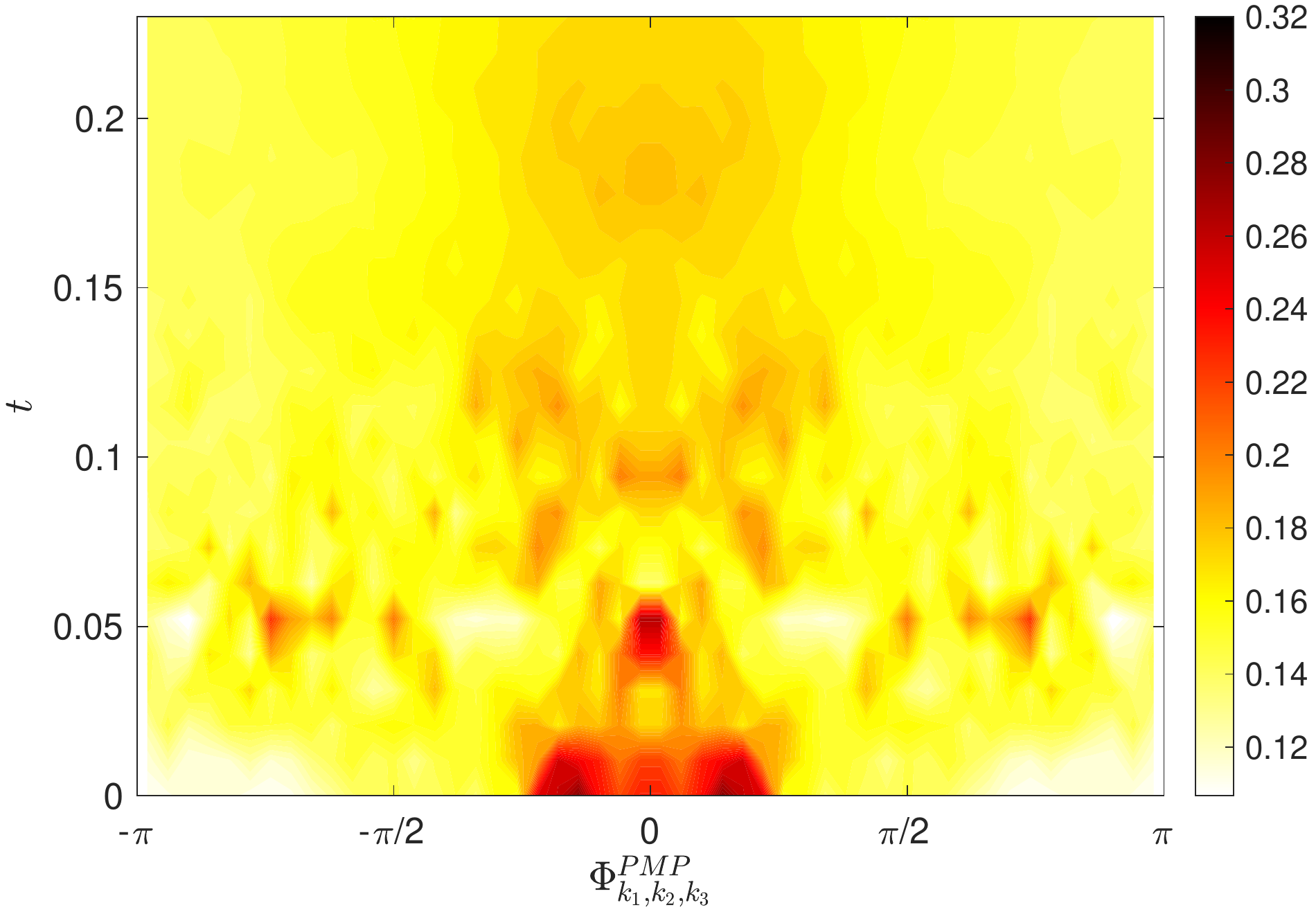}}\qquad
\subfigure[${\mathcal{W}}_{\mathcal{C}_{10}}^{PMM}(\Phi)(t)$]{\includegraphics[width=0.45\textwidth]{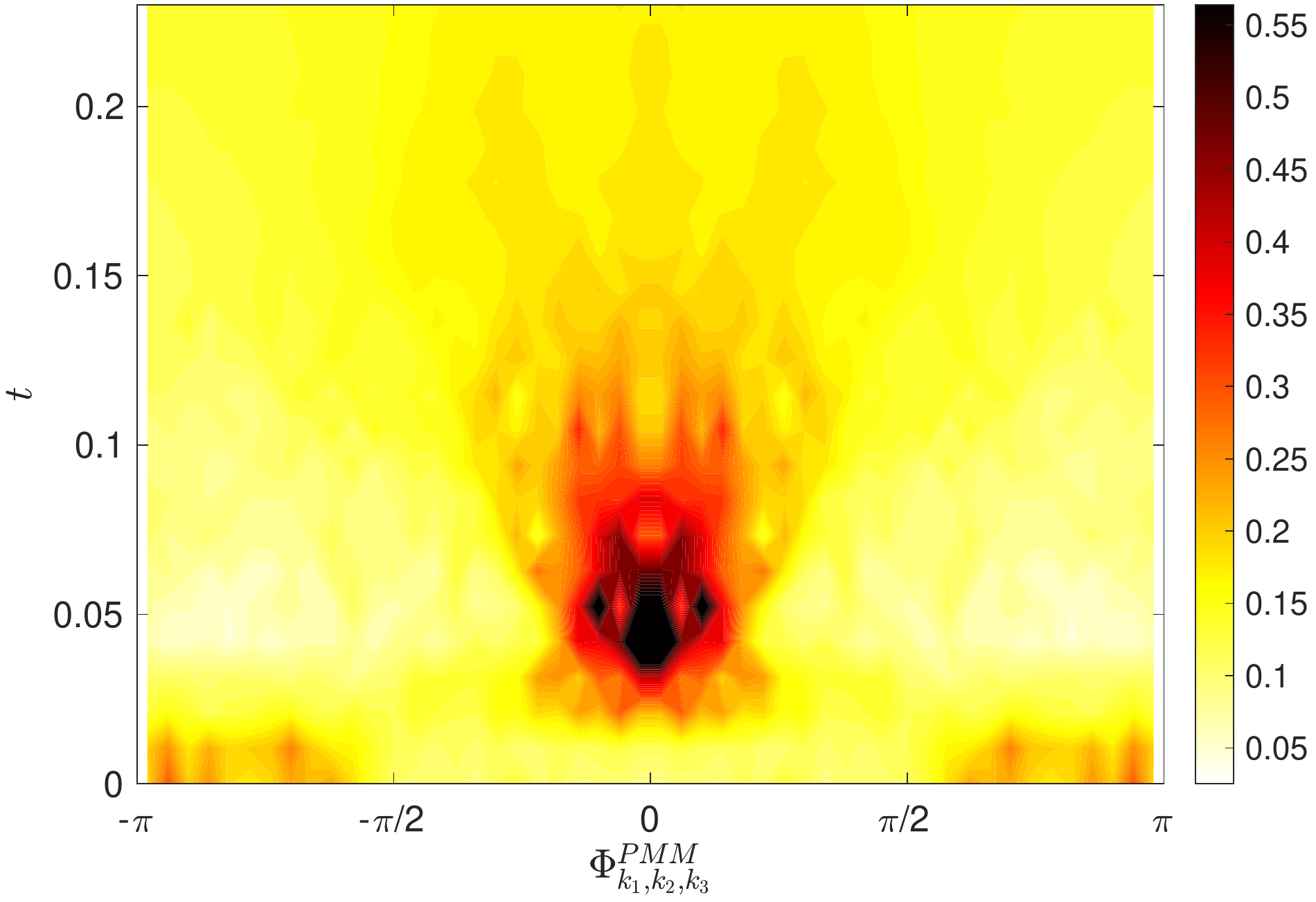}}}
\mbox{\subfigure[${\mathcal{W}}_{\mathcal{C}_{10}}^{P(PM)}(\Phi)(t)$]{\includegraphics[width=0.45\textwidth]{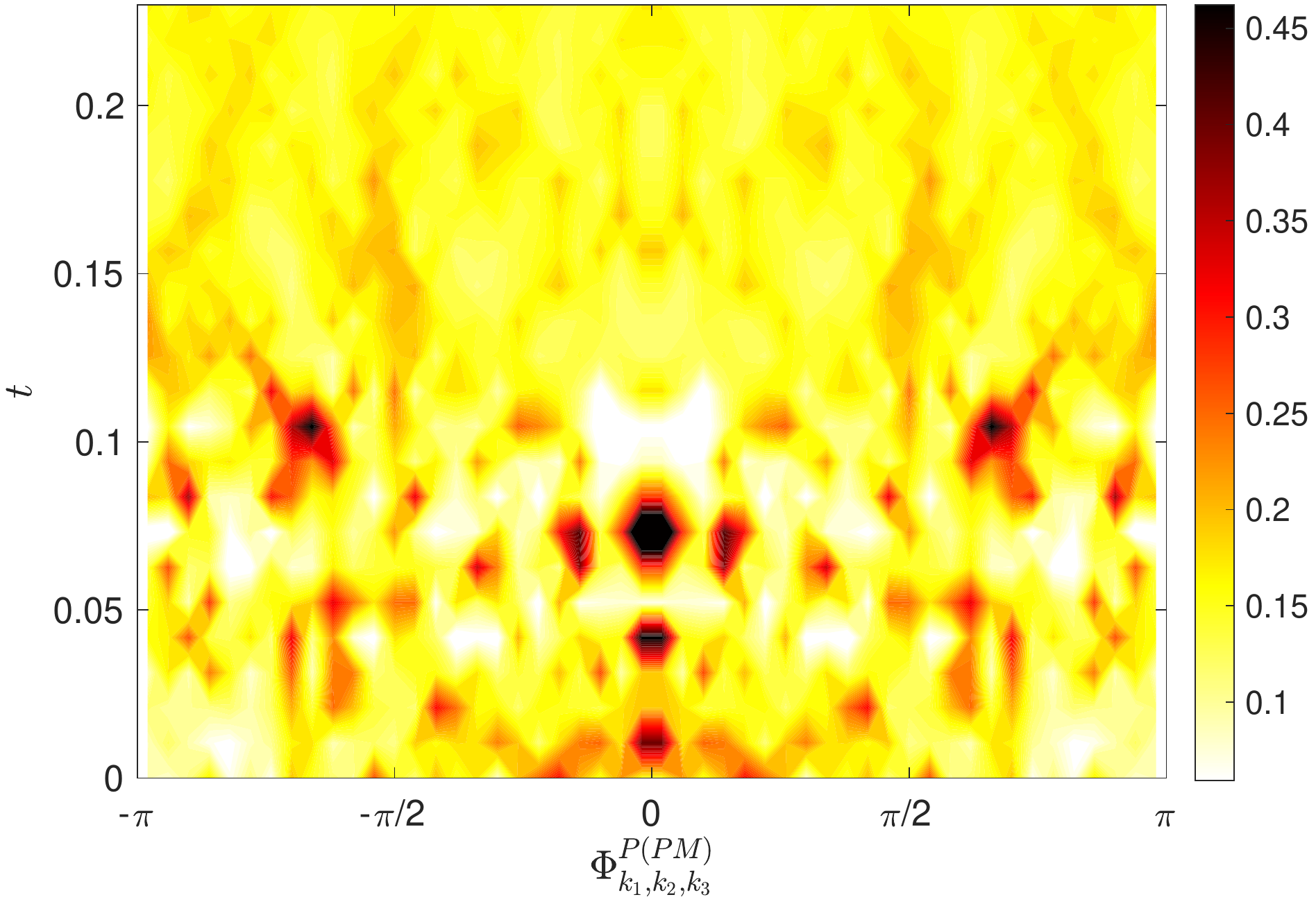}}\qquad
\subfigure[${\mathcal{W}}_{\mathcal{C}_{10}}^{(PM)P}(\Phi)(t)$]{\includegraphics[width=0.45\textwidth]{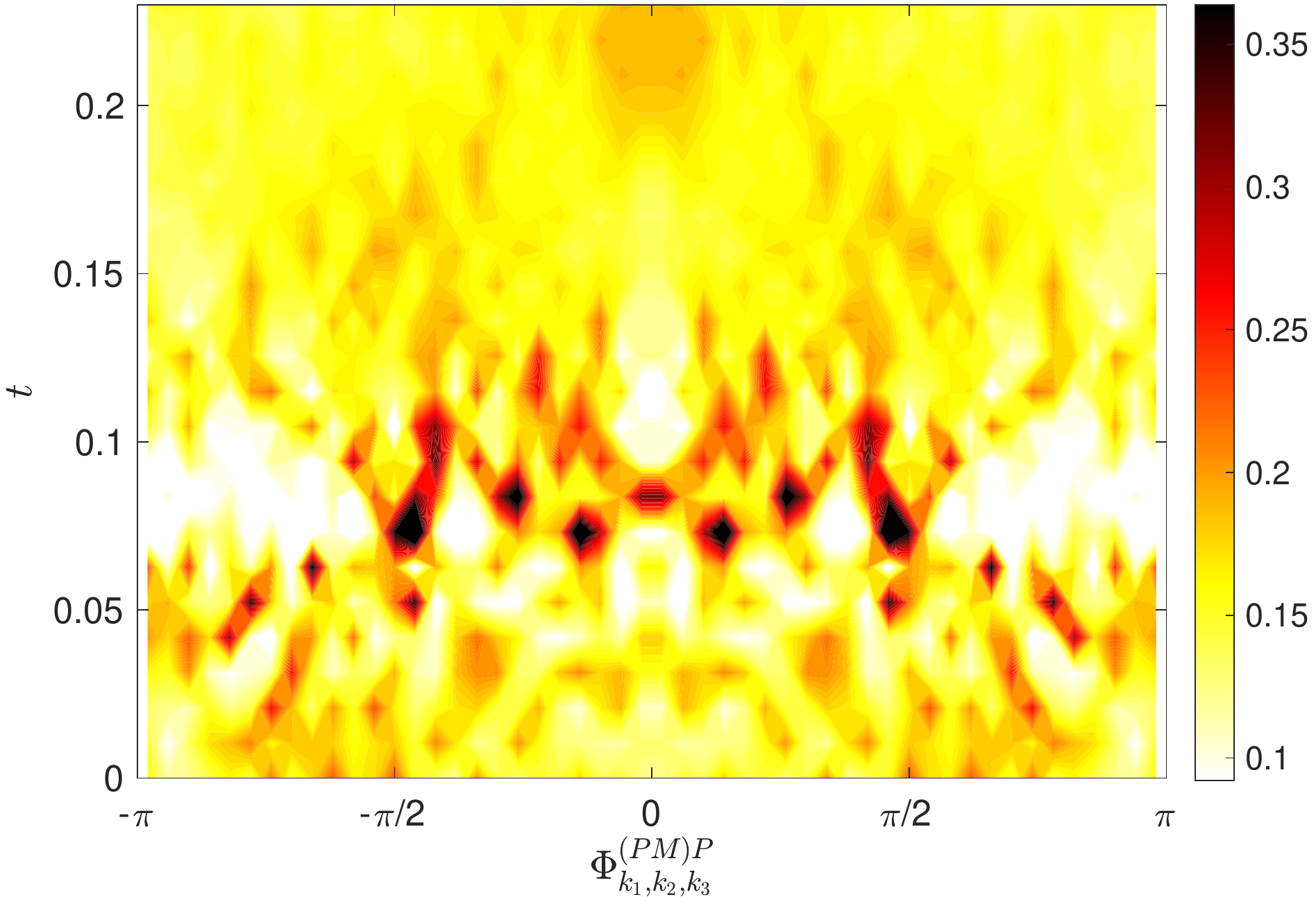}}}

\caption{{[Generic case, $k>10$:]} Time evolution of the weighted
  PDF of the triad phase angle ${\mathcal{W}}_{\mathcal{C}_{10}}^{s_1
    s_2 s_3}(\Phi)$, cf.~expression \eqref{eq:WC}, in the
  Navier-Stokes flow with the generic initial data for the triad types
  (a) ``PPP'', (b) ``PPM'', (c) ``PMP'' , (d) ``PMM'', (e) ``P(PM)''
  and (f) ``(PM)P''. In these plots the uniform distribution
  corresponding to $1/2\pi \approx 0.16$ is depicted by yellow ($25\%$
  of the colour bars).}
\label{fig:NSWpdfkb10_generic}
\end{center}
\end{figure}

Having shown that the generic and extreme cases display comparable
behaviours in terms of the helical-triad assessment of the flux
towards wavenumbers $k>2$, we now turn our attention to the
corresponding study of the flux towards wavenumbers $k>10$. Firstly,
we look at the time evolution of the fluxes for each triad type, shown
in figure \ref{fig:NSFlux_cases}(d) for the generic case. As in the
extreme case, in the generic case the main contributors to the flux
are, in order of importance, the triad types PMM, PPP and PMP. The
time evolution of the fluxes for these triad types is qualitatively
similar to the corresponding evolution in the extreme case, cf.~figure
\ref{fig:NSFlux_cases}(d), with again small differences in the timing
and values of the flux maxima. Importantly, in these plots it is
always the case that a flux maximum is attained earlier in the generic
case, with the extreme case attaining a higher value of the maximum
flux. This difference of the timing and maximum value attained could
be attributed to the lack of coherence in the generic case. Secondly,
we {assess the ranges of the weighted PDFs characterizing the
  main contributors to the flux.}  Comparing the range of values in
the extreme case (table \ref{tab:PDFbounds_extreme}, last column)
against those of the generic case (table \ref{tab:PDFbounds_generic},
last column), we find that in contrast to the case with $k>2$, when
$k>10$ it is the generic case that exhibits a broader range: for the
PMM and PPP triads this range is about 50\% to 70\% larger than in the
extreme case, while the range in the generic case for the PMP triads
shows just a 3\% reduction with respect to the extreme case. Does this
contradict our {earlier observations} that the generic case is
less coherent than the extreme case?  Not necessarily, because even
though the generic case shows a stronger coherence peak at early
times, the extreme case is designed in order to maximize the late-time
flux, cf.~\S\,\ref{sec:ext}.  Comparing panels (b) and (d) of figure
\ref{fig:NSFlux_cases} and focusing on the main flux contributor---the
PMM triads (magenta dash-dotted lines)---the extreme case indeed
presents a slightly lower flux during early and {intermediate
  times, but by the end of the time window the flux in the extreme
  case is larger} and still growing, while in the generic case the
flux is smaller and is already reaching a plateau.

A more detailed comparison between the generic and extreme cases
regarding flux towards wavenumbers $k>10$ involves analyzing the
weighted PDFs {shown in figures \ref{fig:NSWpdfkb10} and
  \ref{fig:NSWpdfkb10_generic} for the extreme and the generic case,
  respectively}. Focusing on the main flux contributors---the
  triads PMM, PPP and PMP in panels (d), (a) and (c) of these
  figures---we see that the generic case exhibits stronger coherence
peaks, but they occur earlier and are less sustained than in the
extreme case. This is consistent with the analysis given in the
preceding paragraph: while the generic case shows stronger phase
coherence at early times, it fails to provide the delicate balance of
phases required in order to {sustain} a large late-time flux.
Maximizing this flux is what the extreme case is designed for.

Finally, we turn to the study of the flow corresponding to the
Taylor-Green initial condition, cf.~\S\,\ref{sec:unimodal}, from the
point of view of helical triads. In this case, the evolution is
strikingly different from the extreme and generic cases. Let us recall
that this initial condition contains Fourier components with only a
few wavevectors, all with wavenumber $k=1$. Therefore, {at the
  beginning it is expected to see low-dimensional dynamics confined
  to} the low-wavenumber region.  We first look at the flux towards
the wavenumber region $k>2$. Data in table \ref{tab:PDFbounds_TG},
second column, shows little difference with respect to the extreme and
generic cases: again, only three triad types {exhibit PDFs with
  variability of at least $\pm5\%$ with respect to the uniform
  distribution}, and again the (PM)P boundary triad type is the one
showing the most variability. The PDFs of these three triad types are
shown in figure \ref{fig:NSpdfkb2_TG}.  In comparison against the
corresponding PDFs in the extreme and generic cases, here we see a
clear difference regarding the (PM)P boundary triad type: on top of
the typical coherent patterns also encountered in the other two cases,
in the Taylor-Green case there is an extra band of densely populated
phase values which persists throughout the time evolution. This band
is clearly correlated with the band that dominates the weighted PDF
for the (PM)P triad type shown in figure \ref{fig:NSWpdfkb2_TG}(f).
Moreover, the actual flux contribution from the (PM)P triads,
cf.~figure \ref{fig:NSFlux_cases}(e) (cyan dotted line), shows that
this triad type contributes over $80\%$ of the total flux, while
representing less than $0.0001\%$ of the total number of triads. Thus,
we are led to the conclusion that the band structure of the PDFs and
weighted PDFs is directly responsible for the observed {robust
  flux behaviour for this triad type}.  Looking at the weighted PDFs of
the four {dominating triad types shown in figures
  \ref{fig:NSWpdfkb2_TG}(a--d),} we again observe the presence of
persistent bands, some of which seem to interact and evolve while
keeping a strong coherence.  However, the corresponding fluxes shown
in figure \ref{fig:NSFlux_cases}(e) are not particularly strong,
suggesting that the coherence observed for these triad types is
perhaps a secondary effect associated with their interaction with the
flux-carrying (PM)P triads.  At this stage it is relevant to note an
important difference between the extreme and generic cases versus the
Taylor-Green case evident in figures \ref{fig:NSFlux_cases}(a,c,e):
while both the extreme and the generic cases show a gradual
development of their most important flux contributor, the triad type
PMP, with a late-time maximum, the Taylor-Green case reveals an
immediate growth of its most important flux contributor, the boundary
triad type (PM)P, with an early-time maximum at half the strength as
compared to the extreme and generic cases.  Finally, a remarkable
aspect of the flux to the wavenumber region $k>2$ in the Taylor-Green
case is that the weighted PDF of the P(PM) triad type shown in figure
\ref{fig:NSWpdfkb2_TG}(e) displays an extra symmetry about
$\Phi=\pi/2$ (on top of the usual symmetry about $\Phi=0$), which
implies that the flux contribution from this triad type is in fact
identically zero. This extra symmetry could be related to a low
effective dimension of the dynamical system associated with the
Fourier modes with $k\leq 2$; {the flux to the wavenumber region
  $k>10$,} to be discussed next, does not exhibit this symmetry
anymore.
  
We now discuss the flux towards wavenumbers $k>10$ in the flow with
the Taylor-Green initial condition. Comparing the data in table
\ref{tab:PDFbounds_TG}, last two columns, with the corresponding
columns for the other two cases in tables \ref{tab:PDFbounds_extreme}
and \ref{tab:PDFbounds_generic}, we see no significant differences in
the PDFs for all triad types. However, for the weighted PDFs, the
Taylor-Green case displays a slightly larger range of values than in
the extreme or generic cases for the four {dominating} triad
types, while for the two boundary triad types, the range in the
Taylor-Green case is several times larger than in the extreme or
generic cases. Looking at the weighted PDFs for each triad type in
figure \ref{fig:NSWpdfkb10_TG}, it is evident that the two boundary
triad types show the band structure characteristic of the flux to
wavenumbers $k>2$ discussed earlier. This band structure is
responsible for an increased range in the weighted PDFs for these
triad types as compared against the extreme and generic cases
(however, as we shall see, in terms of total fluxes these triads do
not contribute much). We also notice that the extra symmetry of the
weighted PDF for the triad P(PM) is now broken. As for the four
{dominating} triad types, cf.~figure
\ref{fig:NSWpdfkb10_TG}(a--d), these display a behaviour that is
closer to the extreme and generic cases, in that the maxima of the
weighted PDFs are more localized in time and the coherent patterns
exhibit more complex dynamics, with no obviously persistent bands, a
fact that is probably related to an increased number of degrees of
freedom involved in the {flux towards modes with wavenumbers
  $k>10$.}  An important difference with respect to the extreme and
generic cases is that local maxima appear in the weighted PDFs at
early times ($t=0.01$---$0.05$). This is further reflected in the
{evolution of the total flux shown in} figure
\ref{fig:NSFlux_cases}(f) which reveals an early growth phase of the
fluxes (starting at about $t=0.05$), much earlier than in the extreme
and generic cases {shown in figures \ref{fig:NSFlux_cases}(b,d).
  These figures} also show that the flux {towards modes with
  wavenumbers $k > 10$} for each triad type in the Taylor-Green case
{attains maxima at earlier times with} peak values about $7$
times smaller than in the corresponding fluxes in the extreme and
generic cases. Perhaps the only obvious similarity we could infer from
a comparison of the three cases {in figures
  \ref{fig:NSFlux_cases}(b,d,f)} is that the three {dominating}
flux contributors are the same in these cases, namely, PMM, then PPP,
followed by the PMP helical triads, with similar peak ratios. Another
similarity---looking at the same figures---is that the three cases show a small negative flux
contribution from the PPM triads near the end of the time window. This
is also evident from the weighted PDFs shown for each case in figures
\ref{fig:NSWpdfkb10}(b), \ref{fig:NSWpdfkb10_generic}(b) and
\ref{fig:NSWpdfkb10_TG}(b) where it can be seen that at late times the
distributions become more concentrated near $\Phi=\pm \pi$. With an
even tinier (but positive) flux contribution, we find the boundary
triad type (PM)P, {for which the time evolutions of the PDFs} are
compared across the three cases in figure \ref{fig:NSpdfkb10},
{to be qualitatively similar,} except for a persistent band
around $\Phi=0$ in the Taylor-Green case.

% 1/2/pi*1.02 = 0.162
% 1/2/pi*0.98 = 0.156

%1/2/pi*1.05 = 0.167
%1/2/pi*0.95 = 0.151

\begin{table}
  \begin{center}
    % \hspace*{-1.1cm}
    \begin{tabular}{|l|c|c|c|c|} \hline
      \backslashbox{triad}{cases}   & $k=2$, PDF  & $k=2$, wPDF  & $k=10$, PDF &  $k=10$, wPDF \\ \hline
     PPP   & \textcolor[rgb]{0.7,0.7,0.7}{[0.158,0.161]}   &  [0,1.785]  & \textcolor[rgb]{0.7,0.7,0.7}{[0.159,0.160]}  &  [0.004,0.859] \\ 
      PPM  &  \textcolor[rgb]{0.7,0.7,0.7}{[0.158,0.161]}  &  [0,2.281]  & \textcolor[rgb]{0.7,0.7,0.7}{[0.159,0.160]}  &  [0.019,0.876] \\    
      PMP    & [0.153,0.168] &  [0,1.831] &  \textcolor[rgb]{0.7,0.7,0.7}{[0.159,0.161]} &   $[0.012,0.616]$ \\
      PMM   & \textcolor[rgb]{0.7,0.7,0.7}{[0.158,0.162]} & [0,1.740]  & \textcolor[rgb]{0.7,0.7,0.7}{[0.159,0.160]} &  $[0.003,0.918]$  \\ 
	 P(PM)    & [0.149,0.173] &  [0,0.673] &  \textcolor[rgb]{0.7,0.7,0.7}{[0.157,0.163]} &   $[0.000,2.230]$ \\
	 (PM)P    & [0.000,0.750] &  [0.000,3.979] &  [0.148,0.175] &   $[0.000,2.239]$ \\\hline
    \end{tabular}
  \end{center}
  \caption{{[Taylor-Green case:] Upper and lower bounds on the PDFs and wPDFs of triad phases of the different types in the Navier-Stokes flow with the Taylor-Green initial data, cf.~figures \ref{fig:NSpdfkb2_TG}--\ref{fig:NSWpdfkb10_TG}. {Shaded} intervals (light gray) represent PDFs that are very close (within $\pm5\%$) to the uniform distribution $1/2\pi \approx 0.16$.}}
  \label{tab:PDFbounds_TG}
\end{table}

\begin{figure}
\begin{center}
\mbox{\subfigure[${\mathcal{P}}_{\mathcal{C}_2}^{PMP}(\Phi)(t)$]{\includegraphics[width=0.45\textwidth]{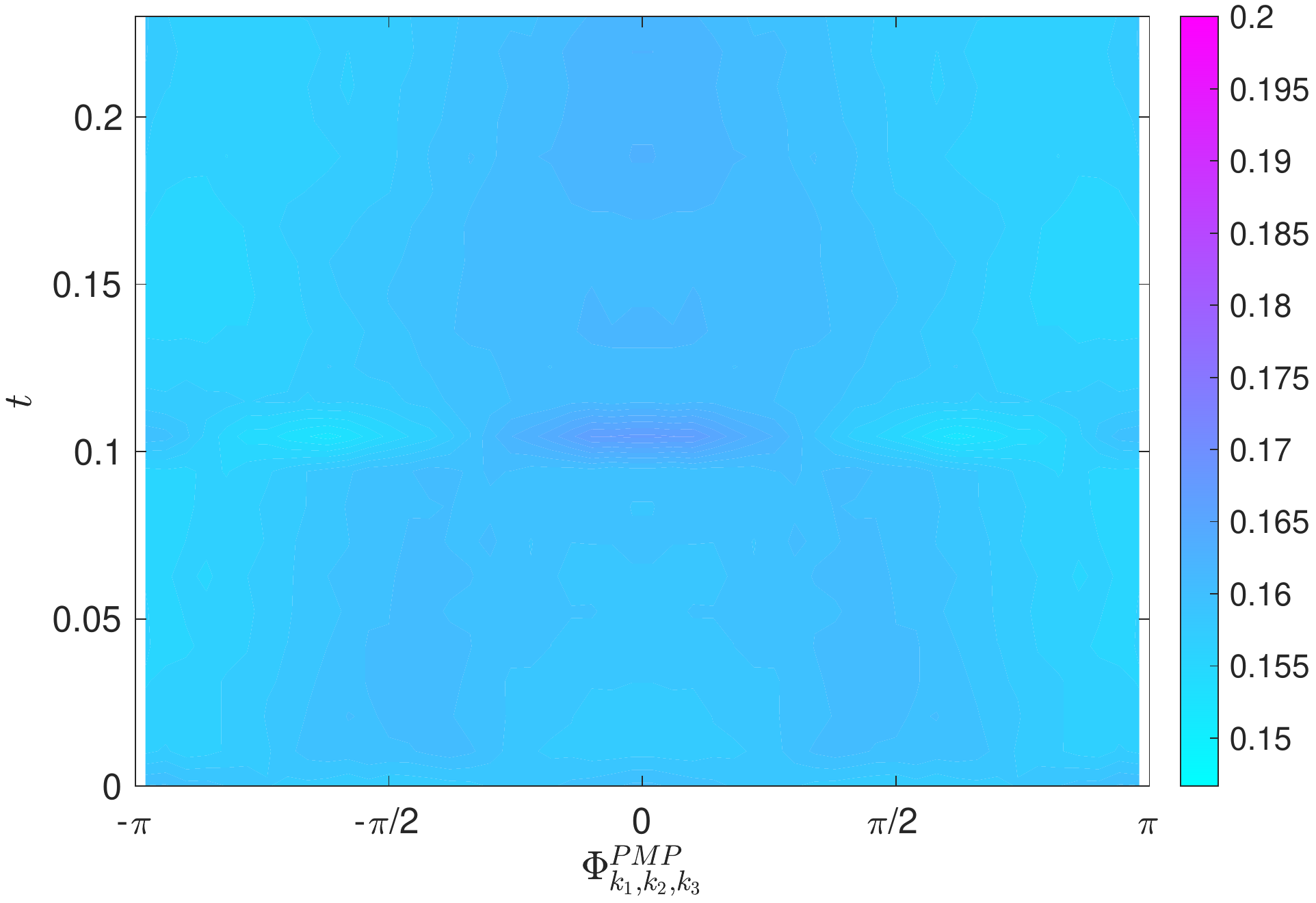}}}
\mbox{\subfigure[${\mathcal{P}}_{\mathcal{C}_2}^{P(PM)}(\Phi)(t)$]{\includegraphics[width=0.45\textwidth]{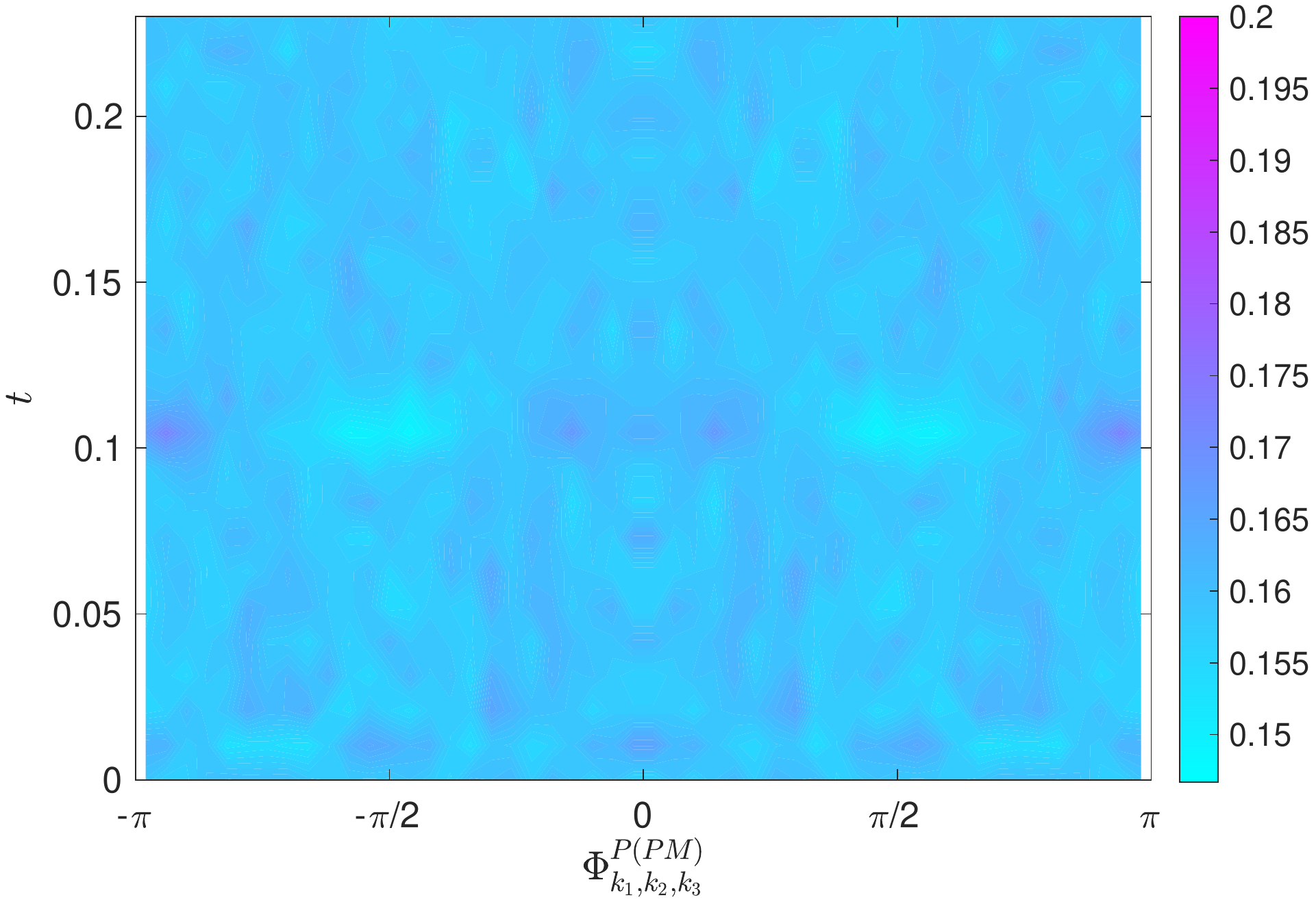}}\qquad
\subfigure[${\mathcal{P}}_{\mathcal{C}_2}^{(PM)P}(\Phi)(t)$]{\includegraphics[width=0.45\textwidth]{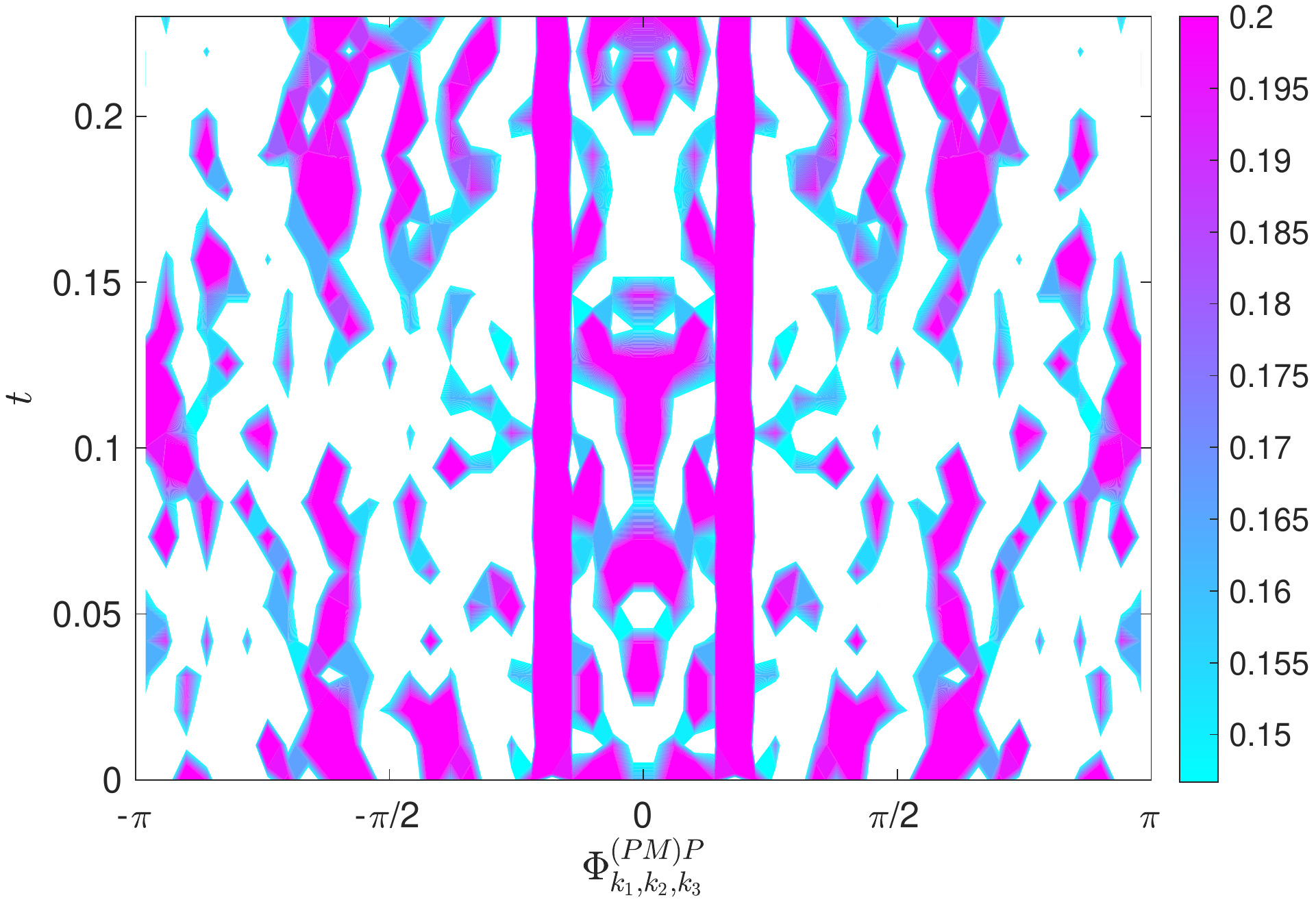}}}
\caption{{[Taylor-Green case, $k>2$:]} Time evolution of the PDFs
  of the triad phase angle ${\mathcal{P}}_{\mathcal{C}_2}^{s_1 s_2
    s_3}(\Phi)$ {for triads of different types in the
    Navier-Stokes flow with the Taylor-Green initial condition.  Only
    triad types with PDFs revealing variability of at least $\pm5\%$
    with respect to the uniform distribution equal to $1/2\pi \approx
    0.16$ (corresponding to light blue colour in the plots), are
    shown:} (a) ``PMP'', (b) ``P(PM)'' and (c) ``(PM)P ''. The PDFs
  for the triad types ``PPP'', ``PPM'' and ``PMM'' are not shown, as
  they are essentially uniform (see table \ref{tab:PDFbounds_TG}).}
\label{fig:NSpdfkb2_TG}
\end{center}
\end{figure}

%\begin{figure}
%\begin{center}
%\mbox{\subfigure[]{\includegraphics[width=0.45\textwidth]{Figs2/Wpdfpppkb2TG.pdf}}\qquad
%\subfigure[]{\includegraphics[width=0.45\textwidth]{Figs2/Wpdfppmkb2TG.pdf}}}\\
%\mbox{\subfigure[]{\includegraphics[width=0.45\textwidth]{Figs2/Wpdfpmpkb2TG.pdf}}\qquad
%\subfigure[]{\includegraphics[width=0.45\textwidth]{Figs2/Wpdfpmmkb2TG.pdf}}}
%\mbox{\subfigure[]{\includegraphics[width=0.45\textwidth]{Figs2/Wpdfp_pmkb2TG.pdf}}\qquad
%\subfigure[]{\includegraphics[width=0.45\textwidth]{Figs2/Wpdfpm_pkb2TG.pdf}}}
%
%\caption{Time evolution of the weighted PDF of the triad phase
%  $\varphi_{k_1,k_2}^{k_3}$ in the solution of the Navier-Stokes equation with Taylor-Green vortex initial condition with triad type 
% (a) ``PPP'', (b)``PPM'', (c) ``PMP''  , (d) ``PMM'',(e) ``P(PM)'' and (f) ``(PM)P ''  for $\E_0 = 250$ and $k=2$.}
%\label{fig:NSWpdfkb2_TG}
%\end{center}
%\end{figure}

\begin{figure}
\begin{center}
\mbox{\subfigure[${\mathcal{W}}_{\mathcal{C}_2}^{PPP}(\Phi)(t)$]{\includegraphics[width=0.45\textwidth]{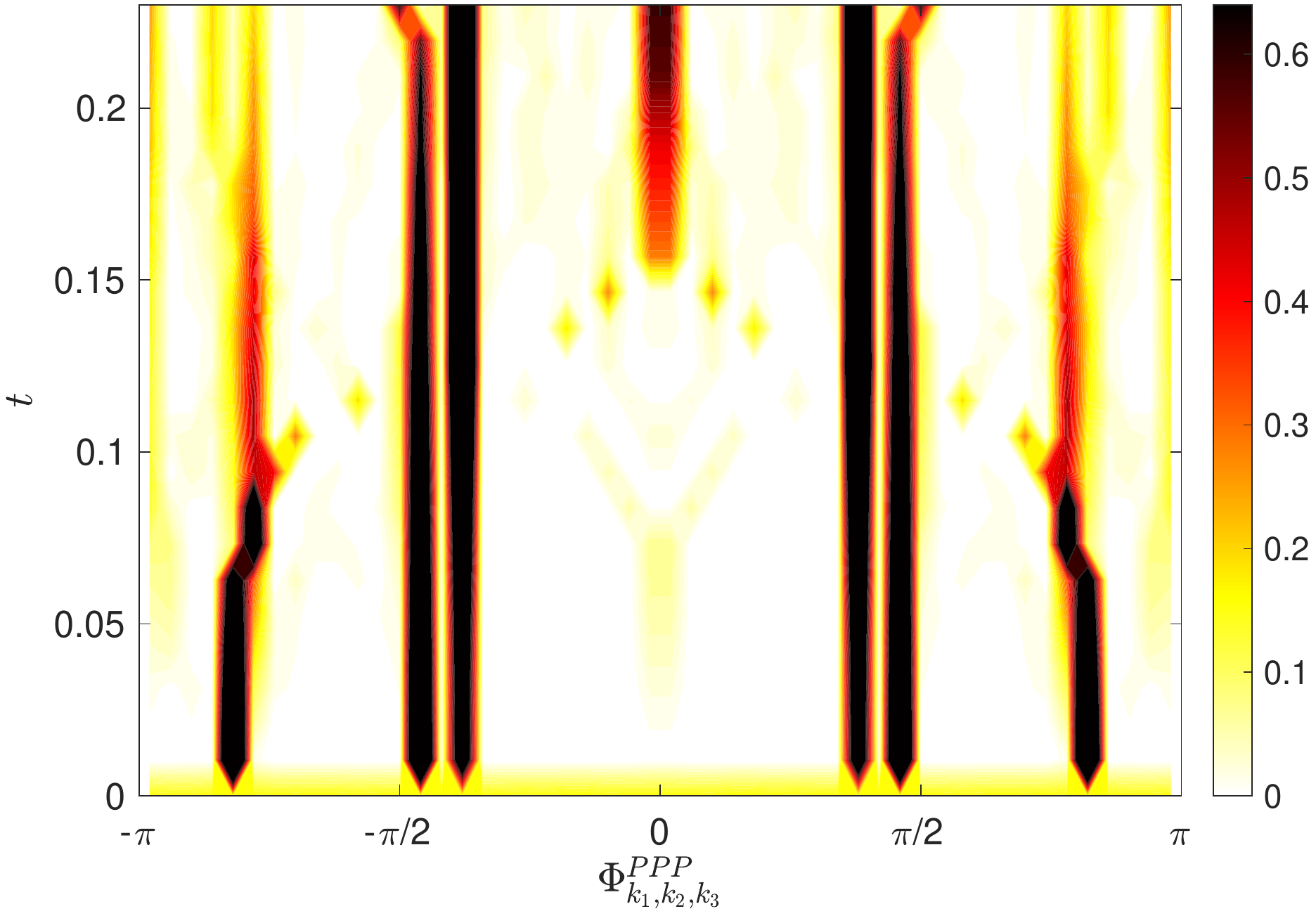}}\qquad
\subfigure[${\mathcal{W}}_{\mathcal{C}_2}^{PPM}(\Phi)(t)$]{\includegraphics[width=0.45\textwidth]{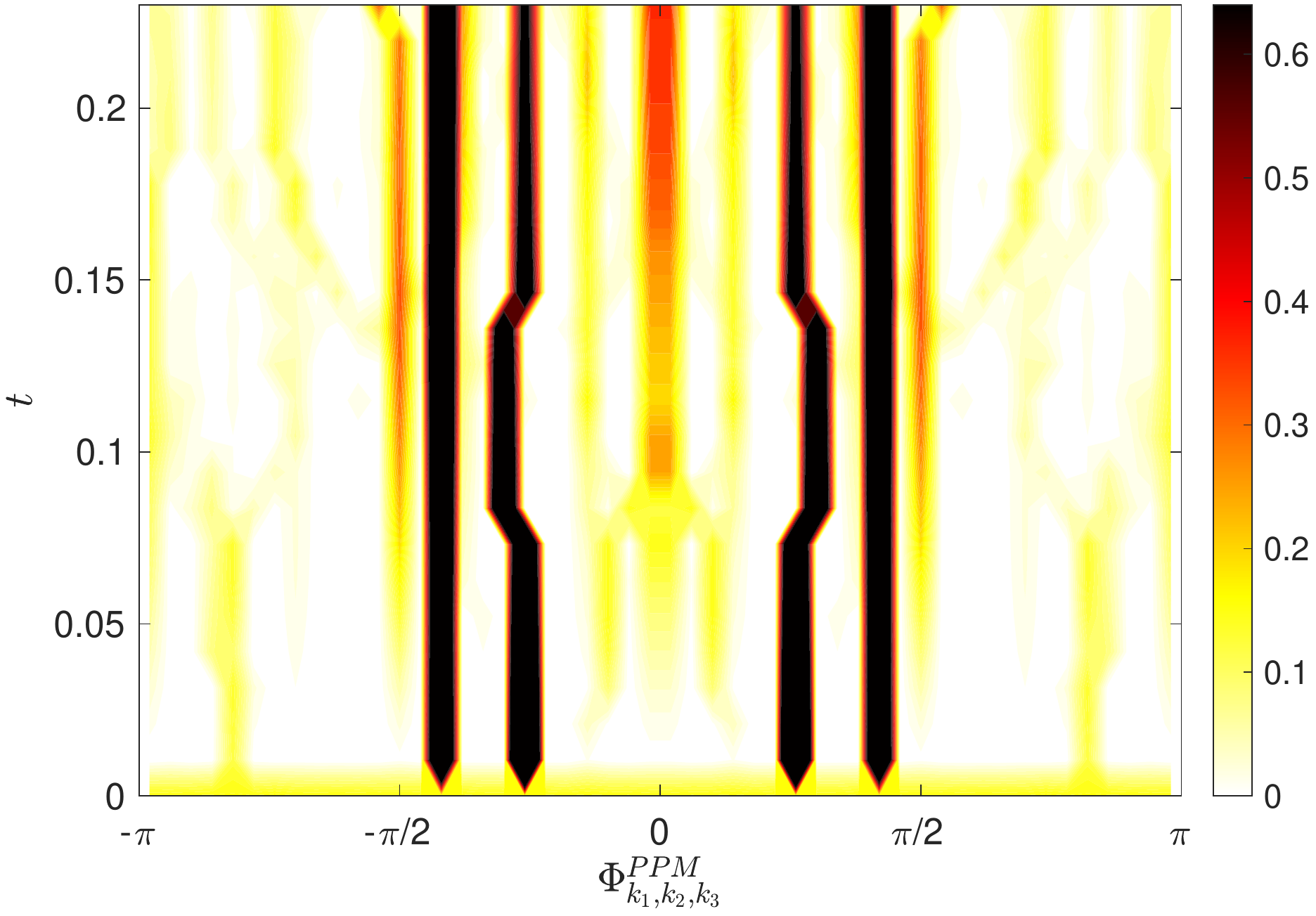}}}\\
\mbox{\subfigure[${\mathcal{W}}_{\mathcal{C}_2}^{PMP}(\Phi)(t)$]{\includegraphics[width=0.45\textwidth]{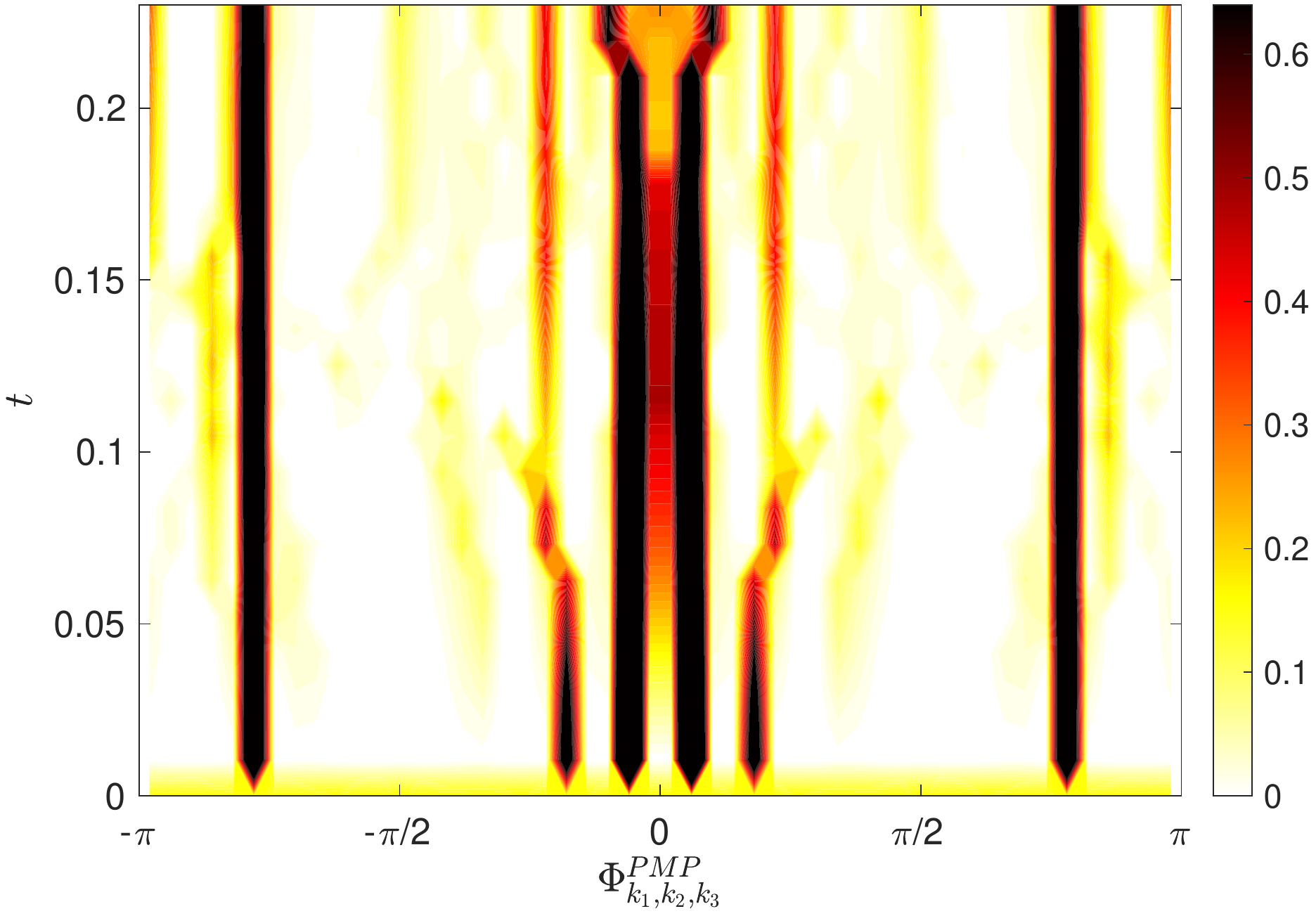}}\qquad
\subfigure[${\mathcal{W}}_{\mathcal{C}_2}^{PMM}(\Phi)(t)$]{\includegraphics[width=0.45\textwidth]{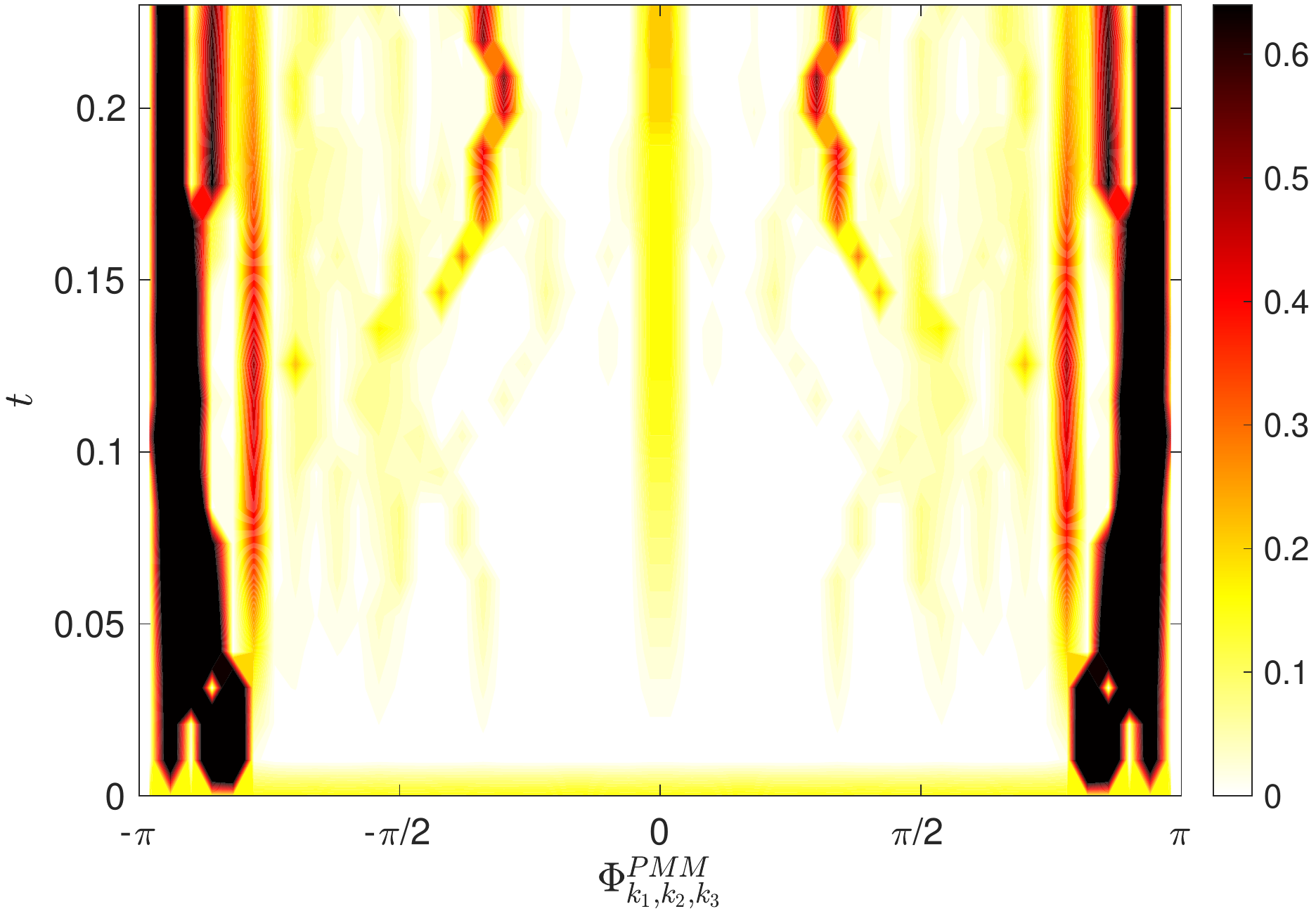}}}
\mbox{\subfigure[${\mathcal{W}}_{\mathcal{C}_2}^{P(PM)}(\Phi)(t)$]{\includegraphics[width=0.45\textwidth]{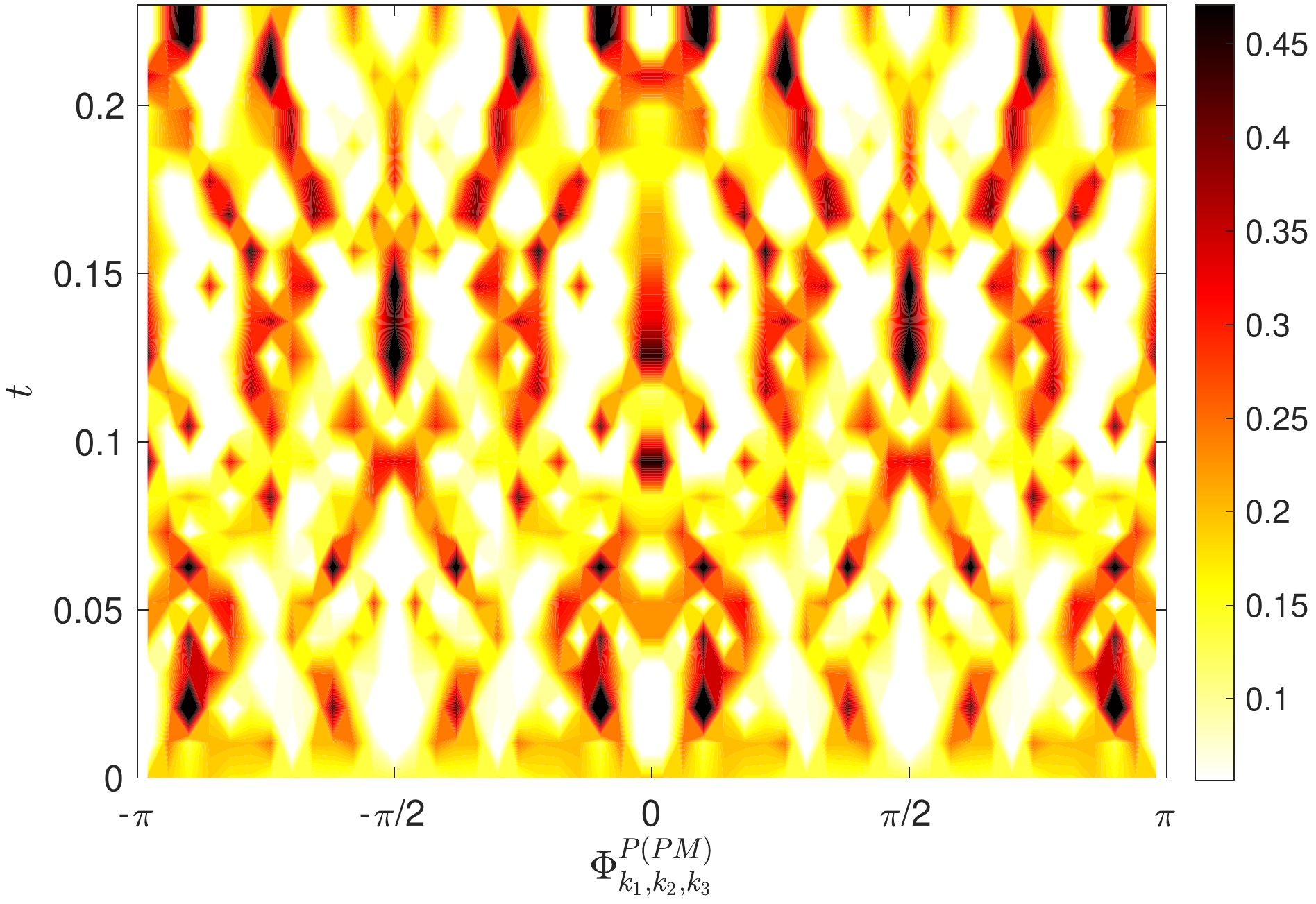}}\qquad
\subfigure[${\mathcal{W}}_{\mathcal{C}_2}^{(PM)P}(\Phi)(t)$]{\includegraphics[width=0.45\textwidth]{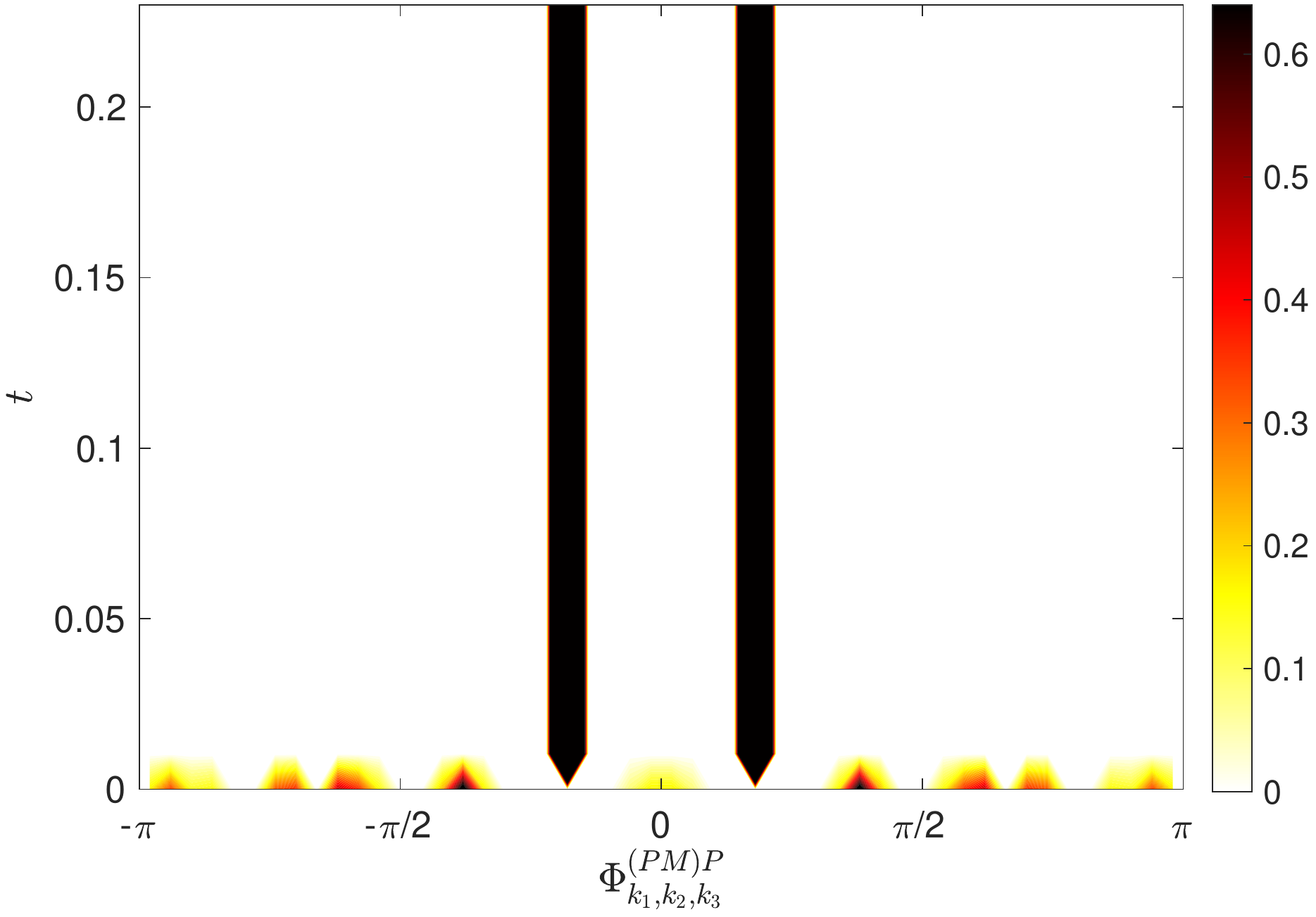}}}

\caption{{[Taylor-Green case, $k>2$:]} Time evolution of the
  weighted PDF of the triad phase angle
  ${\mathcal{W}}_{\mathcal{C}_{2}}^{s_1 s_2 s_3}(\Phi)$,
  cf.~expression \eqref{eq:WC}, in the Navier-Stokes flow with the
  Taylor-Green initial condition for the triad types (a) ``PPP'', (b)
  ``PPM'', (c) ``PMP'' , (d) ``PMM'', (e) ``P(PM)'' and (f) ``(PM)P''.
  In these plots the uniform distribution corresponding to $1/2\pi
  \approx 0.16$ is depicted by yellow ($25\%$ of the colour bars).}
\label{fig:NSWpdfkb2_TG}
\end{center}
\end{figure}

\begin{figure}
\begin{center}
\mbox{\subfigure[${\mathcal{W}}_{\mathcal{C}_{10}}^{PPP}(\Phi)(t)$]{\includegraphics[width=0.45\textwidth]{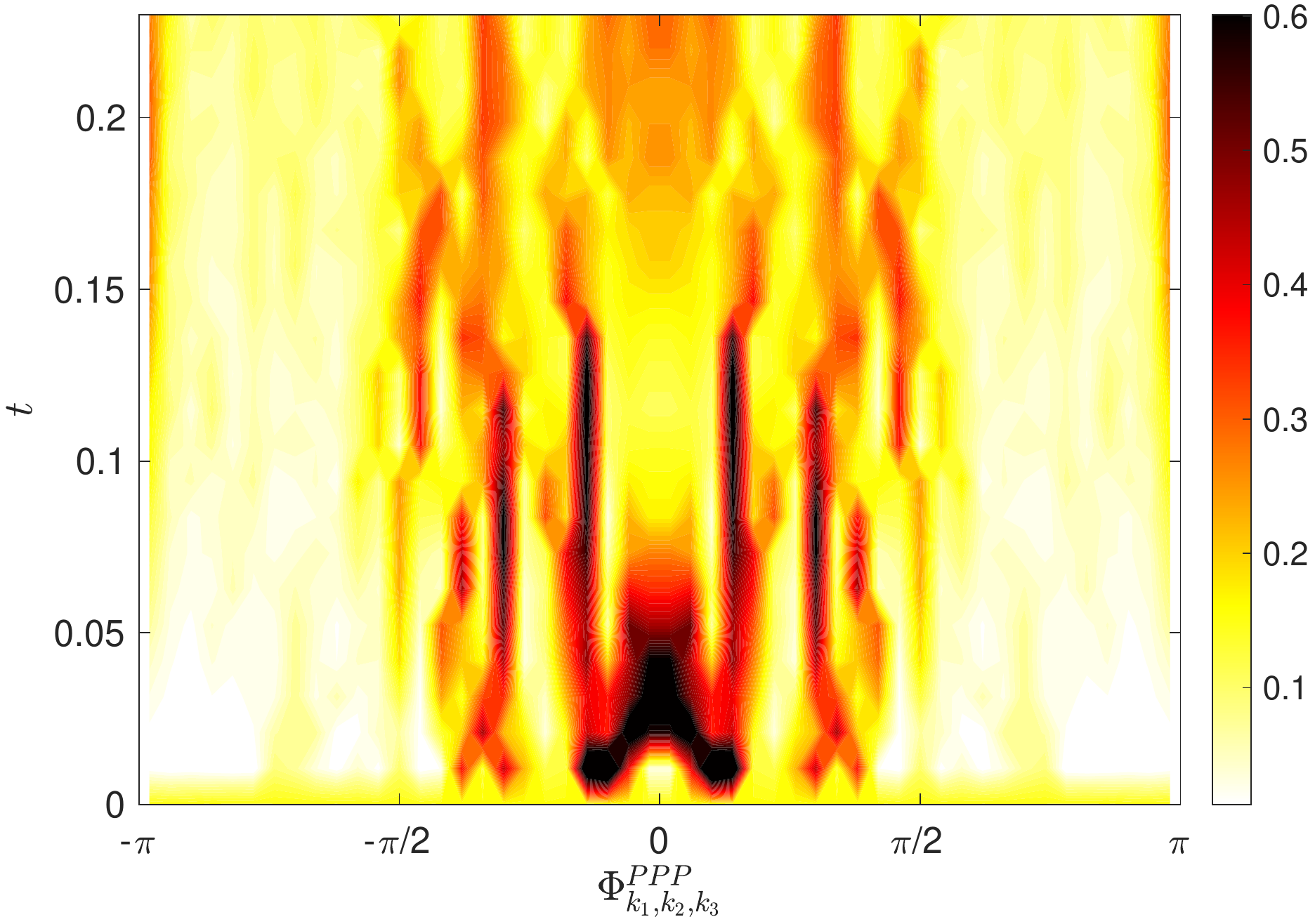}}\qquad
\subfigure[${\mathcal{W}}_{\mathcal{C}_{10}}^{PPM}(\Phi)(t)$]{\includegraphics[width=0.45\textwidth]{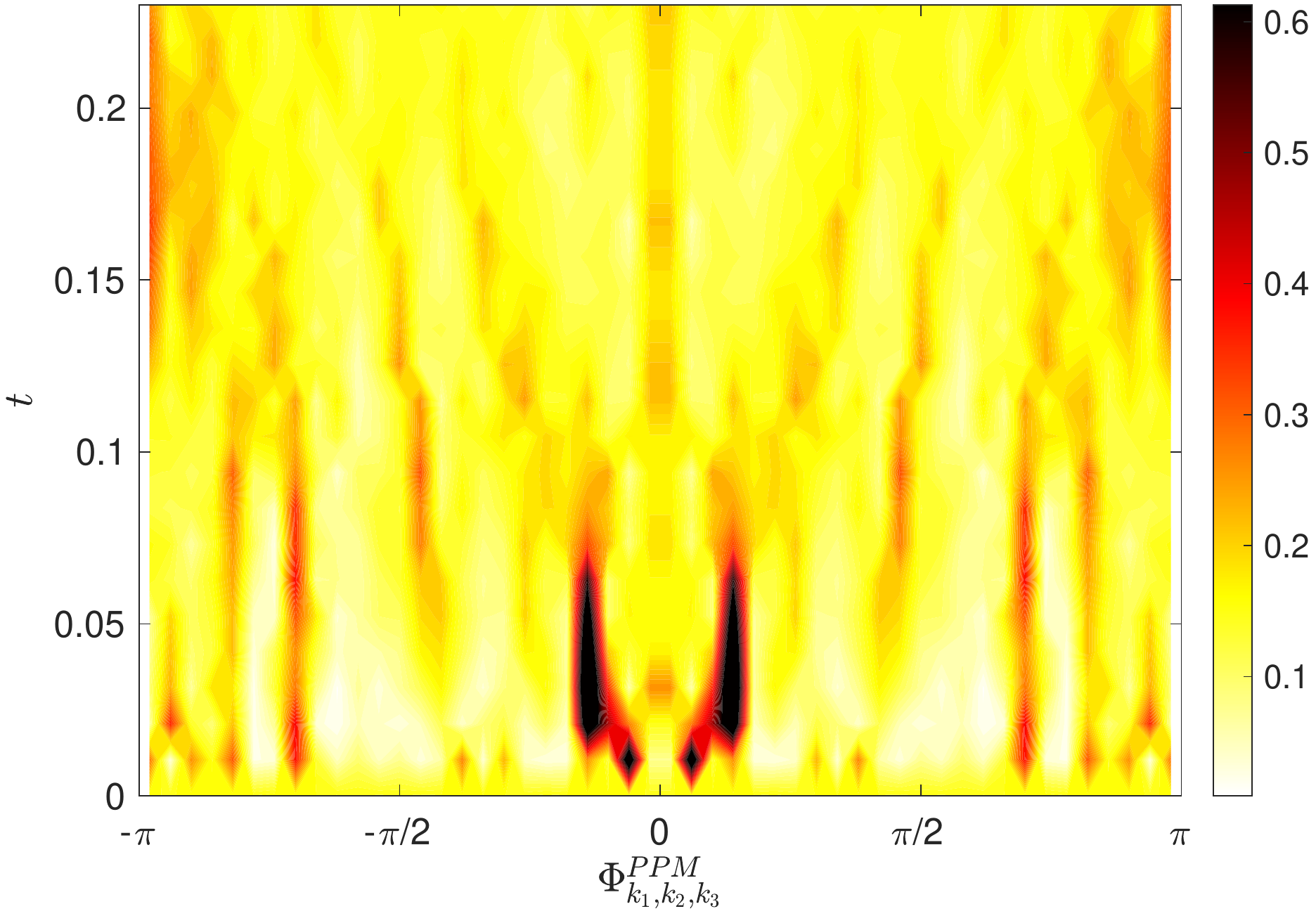}}}\\
\mbox{\subfigure[${\mathcal{W}}_{\mathcal{C}_{10}}^{PMP}(\Phi)(t)$]{\includegraphics[width=0.45\textwidth]{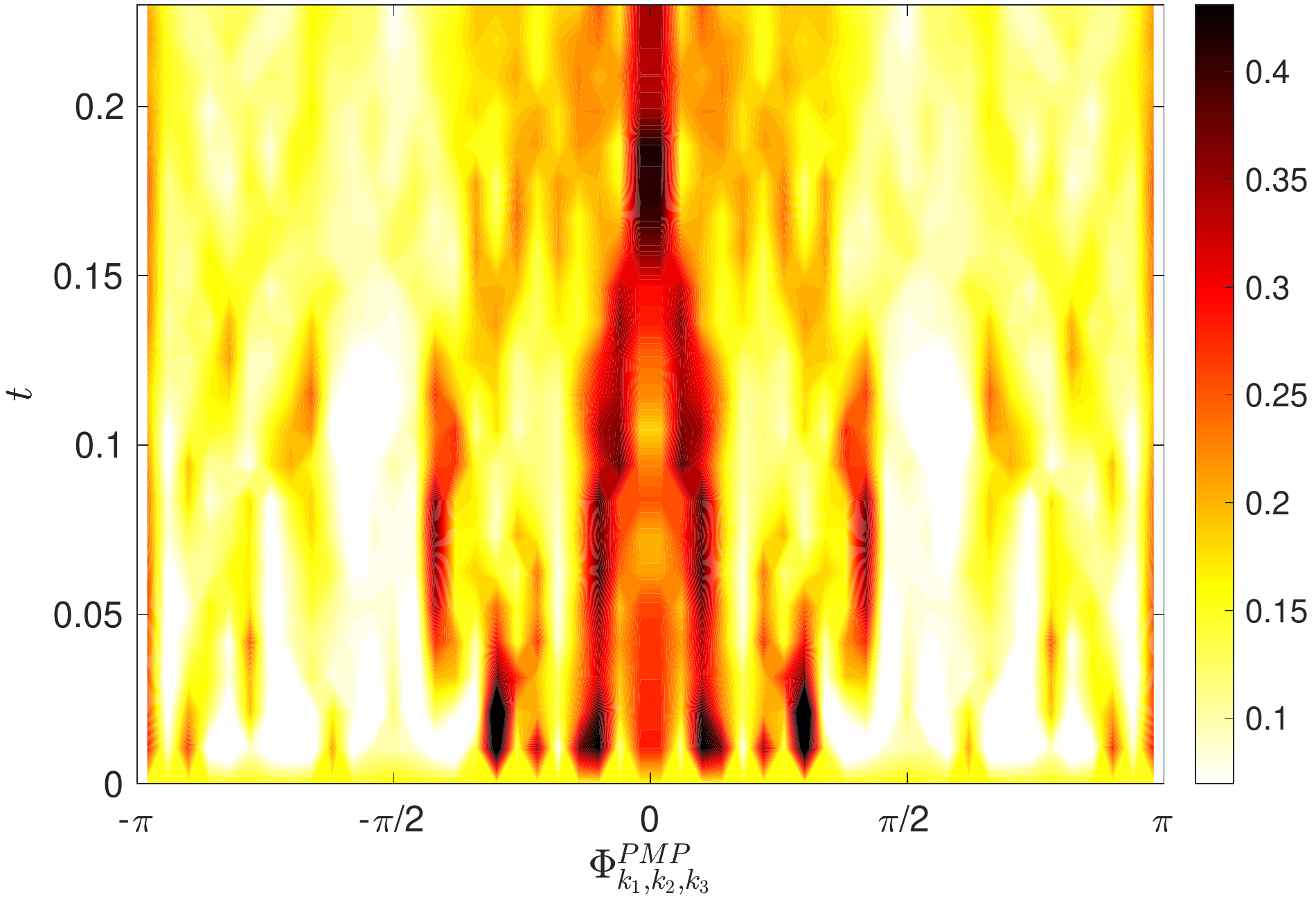}}\qquad
\subfigure[${\mathcal{W}}_{\mathcal{C}_{10}}^{PMM}(\Phi)(t)$]{\includegraphics[width=0.45\textwidth]{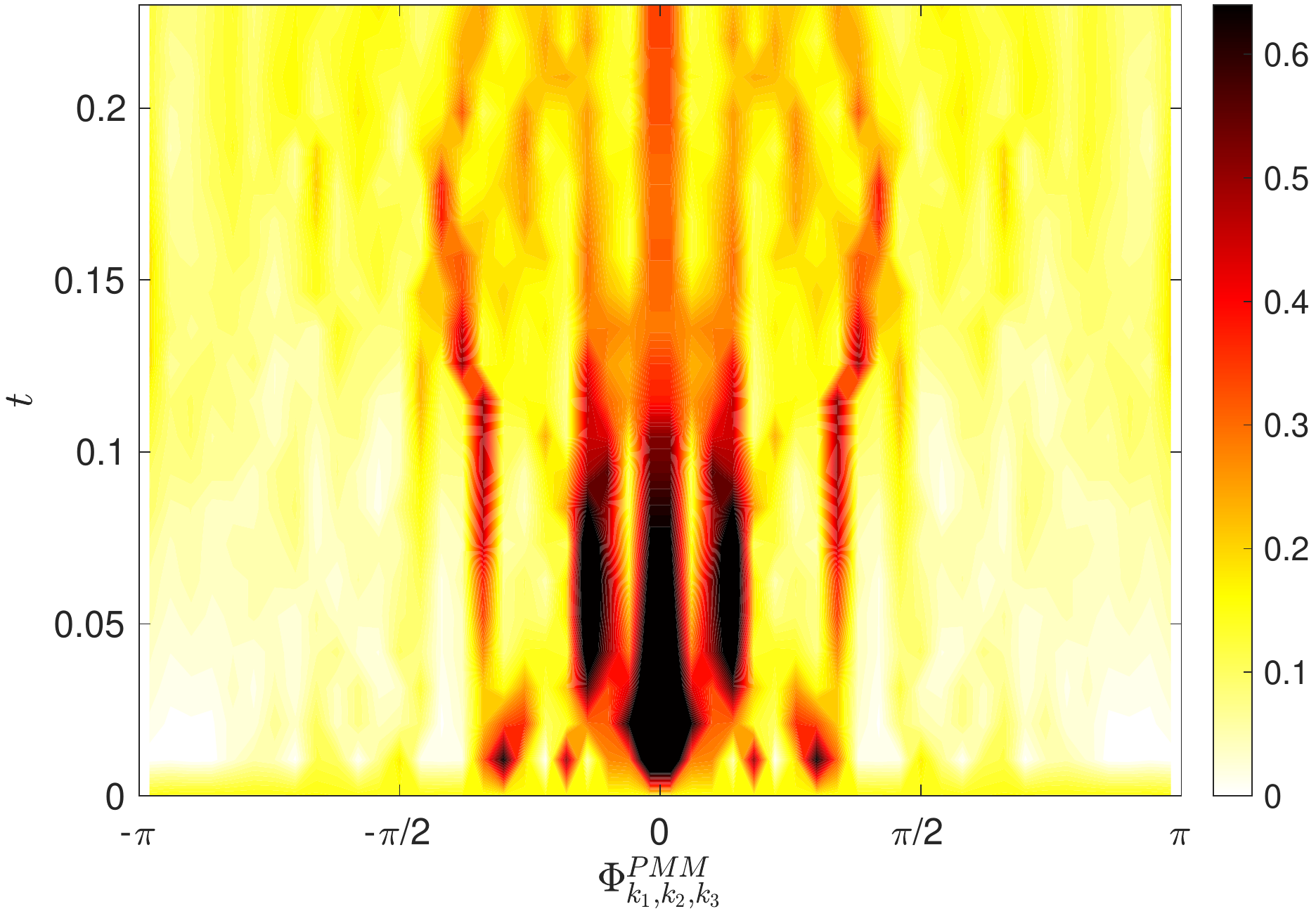}}}
\mbox{\subfigure[${\mathcal{W}}_{\mathcal{C}_{10}}^{P(PM)}(\Phi)(t)$]{\includegraphics[width=0.45\textwidth]{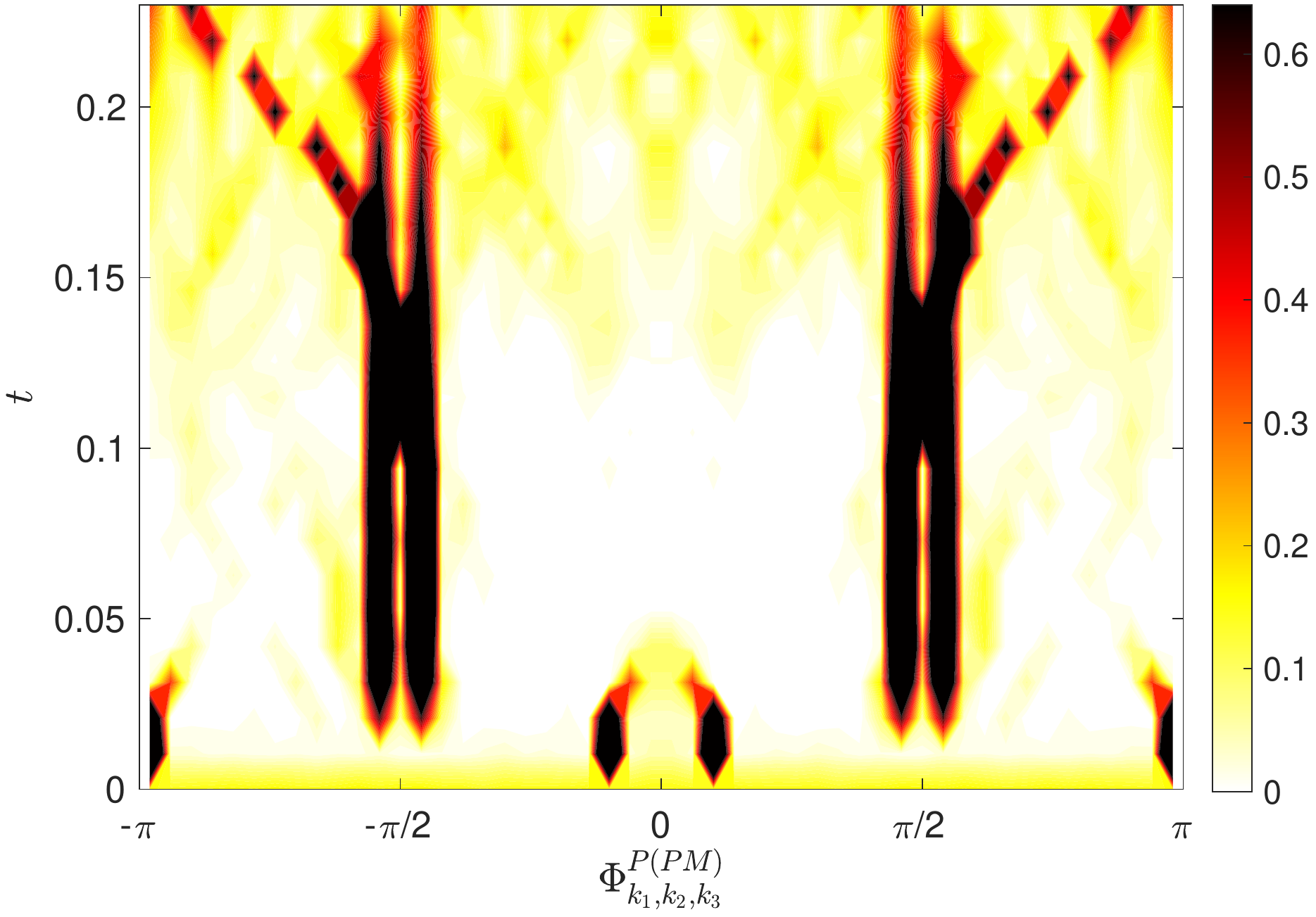}}\qquad
\subfigure[${\mathcal{W}}_{\mathcal{C}_{10}}^{(PM)P}(\Phi)(t)$]{\includegraphics[width=0.45\textwidth]{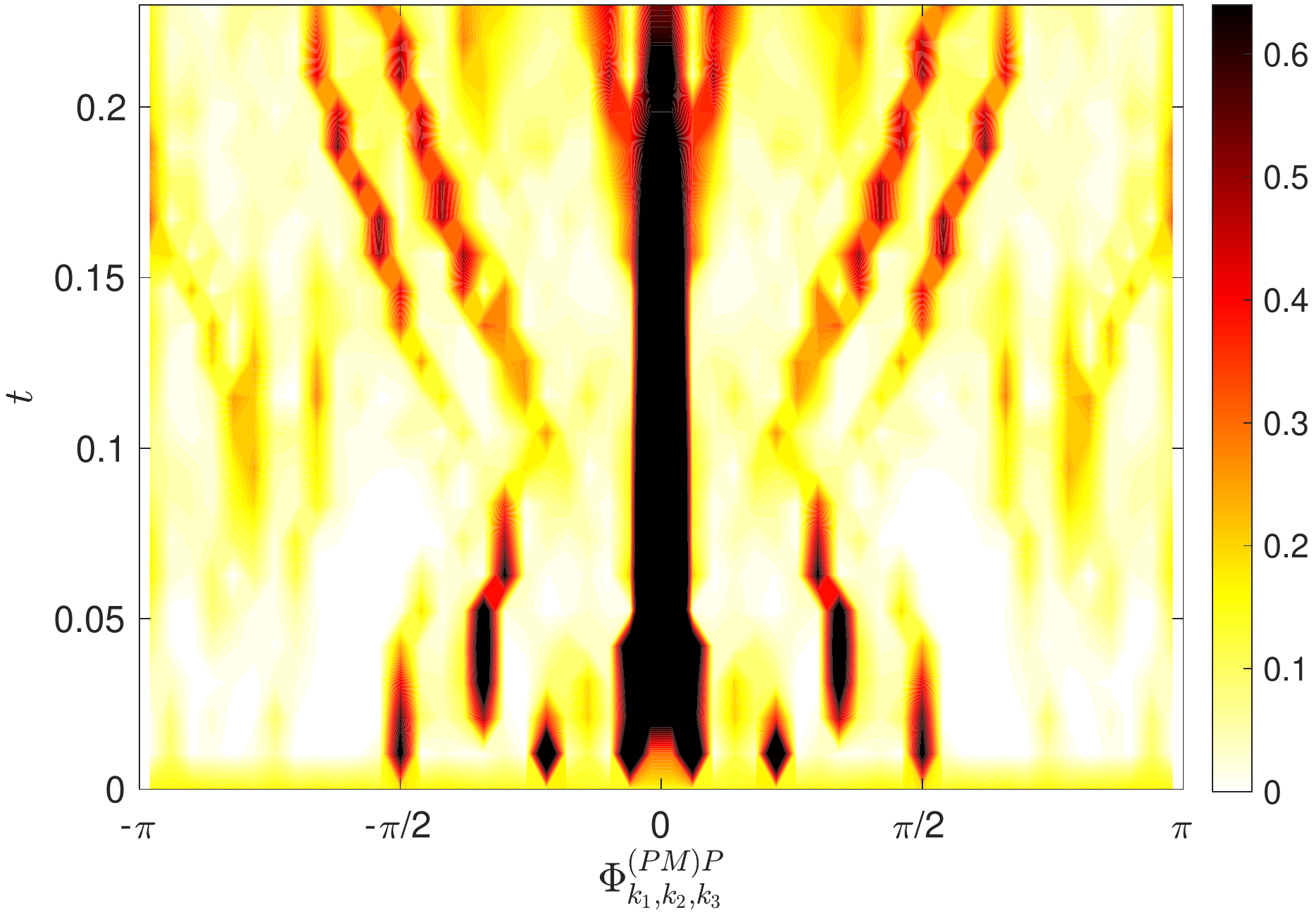}}}

\caption{{[Taylor-Green case, $k>10$:]} Time evolution of the
  weighted PDF of the triad phase angle
  ${\mathcal{W}}_{\mathcal{C}_{10}}^{s_1 s_2 s_3}(\Phi)$,
  cf.~expression \eqref{eq:WC}, in the Navier-Stokes flow with the
  Taylor-Green initial condition for the triad types (a) ``PPP'', (b)
  ``PPM'', (c) ``PMP'' , (d) ``PMM'', (e) ``P(PM)'' and (f) ``(PM)P''.
  In these plots the uniform distribution corresponding to $1/2\pi
  \approx 0.16$ is depicted by yellow ($25\%$ of the colour bars).}
\label{fig:NSWpdfkb10_TG}
\end{center}
\end{figure}

\FloatBarrier
%%%%%%%%%%%%%%%%%%%%%%%%%%%%%%%%%%%%%%%%%%%%%%%%%%%%%%%%%%%%%%%%%%%%%%%%%%%%%%%
%%%%%%%%%%%%%%%%%%%%%%%%%%%%%%%%%%%%%%%%%%%%%%%%%%%%%%%%%%%%%%%%%%%%%%%%%%%%%%%
%%%%%%%%%%%%%%%%%%%%%%%%%%%%%%%%%%%%%%%%%%%%%%%%%%%%%%%%%%%%%%%%%%%%%%%%%%%%%%%
\section{Discussion and Conclusions}
\label{sec:final}

The goal of our study is to shed light on the nature of nonlinear
modal interactions realizing the most extreme (in terms of
amplification of the enstrophy) scenarios achievable in 1D Burgers and
3D Navier-Stokes flows. To provide context, we compare these
interactions against the modal interactions occurring under generic
and unimodal initial conditions in such flows. To address these
questions we introduce analysis tools specifically designed to probe
elementary triadic interactions responsible for the transfer of energy
across scales and deploy these diagnostic tools to investigate Burgers
and Navier-Stokes flows with different initial conditions given in
table \ref{tab:cases}. The extreme initial conditions are constructed
to maximize the finite-time growth of enstrophy, which serves as a
measure of the regularity of the solution, by solving PDE-constrained
optimization problems of the type \eqref{eq:maxE}
\citep{ap11a,KangYumProtas2020}. A remarkable feature of the 1D
viscous Burgers and 3D Navier-Stokes flows obtained form such extreme
initial conditions is that they are characterized by the same
power-law dependence of the maximum attained enstrophy on the initial
enstrophy $\E_0$ given by relation \eqref{eq:maxT_vs_E0}. The generic
initial conditions are obtained from the extreme initial data by
removing the spatial coherence while retaining its energy spectrum.
The unimodal initial data has a small number of Fourier components and
is taken in the form of the Taylor-Green vortex \citep{tg37} in the
case of the 3D Navier-Stokes flow.

The first main finding is that while under the extreme initial data 1D
viscous Burgers flows and 3D Navier-Stokes flows reveal the same
relative level of enstrophy amplification by nonlinear effects, this
behaviour is realized by entirely different forms of modal
interactions. More precisely, 1D Burgers flows, both inviscid and
viscous, reveal highly coherent behaviour with triad phases displaying
preferential alignment at two values, namely, $\pm \pi/2$, cf.~figures
\ref{fig:B0pdf}(b) and \ref{fig:Bpdf}(b). However, the {\em
  flux-carrying} triads have phase values predominantly aligned %only
at $\pi/2$, cf.~figures \ref{fig:B0Wpdf}(b) and \ref{fig:BWpdf}(b),
which saturates the nonlinearity and therefore maximizes the energy
flux towards small scales. On the other hand, in 3D Navier-Stokes
flows with the extreme initial data the level of coherence is
significantly lower: the various classes of helical triads
{exhibit phase distributions} much closer to uniform, cf.~table
\ref{tab:PDFbounds_extreme} and figure \ref{fig:NSpdfkb2}.  Again, as
in the 1D Burgers case, the {\em flux-carrying} triads in 3D
Navier-Stokes flows do show a high level of coherence, but their phase
values take a broader range of preferred values forming complicated
time-dependent patterns, cf.~figures \ref{fig:NSWpdfkb2} and
\ref{fig:NSWpdfkb10}. Thus, in 3D Navier-Stokes flows the energy flux
to small scales is realized by a small only subset of helical triads.
Interestingly, the solutions of both the 1D viscous Burgers system and
the 3D Navier-Stokes system with the extreme initial data exhibit
negative flux corresponding to the inverse energy cascade at early
times (figures \ref{fig:BPin}(b) and \ref{fig:NSflux}(b)). We refer to
this as the ``slingshot'' effect which allows the flow to reorganize
in order to reduce enstrophy production and conserve energy at early
stages before a burst of flux is released towards small scales in
anticipation of the instant of time when the enstrophy is to be
maximized. The triadic interactions realizing this early-time negative
flux have again quite different properties in the two cases, with the
1D Burgers flow exhibiting a significantly higher level of coherence.

The second main finding concerns the role of initial coherence.
Comparison of the results for the 1D viscous Burgers flows and 3D
Navier-Stokes flows corresponding to the extreme versus the generic
initial conditions {(i.e., the latter having the initial phases
  uniformly randomized)}, shows striking similarities between these
two types of flows. The maximum attained enstrophy is reduced only
slightly in the flows with the generic initial condition, cf.~figures
\ref{fig:BsEt}(b) and \ref{fig:NSEt}(b), which also reveal the
``slingshot'' effect, albeit weaker than the case with the extreme
initial data. Contributions to the total flux from helical triads of
different types are also remarkably similar in 3D Navier-Stokes flows
with the extreme and generic initial conditions, cf.~figures
\ref{fig:NSFlux_cases}(a--f), as are the corresponding time evolutions
of the helical triad phases.

The third main finding, focusing now on the 3D Navier-Stokes flows,
concerns the comparison between the flows {with the extreme or generic
  initial condition on the one hand and the flow with the Taylor-Green
  initial condition on the other.} As expected, when compared with the
extreme or generic cases, the Taylor-Green flow shows a smaller level
of flux {in the inertial range (i.e., flux towards $k>2$)}, and an
even smaller proportion of flux {in the dissipative range} (i.e., flux
towards $k>10$). These differences are consistent with the fact that
the flow with the extreme initial condition is designed to maximize
the flux towards small scales: energy does not spend too long
transitioning from the inertial range to the dissipative range {so as
  not to be} dissipated too early. In contrast, a salient feature of
the Taylor-Green flow, observed at all spatial scales, is that the
flux-carrying helical triads display a less dynamical, more rigid
pattern of phase coherence, with very persistent narrow {bands} of
preferred phase values, cf.~figures \ref{fig:NSWpdfkb2_TG} and
\ref{fig:NSWpdfkb10_TG}. Moreover, in the inertial range, represented
by the flux towards wavenumbers $k>2$, the Taylor-Green flow shows
quite a different distribution of the main helical triad contributors
to the direct cascade (i.e., the flux towards small scales), with the
boundary triad type (PM)P being the most important contributor by far,
cf.~figures \ref{fig:NSFlux_cases}(a,c,e). With reference to the same
figures, in the Taylor-Green case the main contributor to the inverse
cascade (i.e., the flux towards large scales) is the PMM helical triad
type, and not the PPP type, which is the main contributor in the
extreme and generic cases.  In contrast, in the dissipative range,
represented by the flux towards wavenumbers $k>10$, the behaviour of
the Taylor-Green flow is qualitatively similar to the extreme and
generic cases, and even quantitatively similar in terms of the
relative contribution from the different helical triad types to the
direct cascade, cf.~figures \ref{fig:NSFlux_cases}(b,d,f).

We provide a quick comparison of the results for the {flows} with
the extreme and generic initial conditions against the paradigm
suggested by Waleffe's instability assumption and the classification
of helical triads \citep{Waleffe1992}. First, we focus on the inertial
range, where most of the flux occurs. In particular, regarding the
flux towards {the region with} $k>2$, cf.~figures \ref{fig:NSFlux_cases}(a,c,e), its main contributor is the
set of Class III helical triads, or PMP in our notation; this is quite
in line with Waleffe's paradigm, according to which the strength of
the transfers towards small scales is the largest in that class (see
the thick arrows in figure \ref{fig:Waleffe}). The second main
contributor is the new set of helical triads we introduced, the (PM)P
triads, corresponding to the intersection
$|\mathbf{k}_1|=|\mathbf{k}_2|$ between Class II and Class III helical
triads. Again, {referring to} figure \ref{fig:Waleffe}, this is
quite in line with Waleffe's paradigm.

The contribution by the PPP helical triads (Class I in Waleffe's
notation) to the flux is, fortunately, a feature that can be compared
against all quantitative works on the subject (the other triad types
are more difficult to compare due to differences in {their
  definitions used by different} research groups). It has been
recently found in studies of forced statistically stationary
turbulence \citep{sahoo2018energy, alexakis2017helically,
  alexakis2018cascades} that the PPP helical triads contribute with an
inverse-cascade flux at wavenumbers in the inertial range and with a
direct-cascade flux in the dissipative range. Our results, being based
on freely decaying turbulence from extreme or generic initial
conditions, confirm this, even quantitatively, if one assumes (see,
for example, \cite{frisch1995turbulence}) that the state of the
maximum dissipation in a transient evolution is a proxy for a
statistically stationary turbulence in a forced evolution, at least
considering energy spectra and energy fluxes. Regarding the flux
towards {the region with} $k>2$, which sits in the inertial
range, we find that the PPP contribution is negative and of the order
of 10\% of the total flux in absolute terms near the time of the
maximum dissipation, {which is} in quantitative agreement with
the works cited above. Regarding the flux towards {the region
  with} $k>10$, which sits in the dissipative range, we find that the
PPP contribution is positive and of the order of 30\% of the total
flux in absolute terms near the time of the maximum dissipation, again
in quantitative agreement with the works cited above. It is worth
mentioning that the case of the flow with the Taylor-Green initial
condition leads to qualitatively different results in the
inertial-range fluxes (namely, the fluxes towards {the region
  with} $k>2$): the PPP contribution to the inverse cascade is there
but is small and the main contribution comes from the PMM helical
triad type (Class II in Waleffe's notation), cf.~figures
\ref{fig:NSFlux_cases}(a,c,e).

It is well known that appearance of small scales and singularity
formation in solutions of time-dependent PDE problems can be
conveniently characterized in terms of the evolution of singularities
in the complex extensions of these fields
\citep{Weideman2003,mbf08,sc09}. In particular, blow-up of the
solution is signaled by the collapse of some of these complex-plane
singularities onto the real axis (i.e., when the width of the
analyticity strip shrinks to zero). As regards the diagnostics
employed in the present study, it is an interesting open question how
properties of triad phases in a flow can be deduced from the evolution
of complex-plane singularities characterizing these solutions, and
vice versa.

\section*{Acknowledgments}

DK, BP and MDB acknowledge financial support from University College
Dublin via Seed Funding project SF1568, and the hospitality of this
institution where a part of this research was carried out. DK and BP
were also partially funded through an NSERC (Canada) Discovery Grant.
Computational resources were provided by Compute Canada under its
Resource Allocation Competition.

\bigskip
\noindent
{\bf Declaration of Interests.} The authors report no conflict of interest.

\FloatBarrier 
\appendix

%%%%%%%%%%%%%%%%%%%%%%%%%%%%%%%%%%%%%%%%%%%%%%%%%%%%%%%%%%%%%%%%%%%%%%%%%%%%%%%
%%%%%%%%%%%%%%%%%%%%%%%%%%%%%%%%%%%%%%%%%%%%%%%%%%%%%%%%%%%%%%%%%%%%%%%%%%%%%%%
%%%%%%%%%%%%%%%%%%%%%%%%%%%%%%%%%%%%%%%%%%%%%%%%%%%%%%%%%%%%%%%%%%%%%%%%%%%%%%%
\section{Minimizing the Blow-up Time in the 1D Inviscid Burgers Equation}
\label{sec:Bu0ext0}

In this appendix we demonstrate that the problem of finding an initial
condition $u_0 \in H^1(\Omega)$ for the inviscid Burgers system
\eqref{eq:Burgers0} with a fixed enstrophy $\E_0> 0$ such that the
corresponding solution blows up in the shortest time $t^*(u_0)$ does
not in fact admit a solution. In other words, it is possible to find
initial data $u_0$ with $\E(u_0) = \E_0$ such that the corresponding
solution of \eqref{eq:Burgers0} blows up arbitrarily quickly. We thus
have
\begin{theorem}
The minimization problem 
\begin{equation}
\min_{u_0 \in H^1(\Omega)} t^*(u_0) \quad \text{subject to} \quad \E(u_0) = \E_0
\label{eq:mint}
\end{equation}
has no solution.
\label{thm:mint}
\end{theorem}
\begin{proof}
  Since for the inviscid Burgers equation \eqref{eq:Burgers0} we have
  an explicit expression for the blow-up time $t^* = -1 / \left[ \inf_{x \in
    \Omega} \Dpartial{u_0}{x}(x) \right]$, problem \eqref{eq:mint} is
  equivalent to
\begin{equation}
\max_{g \in L^2(\Omega)} \| g \|_{L^{\infty}} \quad \text{subject to} \quad \frac{1}{2} \|g\|_{L^2}^2 = \E_0,
\label{eq:maxg}
\end{equation}
where $g := \Dpartial{u_0}{x}$. It is easy to construct a sequence of
functions $g_n$ with a fixed $L^2$ norm and vanishing support $\lim_{n
  \rightarrow \infty} |\supp_{x \in \Omega} g_n(x)| = 0$ such that
their $L^{\infty}$ norms increase without bound, $\lim_{n \rightarrow
  \infty} \| g_n \|_{L^{\infty}} = \infty$, thus demonstrating that a
maximum is not attained in problem \eqref{eq:maxg}.
\end{proof}

%%%%%%%%%%%%%%%%%%%%%%%%%%%%%%%%%%%%%%%%%%%%%%%%%%%%%%%%%%%%%%%%%%%%%%%%%%%%%%%
%%%%%%%%%%%%%%%%%%%%%%%%%%%%%%%%%%%%%%%%%%%%%%%%%%%%%%%%%%%%%%%%%%%%%%%%%%%%%%%
%%%%%%%%%%%%%%%%%%%%%%%%%%%%%%%%%%%%%%%%%%%%%%%%%%%%%%%%%%%%%%%%%%%%%%%%%%%%%%%
\section{Proof of the Symmetry of Conjugate Triads under Zero Helicity
  Assumption}
\label{sec:conjugate_triads}

In this appendix we state and prove the following theorem implying the
symmetry of conjugate triads.

\begin{theorem}
  If the incompressible velocity field $\u$ satisfies the condition of
  oddity under the  parity transformation, namely,
 \begin{equation}
\label{eq:oddity}
\mathbf{u}(-\mathbf{x},t) = -\mathbf{u}(\mathbf{x},t) \qquad \text{for all} \qquad \mathbf{x} \in \mathbb{R}^3, \quad t \in \mathbb{R},
 \end{equation}
then:
\begin{enumerate}
\item the helical modes $u_\mathbf{k}^+, u_\mathbf{k}^-$ satisfy the identity
$$u_{-\mathbf{k}}^s = - u_{\mathbf{k}}^{-s},\qquad \text{for all} \qquad \mathbf{k} \in \mathbb{Z}^3 \setminus \{\mathbf{0}\}, \quad s = \pm,$$
\item the helicity spectrum is zero: $\H_{\k} = 0, \quad \text{for all} \quad \k \in \mathbb{Z}^3 \setminus \{\mathbf{0}\}$, cf.~\eqref{eq:heli_H_k}, and
\item the generalized helical triad phases \eqref{eq:GHTP} exhibit the following invariance with respect to the transformations $s_j \rightarrow -s_j$, $j=1,2,3$:
\begin{equation}
\Phi_{{\mathbf{k}_1} {\mathbf{k}_2} {\mathbf{k}_3}}^{-s_{1} \, -s_{2}  \, -s_{3}} = 
-\Phi_{{\mathbf{k}_1} {\mathbf{k}_2} {\mathbf{k}_3}}^{s_{1} \,\, s_{2} \,\, s_{3}}, \qquad \text{whenever} \qquad  \k_1 + \k_2 + \k_3 = \0.
\label{eq:GHTP0}
\end{equation}
\end{enumerate}
\end{theorem}
\begin{proof}
Assuming condition (\ref{eq:oddity}) is satisfied, let us first establish the corresponding relation for the Fourier components of the velocity field. From equation (\ref{eq:Fourier3D}) we have
$$\mathbf{u}(\mathbf{x},t) = \sum_{\mathbf{k}\in \mathbb{Z}^3\setminus \mathbf{0}} \mathbf{\widehat{u}}_{\mathbf{k}}(t) \exp(2\pi i \mathbf{k}\cdot\mathbf{x})\, ,
$$
so
$$\mathbf{u}(-\mathbf{x},t) = \sum_{\mathbf{k}\in \mathbb{Z}^3\setminus \mathbf{0}} \mathbf{\widehat{u}}_{\mathbf{k}}(t) \exp(-2\pi i \mathbf{k}\cdot\mathbf{x}) = \sum_{\mathbf{k}\in \mathbb{Z}^3\setminus \mathbf{0}} \mathbf{\widehat{u}}_{-\mathbf{k}}(t) \exp(2\pi i \mathbf{k}\cdot\mathbf{x})\, ,
$$
where in the last equality we relabelled the wavevectors from $\mathbf{k}$ to $-\mathbf{k}$. Equating this expression with the Fourier expansion of $-\mathbf{u}(\mathbf{x},t)$ we obtain, from uniqueness of the expansion,
\begin{equation}
\label{eq:oddityFourier}
\mathbf{\widehat{u}}_{-\mathbf{k}}(t) = - \mathbf{\widehat{u}}_{\mathbf{k}}(t), \qquad \text{for \,\, all} \qquad \mathbf{k} \in \mathbb{Z}^3 \setminus \{\mathbf{0}\}\,.
\end{equation}
Now we prove each of the claims in the theorem:

\begin{enumerate}
\item 
\label{part:1}
The helical mode decomposition (\ref{eq:heli_decomp}) gives 
$$\mathbf{\widehat{u}}_{\mathbf{k}}(t) = \mathbf{h}_{\mathbf{k}}^{+} u_{\mathbf{k}}^{+}(t)+ \mathbf{h}_{\mathbf{k}}^- u_{\mathbf{k}}^-(t)\,$$
and noting that $\mathbf{h}_{-\mathbf{k}}^{s} = \mathbf{h}_{\mathbf{k}}^{-s}$ we obtain
$$\mathbf{\widehat{u}}_{-\mathbf{k}}(t)  = \mathbf{h}_{-\mathbf{k}}^{+} u_{-\mathbf{k}}^{+}(t)+ \mathbf{h}_{-\mathbf{k}}^- u_{-\mathbf{k}}^-(t) = \mathbf{h}_{\mathbf{k}}^{-} u_{-\mathbf{k}}^{+}(t)+ \mathbf{h}_{\mathbf{k}}^+ u_{-\mathbf{k}}^-(t).$$
Invoking uniqueness of the helical decomposition, relation \eqref{eq:oddityFourier} is equivalent to
$$u_{-\mathbf{k}}^{-}(t) = -u_{\mathbf{k}}^{+}(t) \,, \qquad \mathbf{k}\in \mathbb{Z}^3\setminus  \{\mathbf{0}\}\,.$$
This of course also implies that  $u_{-\mathbf{k}}^{+}(t) = -u_{\mathbf{k}}^{-}(t)$.

\item The helicity spectrum is defined as
$$\mathcal{H}_{\mathbf{k}}(t) = |u_{\mathbf{k}}^{+}(t)|^2 - |u_{\mathbf{k}}^{-}(t)|^2\,.$$
Since the original velocity field is real-valued, we obtain
$u_{-\mathbf{k}}^{s}(t) = [u_{\mathbf{k}}^{s}(t)]^*$, such that
$$|u_{\mathbf{k}}^{-}(t)| = |u_{\mathbf{-k}}^{-}(t)| = |{-u_{\mathbf{k}}^{+}(t)}| = |u_{\mathbf{k}}^{+}(t)|,$$
implies $\mathcal{H}_{\mathbf{k}}(t) = 0\,, \quad \mathbf{k}\in \mathbb{Z}^3\setminus  \{\mathbf{0}\}$.

\item The generalized helical triad phases are obtained, mod $2\pi$, from combining equations \eqref{eq:GHTP} and \eqref{eq:delta_ENS}:
$$\Phi_{{\mathbf{k}_1} {\mathbf{k}_2} {\mathbf{k}_3}}^{s_{1} \,\, s_{2} \,\, s_{3}} =  \phi_{\mathbf{k}_1}^{s_{1}} + \phi_{\mathbf{k}_2}^{s_{2}} + \phi_{\mathbf{k}_3}^{s_{3}} + \arg[\mathbf{h}_{\mathbf{k}_1}^{s_{1}}\times \mathbf{h}_{\mathbf{k}_2}^{s_{2}} \cdot \mathbf{h}_{\mathbf{k}_3}^{s_{3}}] + \begin{cases}
\arg(s_{2} k_2 - s_{1} k_1), & \text{if} \quad \mathbf{k}_3 \in \mathcal{C}; \mathbf{k}_1, \mathbf{k}_2 \in \mathcal{U}\setminus \mathcal{C}  \\
 \arg(s_{2} k_2 - s_{3} k_3), & \text{if} \quad \mathbf{k}_3, \mathbf{k}_2 \in \mathcal{C}; \mathbf{k}_1 \in \mathcal{U}\setminus \mathcal{C}
 \end{cases}\, ,
$$
where $\phi_{\mathbf{k}}^{s} = \arg u_{\mathbf{k}}^{s}(t)$. We now analyze the behaviour of this expression under the transformation $s_j \to -s_j, \,\, j=1,2,3$. First, using part \ref{part:1} of the theorem, we get:
$$\phi_{\mathbf{k}}^{-s} = \arg u_{\mathbf{k}}^{-s}(t) = \arg (-u_{-\mathbf{k}}^{s}(t))  = \pi +  \phi_{-\mathbf{k}}^{s}=  \pi -  \phi_{\mathbf{k}}^{s}\,,$$
where the last equality follows from the reality condition. Second, from the definition of the helical basis vectors we have
$$\mathbf{h}_{\mathbf{k}}^{-s} =[\mathbf{h}_{\mathbf{k}}^{s}]^*\,.$$
Third, we have, mod $2\pi$, $\arg(-s_{2} k_2 - (-s_{1}) k_1) = \pi - \arg(s_{2} k_2 - s_{1} k_1)$,  and so on. Therefore, we obtain
$$\hspace{-15mm}\Phi_{{\mathbf{k}_1} {\mathbf{k}_2} {\mathbf{k}_3}}^{-s_{1} \,\, -s_{2} \,\, -s_{3}} =  \phi_{\mathbf{k}_1}^{-s_{1}} + \phi_{\mathbf{k}_2}^{-s_{2}} + \phi_{\mathbf{k}_3}^{-s_{3}} + \arg[\mathbf{h}_{\mathbf{k}_1}^{-s_{1}}\times \mathbf{h}_{\mathbf{k}_2}^{-s_{2}} \cdot \mathbf{h}_{\mathbf{k}_3}^{-s_{3}}] + \begin{cases}
\arg(-s_{2} k_2 + s_{1} k_1), & \text{if} \quad \mathbf{k}_3 \in \mathcal{C}; \mathbf{k}_1, \mathbf{k}_2 \in \mathcal{U}\setminus \mathcal{C}\,,  \\
 \arg(-s_{2} k_2 + s_{3} k_3), & \text{if} \quad \mathbf{k}_3, \mathbf{k}_2 \in \mathcal{C}; \mathbf{k}_1 \in \mathcal{U}\setminus \mathcal{C} \,,
 \end{cases}
$$
$$= 3\pi -\left(\phi_{\mathbf{k}_1}^{s_{1}} + \phi_{\mathbf{k}_2}^{s_{2}} + \phi_{\mathbf{k}_3}^{s_{3}}\right) + \arg[\mathbf{h}_{\mathbf{k}_1}^{s_{1}}\times \mathbf{h}_{\mathbf{k}_2}^{s_{2}} \cdot \mathbf{h}_{\mathbf{k}_3}^{s_{3}}]^* + \pi - \begin{cases}
\arg(s_{2} k_2 - s_{1} k_1), & \text{if} \quad \mathbf{k}_3 \in \mathcal{C}; \mathbf{k}_1, \mathbf{k}_2 \in \mathcal{U}\setminus \mathcal{C}\,,  \\
 \arg(s_{2} k_2 - s_{3} k_3), & \text{if} \quad \mathbf{k}_3, \mathbf{k}_2 \in \mathcal{C}; \mathbf{k}_1 \in \mathcal{U}\setminus \mathcal{C} \,,
 \end{cases}
$$
$$= 4\pi -\left(\phi_{\mathbf{k}_1}^{s_{1}} + \phi_{\mathbf{k}_2}^{s_{2}} + \phi_{\mathbf{k}_3}^{s_{3}}\right) - \arg[\mathbf{h}_{\mathbf{k}_1}^{s_{1}}\times \mathbf{h}_{\mathbf{k}_2}^{s_{2}} \cdot \mathbf{h}_{\mathbf{k}_3}^{s_{3}}] - \begin{cases}
\arg(s_{2} k_2 - s_{1} k_1), & \text{if} \quad \mathbf{k}_3 \in \mathcal{C}; \mathbf{k}_1, \mathbf{k}_2 \in \mathcal{U}\setminus \mathcal{C}\,,  \\
 \arg(s_{2} k_2 - s_{3} k_3), & \text{if} \quad \mathbf{k}_3, \mathbf{k}_2 \in \mathcal{C}; \mathbf{k}_1 \in \mathcal{U}\setminus \mathcal{C} \,,
 \end{cases}
$$
so that, mod $2\pi$, we finally arrive at
$$\Phi_{{\mathbf{k}_1} {\mathbf{k}_2} {\mathbf{k}_3}}^{-s_{1} \,\, -s_{2} \,\, -s_{3}}= - \Phi_{{\mathbf{k}_1} {\mathbf{k}_2} {\mathbf{k}_3}}^{s_{1} \,\, s_{2} \,\, s_{3}}\,.$$

\end{enumerate}

\end{proof}

\section{Fourier-Lagrange Formula for the Inviscid Burgers Equation}
\label{sec:FL}

In order to determine the solutions of the inviscid Burgers system
\eqref{eq:Burgers0} with initial conditions \eqref{eq:Bu0sin} and
\eqref{eq:Bu0ext0}, we will use an adaptation of the Fourier-Lagrange
formula originally proposed by \citet{FournierFrisch1983} which
provides an exact expression for the Fourier coefficients
$\widehat{u}_{k}(t)$ of the solution for $0 \le t < t^*$. One
difference with respect to the approach in \citet{FournierFrisch1983}
is that here we work on a periodic, rather than unbounded, spatial
domain. As a starting point, we write the solution of system
\eqref{eq:Burgers0} in terms of the flowmap $\xi = \xi(a,t)$, where $a
\in [0,1]$ is the Lagrangian coordinate, as \cite{kl04}
\begin{equation}
u(\xi(a,t),t) = u_0(a), \quad \xi(a,t) = a + t u_0(a), \quad a \in [0,1], \quad 0 \le t < t^*.
\label{eq:Blag}
\end{equation}
Then, the Fourier coefficients of the solution can be expressed as,
cf.~\eqref{eq:Fourier1D},
\begin{align*}
\widehat{u}_k(t) &= \int_0^1 e^{-2 \pi i k x} u(x,t) \, dx = \int_0^1 e^{-2 \pi i k \xi(a,t)} u(\xi(a,t),t) \, d\xi \\
& = \int_0^1 \frac{d}{d\xi} \left( \frac{-1}{2 \pi i k} e^{-2 \pi i k \xi(a,t)} \right) u(\xi(a,t),t) \, d\xi 
= \frac{1}{2 \pi i k} \int_0^1 e^{-2 \pi i k \xi(a,t)} \frac{d u}{d\xi} (\xi(a,t),t) \, d\xi \\
& = \frac{1}{2 \pi i k} \int_0^1 e^{-2 \pi i k \left[a + t u_0(a)\right]} \frac{d u_0}{da} (a) \, da,
\end{align*}
where we used \eqref{eq:Blag}, performed integration by parts with
respect to $\xi$ and employed the identity $\frac{du}{da}(\xi(a,t),t)
= \frac{du_0}{da}(a)$ which follows from \eqref{eq:Blag} together
with the chain rule. Next, defining $\eta(a,t) := \xi(a,t) - a = t
u_0(a)$, we have
\begin{align}
\widehat{u}_k(t) &= \frac{1}{2 \pi i k t} \int_0^1 \Dpartial{\eta}{a}(a,t)  e^{-2 \pi i k \left[a + \eta(a,t)\right]} \, da \nonumber \\
& = \frac{1}{2 \pi i k t} \int_0^1 \left( \frac{-1}{2 \pi i k} \right) \Dpartial{}{a} \left[
e^{-2 \pi i k  \eta(a,t) } \right] e^{-2 \pi i k a}  \, da \nonumber \\
& = \frac{-1}{2 \pi i k t} \int_0^1 e^{-2 \pi i k \left[ a +  t u_0(a) \right]} \, da, \quad k \neq 0, \quad 0 < t < t^*,
\label{eq:huk}
\end{align}
where we again performed integration by parts, this time with respect
to $a$. 

For a given initial condition $u_0$, relation \eqref{eq:huk} allows
for a straightforward evaluation of the Fourier coefficients of the
solution of \eqref{eq:Burgers0} by performing an integral. For
example, using the unimodal initial condition \eqref{eq:Bu0sin}, we
obtain
\begin{align}
\widehat{u}_k(t) &= \frac{-1}{2 \pi i k t} \int_0^1 e^{-2 \pi i k \left[ a +  t A \sin(2\pi a) \right]} \, da, \nonumber \\ 
& = \frac{-1}{2 \pi i k t} J_k(- kAt), \quad k \neq 0, \quad 0 < t < t^*,
\label{eq:huk0}
\end{align}
where we have used the integral representation of the Bessel functions
of the first kind $J_k(x) = \frac{1}{2 \pi} \int_{-\pi}^{\pi} e^{i [
  x \sin(\tau) - n \tau]}\, d\tau$ \citep{NIST:DLMF}. An expression for
the Fourier coefficients of the solution corresponding to the extreme
initial condition \eqref{eq:Bu0ext0} can be obtained in a similar
manner.

%\bibliographystyle{plain}
%%%%%%%%%%%%%%
% SWITCH BETWEEN THESE TWO BIBLIOGRAPHY COMMANDS
%\bibliographystyle{jfm}
%\bibliography{./Bib/stochastic,./Bib/allPROTAS,./Bib/maxdEdt_biblio,./Bib/control,./Bib/fundam}
%\bibliography{../Bib/stochastic,../Bib/allPROTAS,../Bib/maxdEdt_biblio,../Bib/control,../Bib/fundam}
%
%\bibliography{stochastic,allPROTAS,maxdEdt_biblio,control,fundam,3D_Phases_paper_JFM}
%%%%%%%%%%%%%%

\end{document}